\documentclass[11pt]{article} 

\usepackage{textcomp, amssymb}  
\usepackage{anysize}            
\usepackage{amsfonts}
\usepackage{amsmath}
\usepackage{dsfont}
\usepackage{MnSymbol}
\usepackage{mathtools}
\usepackage{graphicx}
\usepackage{epsfig}
\usepackage{epstopdf}
\usepackage{tikz}
\usepackage{amsthm}
\usepackage{ifpdf}

\usetikzlibrary{calc,decorations.pathmorphing,shapes,graphs}

\newcounter{sarrow}
\newcommand\xrsquigarrow[1]{%
\stepcounter{sarrow}%
\mathrel{\begin{tikzpicture}[baseline= {( $ (current bounding box.south) + (0,-0.5ex) $ )}]
\node[inner sep=.5ex] (\thesarrow) {$\scriptstyle #1$};
\path[draw,<-,decorate,
  decoration={zigzag,amplitude=0.7pt,segment length=1.2mm,pre=lineto,pre length=4pt}]
    (\thesarrow.south east) -- (\thesarrow.south west);
    \end{tikzpicture}}%
}

\newcounter{sarrow1}
\newcommand\xnrsquigarrow[1]{%
\stepcounter{sarrow1}%
\mathrel{\begin{tikzpicture}[baseline= {( $ (current bounding box.south) + (0,-0.5ex) $ )}]
\node[inner sep=.5ex] (\thesarrow) {$\scriptstyle #1$};
\path[draw,<-,decorate,
  decoration={zigzag,amplitude=0.7pt,segment length=1.2mm,pre=lineto,pre length=4pt}]
    (\thesarrow1.south east) -- (\thesarrow1.south west);
    $\slashedarrowfill@\relbar\relbar/$
    \end{tikzpicture}}%
}

\makeatletter
\def\slashedarrowfill@#1#2#3#4#5{%
  $\m@th\thickmuskip0mu\medmuskip\thickmuskip\thinmuskip\thickmuskip
   \relax#5#1\mkern-7mu%
   \cleaders\hbox{$#5\mkern-2mu#2\mkern-2mu$}\hfill
   \mathclap{#3}\mathclap{#2}%
   \cleaders\hbox{$#5\mkern-2mu#2\mkern-2mu$}\hfill
   \mkern-7mu#4$%
}
\def\rightslashedarrowfillb@{%
  \slashedarrowfill@\relbar\relbar/\rightarrow}
\newcommand\xnrightarrow[2][]{%
  \ext@arrow 0055{\rightslashedarrowfillb@}{#1}{#2}}

\def\rightslashedarrowfille@{%
  \slashedarrowfill@\relbar\relbar/\twoheadrightarrow}
\newcommand\xntworightarrow[2][]{%
  \ext@arrow 0055{\rightslashedarrowfille@}{#1}{#2}}

\def\rightslashedarrowfillg@{%
  \slashedarrowfill@\relbar\relbar{\raisebox{.12em}{}}\twoheadrightarrow}
\newcommand\xtworightarrow[2][]{%
  \ext@arrow 0055{\rightslashedarrowfillg@}{#1}{#2}}

\def\rightslashedarrowfillx@{%
  \slashedarrowfill@\Relbar\Relbar/\rightrightarrows}
\newcommand\xnTworightarrow[2][]{%
  \ext@arrow 0055{\rightslashedarrowfillx@}{#1}{#2}}

\def\rightslashedarrowfilly@{%
  \slashedarrowfill@\Relbar\Relbar{\raisebox{.12em}{}}\rightrightarrows}
\newcommand\xTworightarrow[2][]{%
  \ext@arrow 0055{\rightslashedarrowfilly@}{#1}{#2}}

\pgfdeclareshape{slash underlined}
{
  \inheritsavedanchors[from=rectangle] 
  \inheritanchorborder[from=rectangle]
  \inheritanchor[from=rectangle]{north}
  \inheritanchor[from=rectangle]{north west}
  \inheritanchor[from=rectangle]{north east}
  \inheritanchor[from=rectangle]{center}
  \inheritanchor[from=rectangle]{west}
  \inheritanchor[from=rectangle]{east}
  \inheritanchor[from=rectangle]{mid}
  \inheritanchor[from=rectangle]{mid west}
  \inheritanchor[from=rectangle]{mid east}
  \inheritanchor[from=rectangle]{base}
  \inheritanchor[from=rectangle]{base west}
  \inheritanchor[from=rectangle]{base east}
  \inheritanchor[from=rectangle]{south}
  \inheritanchor[from=rectangle]{south west}
  \inheritanchor[from=rectangle]{south east}
  \inheritanchorborder[from=rectangle]
  \foregroundpath{
    \southwest \pgf@xa=\pgf@x \pgf@ya=\pgf@y
    \northeast \pgf@xb=\pgf@x \pgf@yb=\pgf@y
    \pgf@xc=\pgf@xa
    \advance\pgf@xc by .5\pgf@xb
    \pgf@yc=\pgf@ya
    \advance\pgf@xc by -1.3pt
    \advance\pgf@yc by -1.8pt
    \pgfpathmoveto{\pgfqpoint{\pgf@xc}{\pgf@yc}}
    \advance\pgf@xc by  2.6pt
    \advance\pgf@yc by  3.6pt
    \pgfpathlineto{\pgfqpoint{\pgf@xc}{\pgf@yc}}
    \pgfpathmoveto{\pgfqpoint{\pgf@xa}{\pgf@ya}}
    \pgfpathlineto{\pgfqpoint{\pgf@xb}{\pgf@ya}}
 }
}
\tikzset{nomorepostaction/.code=\let\tikz@postactions\pgfutil@empty}

\newtheorem{theorem}{Theorem}[section]
\newtheorem{definition}[theorem]{Definition}
\newtheorem{proposition}[theorem]{Proposition}

\begin{document}

\begin{titlepage}
\thispagestyle{empty}

\hrule
\begin{center}
{\bf\LARGE Truly Concurrent Pi-Calculi with Reversibility, Probabilism and Guards}

\vspace{0.7cm}
--- Yong Wang ---

\vspace{2cm}
\begin{figure}[!htbp]
 \centering
 \includegraphics[width=1.0\textwidth]{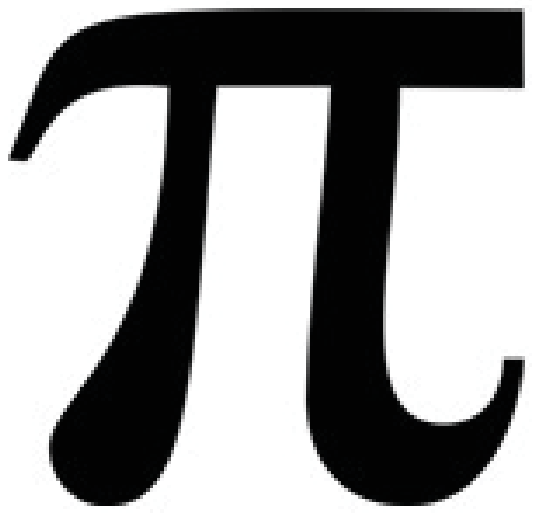}
\end{figure}

\end{center}
\end{titlepage}

\newpage 

\setcounter{page}{1}\pagenumbering{roman}

\tableofcontents

\newpage

\setcounter{page}{1}\pagenumbering{arabic}

        \section{Introduction}

The well-known process algebras, such as CCS \cite{CC} \cite{CCS}, ACP \cite{ACP} and $\pi$-calculus \cite{PI1} \cite{PI2}, capture the interleaving concurrency based on bisimilarity semantics.
We did some work on truly concurrent process algebras, such as CTC \cite{CTC}, APTC \cite{ATC} and $\pi_{tc}$ \cite{PITC}, capture the true concurrency based on truly concurrent bisimilarities, such as
pomset bisimilarity, step bisimilarity, history-preserving (hp-) bisimilarity and hereditary history-preserving (hhp-) bisimilarity. Truly concurrent process algebras are generalizations
of the corresponding traditional process algebras.

In this book, we introduce reversibility, probabilism, and guards into truly concurrent calculus $\pi_{tc}$, based on the work on $\pi$-calculus \cite{PI1} \cite{PI2}, probabilistic process algebra \cite{PPA}
\cite{PPA2} \cite{PPA3}, reversible process algebra \cite{APRTC}, and process algebra with guards \cite{HLPA}. We introduce the
preliminaries in chapter \ref{bg}. $\pi_{tc}$ with guards in chapter \ref{pitcg}, $\pi_{tc}$ with reversibility in chapter \ref{pitcr}, $\pi_{tc}$ with probabilism and reversibility in
chapter \ref{pitcpr}, $\pi_{tc}$ with probabilism and guards in chapter \ref{pitcpg}, $\pi_{tc}$ reversibility and guards in chapter \ref{pitcrg} and $\pi_{}tc$ with reversibility,
probabilism and guards all together in chapter \ref{pa}. For $\pi_{tc}$ with probabilistic, please refer to \cite{APPTC}.

\newpage\section{Backgrounds}\label{bg}

To make this book self-satisfied, we introduce some preliminaries in this chapter, including some introductions on operational semantics and $\pi_{tc}$ \cite{PITC}.

\subsection{Operational Semantics}\label{OS}

Firstly, in this subsection, we introduce concepts of (strongly) truly concurrent bisimilarities, including pomset bisimilarity, step bisimilarity, history-preserving (hp-)bisimilarity
and hereditary history-preserving (hhp-)bisimilarity. In contrast to traditional truly concurrent bisimilarities in CTC and APTC, these versions in $\pi_{tc}$ must take care of actions
with bound objects. Note that, these truly concurrent bisimilarities are defined as late bisimilarities, but not early bisimilarities, as defined in $\pi$-calculus \cite{PI1} \cite{PI2}.
Note that, here, a PES $\mathcal{E}$ is deemed as a process.

\begin{definition}[Bisimulation]
A bisimulation relation $R$ is a binary relation on processes such that: (1) if $p R q$ and $p\xrightarrow{a}p'$ then $q\xrightarrow{a}q'$ with $p' R q'$; (2) if $p R q$ and
$q\xrightarrow{a}q'$ then $p\xrightarrow{a}p'$ with $p' R q'$; (3) if $p R q$ and $pP$, then $qP$; (4) if $p R q$ and $qP$, then $pP$. Two processes $p$ and $q$ are bisimilar,
denoted by $p\sim_{HM} q$, if there is a bisimulation relation $R$ such that $p R q$.
\end{definition}

\begin{definition}[Congruence]
Let $\Sigma$ be a signature. An equivalence relation $R$ on $\mathcal{T}(\Sigma)$ is a congruence if for each $f\in\Sigma$, if $s_i R t_i$ for $i\in\{1,\cdots,ar(f)\}$, then
$f(s_1,\cdots,s_{ar(f)}) R f(t_1,\cdots,t_{ar(f)})$.
\end{definition}

\begin{definition}[Prime event structure with silent event]\label{PES}
Let $\Lambda$ be a fixed set of labels, ranged over $a,b,c,\cdots$ and $\tau$. A ($\Lambda$-labelled) prime event structure with silent event $\tau$ is a tuple
$\mathcal{E}=\langle \mathbb{E}, \leq, \sharp, \lambda\rangle$, where $\mathbb{E}$ is a denumerable set of events, including the silent event $\tau$. Let
$\hat{\mathbb{E}}=\mathbb{E}\backslash\{\tau\}$, exactly excluding $\tau$, it is obvious that $\hat{\tau^*}=\epsilon$, where $\epsilon$ is the empty event.
Let $\lambda:\mathbb{E}\rightarrow\Lambda$ be a labelling function and let $\lambda(\tau)=\tau$. And $\leq$, $\sharp$ are binary relations on $\mathbb{E}$,
called causality and conflict respectively, such that:

\begin{enumerate}
  \item $\leq$ is a partial order and $\lceil e \rceil = \{e'\in \mathbb{E}|e'\leq e\}$ is finite for all $e\in \mathbb{E}$. It is easy to see that
  $e\leq\tau^*\leq e'=e\leq\tau\leq\cdots\leq\tau\leq e'$, then $e\leq e'$.
  \item $\sharp$ is irreflexive, symmetric and hereditary with respect to $\leq$, that is, for all $e,e',e''\in \mathbb{E}$, if $e\sharp e'\leq e''$, then $e\sharp e''$.
\end{enumerate}

Then, the concepts of consistency and concurrency can be drawn from the above definition:

\begin{enumerate}
  \item $e,e'\in \mathbb{E}$ are consistent, denoted as $e\frown e'$, if $\neg(e\sharp e')$. A subset $X\subseteq \mathbb{E}$ is called consistent, if $e\frown e'$ for all
  $e,e'\in X$.
  \item $e,e'\in \mathbb{E}$ are concurrent, denoted as $e\parallel e'$, if $\neg(e\leq e')$, $\neg(e'\leq e)$, and $\neg(e\sharp e')$.
\end{enumerate}
\end{definition}

\begin{definition}[Configuration]
Let $\mathcal{E}$ be a PES. A (finite) configuration in $\mathcal{E}$ is a (finite) consistent subset of events $C\subseteq \mathcal{E}$, closed with respect to causality
(i.e. $\lceil C\rceil=C$). The set of finite configurations of $\mathcal{E}$ is denoted by $\mathcal{C}(\mathcal{E})$. We let $\hat{C}=C\backslash\{\tau\}$.
\end{definition}

A consistent subset of $X\subseteq \mathbb{E}$ of events can be seen as a pomset. Given $X, Y\subseteq \mathbb{E}$, $\hat{X}\sim \hat{Y}$ if $\hat{X}$ and $\hat{Y}$ are
isomorphic as pomsets. In the following of the paper, we say $C_1\sim C_2$, we mean $\hat{C_1}\sim\hat{C_2}$.

\begin{definition}[Pomset transitions and step]
Let $\mathcal{E}$ be a PES and let $C\in\mathcal{C}(\mathcal{E})$, and $\emptyset\neq X\subseteq \mathbb{E}$, if $C\cap X=\emptyset$ and $C'=C\cup X\in\mathcal{C}(\mathcal{E})$, then $C\xrightarrow{X} C'$ is called a pomset transition from $C$ to $C'$. When the events in $X$ are pairwise concurrent, we say that $C\xrightarrow{X}C'$ is a step.
\end{definition}

\begin{definition}[Strong pomset, step bisimilarity]\label{PSB5}
Let $\mathcal{E}_1$, $\mathcal{E}_2$ be PESs. A strong pomset bisimulation is a relation $R\subseteq\mathcal{C}(\mathcal{E}_1)\times\mathcal{C}(\mathcal{E}_2)$, such that if $(C_1,C_2)\in R$, and $C_1\xrightarrow{X_1}C_1'$ (with $\mathcal{E}_1\xrightarrow{X_1}\mathcal{E}_1'$) then $C_2\xrightarrow{X_2}C_2'$ (with $\mathcal{E}_2\xrightarrow{X_2}\mathcal{E}_2'$), with $X_1\subseteq \mathbb{E}_1$, $X_2\subseteq \mathbb{E}_2$, $X_1\sim X_2$ and $(C_1',C_2')\in R$:
\begin{enumerate}
  \item for each fresh action $\alpha\in X_1$, if $C_1''\xrightarrow{\alpha}C_1'''$ (with $\mathcal{E}_1''\xrightarrow{\alpha}\mathcal{E}_1'''$), then for some $C_2''$ and $C_2'''$, $C_2''\xrightarrow{\alpha}C_2'''$ (with $\mathcal{E}_2''\xrightarrow{\alpha}\mathcal{E}_2'''$), such that if $(C_1'',C_2'')\in R$ then $(C_1''',C_2''')\in R$;
  \item for each $x(y)\in X_1$ with ($y\notin n(\mathcal{E}_1, \mathcal{E}_2)$), if $C_1''\xrightarrow{x(y)}C_1'''$ (with $\mathcal{E}_1''\xrightarrow{x(y)}\mathcal{E}_1'''\{w/y\}$) for all $w$, then for some $C_2''$ and $C_2'''$, $C_2''\xrightarrow{x(y)}C_2'''$ (with $\mathcal{E}_2''\xrightarrow{x(y)}\mathcal{E}_2'''\{w/y\}$) for all $w$, such that if $(C_1'',C_2'')\in R$ then $(C_1''',C_2''')\in R$;
  \item for each two $x_1(y),x_2(y)\in X_1$ with ($y\notin n(\mathcal{E}_1, \mathcal{E}_2)$), if $C_1''\xrightarrow{\{x_1(y),x_2(y)\}}C_1'''$ (with $\mathcal{E}_1''\xrightarrow{\{x_1(y),x_2(y)\}}\mathcal{E}_1'''\{w/y\}$) for all $w$, then for some $C_2''$ and $C_2'''$, $C_2''\xrightarrow{\{x_1(y),x_2(y)\}}C_2'''$ (with $\mathcal{E}_2''\xrightarrow{\{x_1(y),x_2(y)\}}\mathcal{E}_2'''\{w/y\}$) for all $w$, such that if $(C_1'',C_2'')\in R$ then $(C_1''',C_2''')\in R$;
  \item for each $\overline{x}(y)\in X_1$ with $y\notin n(\mathcal{E}_1, \mathcal{E}_2)$, if $C_1''\xrightarrow{\overline{x}(y)}C_1'''$ (with $\mathcal{E}_1''\xrightarrow{\overline{x}(y)}\mathcal{E}_1'''$), then for some $C_2''$ and $C_2'''$, $C_2''\xrightarrow{\overline{x}(y)}C_2'''$ (with $\mathcal{E}_2''\xrightarrow{\overline{x}(y)}\mathcal{E}_2'''$), such that if $(C_1'',C_2'')\in R$ then $(C_1''',C_2''')\in R$.
\end{enumerate}
 and vice-versa.

We say that $\mathcal{E}_1$, $\mathcal{E}_2$ are strong pomset bisimilar, written $\mathcal{E}_1\sim_p\mathcal{E}_2$, if there exists a strong pomset bisimulation $R$, such that $(\emptyset,\emptyset)\in R$. By replacing pomset transitions with steps, we can get the definition of strong step bisimulation. When PESs $\mathcal{E}_1$ and $\mathcal{E}_2$ are strong step bisimilar, we write $\mathcal{E}_1\sim_s\mathcal{E}_2$.
\end{definition}

\begin{definition}[Posetal product]
Given two PESs $\mathcal{E}_1$, $\mathcal{E}_2$, the posetal product of their configurations, denoted $\mathcal{C}(\mathcal{E}_1)\overline{\times}\mathcal{C}(\mathcal{E}_2)$, is defined as

$$\{(C_1,f,C_2)|C_1\in\mathcal{C}(\mathcal{E}_1),C_2\in\mathcal{C}(\mathcal{E}_2),f:C_1\rightarrow C_2 \textrm{ isomorphism}\}.$$

A subset $R\subseteq\mathcal{C}(\mathcal{E}_1)\overline{\times}\mathcal{C}(\mathcal{E}_2)$ is called a posetal relation. We say that $R$ is downward closed when for any $(C_1,f,C_2),(C_1',f',C_2')\in \mathcal{C}(\mathcal{E}_1)\overline{\times}\mathcal{C}(\mathcal{E}_2)$, if $(C_1,f,C_2)\subseteq (C_1',f',C_2')$ pointwise and $(C_1',f',C_2')\in R$, then $(C_1,f,C_2)\in R$.

For $f:X_1\rightarrow X_2$, we define $f[x_1\mapsto x_2]:X_1\cup\{x_1\}\rightarrow X_2\cup\{x_2\}$, $z\in X_1\cup\{x_1\}$,(1)$f[x_1\mapsto x_2](z)=
x_2$,if $z=x_1$;(2)$f[x_1\mapsto x_2](z)=f(z)$, otherwise. Where $X_1\subseteq \mathbb{E}_1$, $X_2\subseteq \mathbb{E}_2$, $x_1\in \mathbb{E}_1$, $x_2\in \mathbb{E}_2$.
\end{definition}

\begin{definition}[Strong (hereditary) history-preserving bisimilarity]\label{HHPB5}
A strong history-preserving (hp-) bisimulation is a posetal relation $R\subseteq\mathcal{C}(\mathcal{E}_1)\overline{\times}\mathcal{C}(\mathcal{E}_2)$ such that if $(C_1,f,C_2)\in R$, and
\begin{enumerate}
  \item for $e_1=\alpha$ a fresh action, if $C_1\xrightarrow{\alpha}C_1'$ (with $\mathcal{E}_1\xrightarrow{\alpha}\mathcal{E}_1'$), then for some $C_2'$ and $e_2=\alpha$, $C_2\xrightarrow{\alpha}C_2'$ (with $\mathcal{E}_2\xrightarrow{\alpha}\mathcal{E}_2'$), such that $(C_1',f[e_1\mapsto e_2],C_2')\in R$;
  \item for $e_1=x(y)$ with ($y\notin n(\mathcal{E}_1, \mathcal{E}_2)$), if $C_1\xrightarrow{x(y)}C_1'$ (with $\mathcal{E}_1\xrightarrow{x(y)}\mathcal{E}_1'\{w/y\}$) for all $w$, then for some $C_2'$ and $e_2=x(y)$, $C_2\xrightarrow{x(y)}C_2'$ (with $\mathcal{E}_2\xrightarrow{x(y)}\mathcal{E}_2'\{w/y\}$) for all $w$, such that $(C_1',f[e_1\mapsto e_2],C_2')\in R$;
  \item for $e_1=\overline{x}(y)$ with $y\notin n(\mathcal{E}_1, \mathcal{E}_2)$, if $C_1\xrightarrow{\overline{x}(y)}C_1'$ (with $\mathcal{E}_1\xrightarrow{\overline{x}(y)}\mathcal{E}_1'$), then for some $C_2'$ and $e_2=\overline{x}(y)$, $C_2\xrightarrow{\overline{x}(y)}C_2'$ (with $\mathcal{E}_2\xrightarrow{\overline{x}(y)}\mathcal{E}_2'$), such that $(C_1',f[e_1\mapsto e_2],C_2')\in R$.
\end{enumerate}

and vice-versa. $\mathcal{E}_1,\mathcal{E}_2$ are strong history-preserving (hp-)bisimilar and are written $\mathcal{E}_1\sim_{hp}\mathcal{E}_2$ if there exists a strong hp-bisimulation $R$ such that $(\emptyset,\emptyset,\emptyset)\in R$.

A strongly hereditary history-preserving (hhp-)bisimulation is a downward closed strong hp-bisimulation. $\mathcal{E}_1,\mathcal{E}_2$ are strongly hereditary history-preserving (hhp-)bisimilar and are written $\mathcal{E}_1\sim_{hhp}\mathcal{E}_2$.
\end{definition}

\subsection{$\pi_{tc}$}

$\pi_{tc}$ \cite{PITC} is a calculus of truly concurrent mobile processes. It includes syntax and semantics:

\begin{enumerate}
  \item Its syntax includes actions, process constant, and operators acting between actions, like Prefix, Summation, Composition, Restriction, Input and Output.
  \item Its semantics is based on labeled transition systems, Prefix, Summation, Composition, Restriction, Input and Output have their transition rules. $\pi_{tc}$ has good semantic properties
  based on the truly concurrent bisimulations. These properties include summation laws, identity laws, restriction laws, parallel laws, expansion laws, congruences, and also unique solution for recursion.
\end{enumerate}

\newpage\section{$\pi_{tc}$ with Guards}\label{pitcg}

In this chapter, we design $\pi_{tc}$ with guards. This chapter is organized as follows. In section \ref{os3}, we introduce the truly concurrent operational semantics. Then, we introduce
the syntax and operational semantics, laws modulo strongly truly concurrent bisimulations, and algebraic theory of $\pi_{tc}$ with guards in section \ref{sos3},
\ref{s3} and \ref{a3} respectively.

\subsection{Operational Semantics}\label{os3}

Firstly, in this section, we introduce concepts of (strongly) truly concurrent bisimilarities, including pomset bisimilarity, step
bisimilarity, history-preserving (hp-)bisimilarity and hereditary history-preserving (hhp-)bisimilarity. In contrast to traditional truly
concurrent bisimilarities in section \ref{bg}, these versions in $\pi_{ptc}$ must take care of actions with bound objects. Note that, these truly concurrent bisimilarities
are defined as late bisimilarities, but not early bisimilarities, as defined in $\pi$-calculus \cite{PI1} \cite{PI2}. Note that, here, a PES $\mathcal{E}$ is deemed as a process.

\begin{definition}[Prime event structure with silent event and empty event]
Let $\Lambda$ be a fixed set of labels, ranged over $a,b,c,\cdots$ and $\tau,\epsilon$. A ($\Lambda$-labelled) prime event structure with silent event $\tau$ and empty event
$\epsilon$ is a tuple $\mathcal{E}=\langle \mathbb{E}, \leq, \sharp, \lambda\rangle$, where $\mathbb{E}$ is a denumerable set of events, including the silent event $\tau$ and empty
event $\epsilon$. Let $\hat{\mathbb{E}}=\mathbb{E}\backslash\{\tau,\epsilon\}$, exactly excluding $\tau$ and $\epsilon$, it is obvious that $\hat{\tau^*}=\epsilon$. Let
$\lambda:\mathbb{E}\rightarrow\Lambda$ be a labelling function and let $\lambda(\tau)=\tau$ and $\lambda(\epsilon)=\epsilon$. And $\leq$, $\sharp$ are binary relations on $\mathbb{E}$,
called causality and conflict respectively, such that:

\begin{enumerate}
  \item $\leq$ is a partial order and $\lceil e \rceil = \{e'\in \mathbb{E}|e'\leq e\}$ is finite for all $e\in \mathbb{E}$. It is easy to see that
  $e\leq\tau^*\leq e'=e\leq\tau\leq\cdots\leq\tau\leq e'$, then $e\leq e'$.
  \item $\sharp$ is irreflexive, symmetric and hereditary with respect to $\leq$, that is, for all $e,e',e''\in \mathbb{E}$, if $e\sharp e'\leq e''$, then $e\sharp e''$.
\end{enumerate}

Then, the concepts of consistency and concurrency can be drawn from the above definition:

\begin{enumerate}
  \item $e,e'\in \mathbb{E}$ are consistent, denoted as $e\frown e'$, if $\neg(e\sharp e')$. A subset $X\subseteq \mathbb{E}$ is called consistent, if $e\frown e'$ for all
  $e,e'\in X$.
  \item $e,e'\in \mathbb{E}$ are concurrent, denoted as $e\parallel e'$, if $\neg(e\leq e')$, $\neg(e'\leq e)$, and $\neg(e\sharp e')$.
\end{enumerate}
\end{definition}

\begin{definition}[Configuration]
Let $\mathcal{E}$ be a PES. A (finite) configuration in $\mathcal{E}$ is a (finite) consistent subset of events $C\subseteq \mathcal{E}$, closed with respect to causality (i.e.
$\lceil C\rceil=C$), and a data state $s\in S$ with $S$ the set of all data states, denoted $\langle C, s\rangle$. The set of finite configurations of $\mathcal{E}$ is denoted by
$\langle\mathcal{C}(\mathcal{E}), S\rangle$. We let $\hat{C}=C\backslash\{\tau\}\cup\{\epsilon\}$.
\end{definition}

A consistent subset of $X\subseteq \mathbb{E}$ of events can be seen as a pomset. Given $X, Y\subseteq \mathbb{E}$, $\hat{X}\sim \hat{Y}$ if $\hat{X}$ and $\hat{Y}$ are isomorphic as
pomsets. In the following of the paper, we say $C_1\sim C_2$, we mean $\hat{C_1}\sim\hat{C_2}$.

\begin{definition}[Strongly pomset, step bisimilarity]
Let $\mathcal{E}_1$, $\mathcal{E}_2$ be PESs. A strongly pomset bisimulation is a relation $R\subseteq\langle\mathcal{C}(\mathcal{E}_1),s\rangle\times\langle\mathcal{C}(\mathcal{E}_2),s\rangle$,
such that if $(\langle C_1,s\rangle,\langle C_2,s\rangle)\in R$, and $\langle C_1,s\rangle\xrightarrow{X_1}\langle C_1',s'\rangle$ (with $\mathcal{E}_1\xrightarrow{X_1}\mathcal{E}_1'$) then $\langle C_2,s\rangle\xrightarrow{X_2}\langle C_2',s'\rangle$ (with
$\mathcal{E}_2\xrightarrow{X_2}\mathcal{E}_2'$), with $X_1\subseteq \mathbb{E}_1$, $X_2\subseteq \mathbb{E}_2$, $X_1\sim X_2$ and $(\langle C_1',s'\rangle,\langle C_2',s'\rangle)\in R$:
\begin{enumerate}
  \item for each fresh action $\alpha\in X_1$, if $\langle C_1'',s''\rangle\xrightarrow{\alpha}\langle C_1''',s'''\rangle$ (with $\mathcal{E}_1''\xrightarrow{\alpha}\mathcal{E}_1'''$),
  then for some $C_2''$ and $\langle C_2''',s'''\rangle$, $\langle C_2'',s''\rangle\xrightarrow{\alpha}\langle C_2''',s'''\rangle$ (with
  $\mathcal{E}_2''\xrightarrow{\alpha}\mathcal{E}_2'''$), such that if $(\langle C_1'',s''\rangle,\langle C_2'',s''\rangle)\in R$ then $(\langle C_1''',s'''\rangle,\langle C_2''',s'''\rangle)\in R$;
  \item for each $x(y)\in X_1$ with ($y\notin n(\mathcal{E}_1, \mathcal{E}_2)$), if $\langle C_1'',s''\rangle\xrightarrow{x(y)}\langle C_1''',s'''\rangle$ (with
  $\mathcal{E}_1''\xrightarrow{x(y)}\mathcal{E}_1'''\{w/y\}$) for all $w$, then for some $C_2''$ and $C_2'''$, $\langle C_2'',s''\rangle\xrightarrow{x(y)}\langle C_2''',s'''\rangle$
  (with $\mathcal{E}_2''\xrightarrow{x(y)}\mathcal{E}_2'''\{w/y\}$) for all $w$, such that if $(\langle C_1'',s''\rangle,\langle C_2'',s''\rangle)\in R$ then $(\langle C_1''',s'''\rangle,\langle C_2''',s'''\rangle)\in R$;
  \item for each two $x_1(y),x_2(y)\in X_1$ with ($y\notin n(\mathcal{E}_1, \mathcal{E}_2)$), if $\langle C_1'',s''\rangle\xrightarrow{\{x_1(y),x_2(y)\}}\langle C_1''',s'''\rangle$
  (with $\mathcal{E}_1''\xrightarrow{\{x_1(y),x_2(y)\}}\mathcal{E}_1'''\{w/y\}$) for all $w$, then for some $C_2''$ and $C_2'''$,
  $\langle C_2'',s''\rangle\xrightarrow{\{x_1(y),x_2(y)\}}\langle C_2''',s'''\rangle$ (with $\mathcal{E}_2''\xrightarrow{\{x_1(y),x_2(y)\}}\mathcal{E}_2'''\{w/y\}$) for all $w$, such
  that if $(\langle C_1'',s''\rangle,\langle C_2'',s''\rangle)\in R$ then $(\langle C_1''',s'''\rangle,\langle C_2''',s'''\rangle)\in R$;
  \item for each $\overline{x}(y)\in X_1$ with $y\notin n(\mathcal{E}_1, \mathcal{E}_2)$, if $\langle C_1'',s''\rangle\xrightarrow{\overline{x}(y)}\langle C_1''',s'''\rangle$
  (with $\mathcal{E}_1''\xrightarrow{\overline{x}(y)}\mathcal{E}_1'''$), then for some $C_2''$ and $C_2'''$, $\langle C_2'',s''\rangle\xrightarrow{\overline{x}(y)}\langle C_2''',s'''\rangle$
  (with $\mathcal{E}_2''\xrightarrow{\overline{x}(y)}\mathcal{E}_2'''$), such that if $(\langle C_1'',s''\rangle,\langle C_2'',s''\rangle)\in R$ then $(\langle C_1''',s'''\rangle,\langle C_2''',s'''\rangle)\in R$.
\end{enumerate}
 and vice-versa.
 
We say that $\mathcal{E}_1$, $\mathcal{E}_2$ are strongly pomset bisimilar, written $\mathcal{E}_1\sim_p\mathcal{E}_2$, if there exists a strongly pomset
bisimulation $R$, such that $(\emptyset,\emptyset)\in R$. By replacing pomset transitions with steps, we can get the definition of strongly step bisimulation.
When PESs $\mathcal{E}_1$ and $\mathcal{E}_2$ are strongly step bisimilar, we write $\mathcal{E}_1\sim_s\mathcal{E}_2$.
\end{definition}

\begin{definition}[Posetal product]
Given two PESs $\mathcal{E}_1$, $\mathcal{E}_2$, the posetal product of their configurations, denoted
$\langle\mathcal{C}(\mathcal{E}_1),S\rangle\overline{\times}\langle\mathcal{C}(\mathcal{E}_2),S\rangle$, is defined as

$$\{(\langle C_1,s\rangle,f,\langle C_2,s\rangle)|C_1\in\mathcal{C}(\mathcal{E}_1),C_2\in\mathcal{C}(\mathcal{E}_2),f:C_1\rightarrow C_2 \textrm{ isomorphism}\}.$$

A subset $R\subseteq\langle\mathcal{C}(\mathcal{E}_1),S\rangle\overline{\times}\langle\mathcal{C}(\mathcal{E}_2),S\rangle$ is called a posetal relation. We say that $R$ is downward
closed when for any
$(\langle C_1,s\rangle,f,\langle C_2,s\rangle),(\langle C_1',s'\rangle,f',\langle C_2',s'\rangle)\in \langle\mathcal{C}(\mathcal{E}_1),S\rangle\overline{\times}\langle\mathcal{C}(\mathcal{E}_2),S\rangle$,
if $(\langle C_1,s\rangle,f,\langle C_2,s\rangle)\subseteq (\langle C_1',s'\rangle,f',\langle C_2',s'\rangle)$ pointwise and $(\langle C_1',s'\rangle,f',\langle C_2',s'\rangle)\in R$,
then $(\langle C_1,s\rangle,f,\langle C_2,s\rangle)\in R$.

For $f:X_1\rightarrow X_2$, we define $f[x_1\mapsto x_2]:X_1\cup\{x_1\}\rightarrow X_2\cup\{x_2\}$, $z\in X_1\cup\{x_1\}$,(1)$f[x_1\mapsto x_2](z)=
x_2$,if $z=x_1$;(2)$f[x_1\mapsto x_2](z)=f(z)$, otherwise. Where $X_1\subseteq \mathbb{E}_1$, $X_2\subseteq \mathbb{E}_2$, $x_1\in \mathbb{E}_1$, $x_2\in \mathbb{E}_2$.
\end{definition}

\begin{definition}[Strongly (hereditary) history-preserving bisimilarity]
A strongly history-preserving (hp-) bisimulation is a posetal relation $R\subseteq\mathcal{C}(\mathcal{E}_1)\overline{\times}\mathcal{C}(\mathcal{E}_2)$ such that
if $(\langle C_1,s\rangle,f,\langle C_2,s\rangle)\in R$, and
\begin{enumerate}
  \item for $e_1=\alpha$ a fresh action, if $\langle C_1,s\rangle\xrightarrow{\alpha}\langle C_1',s'\rangle$ (with $\mathcal{E}_1\xrightarrow{\alpha}\mathcal{E}_1'$), then for some
  $C_2'$ and $e_2=\alpha$, $\langle C_2,s\rangle\xrightarrow{\alpha}\langle C_2',s'\rangle$ (with $\mathcal{E}_2\xrightarrow{\alpha}\mathcal{E}_2'$), such that
  $(\langle C_1',s'\rangle,f[e_1\mapsto e_2],\langle C_2',s'\rangle)\in R$;
  \item for $e_1=x(y)$ with ($y\notin n(\mathcal{E}_1, \mathcal{E}_2)$), if $\langle C_1,s\rangle\xrightarrow{x(y)}\langle C_1',s'\rangle$ (with
  $\mathcal{E}_1\xrightarrow{x(y)}\mathcal{E}_1'\{w/y\}$) for all $w$, then for some $C_2'$ and $e_2=x(y)$, $\langle C_2,s\rangle\xrightarrow{x(y)}\langle C_2',s'\rangle$ (with
  $\mathcal{E}_2\xrightarrow{x(y)}\mathcal{E}_2'\{w/y\}$) for all $w$, such that $(\langle C_1',s'\rangle,f[e_1\mapsto e_2],\langle C_2',s'\rangle)\in R$;
  \item for $e_1=\overline{x}(y)$ with $y\notin n(\mathcal{E}_1, \mathcal{E}_2)$, if $\langle C_1,s\rangle\xrightarrow{\overline{x}(y)}\langle C_1',s'\rangle$ (with
  $\mathcal{E}_1\xrightarrow{\overline{x}(y)}\mathcal{E}_1'$), then for some $C_2'$ and $e_2=\overline{x}(y)$, $\langle C_2,s\rangle\xrightarrow{\overline{x}(y)}\langle C_2',s'\rangle$
  (with $\mathcal{E}_2\xrightarrow{\overline{x}(y)}\mathcal{E}_2'$), such that $(\langle C_1',s'\rangle,f[e_1\mapsto e_2],\langle C_2',s'\rangle)\in R$.
\end{enumerate}
and vice-versa. $\mathcal{E}_1,\mathcal{E}_2$
are strongly history-preserving (hp-)bisimilar and are written $\mathcal{E}_1\sim_{hp}\mathcal{E}_2$ if there exists a strongly hp-bisimulation
$R$ such that $(\emptyset,\emptyset,\emptyset)\in R$.

A strongly hereditary history-preserving (hhp-)bisimulation is a downward closed strongly hp-bisimulation. $\mathcal{E}_1,\mathcal{E}_2$ are FR
strongly hereditary history-preserving (hhp-)bisimilar and are written $\mathcal{E}_1\sim_{hhp}\mathcal{E}_2$.
\end{definition}

\subsection{Syntax and Operational Semantics}\label{sos3}

We assume an infinite set $\mathcal{N}$ of (action or event) names, and use $a,b,c,\cdots$ to range over $\mathcal{N}$, use $x,y,z,w,u,v$ as meta-variables over names. We denote by
$\overline{\mathcal{N}}$ the set of co-names and let $\overline{a},\overline{b},\overline{c},\cdots$ range over $\overline{\mathcal{N}}$. Then we set
$\mathcal{L}=\mathcal{N}\cup\overline{\mathcal{N}}$ as the set of labels, and use $l,\overline{l}$ to range over $\mathcal{L}$. We extend complementation to $\mathcal{L}$ such that
$\overline{\overline{a}}=a$. Let $\tau$ denote the silent step (internal action or event) and define $Act=\mathcal{L}\cup\{\tau\}$ to be the set of actions, $\alpha,\beta$ range over
$Act$. And $K,L$ are used to stand for subsets of $\mathcal{L}$ and $\overline{L}$ is used for the set of complements of labels in $L$.

Further, we introduce a set $\mathcal{X}$ of process variables, and a set $\mathcal{K}$ of process constants, and let $X,Y,\cdots$ range over $\mathcal{X}$, and $A,B,\cdots$ range over
$\mathcal{K}$. For each process constant $A$, a nonnegative arity $ar(A)$ is assigned to it. Let $\widetilde{x}=x_1,\cdots,x_{ar(A)}$ be a tuple of distinct name variables, then
$A(\widetilde{x})$ is called a process constant. $\widetilde{X}$ is a tuple of distinct process variables, and also $E,F,\cdots$ range over the recursive expressions. We write
$\mathcal{P}$ for the set of processes. Sometimes, we use $I,J$ to stand for an indexing set, and we write $E_i:i\in I$ for a family of expressions indexed by $I$. $Id_D$ is the
identity function or relation over set $D$. The symbol $\equiv_{\alpha}$ denotes equality under standard alpha-convertibility, note that the subscript $\alpha$ has no relation to the
action $\alpha$.

Let $G_{at}$ be the set of atomic guards, $\delta$ be the deadlock constant, and $\epsilon$ be the empty action, and extend $Act$ to $Act\cup\{\epsilon\}\cup\{\delta\}$. We extend
$G_{at}$ to the set of basic guards $G$ with element $\phi,\psi,\cdots$, which is generated by the following formation rules:

$$\phi::=\delta|\epsilon|\neg\phi|\psi\in G_{at}|\phi+\psi|\phi\cdot\psi$$

The predicate $test(\phi,s)$ represents that $\phi$ holds in the state $s$, and $test(\epsilon,s)$ holds and $test(\delta,s)$ does not hold. $effect(e,s)\in S$ denotes $s'$ in
$s\xrightarrow{e}s'$. The predicate weakest precondition $wp(e,\phi)$ denotes that $\forall s,s'\in S, test(\phi,effect(e,s))$ holds.

\subsubsection{Syntax}

We use the Prefix $.$ to model the causality relation $\leq$ in true concurrency, the Summation $+$ to model the conflict relation $\sharp$ in true concurrency, and the Composition $\parallel$ to explicitly model concurrent relation in true concurrency. And we follow the
conventions of process algebra.

\begin{definition}[Syntax]\label{syntax3}
A truly concurrent process $\pi_{tc}$ with guards is defined inductively by the following formation rules:

\begin{enumerate}
  \item $A(\widetilde{x})\in\mathcal{P}$;
  \item $\phi\in\mathcal{P}$;
  \item $\textbf{nil}\in\mathcal{P}$;
  \item if $P\in\mathcal{P}$, then the Prefix $\tau.P\in\mathcal{P}$, for $\tau\in Act$ is the silent action;
  \item if $P\in\mathcal{P}$, then the Prefix $\phi.P\in\mathcal{P}$, for $\phi\in G_{at}$;
  \item if $P\in\mathcal{P}$, then the Output $\overline{x}y.P\in\mathcal{P}$, for $x,y\in Act$;
  \item if $P\in\mathcal{P}$, then the Input $x(y).P\in\mathcal{P}$, for $x,y\in Act$;
  \item if $P\in\mathcal{P}$, then the Restriction $(x)P\in\mathcal{P}$, for $x\in Act$;
  \item if $P,Q\in\mathcal{P}$, then the Summation $P+Q\in\mathcal{P}$;
  \item if $P,Q\in\mathcal{P}$, then the Composition $P\parallel Q\in\mathcal{P}$;
\end{enumerate}

The standard BNF grammar of syntax of $\pi_{tc}$ with guards can be summarized as follows:

$$P::=A(\widetilde{x})|\textbf{nil}|\tau.P| \overline{x}y.P | x(y).P| (x)P  |\phi.P|  P+P| P\parallel P.$$
\end{definition}

In $\overline{x}y$, $x(y)$ and $\overline{x}(y)$, $x$ is called the subject, $y$ is called the object and it may be free or bound.

\begin{definition}[Free variables]
The free names of a process $P$, $fn(P)$, are defined as follows.

\begin{enumerate}
  \item $fn(A(\widetilde{x}))\subseteq\{\widetilde{x}\}$;
  \item $fn(\textbf{nil})=\emptyset$;
  \item $fn(\tau.P)=fn(P)$;
  \item $fn(\phi.P)=fn(P)$;
  \item $fn(\overline{x}y.P)=fn(P)\cup\{x\}\cup\{y\}$;
  \item $fn(x(y).P)=fn(P)\cup\{x\}-\{y\}$;
  \item $fn((x)P)=fn(P)-\{x\}$;
  \item $fn(P+Q)=fn(P)\cup fn(Q)$;
  \item $fn(P\parallel Q)=fn(P)\cup fn(Q)$.
\end{enumerate}
\end{definition}

\begin{definition}[Bound variables]
Let $n(P)$ be the names of a process $P$, then the bound names $bn(P)=n(P)-fn(P)$.
\end{definition}

For each process constant schema $A(\widetilde{x})$, a defining equation of the form

$$A(\widetilde{x})\overset{\text{def}}{=}P$$

is assumed, where $P$ is a process with $fn(P)\subseteq \{\widetilde{x}\}$.

\begin{definition}[Substitutions]\label{subs3}
A substitution is a function $\sigma:\mathcal{N}\rightarrow\mathcal{N}$. For $x_i\sigma=y_i$ with $1\leq i\leq n$, we write $\{y_1/x_1,\cdots,y_n/x_n\}$ or
$\{\widetilde{y}/\widetilde{x}\}$ for $\sigma$. For a process $P\in\mathcal{P}$, $P\sigma$ is defined inductively as follows:
\begin{enumerate}
  \item if $P$ is a process constant $A(\widetilde{x})=A(x_1,\cdots,x_n)$, then $P\sigma=A(x_1\sigma,\cdots,x_n\sigma)$;
  \item if $P=\textbf{nil}$, then $P\sigma=\textbf{nil}$;
  \item if $P=\tau.P'$, then $P\sigma=\tau.P'\sigma$;
  \item if $P=\phi.P'$, then $P\sigma=\phi.P'\sigma$;
  \item if $P=\overline{x}y.P'$, then $P\sigma=\overline{x\sigma}y\sigma.P'\sigma$;
  \item if $P=x(y).P'$, then $P\sigma=x\sigma(y).P'\sigma$;
  \item if $P=(x)P'$, then $P\sigma=(x\sigma)P'\sigma$;
  \item if $P=P_1+P_2$, then $P\sigma=P_1\sigma+P_2\sigma$;
  \item if $P=P_1\parallel P_2$, then $P\sigma=P_1\sigma \parallel P_2\sigma$.
\end{enumerate}
\end{definition}

\subsubsection{Operational Semantics}

The operational semantics is defined by LTSs (labelled transition systems), and it is detailed by the following definition.

\begin{definition}[Semantics]\label{semantics3}
The operational semantics of $\pi_{tc}$ with guards corresponding to the syntax in Definition \ref{syntax3} is defined by a series of transition rules, named $\textbf{ACT}$, $\textbf{SUM}$,
$\textbf{IDE}$, $\textbf{PAR}$, $\textbf{COM}$, $\textbf{CLOSE}$, $\textbf{RES}$, $\textbf{OPEN}$ indicate that the rules are associated respectively with Prefix, Summation,
Identity, Parallel Composition, Communication, and Restriction in Definition \ref{syntax3}. They are shown in \ref{TRForPITC3}.

\begin{center}
    \begin{table}
        \[\textbf{TAU-ACT}\quad \frac{}{\langle\tau.P,s\rangle\xrightarrow{\tau}\langle P,\tau(s)\rangle}\]

        \[\textbf{OUTPUT-ACT}\quad \frac{}{\langle\overline{x}y.P,s\rangle\xrightarrow{\overline{x}y}\langle P,s'\rangle}\]

        \[\textbf{INPUT-ACT}\quad \frac{}{\langle x(z).P,s\rangle\xrightarrow{x(w)}\langle P\{w/z\},s'\rangle}\quad (w\notin fn((z)P))\]

        \[\textbf{PAR}_1\quad \frac{\langle P,s\rangle\xrightarrow{\alpha}\langle P',s'\rangle\quad \langle Q,s\rangle\nrightarrow}{\langle P\parallel Q,s\rangle\xrightarrow{\alpha}\langle P'\parallel Q,s'\rangle}\quad (bn(\alpha)\cap fn(Q)=\emptyset)\]

        \[\textbf{PAR}_2\quad \frac{\langle Q,s\rangle\xrightarrow{\alpha}\langle Q',s'\rangle\quad \langle P,s\rangle\nrightarrow}{\langle P\parallel Q,s\rangle\xrightarrow{\alpha}\langle P\parallel Q',s'\rangle}\quad (bn(\alpha)\cap fn(P)=\emptyset)\]

        \[\textbf{PAR}_3\quad \frac{\langle P,s\rangle\xrightarrow{\alpha}\langle P',s'\rangle\quad \langle Q,s\rangle\xrightarrow{\beta}\langle Q',s''\rangle}{\langle P\parallel Q,s\rangle\xrightarrow{\{\alpha,\beta\}}\langle P'\parallel Q',s'\cup s''\rangle}\] $(\beta\neq\overline{\alpha}, bn(\alpha)\cap bn(\beta)=\emptyset, bn(\alpha)\cap fn(Q)=\emptyset,bn(\beta)\cap fn(P)=\emptyset)$

        \[\textbf{PAR}_4\quad \frac{\langle P,s\rangle\xrightarrow{x_1(z)}\langle P',s'\rangle\quad \langle Q,s\rangle\xrightarrow{x_2(z)}\langle Q',s''\rangle}{\langle P\parallel Q,s\rangle\xrightarrow{\{x_1(w),x_2(w)\}}\langle P'\{w/z\}\parallel Q'\{w/z\},s'\cup s''\rangle}\quad (w\notin fn((z)P)\cup fn((z)Q))\]

        \[\textbf{COM}\quad \frac{\langle P,s\rangle\xrightarrow{\overline{x}y}\langle P',s'\rangle\quad \langle Q,s\rangle\xrightarrow{x(z)}\langle Q',s''\rangle}{\langle P\parallel Q,s\rangle\xrightarrow{\tau}\langle P'\parallel Q'\{y/z\},s'\cup s''\rangle}\]

        \[\textbf{CLOSE}\quad \frac{\langle P,s\rangle\xrightarrow{\overline{x}(w)}\langle P',s'\rangle\quad \langle Q,s\rangle\xrightarrow{x(w)}\langle Q',s''\rangle}{\langle P\parallel Q,s\rangle\xrightarrow{\tau}\langle (w)(P'\parallel Q'),s'\cup s''\rangle}\]

        \caption{Transition rules}
        \label{TRForPITC3}
    \end{table}
\end{center}

\begin{center}
    \begin{table}
        \[\textbf{SUM}_1\quad \frac{\langle P,s\rangle\xrightarrow{\alpha}\langle P',s'\rangle}{\langle P+Q,s\rangle\xrightarrow{\alpha}\langle P',s'\rangle}\]

        \[\textbf{SUM}_2\quad \frac{\langle P,s\rangle\xrightarrow{\{\alpha_1,\cdots,\alpha_n\}}\langle P',s'\rangle}{\langle P+Q,s\rangle\xrightarrow{\{\alpha_1,\cdots,\alpha_n\}}\langle P',s'\rangle}\]

        \[\textbf{IDE}_1\quad\frac{\langle P\{\widetilde{y}/\widetilde{x}\},s\rangle\xrightarrow{\alpha}\langle P',s'\rangle}{\langle A(\widetilde{y}),s\rangle\xrightarrow{\alpha}\langle P',s'\rangle}\quad (A(\widetilde{x})\overset{\text{def}}{=}P)\]

        \[\textbf{IDE}_2\quad\frac{\langle P\{\widetilde{y}/\widetilde{x}\},s\rangle\xrightarrow{\{\alpha_1,\cdots,\alpha_n\}}\langle P',s'\rangle} {\langle A(\widetilde{y}),s\rangle\xrightarrow{\{\alpha_1,\cdots,\alpha_n\}}\langle P',s'\rangle}\quad (A(\widetilde{x})\overset{\text{def}}{=}P)\]

        \[\textbf{RES}_1\quad \frac{\langle P,s\rangle\xrightarrow{\alpha}\langle P',s'\rangle}{\langle (y)P,s\rangle\xrightarrow{\alpha}\langle (y)P',s'\rangle}\quad (y\notin n(\alpha))\]

        \[\textbf{RES}_2\quad \frac{\langle P,s\rangle\xrightarrow{\{\alpha_1,\cdots,\alpha_n\}}\langle P',s'\rangle}{\langle (y)P,s\rangle\xrightarrow{\{\alpha_1,\cdots,\alpha_n\}}\langle (y)P',s'\rangle}\quad (y\notin n(\alpha_1)\cup\cdots\cup n(\alpha_n))\]

        \[\textbf{OPEN}_1\quad \frac{\langle P,s\rangle\xrightarrow{\overline{x}y}\langle P',s'\rangle}{\langle (y)P,s\rangle\xrightarrow{\overline{x}(w)}\langle P'\{w/y\},s'\rangle} \quad (y\neq x, w\notin fn((y)P'))\]

        \[\textbf{OPEN}_2\quad \frac{\langle P,s\rangle\xrightarrow{\{\overline{x}_1 y,\cdots,\overline{x}_n y\}}\langle P',s'\rangle}{\langle(y)P,s\rangle\xrightarrow{\{\overline{x}_1(w),\cdots,\overline{x}_n(w)\}}\langle P'\{w/y\},s'\rangle} \quad (y\neq x_1\neq\cdots\neq x_n, w\notin fn((y)P'))\]

        \caption{Transition rules (continuing)}
        \label{TRForPITC32}
    \end{table}
\end{center}
\end{definition}

\subsubsection{Properties of Transitions}

\begin{proposition}
\begin{enumerate}
  \item If $\langle P,s\rangle\xrightarrow{\alpha}\langle P',s'\rangle$ then
  \begin{enumerate}
    \item $fn(\alpha)\subseteq fn(P)$;
    \item $fn(P')\subseteq fn(P)\cup bn(\alpha)$;
  \end{enumerate}
  \item If $\langle P,s\rangle\xrightarrow{\{\alpha_1,\cdots,\alpha_n\}}\langle P',s\rangle$ then
  \begin{enumerate}
    \item $fn(\alpha_1)\cup\cdots\cup fn(\alpha_n)\subseteq fn(P)$;
    \item $fn(P')\subseteq fn(P)\cup bn(\alpha_1)\cup\cdots\cup bn(\alpha_n)$.
  \end{enumerate}
\end{enumerate}
\end{proposition}

\begin{proof}
By induction on the depth of inference.
\end{proof}

\begin{proposition}
Suppose that $\langle P,s\rangle\xrightarrow{\alpha(y)}\langle P',s'\rangle$, where $\alpha=x$ or $\alpha=\overline{x}$, and $x\notin n(P)$, then there exists some $P''\equiv_{\alpha}P'\{z/y\}$,
$\langle P,s\rangle\xrightarrow{\alpha(z)}\langle P'',s''\rangle$.
\end{proposition}

\begin{proof}
By induction on the depth of inference.
\end{proof}

\begin{proposition}
If $\langle P,s\rangle\xrightarrow{\alpha} \langle P',s'\rangle$, $bn(\alpha)\cap fn(P'\sigma)=\emptyset$, and $\sigma\lceil bn(\alpha)=id$, then there exists some $P''\equiv_{\alpha}P'\sigma$,
$\langle P,s\rangle\sigma\xrightarrow{\alpha\sigma}\langle P'',s''\rangle$.
\end{proposition}

\begin{proof}
By the definition of substitution (Definition \ref{subs3}) and induction on the depth of inference.
\end{proof}

\begin{proposition}
\begin{enumerate}
  \item If $\langle P\{w/z\},s\rangle\xrightarrow{\alpha}\langle P',s'\rangle$, where $w\notin fn(P)$ and $bn(\alpha)\cap fn(P,w)=\emptyset$, then there exist some $Q$ and $\beta$ with $Q\{w/z\}\equiv_{\alpha}P'$ and
  $\beta\sigma=\alpha$, $\langle P,s\rangle\xrightarrow{\beta}\langle Q,s'\rangle$;
  \item If $\langle P\{w/z\},s\rangle\xrightarrow{\{\alpha_1,\cdots,\alpha_n\}}\langle P',s'\rangle$, where $w\notin fn(P)$ and $bn(\alpha_1)\cap\cdots\cap bn(\alpha_n)\cap fn(P,w)=\emptyset$, then there exist some $Q$
  and $\beta_1,\cdots,\beta_n$ with $Q\{w/z\}\equiv_{\alpha}P'$ and $\beta_1\sigma=\alpha_1,\cdots,\beta_n\sigma=\alpha_n$, $\langle P,s\rangle\xrightarrow{\{\beta_1,\cdots,\beta_n\}}\langle Q,s'\rangle$.
\end{enumerate}

\end{proposition}

\begin{proof}
By the definition of substitution (Definition \ref{subs3}) and induction on the depth of inference.
\end{proof}

\subsection{Strong Bisimilarities}\label{s3}

\subsubsection{Laws and Congruence}

\begin{theorem}
$\equiv_{\alpha}$ are strongly truly concurrent bisimulations. That is, if $P\equiv_{\alpha}Q$, then,
\begin{enumerate}
  \item $P\sim_p  Q$;
  \item $P\sim_s  Q$;
  \item $P\sim_{hp}  Q$;
  \item $P\sim_{hhp}  Q$.
\end{enumerate}
\end{theorem}

\begin{proof}
By induction on the depth of inference, we can get the following facts:

\begin{enumerate}
  \item If $\alpha$ is a free action and $\langle P,s\rangle\xrightarrow{\alpha}\langle P',s'\rangle$, then equally for some $Q'$ with $P'\equiv_{\alpha}Q'$,
  $\langle Q,s\rangle\xrightarrow{\alpha}\langle Q',s'\rangle$;
  \item If $\langle P,s\rangle\xrightarrow{a(y)}\langle P',s'\rangle$ with $a=x$ or $a=\overline{x}$ and $z\notin n(Q)$, then equally for some $Q'$ with $P'\{z/y\}\equiv_{\alpha}Q'$,
  $\langle Q,s\rangle\xrightarrow{a(z)}\langle Q',s'\rangle$.
\end{enumerate}

Then, we can get:

\begin{enumerate}
  \item by the definition of strongly pomset bisimilarity, $P\sim_p  Q$;
  \item by the definition of strongly step bisimilarity, $P\sim_s  Q$;
  \item by the definition of strongly hp-bisimilarity, $P\sim_{hp}  Q$;
  \item by the definition of strongly hhp-bisimilarity, $P\sim_{hhp}  Q$.
\end{enumerate}
\end{proof}

\begin{proposition}[Summation laws for strongly pomset bisimulation] The Summation laws for strongly pomset bisimulation are as follows.

\begin{enumerate}
  \item $P+Q\sim_p  Q+P$;
  \item $P+(Q+R)\sim_p  (P+Q)+R$;
  \item $P+P\sim_p  P$;
  \item $P+\textbf{nil}\sim_p  P$.
\end{enumerate}

\end{proposition}

\begin{proof}
\begin{enumerate}
  \item $P+Q\sim_p  Q+P$. It is sufficient to prove the relation $R=\{(P+Q, Q+P)\}\cup \textbf{Id}$ is a strongly pomset bisimulation, we omit it;
  \item $P+(Q+R)\sim_p  (P+Q)+R$. It is sufficient to prove the relation $R=\{(P+(Q+R), (P+Q)+R)\}\cup \textbf{Id}$ is a strongly pomset bisimulation, we omit it;
  \item $P+P\sim_p  P$. It is sufficient to prove the relation $R=\{(P+P, P)\}\cup \textbf{Id}$ is a strongly pomset bisimulation, we omit it;
  \item $P+\textbf{nil}\sim_p  P$. It is sufficient to prove the relation $R=\{(P+\textbf{nil}, P)\}\cup \textbf{Id}$ is a strongly pomset bisimulation, we omit it.
\end{enumerate}
\end{proof}

\begin{proposition}[Summation laws for strongly step bisimulation] The Summation laws for strongly step bisimulation are as follows.
\begin{enumerate}
  \item $P+Q\sim_s  Q+P$;
  \item $P+(Q+R)\sim_s  (P+Q)+R$;
  \item $P+P\sim_s  P$;
  \item $P+\textbf{nil}\sim_s  P$.
\end{enumerate}
\end{proposition}

\begin{proof}
\begin{enumerate}
  \item $P+Q\sim_s  Q+P$. It is sufficient to prove the relation $R=\{(P+Q, Q+P)\}\cup \textbf{Id}$ is a strongly step bisimulation, we omit it;
  \item $P+(Q+R)\sim_s  (P+Q)+R$. It is sufficient to prove the relation $R=\{(P+(Q+R), (P+Q)+R)\}\cup \textbf{Id}$ is a strongly step bisimulation, we omit it;
  \item $P+P\sim_s  P$. It is sufficient to prove the relation $R=\{(P+P, P)\}\cup \textbf{Id}$ is a strongly step bisimulation, we omit it;
  \item $P+\textbf{nil}\sim_s  P$. It is sufficient to prove the relation $R=\{(P+\textbf{nil}, P)\}\cup \textbf{Id}$ is a strongly step bisimulation, we omit it.
\end{enumerate}
\end{proof}

\begin{proposition}[Summation laws for strongly hp-bisimulation] The Summation laws for strongly hp-bisimulation are as follows.
\begin{enumerate}
  \item $P+Q\sim_{hp}  Q+P$;
  \item $P+(Q+R)\sim_{hp}  (P+Q)+R$;
  \item $P+P\sim_{hp}  P$;
  \item $P+\textbf{nil}\sim_{hp}  P$.
\end{enumerate}
\end{proposition}

\begin{proof}
\begin{enumerate}
  \item $P+Q\sim_{hp}  Q+P$. It is sufficient to prove the relation $R=\{(P+Q, Q+P)\}\cup \textbf{Id}$ is a strongly hp-bisimulation, we omit it;
  \item $P+(Q+R)\sim_{hp}  (P+Q)+R$. It is sufficient to prove the relation $R=\{(P+(Q+R), (P+Q)+R)\}\cup \textbf{Id}$ is a strongly hp-bisimulation, we omit it;
  \item $P+P\sim_{hp}  P$. It is sufficient to prove the relation $R=\{(P+P, P)\}\cup \textbf{Id}$ is a strongly hp-bisimulation, we omit it;
  \item $P+\textbf{nil}\sim_{hp}  P$. It is sufficient to prove the relation $R=\{(P+\textbf{nil}, P)\}\cup \textbf{Id}$ is a strongly hp-bisimulation, we omit it.
\end{enumerate}
\end{proof}

\begin{proposition}[Summation laws for strongly hhp-bisimulation] The Summation laws for strongly hhp-bisimulation are as follows.
\begin{enumerate}
  \item $P+Q\sim_{hhp}  Q+P$;
  \item $P+(Q+R)\sim_{hhp}  (P+Q)+R$;
  \item $P+P\sim_{hhp}  P$;
  \item $P+\textbf{nil}\sim_{hhp}  P$.
\end{enumerate}
\end{proposition}

\begin{proof}
\begin{enumerate}
  \item $P+Q\sim_{hhp}  Q+P$. It is sufficient to prove the relation $R=\{(P+Q, Q+P)\}\cup \textbf{Id}$ is a strongly hhp-bisimulation, we omit it;
  \item $P+(Q+R)\sim_{hhp}  (P+Q)+R$. It is sufficient to prove the relation $R=\{(P+(Q+R), (P+Q)+R)\}\cup \textbf{Id}$ is a strongly hhp-bisimulation, we omit it;
  \item $P+P\sim_{hhp}  P$. It is sufficient to prove the relation $R=\{(P+P, P)\}\cup \textbf{Id}$ is a strongly hhp-bisimulation, we omit it;
  \item $P+\textbf{nil}\sim_{hhp}  P$. It is sufficient to prove the relation $R=\{(P+\textbf{nil}, P)\}\cup \textbf{Id}$ is a strongly hhp-bisimulation, we omit it.
\end{enumerate}
\end{proof}

\begin{theorem}[Identity law for strongly truly concurrent bisimilarities]
If $A(\widetilde{x})\overset{\text{def}}{=}P$, then

\begin{enumerate}
  \item $A(\widetilde{y})\sim_p  P\{\widetilde{y}/\widetilde{x}\}$;
  \item $A(\widetilde{y})\sim_s  P\{\widetilde{y}/\widetilde{x}\}$;
  \item $A(\widetilde{y})\sim_{hp}  P\{\widetilde{y}/\widetilde{x}\}$;
  \item $A(\widetilde{y})\sim_{hhp}  P\{\widetilde{y}/\widetilde{x}\}$.
\end{enumerate}
\end{theorem}

\begin{proof}
\begin{enumerate}
  \item $A(\widetilde{y})\sim_p  P\{\widetilde{y}/\widetilde{x}\}$. It is sufficient to prove the relation $R=\{(A(\widetilde{y}), P\{\widetilde{y}/\widetilde{x}\})\}\cup \textbf{Id}$ is a strongly pomset bisimulation, we omit it;
  \item $A(\widetilde{y})\sim_s  P\{\widetilde{y}/\widetilde{x}\}$. It is sufficient to prove the relation $R=\{(A(\widetilde{y}), P\{\widetilde{y}/\widetilde{x}\})\}\cup \textbf{Id}$ is a strongly step bisimulation, we omit it;
  \item $A(\widetilde{y})\sim_{hp}  P\{\widetilde{y}/\widetilde{x}\}$. It is sufficient to prove the relation $R=\{(A(\widetilde{y}), P\{\widetilde{y}/\widetilde{x}\})\}\cup \textbf{Id}$ is a strongly hp-bisimulation, we omit it;
  \item $A(\widetilde{y})\sim_{hhp}  P\{\widetilde{y}/\widetilde{x}\}$. It is sufficient to prove the relation $R=\{(A(\widetilde{y}), P\{\widetilde{y}/\widetilde{x}\})\}\cup \textbf{Id}$ is a strongly hhp-bisimulation, we omit it.
\end{enumerate}
\end{proof}

\begin{theorem}[Restriction Laws for strongly pomset bisimilarity]
The restriction laws for strongly pomset bisimilarity are as follows.

\begin{enumerate}
  \item $(y)P\sim_p  P$, if $y\notin fn(P)$;
  \item $(y)(z)P\sim_p  (z)(y)P$;
  \item $(y)(P+Q)\sim_p  (y)P+(y)Q$;
  \item $(y)\alpha.P\sim_p  \alpha.(y)P$ if $y\notin n(\alpha)$;
  \item $(y)\alpha.P\sim_p  \textbf{nil}$ if $y$ is the subject of $\alpha$.
\end{enumerate}
\end{theorem}

\begin{proof}
\begin{enumerate}
  \item $(y)P\sim_p  P$, if $y\notin fn(P)$. It is sufficient to prove the relation $R=\{((y)P, P)\}\cup \textbf{Id}$, if $y\notin fn(P)$, is a strongly pomset bisimulation, we omit it;
  \item $(y)(z)P\sim_p  (z)(y)P$. It is sufficient to prove the relation $R=\{((y)(z)P, (z)(y)P)\}\cup \textbf{Id}$ is a strongly pomset bisimulation, we omit it;
  \item $(y)(P+Q)\sim_p  (y)P+(y)Q$. It is sufficient to prove the relation $R=\{((y)(P+Q), (y)P+(y)Q)\}\cup \textbf{Id}$ is a strongly pomset bisimulation, we omit it;
  \item $(y)\alpha.P\sim_p  \alpha.(y)P$ if $y\notin n(\alpha)$. It is sufficient to prove the relation $R=\{((y)\alpha.P, \alpha.(y)P)\}\cup \textbf{Id}$, if $y\notin n(\alpha)$, is a strongly pomset bisimulation, we omit it;
  \item $(y)\alpha.P\sim_p  \textbf{nil}$ if $y$ is the subject of $\alpha$. It is sufficient to prove the relation $R=\{((y)\alpha.P, \textbf{nil})\}\cup \textbf{Id}$, if $y$ is the subject of $\alpha$, is a strongly pomset bisimulation, we omit it.
\end{enumerate}
\end{proof}

\begin{theorem}[Restriction Laws for strongly step bisimilarity]
The restriction laws for strongly step bisimilarity are as follows.

\begin{enumerate}
  \item $(y)P\sim_s  P$, if $y\notin fn(P)$;
  \item $(y)(z)P\sim_s  (z)(y)P$;
  \item $(y)(P+Q)\sim_s  (y)P+(y)Q$;
  \item $(y)\alpha.P\sim_s  \alpha.(y)P$ if $y\notin n(\alpha)$;
  \item $(y)\alpha.P\sim_s  \textbf{nil}$ if $y$ is the subject of $\alpha$.
\end{enumerate}
\end{theorem}

\begin{proof}
\begin{enumerate}
  \item $(y)P\sim_s  P$, if $y\notin fn(P)$. It is sufficient to prove the relation $R=\{((y)P, P)\}\cup \textbf{Id}$, if $y\notin fn(P)$, is a strongly step bisimulation, we omit it;
  \item $(y)(z)P\sim_s  (z)(y)P$. It is sufficient to prove the relation $R=\{((y)(z)P, (z)(y)P)\}\cup \textbf{Id}$ is a strongly step bisimulation, we omit it;
  \item $(y)(P+Q)\sim_s  (y)P+(y)Q$. It is sufficient to prove the relation $R=\{((y)(P+Q), (y)P+(y)Q)\}\cup \textbf{Id}$ is a strongly step bisimulation, we omit it;
  \item $(y)\alpha.P\sim_s  \alpha.(y)P$ if $y\notin n(\alpha)$. It is sufficient to prove the relation $R=\{((y)\alpha.P, \alpha.(y)P)\}\cup \textbf{Id}$, if $y\notin n(\alpha)$, is a strongly step bisimulation, we omit it;
  \item $(y)\alpha.P\sim_s  \textbf{nil}$ if $y$ is the subject of $\alpha$. It is sufficient to prove the relation $R=\{((y)\alpha.P, \textbf{nil})\}\cup \textbf{Id}$, if $y$ is the subject of $\alpha$, is a strongly step bisimulation, we omit it.
\end{enumerate}
\end{proof}

\begin{theorem}[Restriction Laws for strongly hp-bisimilarity]
The restriction laws for strongly hp-bisimilarity are as follows.

\begin{enumerate}
  \item $(y)P\sim_{hp}  P$, if $y\notin fn(P)$;
  \item $(y)(z)P\sim_{hp}  (z)(y)P$;
  \item $(y)(P+Q)\sim_{hp}  (y)P+(y)Q$;
  \item $(y)\alpha.P\sim_{hp}  \alpha.(y)P$ if $y\notin n(\alpha)$;
  \item $(y)\alpha.P\sim_{hp}  \textbf{nil}$ if $y$ is the subject of $\alpha$.
\end{enumerate}
\end{theorem}

\begin{proof}
\begin{enumerate}
  \item $(y)P\sim_{hp}  P$, if $y\notin fn(P)$. It is sufficient to prove the relation $R=\{((y)P, P)\}\cup \textbf{Id}$, if $y\notin fn(P)$, is a strongly hp-bisimulation, we omit it;
  \item $(y)(z)P\sim_{hp}  (z)(y)P$. It is sufficient to prove the relation $R=\{((y)(z)P, (z)(y)P)\}\cup \textbf{Id}$ is a strongly hp-bisimulation, we omit it;
  \item $(y)(P+Q)\sim_{hp}  (y)P+(y)Q$. It is sufficient to prove the relation $R=\{((y)(P+Q), (y)P+(y)Q)\}\cup \textbf{Id}$ is a strongly hp-bisimulation, we omit it;
  \item $(y)\alpha.P\sim_{hp}  \alpha.(y)P$ if $y\notin n(\alpha)$. It is sufficient to prove the relation $R=\{((y)\alpha.P, \alpha.(y)P)\}\cup \textbf{Id}$, if $y\notin n(\alpha)$, is a strongly hp-bisimulation, we omit it;
  \item $(y)\alpha.P\sim_{hp}  \textbf{nil}$ if $y$ is the subject of $\alpha$. It is sufficient to prove the relation $R=\{((y)\alpha.P, \textbf{nil})\}\cup \textbf{Id}$, if $y$ is the subject of $\alpha$, is a strongly hp-bisimulation, we omit it.
\end{enumerate}
\end{proof}

\begin{theorem}[Restriction Laws for strongly hhp-bisimilarity]
The restriction laws for strongly hhp-bisimilarity are as follows.

\begin{enumerate}
  \item $(y)P\sim_{hhp}  P$, if $y\notin fn(P)$;
  \item $(y)(z)P\sim_{hhp}  (z)(y)P$;
  \item $(y)(P+Q)\sim_{hhp}  (y)P+(y)Q$;
  \item $(y)\alpha.P\sim_{hhp}  \alpha.(y)P$ if $y\notin n(\alpha)$;
  \item $(y)\alpha.P\sim_{hhp}  \textbf{nil}$ if $y$ is the subject of $\alpha$.
\end{enumerate}
\end{theorem}

\begin{proof}
\begin{enumerate}
  \item $(y)P\sim_{hhp}  P$, if $y\notin fn(P)$. It is sufficient to prove the relation $R=\{((y)P, P)\}\cup \textbf{Id}$, if $y\notin fn(P)$, is a strongly hhp-bisimulation, we omit it;
  \item $(y)(z)P\sim_{hhp}  (z)(y)P$. It is sufficient to prove the relation $R=\{((y)(z)P, (z)(y)P)\}\cup \textbf{Id}$ is a strongly hhp-bisimulation, we omit it;
  \item $(y)(P+Q)\sim_{hhp}  (y)P+(y)Q$. It is sufficient to prove the relation $R=\{((y)(P+Q), (y)P+(y)Q)\}\cup \textbf{Id}$ is a strongly hhp-bisimulation, we omit it;
  \item $(y)\alpha.P\sim_{hhp}  \alpha.(y)P$ if $y\notin n(\alpha)$. It is sufficient to prove the relation $R=\{((y)\alpha.P, \alpha.(y)P)\}\cup \textbf{Id}$, if $y\notin n(\alpha)$, is a strongly hhp-bisimulation, we omit it;
  \item $(y)\alpha.P\sim_{hhp}  \textbf{nil}$ if $y$ is the subject of $\alpha$. It is sufficient to prove the relation $R=\{((y)\alpha.P, \textbf{nil})\}\cup \textbf{Id}$, if $y$ is the subject of $\alpha$, is a strongly hhp-bisimulation, we omit it.
\end{enumerate}
\end{proof}

\begin{theorem}[Parallel laws for strongly pomset bisimilarity]
The parallel laws for strongly pomset bisimilarity are as follows.

\begin{enumerate}
  \item $P\parallel \textbf{nil}\sim_p  P$;
  \item $P_1\parallel P_2\sim_p  P_2\parallel P_1$;
  \item $(P_1\parallel P_2)\parallel P_3\sim_p  P_1\parallel (P_2\parallel P_3)$;
  \item $(y)(P_1\parallel P_2)\sim_p  (y)P_1\parallel (y)P_2$, if $y\notin fn(P_1)\cap fn(P_2)$.
\end{enumerate}
\end{theorem}

\begin{proof}
\begin{enumerate}
  \item $P\parallel \textbf{nil}\sim_p  P$. It is sufficient to prove the relation $R=\{(P\parallel \textbf{nil}, P)\}\cup \textbf{Id}$ is a strongly pomset bisimulation, we omit it;
  \item $P_1\parallel P_2\sim_p  P_2\parallel P_1$. It is sufficient to prove the relation $R=\{(P_1\parallel P_2, P_2\parallel P_1)\}\cup \textbf{Id}$ is a strongly pomset bisimulation, we omit it;
  \item $(P_1\parallel P_2)\parallel P_3\sim_p  P_1\parallel (P_2\parallel P_3)$. It is sufficient to prove the relation $R=\{((P_1\parallel P_2)\parallel P_3, P_1\parallel (P_2\parallel P_3))\}\cup \textbf{Id}$ is a strongly pomset bisimulation, we omit it;
  \item $(y)(P_1\parallel P_2)\sim_p  (y)P_1\parallel (y)P_2$, if $y\notin fn(P_1)\cap fn(P_2)$. It is sufficient to prove the relation $R=\{((y)(P_1\parallel P_2), (y)P_1\parallel (y)P_2)\}\cup \textbf{Id}$, if $y\notin fn(P_1)\cap fn(P_2)$, is a strongly pomset bisimulation, we omit it.
\end{enumerate}
\end{proof}

\begin{theorem}[Parallel laws for strongly step bisimilarity]
The parallel laws for strongly step bisimilarity are as follows.

\begin{enumerate}
  \item $P\parallel \textbf{nil}\sim_s  P$;
  \item $P_1\parallel P_2\sim_s  P_2\parallel P_1$;
  \item $(P_1\parallel P_2)\parallel P_3\sim_s  P_1\parallel (P_2\parallel P_3)$;
  \item $(y)(P_1\parallel P_2)\sim_s  (y)P_1\parallel (y)P_2$, if $y\notin fn(P_1)\cap fn(P_2)$.
\end{enumerate}
\end{theorem}

\begin{proof}
\begin{enumerate}
  \item $P\parallel \textbf{nil}\sim_s  P$. It is sufficient to prove the relation $R=\{(P\parallel \textbf{nil}, P)\}\cup \textbf{Id}$ is a strongly step bisimulation, we omit it;
  \item $P_1\parallel P_2\sim_s  P_2\parallel P_1$. It is sufficient to prove the relation $R=\{(P_1\parallel P_2, P_2\parallel P_1)\}\cup \textbf{Id}$ is a strongly step bisimulation, we omit it;
  \item $(P_1\parallel P_2)\parallel P_3\sim_s  P_1\parallel (P_2\parallel P_3)$. It is sufficient to prove the relation $R=\{((P_1\parallel P_2)\parallel P_3, P_1\parallel (P_2\parallel P_3))\}\cup \textbf{Id}$ is a strongly step bisimulation, we omit it;
  \item $(y)(P_1\parallel P_2)\sim_s  (y)P_1\parallel (y)P_2$, if $y\notin fn(P_1)\cap fn(P_2)$. It is sufficient to prove the relation $R=\{((y)(P_1\parallel P_2), (y)P_1\parallel (y)P_2)\}\cup \textbf{Id}$, if $y\notin fn(P_1)\cap fn(P_2)$, is a strongly step bisimulation, we omit it.
\end{enumerate}
\end{proof}

\begin{theorem}[Parallel laws for strongly hp-bisimilarity]
The parallel laws for strongly hp-bisimilarity are as follows.

\begin{enumerate}
  \item $P\parallel \textbf{nil}\sim_{hp}  P$;
  \item $P_1\parallel P_2\sim_{hp}  P_2\parallel P_1$;
  \item $(P_1\parallel P_2)\parallel P_3\sim_{hp}  P_1\parallel (P_2\parallel P_3)$;
  \item $(y)(P_1\parallel P_2)\sim_{hp}  (y)P_1\parallel (y)P_2$, if $y\notin fn(P_1)\cap fn(P_2)$.
\end{enumerate}
\end{theorem}

\begin{proof}
\begin{enumerate}
  \item $P\parallel \textbf{nil}\sim_{hp}  P$. It is sufficient to prove the relation $R=\{(P\parallel \textbf{nil}, P)\}\cup \textbf{Id}$ is a strongly hp-bisimulation, we omit it;
  \item $P_1\parallel P_2\sim_{hp}  P_2\parallel P_1$. It is sufficient to prove the relation $R=\{(P_1\parallel P_2, P_2\parallel P_1)\}\cup \textbf{Id}$ is a strongly hp-bisimulation, we omit it;
  \item $(P_1\parallel P_2)\parallel P_3\sim_{hp}  P_1\parallel (P_2\parallel P_3)$. It is sufficient to prove the relation $R=\{((P_1\parallel P_2)\parallel P_3, P_1\parallel (P_2\parallel P_3))\}\cup \textbf{Id}$ is a strongly hp-bisimulation, we omit it;
  \item $(y)(P_1\parallel P_2)\sim_{hp}  (y)P_1\parallel (y)P_2$, if $y\notin fn(P_1)\cap fn(P_2)$. It is sufficient to prove the relation $R=\{((y)(P_1\parallel P_2), (y)P_1\parallel (y)P_2)\}\cup \textbf{Id}$, if $y\notin fn(P_1)\cap fn(P_2)$, is a strongly hp-bisimulation, we omit it.
\end{enumerate}
\end{proof}

\begin{theorem}[Parallel laws for strongly hhp-bisimilarity]
The parallel laws for strongly hhp-bisimilarity are as follows.

\begin{enumerate}
  \item $P\parallel \textbf{nil}\sim_{hhp}  P$;
  \item $P_1\parallel P_2\sim_{hhp}  P_2\parallel P_1$;
  \item $(P_1\parallel P_2)\parallel P_3\sim_{hhp}  P_1\parallel (P_2\parallel P_3)$;
  \item $(y)(P_1\parallel P_2)\sim_{hhp}  (y)P_1\parallel (y)P_2$, if $y\notin fn(P_1)\cap fn(P_2)$.
\end{enumerate}
\end{theorem}

\begin{proof}
\begin{enumerate}
  \item $P\parallel \textbf{nil}\sim_{hhp}  P$. It is sufficient to prove the relation $R=\{(P\parallel \textbf{nil}, P)\}\cup \textbf{Id}$ is a strongly hhp-bisimulation, we omit it;
  \item $P_1\parallel P_2\sim_{hhp}  P_2\parallel P_1$. It is sufficient to prove the relation $R=\{(P_1\parallel P_2, P_2\parallel P_1)\}\cup \textbf{Id}$ is a strongly hhp-bisimulation, we omit it;
  \item $(P_1\parallel P_2)\parallel P_3\sim_{hhp}  P_1\parallel (P_2\parallel P_3)$. It is sufficient to prove the relation $R=\{((P_1\parallel P_2)\parallel P_3, P_1\parallel (P_2\parallel P_3))\}\cup \textbf{Id}$ is a strongly hhp-bisimulation, we omit it;
  \item $(y)(P_1\parallel P_2)\sim_{hhp}  (y)P_1\parallel (y)P_2$, if $y\notin fn(P_1)\cap fn(P_2)$. It is sufficient to prove the relation $R=\{((y)(P_1\parallel P_2), (y)P_1\parallel (y)P_2)\}\cup \textbf{Id}$, if $y\notin fn(P_1)\cap fn(P_2)$, is a strongly hhp-bisimulation, we omit it.
\end{enumerate}
\end{proof}

\begin{theorem}[Expansion law for truly concurrent bisimilarities]
Let $P\equiv\sum_i \alpha_{i}.P_{i}$ and $Q\equiv\sum_j\beta_{j}.Q_{j}$, where $bn(\alpha_{i})\cap fn(Q)=\emptyset$ for all $i$, and
  $bn(\beta_{j})\cap fn(P)=\emptyset$ for all $j$. Then,

\begin{enumerate}
  \item $P\parallel Q\sim_p  \sum_i\sum_j (\alpha_{i}\parallel \beta_{j}).(P_{i}\parallel Q_{j})+\sum_{\alpha_{i} \textrm{ comp }\beta_{j}}\tau.R_{ij}$;
  \item $P\parallel Q\sim_s  \sum_i\sum_j (\alpha_{i}\parallel \beta_{j}).(P_{i}\parallel Q_{j})+\sum_{\alpha_{i} \textrm{ comp }\beta_{j}}\tau.R_{ij}$;
  \item $P\parallel Q\sim_{hp}  \sum_i\sum_j (\alpha_{i}\parallel \beta_{j}).(P_{i}\parallel Q_{j})+\sum_{\alpha_{i} \textrm{ comp }\beta_{j}}\tau.R_{ij}$;
  \item $P\parallel Q\nsim_{phhp} \sum_i\sum_j (\alpha_{i}\parallel \beta_{j}).(P_{i}\parallel Q_{j})+\sum_{\alpha_{i} \textrm{ comp }\beta_{j}}\tau.R_{ij}$.
\end{enumerate}

Where $\alpha_i$ comp $\beta_j$ and $R_{ij}$ are defined as follows:
\begin{enumerate}
  \item $\alpha_{i}$ is $\overline{x}u$ and $\beta_{j}$ is $x(v)$, then $R_{ij}=P_{i}\parallel Q_{j}\{u/v\}$;
  \item $\alpha_{i}$ is $\overline{x}(u)$ and $\beta_{j}$ is $x(v)$, then $R_{ij}=(w)(P_{i}\{w/u\}\parallel Q_{j}\{w/v\})$, if $w\notin fn((u)P_{i})\cup fn((v)Q_{j})$;
  \item $\alpha_{i}$ is $x(v)$ and $\beta_{j}$ is $\overline{x}u$, then $R_{ij}=P_{i}\{u/v\}\parallel Q_{j}$;
  \item $\alpha_{i}$ is $x(v)$ and $\beta_{j}$ is $\overline{x}(u)$, then $R_{ij}=(w)(P_{i}\{w/v\}\parallel Q_{j}\{w/u\})$, if $w\notin fn((v)P_{i})\cup fn((u)Q_{j})$.
\end{enumerate}
\end{theorem}

\begin{proof}
According to the definition of strongly truly concurrent bisimulations, we can easily prove the above equations, and we omit the proof.
\end{proof}

\begin{theorem}[Equivalence and congruence for strongly pomset bisimilarity]
We can enjoy the full congruence modulo strongly pomset bisimilarity.

\begin{enumerate}
  \item $\sim_p $ is an equivalence relation;
  \item If $P\sim_p  Q$ then
  \begin{enumerate}
    \item $\alpha.P\sim_p  \alpha.Q$, $\alpha$ is a free action;
    \item $\phi.P\sim_p  \phi.Q$;
    \item $P+R\sim_p  Q+R$;
    \item $P\parallel R\sim_p  Q\parallel R$;
    \item $(w)P\sim_p  (w)Q$;
    \item $x(y).P\sim_p  x(y).Q$.
  \end{enumerate}
\end{enumerate}
\end{theorem}

\begin{proof}
\begin{enumerate}
  \item $\sim_p $ is an equivalence relation, it is obvious;
  \item If $P\sim_p  Q$ then
  \begin{enumerate}
    \item $\alpha.P\sim_p  \alpha.Q$, $\alpha$ is a free action. It is sufficient to prove the relation $R=\{(\alpha.P, \alpha.Q)\}\cup \textbf{Id}$ is a strongly pomset bisimulation, we omit it;
    \item $\phi.P\sim_p  \phi.Q$. It is sufficient to prove the relation $R=\{(\phi.P, \phi.Q)\}\cup \textbf{Id}$ is a strongly pomset bisimulation, we omit it;
    \item $P+R\sim_p  Q+R$. It is sufficient to prove the relation $R=\{(P+R, Q+R)\}\cup \textbf{Id}$ is a strongly pomset bisimulation, we omit it;
    \item $P\parallel R\sim_p  Q\parallel R$. It is sufficient to prove the relation $R=\{(P\parallel R, Q\parallel R)\}\cup \textbf{Id}$ is a strongly pomset bisimulation, we omit it;
    \item $(w)P\sim_p  (w)Q$. It is sufficient to prove the relation $R=\{((w)P, (w)Q)\}\cup \textbf{Id}$ is a strongly pomset bisimulation, we omit it;
    \item $x(y).P\sim_p  x(y).Q$. It is sufficient to prove the relation $R=\{(x(y).P, x(y).Q)\}\cup \textbf{Id}$ is a strongly pomset bisimulation, we omit it.
  \end{enumerate}
\end{enumerate}
\end{proof}

\begin{theorem}[Equivalence and congruence for strongly step bisimilarity]
We can enjoy the full congruence modulo strongly step bisimilarity.

\begin{enumerate}
  \item $\sim_s $ is an equivalence relation;
  \item If $P\sim_s  Q$ then
  \begin{enumerate}
    \item $\alpha.P\sim_s  \alpha.Q$, $\alpha$ is a free action;
    \item $\phi.P\sim_s  \phi.Q$;
    \item $P+R\sim_s  Q+R$;
    \item $P\parallel R\sim_s  Q\parallel R$;
    \item $(w)P\sim_s  (w)Q$;
    \item $x(y).P\sim_s  x(y).Q$.
  \end{enumerate}
\end{enumerate}
\end{theorem}

\begin{proof}
\begin{enumerate}
  \item $\sim_s $ is an equivalence relation, it is obvious;
  \item If $P\sim_s  Q$ then
  \begin{enumerate}
    \item $\alpha.P\sim_s  \alpha.Q$, $\alpha$ is a free action. It is sufficient to prove the relation $R=\{(\alpha.P, \alpha.Q)\}\cup \textbf{Id}$ is a strongly step bisimulation, we omit it;
    \item $\phi.P\sim_s  \phi.Q$. It is sufficient to prove the relation $R=\{(\phi.P, \phi.Q)\}\cup \textbf{Id}$ is a strongly step bisimulation, we omit it;
    \item $P+R\sim_s  Q+R$. It is sufficient to prove the relation $R=\{(P+R, Q+R)\}\cup \textbf{Id}$ is a strongly step bisimulation, we omit it;
    \item $P\parallel R\sim_s  Q\parallel R$. It is sufficient to prove the relation $R=\{(P\parallel R, Q\parallel R)\}\cup \textbf{Id}$ is a strongly step bisimulation, we omit it;
    \item $(w)P\sim_s  (w)Q$. It is sufficient to prove the relation $R=\{((w)P, (w)Q)\}\cup \textbf{Id}$ is a strongly step bisimulation, we omit it;
    \item $x(y).P\sim_s  x(y).Q$. It is sufficient to prove the relation $R=\{(x(y).P, x(y).Q)\}\cup \textbf{Id}$ is a strongly step bisimulation, we omit it.
  \end{enumerate}
\end{enumerate}
\end{proof}

\begin{theorem}[Equivalence and congruence for strongly hp-bisimilarity]
We can enjoy the full congruence modulo strongly hp-bisimilarity.

\begin{enumerate}
  \item $\sim_{hp} $ is an equivalence relation;
  \item If $P\sim_{hp}  Q$ then
  \begin{enumerate}
    \item $\alpha.P\sim_{hp}  \alpha.Q$, $\alpha$ is a free action;
    \item $\phi.P\sim_{hp}  \phi.Q$;
    \item $P+R\sim_{hp}  Q+R$;
    \item $P\parallel R\sim_{hp}  Q\parallel R$;
    \item $(w)P\sim_{hp}  (w)Q$;
    \item $x(y).P\sim_{hp}  x(y).Q$.
  \end{enumerate}
\end{enumerate}
\end{theorem}

\begin{proof}
\begin{enumerate}
  \item $\sim_{hp} $ is an equivalence relation, it is obvious;
  \item If $P\sim_{hp}  Q$ then
  \begin{enumerate}
    \item $\alpha.P\sim_{hp}  \alpha.Q$, $\alpha$ is a free action. It is sufficient to prove the relation $R=\{(\alpha.P, \alpha.Q)\}\cup \textbf{Id}$ is a strongly hp-bisimulation, we omit it;
    \item $\phi.P\sim_{hp}  \phi.Q$. It is sufficient to prove the relation $R=\{(\phi.P, \phi.Q)\}\cup \textbf{Id}$ is a strongly hp-bisimulation, we omit it;
    \item $P+R\sim_{hp}  Q+R$. It is sufficient to prove the relation $R=\{(P+R, Q+R)\}\cup \textbf{Id}$ is a strongly hp-bisimulation, we omit it;
    \item $P\parallel R\sim_{hp}  Q\parallel R$. It is sufficient to prove the relation $R=\{(P\parallel R, Q\parallel R)\}\cup \textbf{Id}$ is a strongly hp-bisimulation, we omit it;
    \item $(w)P\sim_{hp}  (w)Q$. It is sufficient to prove the relation $R=\{((w)P, (w)Q)\}\cup \textbf{Id}$ is a strongly hp-bisimulation, we omit it;
    \item $x(y).P\sim_{hp}  x(y).Q$. It is sufficient to prove the relation $R=\{(x(y).P, x(y).Q)\}\cup \textbf{Id}$ is a strongly hp-bisimulation, we omit it.
  \end{enumerate}
\end{enumerate}
\end{proof}

\begin{theorem}[Equivalence and congruence for strongly hhp-bisimilarity]
We can enjoy the full congruence modulo strongly hhp-bisimilarity.

\begin{enumerate}
  \item $\sim_{hhp} $ is an equivalence relation;
  \item If $P\sim_{hhp}  Q$ then
  \begin{enumerate}
    \item $\alpha.P\sim_{hhp}  \alpha.Q$, $\alpha$ is a free action;
    \item $\phi.P\sim_{hhp}  \phi.Q$;
    \item $P+R\sim_{hhp}  Q+R$;
    \item $P\parallel R\sim_{hhp}  Q\parallel R$;
    \item $(w)P\sim_{hhp}  (w)Q$;
    \item $x(y).P\sim_{hhp}  x(y).Q$.
  \end{enumerate}
\end{enumerate}
\end{theorem}

\begin{proof}
\begin{enumerate}
  \item $\sim_{hhp} $ is an equivalence relation, it is obvious;
  \item If $P\sim_{hhp}  Q$ then
  \begin{enumerate}
    \item $\alpha.P\sim_{hhp}  \alpha.Q$, $\alpha$ is a free action. It is sufficient to prove the relation $R=\{(\alpha.P, \alpha.Q)\}\cup \textbf{Id}$ is a strongly hhp-bisimulation, we omit it;
    \item $\phi.P\sim_{hhp}  \phi.Q$. It is sufficient to prove the relation $R=\{(\phi.P, \phi.Q)\}\cup \textbf{Id}$ is a strongly hhp-bisimulation, we omit it;
    \item $P+R\sim_{hhp}  Q+R$. It is sufficient to prove the relation $R=\{(P+R, Q+R)\}\cup \textbf{Id}$ is a strongly hhp-bisimulation, we omit it;
    \item $P\parallel R\sim_{hhp}  Q\parallel R$. It is sufficient to prove the relation $R=\{(P\parallel R, Q\parallel R)\}\cup \textbf{Id}$ is a strongly hhp-bisimulation, we omit it;
    \item $(w)P\sim_{hhp}  (w)Q$. It is sufficient to prove the relation $R=\{((w)P, (w)Q)\}\cup \textbf{Id}$ is a strongly hhp-bisimulation, we omit it;
    \item $x(y).P\sim_{hhp}  x(y).Q$. It is sufficient to prove the relation $R=\{(x(y).P, x(y).Q)\}\cup \textbf{Id}$ is a strongly hhp-bisimulation, we omit it.
  \end{enumerate}
\end{enumerate}
\end{proof}

\subsubsection{Recursion}

\begin{definition}
Let $X$ have arity $n$, and let $\widetilde{x}=x_1,\cdots,x_n$ be distinct names, and $fn(P)\subseteq\{x_1,\cdots,x_n\}$. The replacement of $X(\widetilde{x})$ by $P$ in $E$, written
$E\{X(\widetilde{x}):=P\}$, means the result of replacing each subterm $X(\widetilde{y})$ in $E$ by $P\{\widetilde{y}/\widetilde{x}\}$.
\end{definition}

\begin{definition}
Let $E$ and $F$ be two process expressions containing only $X_1,\cdots,X_m$ with associated name sequences $\widetilde{x}_1,\cdots,\widetilde{x}_m$. Then,
\begin{enumerate}
  \item $E\sim_p  F$ means $E(\widetilde{P})\sim_p  F(\widetilde{P})$;
  \item $E\sim_s  F$ means $E(\widetilde{P})\sim_s  F(\widetilde{P})$;
  \item $E\sim_{hp}  F$ means $E(\widetilde{P})\sim_{hp}  F(\widetilde{P})$;
  \item $E\sim_{hhp}  F$ means $E(\widetilde{P})\sim_{hhp}  F(\widetilde{P})$;
\end{enumerate}

for all $\widetilde{P}$ such that $fn(P_i)\subseteq \widetilde{x}_i$ for each $i$.
\end{definition}

\begin{definition}
A term or identifier is weakly guarded in $P$ if it lies within some subterm $\alpha.Q$ or $(\alpha_1\parallel\cdots\parallel \alpha_n).Q$ of $P$.
\end{definition}

\begin{theorem}
Assume that $\widetilde{E}$ and $\widetilde{F}$ are expressions containing only $X_i$ with $\widetilde{x}_i$, and $\widetilde{A}$ and $\widetilde{B}$ are identifiers with $A_i$, $B_i$. Then, for all $i$,
\begin{enumerate}
  \item $E_i\sim_s  F_i$, $A_i(\widetilde{x}_i)\overset{\text{def}}{=}E_i(\widetilde{A})$, $B_i(\widetilde{x}_i)\overset{\text{def}}{=}F_i(\widetilde{B})$, then
  $A_i(\widetilde{x}_i)\sim_s  B_i(\widetilde{x}_i)$;
  \item $E_i\sim_p  F_i$, $A_i(\widetilde{x}_i)\overset{\text{def}}{=}E_i(\widetilde{A})$, $B_i(\widetilde{x}_i)\overset{\text{def}}{=}F_i(\widetilde{B})$, then
  $A_i(\widetilde{x}_i)\sim_p  B_i(\widetilde{x}_i)$;
  \item $E_i\sim_{hp}  F_i$, $A_i(\widetilde{x}_i)\overset{\text{def}}{=}E_i(\widetilde{A})$, $B_i(\widetilde{x}_i)\overset{\text{def}}{=}F_i(\widetilde{B})$, then
  $A_i(\widetilde{x}_i)\sim_{hp}  B_i(\widetilde{x}_i)$;
  \item $E_i\sim_{hhp}  F_i$, $A_i(\widetilde{x}_i)\overset{\text{def}}{=}E_i(\widetilde{A})$, $B_i(\widetilde{x}_i)\overset{\text{def}}{=}F_i(\widetilde{B})$, then
  $A_i(\widetilde{x}_i)\sim_{hhp}  B_i(\widetilde{x}_i)$.
\end{enumerate}
\end{theorem}

\begin{proof}
\begin{enumerate}
  \item $E_i\sim_s  F_i$, $A_i(\widetilde{x}_i)\overset{\text{def}}{=}E_i(\widetilde{A})$, $B_i(\widetilde{x}_i)\overset{\text{def}}{=}F_i(\widetilde{B})$, then
  $A_i(\widetilde{x}_i)\sim_s  B_i(\widetilde{x}_i)$.

      We will consider the case $I=\{1\}$ with loss of generality, and show the following relation $R$ is a strongly step bisimulation.

      $$R=\{(G(A),G(B)):G\textrm{ has only identifier }X\}.$$

      By choosing $G\equiv X(\widetilde{y})$, it follows that $A(\widetilde{y})\sim_s  B(\widetilde{y})$. It is sufficient to prove the following:
      \begin{enumerate}
        \item If $\langle G(A),s\rangle\xrightarrow{\{\alpha_1,\cdots,\alpha_n\}}\langle P',s'\rangle$, where $\alpha_i(1\leq i\leq n)$ is a free action or bound output action with
        $bn(\alpha_1)\cap\cdots\cap bn(\alpha_n)\cap n(G(A),G(B))=\emptyset$, then $\langle G(B),s\rangle\xrightarrow{\{\alpha_1,\cdots,\alpha_n\}}\langle Q'',s''\rangle$ such that $P'\sim_s  Q''$;
        \item If $\langle G(A),s\rangle\xrightarrow{x(y)}\langle P',s'\rangle$ with $x\notin n(G(A),G(B))$, then $\langle G(B),s\rangle\xrightarrow{x(y)}\langle Q'',s''\rangle$, such that for all $u$,
        $\langle P',s'\rangle\{u/y\}\sim_s  \langle Q''\{u/y\},s''\rangle$.
      \end{enumerate}

      To prove the above properties, it is sufficient to induct on the depth of inference and quite routine, we omit it.
  \item $E_i\sim_p  F_i$, $A_i(\widetilde{x}_i)\overset{\text{def}}{=}E_i(\widetilde{A})$, $B_i(\widetilde{x}_i)\overset{\text{def}}{=}F_i(\widetilde{B})$, then
  $A_i(\widetilde{x}_i)\sim_p  B_i(\widetilde{x}_i)$. It can be proven similarly to the above case.
  \item $E_i\sim_{hp}  F_i$, $A_i(\widetilde{x}_i)\overset{\text{def}}{=}E_i(\widetilde{A})$, $B_i(\widetilde{x}_i)\overset{\text{def}}{=}F_i(\widetilde{B})$, then
  $A_i(\widetilde{x}_i)\sim_{hp}  B_i(\widetilde{x}_i)$. It can be proven similarly to the above case.
  \item $E_i\sim_{hhp}  F_i$, $A_i(\widetilde{x}_i)\overset{\text{def}}{=}E_i(\widetilde{A})$, $B_i(\widetilde{x}_i)\overset{\text{def}}{=}F_i(\widetilde{B})$, then
  $A_i(\widetilde{x}_i)\sim_{hhp}  B_i(\widetilde{x}_i)$. It can be proven similarly to the above case.
\end{enumerate}
\end{proof}

\begin{theorem}[Unique solution of equations]
Assume $\widetilde{E}$ are expressions containing only $X_i$ with $\widetilde{x}_i$, and each $X_i$ is weakly guarded in each $E_j$. Assume that $\widetilde{P}$ and $\widetilde{Q}$ are
processes such that $fn(P_i)\subseteq \widetilde{x}_i$ and $fn(Q_i)\subseteq \widetilde{x}_i$. Then, for all $i$,
\begin{enumerate}
  \item if $P_i\sim_p  E_i(\widetilde{P})$, $Q_i\sim_p  E_i(\widetilde{Q})$, then $P_i\sim_p  Q_i$;
  \item if $P_i\sim_s  E_i(\widetilde{P})$, $Q_i\sim_s  E_i(\widetilde{Q})$, then $P_i\sim_s  Q_i$;
  \item if $P_i\sim_{hp}  E_i(\widetilde{P})$, $Q_i\sim_{hp}  E_i(\widetilde{Q})$, then $P_i\sim_{hp}  Q_i$;
  \item if $P_i\sim_{hhp}  E_i(\widetilde{P})$, $Q_i\sim_{hhp}  E_i(\widetilde{Q})$, then $P_i\sim_{hhp}  Q_i$.
\end{enumerate}
\end{theorem}

\begin{proof}
\begin{enumerate}
  \item It is similar to the proof of unique solution of equations for strongly pomset bisimulation in CTC, please refer to \cite{CTC2} for details, we omit it;
  \item It is similar to the proof of unique solution of equations for strongly step bisimulation in CTC, please refer to \cite{CTC2} for details, we omit it;
  \item It is similar to the proof of unique solution of equations for strongly hp-bisimulation in CTC, please refer to \cite{CTC2} for details, we omit it;
  \item It is similar to the proof of unique solution of equations for strongly hhp-bisimulation in CTC, please refer to \cite{CTC2} for details, we omit it.
\end{enumerate}
\end{proof}

\subsection{Algebraic Theory}\label{a3}

\begin{definition}[STC]
The theory \textbf{STC} is consisted of the following axioms and inference rules:

\begin{enumerate}
  \item Alpha-conversion $\textbf{A}$.
  \[\textrm{if } P\equiv Q, \textrm{ then } P=Q\]
  \item Congruence $\textbf{C}$. If $P=Q$, then,
  \[\tau.P=\tau.Q\quad \overline{x}y.P=\overline{x}y.Q\]
  \[P+R=Q+R\quad P\parallel R=Q\parallel R\]
  \[(x)P=(x)Q\quad x(y).P=x(y).Q\]
  \item Summation $\textbf{S}$.
  \[\textbf{S0}\quad P+\textbf{nil}=P\]
  \[\textbf{S1}\quad P+P=P\]
  \[\textbf{S2}\quad P+Q=Q+P\]
  \[\textbf{S3}\quad P+(Q+R)=(P+Q)+R\]
  \item Restriction $\textbf{R}$.
  \[\textbf{R0}\quad (x)P=P\quad \textrm{ if }x\notin fn(P)\]
  \[\textbf{R1}\quad (x)(y)P=(y)(x)P\]
  \[\textbf{R2}\quad (x)(P+Q)=(x)P+(x)Q\]
  \[\textbf{R3}\quad (x)\alpha.P=\alpha.(x)P\quad \textrm{ if }x\notin n(\alpha)\]
  \[\textbf{R4}\quad (x)\alpha.P=\textbf{nil}\quad \textrm{ if }x\textrm{is the subject of }\alpha\]
  \item Expansion $\textbf{E}$.
  Let $P\equiv\sum_i \alpha_{i}.P_{i}$ and $Q\equiv\sum_j\beta_{j}.Q_{j}$, where $bn(\alpha_{i})\cap fn(Q)=\emptyset$ for all $i$, and
  $bn(\beta_{j})\cap fn(P)=\emptyset$ for all $j$. Then,

\begin{enumerate}
  \item $P\parallel Q\sim_p  \sum_i\sum_j (\alpha_{i}\parallel \beta_{j}).(P_{i}\parallel Q_{j})+\sum_{\alpha_{i} \textrm{ comp }\beta_{j}}\tau.R_{ij}$;
  \item $P\parallel Q\sim_s  \sum_i\sum_j (\alpha_{i}\parallel \beta_{j}).(P_{i}\parallel Q_{j})+\sum_{\alpha_{i} \textrm{ comp }\beta_{j}}\tau.R_{ij}$;
  \item $P\parallel Q\sim_{hp}  \sum_i\sum_j (\alpha_{i}\parallel \beta_{j}).(P_{i}\parallel Q_{j})+\sum_{\alpha_{i} \textrm{ comp }\beta_{j}}\tau.R_{ij}$;
  \item $P\parallel Q\nsim_{phhp} \sum_i\sum_j (\alpha_{i}\parallel \beta_{j}).(P_{i}\parallel Q_{j})+\sum_{\alpha_{i} \textrm{ comp }\beta_{j}}\tau.R_{ij}$.
\end{enumerate}

Where $\alpha_i$ comp $\beta_j$ and $R_{ij}$ are defined as follows:
\begin{enumerate}
  \item $\alpha_{i}$ is $\overline{x}u$ and $\beta_{j}$ is $x(v)$, then $R_{ij}=P_{i}\parallel Q_{j}\{u/v\}$;
  \item $\alpha_{i}$ is $\overline{x}(u)$ and $\beta_{j}$ is $x(v)$, then $R_{ij}=(w)(P_{i}\{w/u\}\parallel Q_{j}\{w/v\})$, if $w\notin fn((u)P_{i})\cup fn((v)Q_{j})$;
  \item $\alpha_{i}$ is $x(v)$ and $\beta_{j}$ is $\overline{x}u$, then $R_{ij}=P_{i}\{u/v\}\parallel Q_{j}$;
  \item $\alpha_{i}$ is $x(v)$ and $\beta_{j}$ is $\overline{x}(u)$, then $R_{ij}=(w)(P_{i}\{w/v\}\parallel Q_{j}\{w/u\})$, if $w\notin fn((v)P_{i})\cup fn((u)Q_{j})$.
\end{enumerate}
  \item Identifier $\textbf{I}$.
  \[\textrm{If }A(\widetilde{x})\overset{\text{def}}{=}P,\textrm{ then }A(\widetilde{y})= P\{\widetilde{y}/\widetilde{x}\}.\]
\end{enumerate}
\end{definition}

\begin{theorem}[Soundness]
If $\textbf{STC}\vdash P=Q$ then
\begin{enumerate}
  \item $P\sim_p  Q$;
  \item $P\sim_s  Q$;
  \item $P\sim_{hp}  Q$;
  \item $P\sim_{hhp}  Q$.
\end{enumerate}
\end{theorem}

\begin{proof}
The soundness of these laws modulo strongly truly concurrent bisimilarities is already proven in Section \ref{s3}.
\end{proof}

\begin{definition}
The agent identifier $A$ is weakly guardedly defined if every agent identifier is weakly guarded in the right-hand side of the definition of $A$.
\end{definition}

\begin{definition}[Head normal form]
A Process $P$ is in head normal form if it is a sum of the prefixes:

$$P\equiv \sum_i(\alpha_{i1}\parallel\cdots\parallel\alpha_{in}).P_{i}$$
\end{definition}

\begin{proposition}
If every agent identifier is weakly guardedly defined, then for any process $P$, there is a head normal form $H$ such that

$$\textbf{STC}\vdash P=H.$$
\end{proposition}

\begin{proof}
It is sufficient to induct on the structure of $P$ and quite obvious.
\end{proof}

\begin{theorem}[Completeness]
For all processes $P$ and $Q$,
\begin{enumerate}
  \item if $P\sim_p  Q$, then $\textbf{STC}\vdash P=Q$;
  \item if $P\sim_s  Q$, then $\textbf{STC}\vdash P=Q$;
  \item if $P\sim_{hp}  Q$, then $\textbf{STC}\vdash P=Q$.
\end{enumerate}
\end{theorem}

\begin{proof}
\begin{enumerate}
  \item if $P\sim_s  Q$, then $\textbf{STC}\vdash P=Q$.

Since $P$ and $Q$ all have head normal forms, let $P\equiv\sum_{i=1}^k\alpha_{i}.P_{i}$ and $Q\equiv\sum_{i=1}^k\beta_{i}.Q_{i}$. Then the depth of
$P$, denoted as $d(P)=0$, if $k=0$; $d(P)=1+max\{d(P_{i})\}$ for $1\leq j,i\leq k$. The depth $d(Q)$ can be defined similarly.

It is sufficient to induct on $d=d(P)+d(Q)$. When $d=0$, $P\equiv\textbf{nil}$ and $Q\equiv\textbf{nil}$, $P=Q$, as desired.

Suppose $d>0$.

\begin{itemize}
  \item If $(\alpha_1\parallel\cdots\parallel\alpha_n).M$ with $\alpha_{i}(1\leq i\leq n)$ free actions is a summand of $P$, then
  $\langle P,s\rangle\xrightarrow{\{\alpha_1,\cdots,\alpha_n\}}\langle M,s'\rangle$.
  Since $Q$ is in head normal form and has a summand $(\alpha_1\parallel\cdots\parallel\alpha_n).N$ such that $M\sim_s  N$, by the induction hypothesis $\textbf{STC}\vdash M=N$,
  $\textbf{STC}\vdash (\alpha_1\parallel\cdots\parallel\alpha_n).M= (\alpha_1\parallel\cdots\parallel\alpha_n).N$;
  \item If $x(y).M$ is a summand of $P$, then for $z\notin n(P, Q)$, $\langle P,s\rangle\xrightarrow{x(z)}\langle M',s'\rangle\equiv \langle M\{z/y\},s'\rangle$. Since $Q$ is in head normal form and has a summand
  $x(w).N$ such that for all $v$, $M'\{v/z\}\sim_s  N'\{v/z\}$ where $N'\equiv N\{z/w\}$, by the induction hypothesis $\textbf{STC}\vdash M'\{v/z\}=N'\{v/z\}$, by the axioms
  $\textbf{C}$ and $\textbf{A}$, $\textbf{STC}\vdash x(y).M=x(w).N$;
  \item If $\overline{x}(y).M$ is a summand of $P$, then for $z\notin n(P,Q)$, $\langle P,s\rangle\xrightarrow{\overline{x}(z)}\langle M',s'\rangle\equiv \langle M\{z/y\},s'\rangle$. Since $Q$ is in head normal form and
  has a summand $\overline{x}(w).N$ such that $M'\sim_s  N'$ where $N'\equiv N\{z/w\}$, by the induction hypothesis $\textbf{STC}\vdash M'=N'$, by the axioms
  $\textbf{A}$ and $\textbf{C}$, $\textbf{STC}\vdash \overline{x}(y).M= \overline{x}(w).N$.
\end{itemize}
  \item if $P\sim_p  Q$, then $\textbf{STC}\vdash P=Q$. It can be proven similarly to the above case.
  \item if $P\sim_{hp}  Q$, then $\textbf{STC}\vdash P=Q$. It can be proven similarly to the above case.
\end{enumerate}
\end{proof}

\newpage\section{$\pi_{tc}$ with Reversibility}\label{pitcr}

In this chapter, we design $\pi_{tc}$ with reversibility. This chapter is organized as follows. In section \ref{os4}, we introduce the truly concurrent operational semantics. Then, we introduce
the syntax and operational semantics, laws modulo strongly truly concurrent bisimulations, and algebraic theory of $\pi_{tc}$ with reversibility in section \ref{sos4},
\ref{s4} and \ref{a4} respectively.

\subsection{Operational Semantics}\label{os4}

Firstly, in this section, we introduce concepts of FR (strongly) truly concurrent bisimilarities, including FR pomset bisimilarity, FR step
bisimilarity, FR history-preserving (hp-)bisimilarity and FR hereditary history-preserving (hhp-)bisimilarity. In contrast to traditional FR truly
concurrent bisimilarities in section \ref{bg}, these versions in $\pi_{ptc}$ must take care of actions with bound objects. Note that, these FR truly concurrent bisimilarities
are defined as late bisimilarities, but not early bisimilarities, as defined in $\pi$-calculus \cite{PI1} \cite{PI2}. Note that, here, a PES $\mathcal{E}$ is deemed as a process.

\begin{definition}[Prime event structure with silent event]\label{PES}
Let $\Lambda$ be a fixed set of labels, ranged over $a,b,c,\cdots$ and $\tau$. A ($\Lambda$-labelled) prime event structure with silent event $\tau$ is a tuple
$\mathcal{E}=\langle \mathbb{E}, \leq, \sharp, \lambda\rangle$, where $\mathbb{E}$ is a denumerable set of events, including the silent event $\tau$. Let
$\hat{\mathbb{E}}=\mathbb{E}\backslash\{\tau\}$, exactly excluding $\tau$, it is obvious that $\hat{\tau^*}=\epsilon$, where $\epsilon$ is the empty event.
Let $\lambda:\mathbb{E}\rightarrow\Lambda$ be a labelling function and let $\lambda(\tau)=\tau$. And $\leq$, $\sharp$ are binary relations on $\mathbb{E}$,
called causality and conflict respectively, such that:

\begin{enumerate}
  \item $\leq$ is a partial order and $\lceil e \rceil = \{e'\in \mathbb{E}|e'\leq e\}$ is finite for all $e\in \mathbb{E}$. It is easy to see that
  $e\leq\tau^*\leq e'=e\leq\tau\leq\cdots\leq\tau\leq e'$, then $e\leq e'$.
  \item $\sharp$ is irreflexive, symmetric and hereditary with respect to $\leq$, that is, for all $e,e',e''\in \mathbb{E}$, if $e\sharp e'\leq e''$, then $e\sharp e''$.
\end{enumerate}

Then, the concepts of consistency and concurrency can be drawn from the above definition:

\begin{enumerate}
  \item $e,e'\in \mathbb{E}$ are consistent, denoted as $e\frown e'$, if $\neg(e\sharp e')$. A subset $X\subseteq \mathbb{E}$ is called consistent, if $e\frown e'$ for all
  $e,e'\in X$.
  \item $e,e'\in \mathbb{E}$ are concurrent, denoted as $e\parallel e'$, if $\neg(e\leq e')$, $\neg(e'\leq e)$, and $\neg(e\sharp e')$.
\end{enumerate}
\end{definition}

\begin{definition}[Configuration]
Let $\mathcal{E}$ be a PES. A (finite) configuration in $\mathcal{E}$ is a (finite) consistent subset of events $C\subseteq \mathcal{E}$, closed with respect to causality
(i.e. $\lceil C\rceil=C$). The set of finite configurations of $\mathcal{E}$ is denoted by $\mathcal{C}(\mathcal{E})$. We let $\hat{C}=C\backslash\{\tau\}$.
\end{definition}

A consistent subset of $X\subseteq \mathbb{E}$ of events can be seen as a pomset. Given $X, Y\subseteq \mathbb{E}$, $\hat{X}\sim \hat{Y}$ if $\hat{X}$ and $\hat{Y}$ are
isomorphic as pomsets. In the following of the paper, we say $C_1\sim C_2$, we mean $\hat{C_1}\sim\hat{C_2}$.

\begin{definition}[Pomset transitions and step]
Let $\mathcal{E}$ be a PES and let $C\in\mathcal{C}(\mathcal{E})$, and $\emptyset\neq X\subseteq \mathbb{E}$, if $C\cap X=\emptyset$ and $C'=C\cup X\in\mathcal{C}(\mathcal{E})$, then $C\xrightarrow{X} C'$ is called a pomset transition from $C$ to $C'$. When the events in $X$ are pairwise concurrent, we say that $C\xrightarrow{X}C'$ is a step.
\end{definition}

\begin{definition}[FR pomset transitions and step]
Let $\mathcal{E}$ be a PES and let $C\in\mathcal{C}(\mathcal{E})$, and $\emptyset\neq X\subseteq \mathbb{E}$, if $C\cap X=\emptyset$ and $C'=C\cup X\in\mathcal{C}(\mathcal{E})$, then
$ C\xrightarrow{X}  C'$ is called a forward pomset transition from $ C$ to $ C'$ and
$ C'\xtworightarrow{X[\mathcal{K}]}  C$ is called a reverse pomset transition from $ C'$ to $ C$. When the events in
$X$ and $X[\mathcal{K}]$ are pairwise
concurrent, we say that $ C\xrightarrow{X} C'$ is a forward step and $ C'\xrightarrow{X[\mathcal{K}]} C$ is a reverse step.
It is obvious that $\rightarrow^*\xrightarrow{X}\rightarrow^*=\xrightarrow{X}$ and
$\rightarrow^*\xrightarrow{e}\rightarrow^*=\xrightarrow{e}$ for any $e\in\mathbb{E}$ and $X\subseteq\mathbb{E}$.
\end{definition}

\begin{definition}[FR strongly pomset, step bisimilarity]
Let $\mathcal{E}_1$, $\mathcal{E}_2$ be PESs. A FR strongly pomset bisimulation is a relation $R\subseteq\mathcal{C}(\mathcal{E}_1)\times\mathcal{C}(\mathcal{E}_2)$,
such that (1) if $( C_1, C_2)\in R$, and $ C_1\xrightarrow{X_1} C_1'$ (with $\mathcal{E}_1\xrightarrow{X_1}\mathcal{E}_1'$) then $ C_2\xrightarrow{X_2} C_2'$ (with
$\mathcal{E}_2\xrightarrow{X_2}\mathcal{E}_2'$), with $X_1\subseteq \mathbb{E}_1$, $X_2\subseteq \mathbb{E}_2$, $X_1\sim X_2$ and $( C_1', C_2')\in R$:
\begin{enumerate}
  \item for each fresh action $\alpha\in X_1$, if $ C_1''\xrightarrow{\alpha} C_1'''$ (with $\mathcal{E}_1''\xrightarrow{\alpha}\mathcal{E}_1'''$),
  then for some $C_2''$ and $ C_2'''$, $ C_2''\xrightarrow{\alpha} C_2'''$ (with
  $\mathcal{E}_2''\xrightarrow{\alpha}\mathcal{E}_2'''$), such that if $( C_1'', C_2'')\in R$ then $( C_1''', C_2''')\in R$;
  \item for each $x(y)\in X_1$ with ($y\notin n(\mathcal{E}_1, \mathcal{E}_2)$), if $ C_1''\xrightarrow{x(y)} C_1'''$ (with
  $\mathcal{E}_1''\xrightarrow{x(y)}\mathcal{E}_1'''\{w/y\}$) for all $w$, then for some $C_2''$ and $C_2'''$, $ C_2''\xrightarrow{x(y)} C_2'''$
  (with $\mathcal{E}_2''\xrightarrow{x(y)}\mathcal{E}_2'''\{w/y\}$) for all $w$, such that if $( C_1'', C_2'')\in R$ then $( C_1''', C_2''')\in R$;
  \item for each two $x_1(y),x_2(y)\in X_1$ with ($y\notin n(\mathcal{E}_1, \mathcal{E}_2)$), if $ C_1''\xrightarrow{\{x_1(y),x_2(y)\}} C_1'''$
  (with $\mathcal{E}_1''\xrightarrow{\{x_1(y),x_2(y)\}}\mathcal{E}_1'''\{w/y\}$) for all $w$, then for some $C_2''$ and $C_2'''$,
  $ C_2''\xrightarrow{\{x_1(y),x_2(y)\}} C_2'''$ (with $\mathcal{E}_2''\xrightarrow{\{x_1(y),x_2(y)\}}\mathcal{E}_2'''\{w/y\}$) for all $w$, such
  that if $( C_1'', C_2'')\in R$ then $( C_1''', C_2''')\in R$;
  \item for each $\overline{x}(y)\in X_1$ with $y\notin n(\mathcal{E}_1, \mathcal{E}_2)$, if $ C_1''\xrightarrow{\overline{x}(y)} C_1'''$
  (with $\mathcal{E}_1''\xrightarrow{\overline{x}(y)}\mathcal{E}_1'''$), then for some $C_2''$ and $C_2'''$, $ C_2''\xrightarrow{\overline{x}(y)} C_2'''$
  (with $\mathcal{E}_2''\xrightarrow{\overline{x}(y)}\mathcal{E}_2'''$), such that if $( C_1'', C_2'')\in R$ then $( C_1''', C_2''')\in R$.
\end{enumerate}
 and vice-versa; (2) if $( C_1, C_2)\in R$, and $ C_1\xtworightarrow{X_1[\mathcal{K}_1]} C_1'$ (with $\mathcal{E}_1\xtworightarrow{X_1[\mathcal{K}_1]}\mathcal{E}_1'$) then $ C_2\xtworightarrow{X_2[\mathcal{K}_2]} C_2'$ (with
$\mathcal{E}_2\xtworightarrow{X_2[\mathcal{K}_2]}\mathcal{E}_2'$), with $X_1\subseteq \mathbb{E}_1$, $X_2\subseteq \mathbb{E}_2$, $X_1\sim X_2$ and $( C_1', C_2')\in R$:
\begin{enumerate}
  \item for each fresh action $\alpha\in X_1$, if $ C_1''\xtworightarrow{\alpha[m]} C_1'''$ (with $\mathcal{E}_1''\xtworightarrow{\alpha[m]}\mathcal{E}_1'''$),
  then for some $C_2''$ and $ C_2'''$, $ C_2''\xtworightarrow{\alpha[m]} C_2'''$ (with
  $\mathcal{E}_2''\xtworightarrow{\alpha[m]}\mathcal{E}_2'''$), such that if $( C_1'', C_2'')\in R$ then $( C_1''', C_2''')\in R$;
  \item for each $x(y)\in X_1$ with ($y\notin n(\mathcal{E}_1, \mathcal{E}_2)$), if $ C_1''\xtworightarrow{x(y)[m]} C_1'''$ (with
  $\mathcal{E}_1''\xtworightarrow{x(y)[m]}\mathcal{E}_1'''\{w/y\}$) for all $w$, then for some $C_2''$ and $C_2'''$, $ C_2''\xtworightarrow{x(y)[m]} C_2'''$
  (with $\mathcal{E}_2''\xtworightarrow{x(y)[m]}\mathcal{E}_2'''\{w/y\}$) for all $w$, such that if $( C_1'', C_2'')\in R$ then $( C_1''', C_2''')\in R$;
  \item for each two $x_1(y),x_2(y)\in X_1$ with ($y\notin n(\mathcal{E}_1, \mathcal{E}_2)$), if $ C_1''\xtworightarrow{\{x_1(y)[m],x_2(y)[n]\}} C_1'''$
  (with $\mathcal{E}_1''\xtworightarrow{\{x_1(y)[m],x_2(y)[n]\}}\mathcal{E}_1'''\{w/y\}$) for all $w$, then for some $C_2''$ and $C_2'''$,
  $ C_2''\xtworightarrow{\{x_1(y)[m],x_2(y)[n]\}} C_2'''$ (with $\mathcal{E}_2''\xtworightarrow{\{x_1(y)[m],x_2(y)[n]\}}\mathcal{E}_2'''\{w/y\}$) for all $w$, such
  that if $( C_1'', C_2'')\in R$ then $( C_1''', C_2''')\in R$;
  \item for each $\overline{x}(y)\in X_1$ with $y\notin n(\mathcal{E}_1, \mathcal{E}_2)$, if $ C_1''\xtworightarrow{\overline{x}(y)[m]} C_1'''$
  (with $\mathcal{E}_1''\xtworightarrow{\overline{x}(y)[m]}\mathcal{E}_1'''$), then for some $C_2''$ and $C_2'''$, $ C_2''\xtworightarrow{\overline{x}(y)[m]} C_2'''$
  (with $\mathcal{E}_2''\xtworightarrow{\overline{x}(y)[m]}\mathcal{E}_2'''$), such that if $( C_1'', C_2'')\in R$ then $( C_1''', C_2''')\in R$.
\end{enumerate}
 and vice-versa.

We say that $\mathcal{E}_1$, $\mathcal{E}_2$ are FR strongly pomset bisimilar, written $\mathcal{E}_1\sim_p^{fr}\mathcal{E}_2$, if there exists a FR strongly pomset
bisimulation $R$, such that $(\emptyset,\emptyset)\in R$. By replacing FR pomset transitions with steps, we can get the definition of FR strongly step bisimulation.
When PESs $\mathcal{E}_1$ and $\mathcal{E}_2$ are FR strongly step bisimilar, we write $\mathcal{E}_1\sim_s^{fr}\mathcal{E}_2$.
\end{definition}

\begin{definition}[Posetal product]
Given two PESs $\mathcal{E}_1$, $\mathcal{E}_2$, the posetal product of their configurations, denoted
$\mathcal{C}(\mathcal{E}_1)\overline{\times}\mathcal{C}(\mathcal{E}_2)$, is defined as

$$\{( C_1,f, C_2)|C_1\in\mathcal{C}(\mathcal{E}_1),C_2\in\mathcal{C}(\mathcal{E}_2),f:C_1\rightarrow C_2 \textrm{ isomorphism}\}.$$

A subset $R\subseteq\mathcal{C}(\mathcal{E}_1)\overline{\times}\mathcal{C}(\mathcal{E}_2)$ is called a posetal relation. We say that $R$ is downward
closed when for any
$( C_1,f, C_2),( C_1',f', C_2')\in \mathcal{C}(\mathcal{E}_1)\overline{\times}\mathcal{C}(\mathcal{E}_2)$,
if $( C_1,f, C_2)\subseteq ( C_1',f', C_2')$ pointwise and $( C_1',f', C_2')\in R$,
then $( C_1,f, C_2)\in R$.

For $f:X_1\rightarrow X_2$, we define $f[x_1\mapsto x_2]:X_1\cup\{x_1\}\rightarrow X_2\cup\{x_2\}$, $z\in X_1\cup\{x_1\}$,(1)$f[x_1\mapsto x_2](z)=
x_2$,if $z=x_1$;(2)$f[x_1\mapsto x_2](z)=f(z)$, otherwise. Where $X_1\subseteq \mathbb{E}_1$, $X_2\subseteq \mathbb{E}_2$, $x_1\in \mathbb{E}_1$, $x_2\in \mathbb{E}_2$.
\end{definition}

\begin{definition}[FR strongly (hereditary) history-preserving bisimilarity]
A FR strongly history-preserving (hp-) bisimulation is a posetal relation $R\subseteq\mathcal{C}(\mathcal{E}_1)\overline{\times}\mathcal{C}(\mathcal{E}_2)$ such that
(1) if $( C_1,f, C_2)\in R$, and
\begin{enumerate}
  \item for $e_1=\alpha$ a fresh action, if $ C_1\xrightarrow{\alpha} C_1'$ (with $\mathcal{E}_1\xrightarrow{\alpha}\mathcal{E}_1'$), then for some
  $C_2'$ and $e_2=\alpha$, $ C_2\xrightarrow{\alpha} C_2'$ (with $\mathcal{E}_2\xrightarrow{\alpha}\mathcal{E}_2'$), such that
  $( C_1',f[e_1\mapsto e_2], C_2')\in R$;
  \item for $e_1=x(y)$ with ($y\notin n(\mathcal{E}_1, \mathcal{E}_2)$), if $ C_1\xrightarrow{x(y)} C_1'$ (with
  $\mathcal{E}_1\xrightarrow{x(y)}\mathcal{E}_1'\{w/y\}$) for all $w$, then for some $C_2'$ and $e_2=x(y)$, $ C_2\xrightarrow{x(y)} C_2'$ (with
  $\mathcal{E}_2\xrightarrow{x(y)}\mathcal{E}_2'\{w/y\}$) for all $w$, such that $( C_1',f[e_1\mapsto e_2], C_2')\in R$;
  \item for $e_1=\overline{x}(y)$ with $y\notin n(\mathcal{E}_1, \mathcal{E}_2)$, if $ C_1\xrightarrow{\overline{x}(y)} C_1'$ (with
  $\mathcal{E}_1\xrightarrow{\overline{x}(y)}\mathcal{E}_1'$), then for some $C_2'$ and $e_2=\overline{x}(y)$, $ C_2\xrightarrow{\overline{x}(y)} C_2'$
  (with $\mathcal{E}_2\xrightarrow{\overline{x}(y)}\mathcal{E}_2'$), such that $( C_1',f[e_1\mapsto e_2], C_2')\in R$.
\end{enumerate}
and vice-versa; (2) if $( C_1,f, C_2)\in R$, and
\begin{enumerate}
  \item for $e_1=\alpha$ a fresh action, if $ C_1\xtworightarrow{\alpha[m]} C_1'$ (with $\mathcal{E}_1\xtworightarrow{\alpha[m]}\mathcal{E}_1'$), then for some
  $C_2'$ and $e_2=\alpha$, $ C_2\xtworightarrow{\alpha[m]} C_2'$ (with $\mathcal{E}_2\xtworightarrow{\alpha[m]}\mathcal{E}_2'$), such that
  $( C_1',f[e_1\mapsto e_2], C_2')\in R$;
  \item for $e_1=x(y)$ with ($y\notin n(\mathcal{E}_1, \mathcal{E}_2)$), if $ C_1\xtworightarrow{x(y)[m]} C_1'$ (with
  $\mathcal{E}_1\xtworightarrow{x(y)[m]}\mathcal{E}_1'\{w/y\}$) for all $w$, then for some $C_2'$ and $e_2=x(y)$, $ C_2\xtworightarrow{x(y)[m]} C_2'$ (with
  $\mathcal{E}_2\xtworightarrow{x(y)[m]}\mathcal{E}_2'\{w/y\}$) for all $w$, such that $( C_1',f[e_1\mapsto e_2], C_2')\in R$;
  \item for $e_1=\overline{x}(y)$ with $y\notin n(\mathcal{E}_1, \mathcal{E}_2)$, if $ C_1\xtworightarrow{\overline{x}(y)[m]} C_1'$ (with
  $\mathcal{E}_1\xtworightarrow{\overline{x}(y)[m]}\mathcal{E}_1'$), then for some $C_2'$ and $e_2=\overline{x}(y)$, $ C_2\xtworightarrow{\overline{x}(y)[m]} C_2'$
  (with $\mathcal{E}_2\xtworightarrow{\overline{x}(y)[m]}\mathcal{E}_2'$), such that $( C_1',f[e_1\mapsto e_2], C_2')\in R$.
\end{enumerate}
and vice-versa. $\mathcal{E}_1,\mathcal{E}_2$
are FR strongly history-preserving (hp-)bisimilar and are written $\mathcal{E}_1\sim_{hp}^{fr}\mathcal{E}_2$ if there exists a FR strongly hp-bisimulation
$R$ such that $(\emptyset,\emptyset,\emptyset)\in R$.

A FR strongly hereditary history-preserving (hhp-)bisimulation is a downward closed FR strongly hp-bisimulation. $\mathcal{E}_1,\mathcal{E}_2$ are FR
strongly hereditary history-preserving (hhp-)bisimilar and are written $\mathcal{E}_1\sim_{hhp}^{fr}\mathcal{E}_2$.
\end{definition}

\subsection{Syntax and Operational Semantics}\label{sos4}

We assume an infinite set $\mathcal{N}$ of (action or event) names, and use $a,b,c,\cdots$ to range over $\mathcal{N}$, use $x,y,z,w,u,v$ as meta-variables over names. We denote by
$\overline{\mathcal{N}}$ the set of co-names and let $\overline{a},\overline{b},\overline{c},\cdots$ range over $\overline{\mathcal{N}}$. Then we set
$\mathcal{L}=\mathcal{N}\cup\overline{\mathcal{N}}$ as the set of labels, and use $l,\overline{l}$ to range over $\mathcal{L}$. We extend complementation to $\mathcal{L}$ such that
$\overline{\overline{a}}=a$. Let $\tau$ denote the silent step (internal action or event) and define $Act=\mathcal{L}\cup\{\tau\}$ to be the set of actions, $\alpha,\beta$ range over
$Act$. And $K,L$ are used to stand for subsets of $\mathcal{L}$ and $\overline{L}$ is used for the set of complements of labels in $L$.

Further, we introduce a set $\mathcal{X}$ of process variables, and a set $\mathcal{K}$ of process constants, and let $X,Y,\cdots$ range over $\mathcal{X}$, and $A,B,\cdots$ range over
$\mathcal{K}$. For each process constant $A$, a nonnegative arity $ar(A)$ is assigned to it. Let $\widetilde{x}=x_1,\cdots,x_{ar(A)}$ be a tuple of distinct name variables, then
$A(\widetilde{x})$ is called a process constant. $\widetilde{X}$ is a tuple of distinct process variables, and also $E,F,\cdots$ range over the recursive expressions. We write
$\mathcal{P}$ for the set of processes. Sometimes, we use $I,J$ to stand for an indexing set, and we write $E_i:i\in I$ for a family of expressions indexed by $I$. $Id_D$ is the
identity function or relation over set $D$. The symbol $\equiv_{\alpha}$ denotes equality under standard alpha-convertibility, note that the subscript $\alpha$ has no relation to the
action $\alpha$.

\subsubsection{Syntax}

We use the Prefix $.$ to model the causality relation $\leq$ in true concurrency, the Summation $+$ to model the conflict relation $\sharp$ in true concurrency, and the Composition $\parallel$ to explicitly model concurrent relation in true concurrency. And we follow the
conventions of process algebra.

\begin{definition}[Syntax]\label{syntax4}
A truly concurrent process $\pi_{tc}$ with reversibility is defined inductively by the following formation rules:

\begin{enumerate}
  \item $A(\widetilde{x})\in\mathcal{P}$;
  \item $\textbf{nil}\in\mathcal{P}$;
  \item if $P\in\mathcal{P}$, then the Prefix $\tau.P\in\mathcal{P}$, for $\tau\in Act$ is the silent action;
  \item if $P\in\mathcal{P}$, then the Output $\overline{x}y.P\in\mathcal{P}$, for $x,y\in Act$;
  \item if $P\in\mathcal{P}$, then the Output $P.\overline{x}y[m]\in\mathcal{P}$, for $x,y\in Act$;
  \item if $P\in\mathcal{P}$, then the Input $x(y).P\in\mathcal{P}$, for $x,y\in Act$;
  \item if $P\in\mathcal{P}$, then the Input $P.x(y)[m]\in\mathcal{P}$, for $x,y\in Act$;
  \item if $P\in\mathcal{P}$, then the Restriction $(x)P\in\mathcal{P}$, for $x\in Act$;
  \item if $P,Q\in\mathcal{P}$, then the Summation $P+Q\in\mathcal{P}$;
  \item if $P,Q\in\mathcal{P}$, then the Composition $P\parallel Q\in\mathcal{P}$;
\end{enumerate}

The standard BNF grammar of syntax of $\pi_{tc}$ with reversibility can be summarized as follows:

$$P::=A(\widetilde{x})|\textbf{nil}|\tau.P| \overline{x}y.P | x(y).P|\overline{x}y[m].P | x(y)[m].P | (x)P  | P+P| P\parallel P.$$
\end{definition}

In $\overline{x}y$, $x(y)$ and $\overline{x}(y)$, $x$ is called the subject, $y$ is called the object and it may be free or bound.

\begin{definition}[Free variables]
The free names of a process $P$, $fn(P)$, are defined as follows.

\begin{enumerate}
  \item $fn(A(\widetilde{x}))\subseteq\{\widetilde{x}\}$;
  \item $fn(\textbf{nil})=\emptyset$;
  \item $fn(\tau.P)=fn(P)$;
  \item $fn(\overline{x}y.P)=fn(P)\cup\{x\}\cup\{y\}$;
  \item $fn(\overline{x}y[m].P)=fn(P)\cup\{x\}\cup\{y\}$;
  \item $fn(x(y).P)=fn(P)\cup\{x\}-\{y\}$;
  \item $fn(x(y)[m].P)=fn(P)\cup\{x\}-\{y\}$;
  \item $fn((x)P)=fn(P)-\{x\}$;
  \item $fn(P+Q)=fn(P)\cup fn(Q)$;
  \item $fn(P\parallel Q)=fn(P)\cup fn(Q)$.
\end{enumerate}
\end{definition}

\begin{definition}[Bound variables]
Let $n(P)$ be the names of a process $P$, then the bound names $bn(P)=n(P)-fn(P)$.
\end{definition}

For each process constant schema $A(\widetilde{x})$, a defining equation of the form

$$A(\widetilde{x})\overset{\text{def}}{=}P$$

is assumed, where $P$ is a process with $fn(P)\subseteq \{\widetilde{x}\}$.

\begin{definition}[Substitutions]\label{subs4}
A substitution is a function $\sigma:\mathcal{N}\rightarrow\mathcal{N}$. For $x_i\sigma=y_i$ with $1\leq i\leq n$, we write $\{y_1/x_1,\cdots,y_n/x_n\}$ or
$\{\widetilde{y}/\widetilde{x}\}$ for $\sigma$. For a process $P\in\mathcal{P}$, $P\sigma$ is defined inductively as follows:
\begin{enumerate}
  \item if $P$ is a process constant $A(\widetilde{x})=A(x_1,\cdots,x_n)$, then $P\sigma=A(x_1\sigma,\cdots,x_n\sigma)$;
  \item if $P=\textbf{nil}$, then $P\sigma=\textbf{nil}$;
  \item if $P=\tau.P'$, then $P\sigma=\tau.P'\sigma$;
  \item if $P=\overline{x}y.P'$, then $P\sigma=\overline{x\sigma}y\sigma.P'\sigma$;
  \item if $P=\overline{x}y[m].P'$, then $P\sigma=\overline{x\sigma}y\sigma[m].P'\sigma$;
  \item if $P=x(y).P'$, then $P\sigma=x\sigma(y).P'\sigma$;
  \item if $P=x(y)[m].P'$, then $P\sigma=x\sigma(y)[m].P'\sigma$;
  \item if $P=(x)P'$, then $P\sigma=(x\sigma)P'\sigma$;
  \item if $P=P_1+P_2$, then $P\sigma=P_1\sigma+P_2\sigma$;
  \item if $P=P_1\parallel P_2$, then $P\sigma=P_1\sigma \parallel P_2\sigma$.
\end{enumerate}
\end{definition}

\subsubsection{Operational Semantics}

The operational semantics is defined by LTSs (labelled transition systems), and it is detailed by the following definition.

\begin{definition}[Semantics]\label{semantics4}
The operational semantics of $\pi_{tc}$ with reversibility corresponding to the syntax in Definition \ref{syntax4} is defined by a series of transition rules, named $\textbf{ACT}$, $\textbf{SUM}$,
$\textbf{IDE}$, $\textbf{PAR}$, $\textbf{COM}$, $\textbf{CLOSE}$, $\textbf{RES}$, $\textbf{OPEN}$ indicate that the rules are associated respectively with Prefix, Summation,
Identity, Parallel Composition, Communication, and Restriction in Definition \ref{syntax4}. They are shown in \ref{TRForPITC4}.

\begin{center}
    \begin{table}
        \[\textbf{TAU-ACT}\quad \frac{}{\tau.P\xrightarrow{\tau}P}\]

        \[\textbf{OUTPUT-ACT}\quad \frac{}{\overline{x}y.P\xrightarrow{\overline{x}y}P}\]

        \[\textbf{INPUT-ACT}\quad \frac{}{x(z).P\xrightarrow{x(w)}P\{w/z\}}\quad (w\notin fn((z)P))\]

        \[\textbf{PAR}_1\quad \frac{P\xrightarrow{\alpha}P'\quad Q\nrightarrow}{P\parallel Q\xrightarrow{\alpha}P'\parallel Q}\quad (bn(\alpha)\cap fn(Q)=\emptyset)\]

        \[\textbf{PAR}_2\quad \frac{Q\xrightarrow{\alpha}Q'\quad P\nrightarrow}{P\parallel Q\xrightarrow{\alpha}P\parallel Q'}\quad (bn(\alpha)\cap fn(P)=\emptyset)\]

        \[\textbf{PAR}_3\quad \frac{P\xrightarrow{\alpha}P'\quad Q\xrightarrow{\beta}Q'}{P\parallel Q\xrightarrow{\{\alpha,\beta\}}P'\parallel Q'}\] $(\beta\neq\overline{\alpha}, bn(\alpha)\cap bn(\beta)=\emptyset, bn(\alpha)\cap fn(Q)=\emptyset,bn(\beta)\cap fn(P)=\emptyset)$

        \[\textbf{PAR}_4\quad \frac{P\xrightarrow{x_1(z)}P'\quad Q\xrightarrow{x_2(z)}Q'}{P\parallel Q\xrightarrow{\{x_1(w),x_2(w)\}}P'\{w/z\}\parallel Q'\{w/z\}}\quad (w\notin fn((z)P)\cup fn((z)Q))\]

        \[\textbf{COM}\quad \frac{P\xrightarrow{\overline{x}y}P'\quad Q\xrightarrow{x(z)}Q'}{P\parallel Q\xrightarrow{\tau}P'\parallel Q'\{y/z\}}\]

        \[\textbf{CLOSE}\quad \frac{P\xrightarrow{\overline{x}(w)}P'\quad Q\xrightarrow{x(w)}Q'}{P\parallel Q\xrightarrow{\tau}(w)(P'\parallel Q')}\]

%
%
%
%
%
%
%
%
        \caption{Forward transition rules}
        \label{TRForPITC4}
    \end{table}
\end{center}

\begin{center}
    \begin{table}
%
%
%
%
%
%
%
%
%
        \[\textbf{SUM}_1\quad \frac{P\xrightarrow{\alpha}P'}{P+Q\xrightarrow{\alpha}P'}\]

        \[\textbf{SUM}_2\quad \frac{P\xrightarrow{\{\alpha_1,\cdots,\alpha_n\}}P'}{P+Q\xrightarrow{\{\alpha_1,\cdots,\alpha_n\}}P'}\]

        \[\textbf{IDE}_1\quad\frac{P\{\widetilde{y}/\widetilde{x}\}\xrightarrow{\alpha}P'}{A(\widetilde{y})\xrightarrow{\alpha}P'}\quad (A(\widetilde{x})\overset{\text{def}}{=}P)\]

        \[\textbf{IDE}_2\quad\frac{P\{\widetilde{y}/\widetilde{x}\}\xrightarrow{\{\alpha_1,\cdots,\alpha_n\}}P'} {A(\widetilde{y})\xrightarrow{\{\alpha_1,\cdots,\alpha_n\}}P'}\quad (A(\widetilde{x})\overset{\text{def}}{=}P)\]

        \[\textbf{RES}_1\quad \frac{P\xrightarrow{\alpha}P'}{(y)P\xrightarrow{\alpha}(y)P'}\quad (y\notin n(\alpha))\]

        \[\textbf{RES}_2\quad \frac{P\xrightarrow{\{\alpha_1,\cdots,\alpha_n\}}P'}{(y)P\xrightarrow{\{\alpha_1,\cdots,\alpha_n\}}(y)P'}\quad (y\notin n(\alpha_1)\cup\cdots\cup n(\alpha_n))\]

        \[\textbf{OPEN}_1\quad \frac{P\xrightarrow{\overline{x}y}P'}{(y)P\xrightarrow{\overline{x}(w)}P'\{w/y\}} \quad (y\neq x, w\notin fn((y)P'))\]

        \[\textbf{OPEN}_2\quad \frac{P\xrightarrow{\{\overline{x}_1 y,\cdots,\overline{x}_n y\}}P'}{(y)P\xrightarrow{\{\overline{x}_1(w),\cdots,\overline{x}_n(w)\}}P'\{w/y\}} \quad (y\neq x_1\neq\cdots\neq x_n, w\notin fn((y)P'))\]

        \caption{Forward transition rules (continuing)}
        \label{TRForPITC42}
    \end{table}
\end{center}

\begin{center}
    \begin{table}
        \[\textbf{RTAU-ACT}\quad \frac{}{\tau.P\xtworightarrow{\tau}P}\]

        \[\textbf{ROUTPUT-ACT}\quad \frac{}{\overline{x}y[m].P\xtworightarrow{\overline{x}y[m]}P}\]

        \[\textbf{RINPUT-ACT}\quad \frac{}{x(z)[m].P\xtworightarrow{x(w)[m]}P\{w/z\}}\quad (w\notin fn((z)P))\]

        \[\textbf{RPAR}_1\quad \frac{P\xtworightarrow{\alpha[m]}P'\quad Q\nrightarrow}{P\parallel Q\xtworightarrow{\alpha[m]}P'\parallel Q}\quad (bn(\alpha)\cap fn(Q)=\emptyset)\]

        \[\textbf{RPAR}_2\quad \frac{Q\xtworightarrow{\alpha[m]}Q'\quad P\nrightarrow}{P\parallel Q\xtworightarrow{\alpha[m]}P\parallel Q'}\quad (bn(\alpha)\cap fn(P)=\emptyset)\]

        \[\textbf{RPAR}_3\quad \frac{P\xtworightarrow{\alpha[m]}P'\quad Q\xtworightarrow{\beta[m]}Q'}{P\parallel Q\xtworightarrow{\{\alpha[m],\beta[m]\}}P'\parallel Q'}\] $(\beta\neq\overline{\alpha}, bn(\alpha)\cap bn(\beta)=\emptyset, bn(\alpha)\cap fn(Q)=\emptyset,bn(\beta)\cap fn(P)=\emptyset)$

        \[\textbf{RPAR}_4\quad \frac{P\xtworightarrow{x_1(z)[m]}P'\quad Q\xtworightarrow{x_2(z)[m]}Q'}{P\parallel Q\xtworightarrow{\{x_1(w)[m],x_2(w)[m]\}}P'\{w/z\}\parallel Q'\{w/z\}}\quad (w\notin fn((z)P)\cup fn((z)Q))\]

        \[\textbf{RCOM}\quad \frac{P\xtworightarrow{\overline{x}y[m]}P'\quad Q\xtworightarrow{x(z)[m]}Q'}{P\parallel Q\xtworightarrow{\tau}P'\parallel Q'\{y/z\}}\]

        \[\textbf{RCLOSE}\quad \frac{P\xtworightarrow{\overline{x}(w)[m]}P'\quad Q\xtworightarrow{x(w)[m]}Q'}{P\parallel Q\xtworightarrow{\tau}(w)(P'\parallel Q')}\]

%
%
%
%
%
%
%
%
        \caption{Reverse transition rules}
        \label{TRForPITC43}
    \end{table}
\end{center}

\begin{center}
    \begin{table}
%
%
%
%
%
%
%
%
%
        \[\textbf{RSUM}_1\quad \frac{P\xtworightarrow{\alpha[m]}P'}{P+Q\xtworightarrow{\alpha[m]}P'}\]

        \[\textbf{RSUM}_2\quad \frac{P\xtworightarrow{\{\alpha_1[m],\cdots,\alpha_n[m]\}}P'}{P+Q\xtworightarrow{\{\alpha_1[m],\cdots,\alpha_n[m]\}}P'}\]

        \[\textbf{RIDE}_1\quad\frac{P\{\widetilde{y}/\widetilde{x}\}\xtworightarrow{\alpha[m]}P'}{A(\widetilde{y})\xtworightarrow{\alpha[m]}P'}\quad (A(\widetilde{x})\overset{\text{def}}{=}P)\]

        \[\textbf{RIDE}_2\quad\frac{P\{\widetilde{y}/\widetilde{x}\}\xtworightarrow{\{\alpha_1[m],\cdots,\alpha_n[m]\}}P'} {A(\widetilde{y})\xtworightarrow{\{\alpha_1[m],\cdots,\alpha_n[m]\}}P'}\quad (A(\widetilde{x})\overset{\text{def}}{=}P)\]

        \[\textbf{RRES}_1\quad \frac{P\xtworightarrow{\alpha[m]}P'}{(y)P\xtworightarrow{\alpha[m]}(y)P'}\quad (y\notin n(\alpha))\]

        \[\textbf{RRES}_2\quad \frac{P\xtworightarrow{\{\alpha_1[m],\cdots,\alpha_n[m]\}}P'}{(y)P\xtworightarrow{\{\alpha_1[m],\cdots,\alpha_n[m]\}}(y)P'}\quad (y\notin n(\alpha_1)\cup\cdots\cup n(\alpha_n))\]

        \[\textbf{ROPEN}_1\quad \frac{P\xtworightarrow{\overline{x}y[m]}P'}{(y)P\xtworightarrow{\overline{x}(w)[m]}P'\{w/y\}} \quad (y\neq x, w\notin fn((y)P'))\]

        \[\textbf{ROPEN}_2\quad \frac{P\xtworightarrow{\{\overline{x}_1 y[m],\cdots,\overline{x}_n y[m]\}}P'}{(y)P\xtworightarrow{\{\overline{x}_1(w)[m],\cdots,\overline{x}_n(w)[m]\}}P'\{w/y\}} \quad (y\neq x_1\neq\cdots\neq x_n, w\notin fn((y)P'))\]

        \caption{Reverse transition rules (continuing)}
        \label{TRForPITC44}
    \end{table}
\end{center}
\end{definition}

\subsubsection{Properties of Transitions}

\begin{proposition}
\begin{enumerate}
  \item If $P\xrightarrow{\alpha}P'$ then
  \begin{enumerate}
    \item $fn(\alpha)\subseteq fn(P)$;
    \item $fn(P')\subseteq fn(P)\cup bn(\alpha)$;
  \end{enumerate}
  \item If $P\xrightarrow{\{\alpha_1,\cdots,\alpha_n\}}P'$ then
  \begin{enumerate}
    \item $fn(\alpha_1)\cup\cdots\cup fn(\alpha_n)\subseteq fn(P)$;
    \item $fn(P')\subseteq fn(P)\cup bn(\alpha_1)\cup\cdots\cup bn(\alpha_n)$.
  \end{enumerate}
\end{enumerate}
\end{proposition}

\begin{proof}
By induction on the depth of inference.
\end{proof}

\begin{proposition}
Suppose that $P\xrightarrow{\alpha(y)}P'$, where $\alpha=x$ or $\alpha=\overline{x}$, and $x\notin n(P)$, then there exists some $P''\equiv_{\alpha}P'\{z/y\}$,
$P\xrightarrow{\alpha(z)}P''$.
\end{proposition}

\begin{proof}
By induction on the depth of inference.
\end{proof}

\begin{proposition}
If $P\xrightarrow{\alpha} P'$, $bn(\alpha)\cap fn(P'\sigma)=\emptyset$, and $\sigma\lceil bn(\alpha)=id$, then there exists some $P''\equiv_{\alpha}P'\sigma$,
$P\sigma\xrightarrow{\alpha\sigma}P''$.
\end{proposition}

\begin{proof}
By the definition of substitution (Definition \ref{subs4}) and induction on the depth of inference.
\end{proof}

\begin{proposition}
\begin{enumerate}
  \item If $P\{w/z\}\xrightarrow{\alpha}P'$, where $w\notin fn(P)$ and $bn(\alpha)\cap fn(P,w)=\emptyset$, then there exist some $Q$ and $\beta$ with $Q\{w/z\}\equiv_{\alpha}P'$ and
  $\beta\sigma=\alpha$, $P\xrightarrow{\beta}Q$;
  \item If $P\{w/z\}\xrightarrow{\{\alpha_1,\cdots,\alpha_n\}}P'$, where $w\notin fn(P)$ and $bn(\alpha_1)\cap\cdots\cap bn(\alpha_n)\cap fn(P,w)=\emptyset$, then there exist some $Q$
  and $\beta_1,\cdots,\beta_n$ with $Q\{w/z\}\equiv_{\alpha}P'$ and $\beta_1\sigma=\alpha_1,\cdots,\beta_n\sigma=\alpha_n$, $P\xrightarrow{\{\beta_1,\cdots,\beta_n\}}Q$.
\end{enumerate}

\end{proposition}

\begin{proof}
By the definition of substitution (Definition \ref{subs4}) and induction on the depth of inference.
\end{proof}

\begin{proposition}
\begin{enumerate}
  \item If $P\xtworightarrow{\alpha[m]}P'$ then
  \begin{enumerate}
    \item $fn(\alpha[m])\subseteq fn(P)$;
    \item $fn(P')\subseteq fn(P)\cup bn(\alpha[m])$;
  \end{enumerate}
  \item If $P\xtworightarrow{\{\alpha_1[m],\cdots,\alpha_n[m]\}}P'$ then
  \begin{enumerate}
    \item $fn(\alpha_1[m])\cup\cdots\cup fn(\alpha_n[m])\subseteq fn(P)$;
    \item $fn(P')\subseteq fn(P)\cup bn(\alpha_1[m])\cup\cdots\cup bn(\alpha_n[m])$.
  \end{enumerate}
\end{enumerate}
\end{proposition}

\begin{proof}
By induction on the depth of inference.
\end{proof}

\begin{proposition}
Suppose that $P\xtworightarrow{\alpha(y)[m]}P'$, where $\alpha=x$ or $\alpha=\overline{x}$, and $x\notin n(P)$, then there exists some $P''\equiv_{\alpha}P'\{z/y\}$,
$P\xtworightarrow{\alpha(z)[m]}P''$.
\end{proposition}

\begin{proof}
By induction on the depth of inference.
\end{proof}

\begin{proposition}
If $P\xtworightarrow{\alpha[m]} P'$, $bn(\alpha[m])\cap fn(P'\sigma)=\emptyset$, and $\sigma\lceil bn(\alpha[m])=id$, then there exists some $P''\equiv_{\alpha}P'\sigma$,
$P\sigma\xtworightarrow{\alpha[m]\sigma}P''$.
\end{proposition}

\begin{proof}
By the definition of substitution (Definition \ref{subs4}) and induction on the depth of inference.
\end{proof}

\begin{proposition}
\begin{enumerate}
  \item If $P\{w/z\}\xtworightarrow{\alpha[m]}P'$, where $w\notin fn(P)$ and $bn(\alpha)\cap fn(P,w)=\emptyset$, then there exist some $Q$ and $\beta$ with $Q\{w/z\}\equiv_{\alpha}P'$ and
  $\beta\sigma[m]=\alpha[m]$, $P\xtworightarrow{\beta[m]}Q$;
  \item If $P\{w/z\}\xtworightarrow{\{\alpha_1[m],\cdots,\alpha_n[m]\}}P'$, where $w\notin fn(P)$ and $bn(\alpha_1[m])\cap\cdots\cap bn(\alpha_n[m])\cap fn(P,w)=\emptyset$, then there exist some $Q$
  and $\beta_1[m],\cdots,\beta_n[m]$ with $Q\{w/z\}\equiv_{\alpha}P'$ and $\beta_1\sigma[m]=\alpha_1[m],\cdots,\beta_n\sigma[m]=\alpha_n[m]$, $P\xtworightarrow{\{\beta_1[m],\cdots,\beta_n[m]\}}Q$.
\end{enumerate}

\end{proposition}

\begin{proof}
By the definition of substitution (Definition \ref{subs4}) and induction on the depth of inference.
\end{proof}

\subsection{Strong Bisimilarities}\label{s4}

\subsubsection{Laws and Congruence}

\begin{theorem}
$\equiv_{\alpha}$ are FR strongly truly concurrent bisimulations. That is, if $P\equiv_{\alpha}Q$, then,
\begin{enumerate}
  \item $P\sim_p^{fr} Q$;
  \item $P\sim_s^{fr} Q$;
  \item $P\sim_{hp}^{fr} Q$;
  \item $P\sim_{hhp}^{fr} Q$.
\end{enumerate}
\end{theorem}

\begin{proof}
By induction on the depth of inference, we can get the following facts:

\begin{enumerate}
  \item If $\alpha$ is a free action and $P\xrightarrow{\alpha}P'$, then equally for some $Q'$ with $P'\equiv_{\alpha}Q'$,
  $Q\xrightarrow{\alpha}Q'$;
  \item If $P\xrightarrow{a(y)}P'$ with $a=x$ or $a=\overline{x}$ and $z\notin n(Q)$, then equally for some $Q'$ with $P'\{z/y\}\equiv_{\alpha}Q'$,
  $Q\xrightarrow{a(z)}Q'$;
  \item If $\alpha[m]$ is a free action and $P\xtworightarrow{\alpha[m]}P'$, then equally for some $Q'$ with $P'\equiv_{\alpha}Q'$,
  $Q\xtworightarrow{\alpha[m]}Q'$;
  \item If $P\xtworightarrow{a(y)[m]}P'$ with $a=x$ or $a=\overline{x}$ and $z\notin n(Q)$, then equally for some $Q'$ with $P'\{z/y\}\equiv_{\alpha}Q'$,
  $Q\xtworightarrow{a(z)[m]}Q'$.
\end{enumerate}

Then, we can get:

\begin{enumerate}
  \item by the definition of FR strongly pomset bisimilarity, $P\sim_p^{fr} Q$;
  \item by the definition of FR strongly step bisimilarity, $P\sim_s^{fr} Q$;
  \item by the definition of FR strongly hp-bisimilarity, $P\sim_{hp}^{fr} Q$;
  \item by the definition of FR strongly hhp-bisimilarity, $P\sim_{hhp}^{fr} Q$.
\end{enumerate}
\end{proof}

\begin{proposition}[Summation laws for FR strongly pomset bisimulation] The Summation laws for FR strongly pomset bisimulation are as follows.

\begin{enumerate}
  \item $P+Q\sim_p^{fr} Q+P$;
  \item $P+(Q+R)\sim_p^{fr} (P+Q)+R$;
  \item $P+P\sim_p^{fr} P$;
  \item $P+\textbf{nil}\sim_p^{fr} P$.
\end{enumerate}

\end{proposition}

\begin{proof}
\begin{enumerate}
  \item $P+Q\sim_p^{fr} Q+P$. It is sufficient to prove the relation $R=\{(P+Q, Q+P)\}\cup \textbf{Id}$ is a FR strongly pomset bisimulation, we omit it;
  \item $P+(Q+R)\sim_p^{fr} (P+Q)+R$. It is sufficient to prove the relation $R=\{(P+(Q+R), (P+Q)+R)\}\cup \textbf{Id}$ is a FR strongly pomset bisimulation, we omit it;
  \item $P+P\sim_p^{fr} P$. It is sufficient to prove the relation $R=\{(P+P, P)\}\cup \textbf{Id}$ is a FR strongly pomset bisimulation, we omit it;
  \item $P+\textbf{nil}\sim_p^{fr} P$. It is sufficient to prove the relation $R=\{(P+\textbf{nil}, P)\}\cup \textbf{Id}$ is a FR strongly pomset bisimulation, we omit it.
\end{enumerate}
\end{proof}

\begin{proposition}[Summation laws for FR strongly step bisimulation] The Summation laws for FR strongly step bisimulation are as follows.
\begin{enumerate}
  \item $P+Q\sim_s^{fr} Q+P$;
  \item $P+(Q+R)\sim_s^{fr} (P+Q)+R$;
  \item $P+P\sim_s^{fr} P$;
  \item $P+\textbf{nil}\sim_s^{fr} P$.
\end{enumerate}
\end{proposition}

\begin{proof}
\begin{enumerate}
  \item $P+Q\sim_s^{fr} Q+P$. It is sufficient to prove the relation $R=\{(P+Q, Q+P)\}\cup \textbf{Id}$ is a FR strongly step bisimulation, we omit it;
  \item $P+(Q+R)\sim_s^{fr} (P+Q)+R$. It is sufficient to prove the relation $R=\{(P+(Q+R), (P+Q)+R)\}\cup \textbf{Id}$ is a FR strongly step bisimulation, we omit it;
  \item $P+P\sim_s^{fr} P$. It is sufficient to prove the relation $R=\{(P+P, P)\}\cup \textbf{Id}$ is a FR strongly step bisimulation, we omit it;
  \item $P+\textbf{nil}\sim_s^{fr} P$. It is sufficient to prove the relation $R=\{(P+\textbf{nil}, P)\}\cup \textbf{Id}$ is a FR strongly step bisimulation, we omit it.
\end{enumerate}
\end{proof}

\begin{proposition}[Summation laws for FR strongly hp-bisimulation] The Summation laws for FR strongly hp-bisimulation are as follows.
\begin{enumerate}
  \item $P+Q\sim_{hp}^{fr} Q+P$;
  \item $P+(Q+R)\sim_{hp}^{fr} (P+Q)+R$;
  \item $P+P\sim_{hp}^{fr} P$;
  \item $P+\textbf{nil}\sim_{hp}^{fr} P$.
\end{enumerate}
\end{proposition}

\begin{proof}
\begin{enumerate}
  \item $P+Q\sim_{hp}^{fr} Q+P$. It is sufficient to prove the relation $R=\{(P+Q, Q+P)\}\cup \textbf{Id}$ is a FR strongly hp-bisimulation, we omit it;
  \item $P+(Q+R)\sim_{hp}^{fr} (P+Q)+R$. It is sufficient to prove the relation $R=\{(P+(Q+R), (P+Q)+R)\}\cup \textbf{Id}$ is a FR strongly hp-bisimulation, we omit it;
  \item $P+P\sim_{hp}^{fr} P$. It is sufficient to prove the relation $R=\{(P+P, P)\}\cup \textbf{Id}$ is a FR strongly hp-bisimulation, we omit it;
  \item $P+\textbf{nil}\sim_{hp}^{fr} P$. It is sufficient to prove the relation $R=\{(P+\textbf{nil}, P)\}\cup \textbf{Id}$ is a FR strongly hp-bisimulation, we omit it.
\end{enumerate}
\end{proof}

\begin{proposition}[Summation laws for FR strongly hhp-bisimulation] The Summation laws for FR strongly hhp-bisimulation are as follows.
\begin{enumerate}
  \item $P+Q\sim_{hhp}^{fr} Q+P$;
  \item $P+(Q+R)\sim_{hhp}^{fr} (P+Q)+R$;
  \item $P+P\sim_{hhp}^{fr} P$;
  \item $P+\textbf{nil}\sim_{hhp}^{fr} P$.
\end{enumerate}
\end{proposition}

\begin{proof}
\begin{enumerate}
  \item $P+Q\sim_{hhp}^{fr} Q+P$. It is sufficient to prove the relation $R=\{(P+Q, Q+P)\}\cup \textbf{Id}$ is a FR strongly hhp-bisimulation, we omit it;
  \item $P+(Q+R)\sim_{hhp}^{fr} (P+Q)+R$. It is sufficient to prove the relation $R=\{(P+(Q+R), (P+Q)+R)\}\cup \textbf{Id}$ is a FR strongly hhp-bisimulation, we omit it;
  \item $P+P\sim_{hhp}^{fr} P$. It is sufficient to prove the relation $R=\{(P+P, P)\}\cup \textbf{Id}$ is a FR strongly hhp-bisimulation, we omit it;
  \item $P+\textbf{nil}\sim_{hhp}^{fr} P$. It is sufficient to prove the relation $R=\{(P+\textbf{nil}, P)\}\cup \textbf{Id}$ is a FR strongly hhp-bisimulation, we omit it.
\end{enumerate}
\end{proof}

\begin{theorem}[Identity law for FR strongly truly concurrent bisimilarities]
If $A(\widetilde{x})\overset{\text{def}}{=}P$, then

\begin{enumerate}
  \item $A(\widetilde{y})\sim_p^{fr} P\{\widetilde{y}/\widetilde{x}\}$;
  \item $A(\widetilde{y})\sim_s^{fr} P\{\widetilde{y}/\widetilde{x}\}$;
  \item $A(\widetilde{y})\sim_{hp}^{fr} P\{\widetilde{y}/\widetilde{x}\}$;
  \item $A(\widetilde{y})\sim_{hhp}^{fr} P\{\widetilde{y}/\widetilde{x}\}$.
\end{enumerate}
\end{theorem}

\begin{proof}
\begin{enumerate}
  \item $A(\widetilde{y})\sim_p^{fr} P\{\widetilde{y}/\widetilde{x}\}$. It is sufficient to prove the relation $R=\{(A(\widetilde{y}), P\{\widetilde{y}/\widetilde{x}\})\}\cup \textbf{Id}$ is a FR strongly pomset bisimulation, we omit it;
  \item $A(\widetilde{y})\sim_s^{fr} P\{\widetilde{y}/\widetilde{x}\}$. It is sufficient to prove the relation $R=\{(A(\widetilde{y}), P\{\widetilde{y}/\widetilde{x}\})\}\cup \textbf{Id}$ is a FR strongly step bisimulation, we omit it;
  \item $A(\widetilde{y})\sim_{hp}^{fr} P\{\widetilde{y}/\widetilde{x}\}$. It is sufficient to prove the relation $R=\{(A(\widetilde{y}), P\{\widetilde{y}/\widetilde{x}\})\}\cup \textbf{Id}$ is a FR strongly hp-bisimulation, we omit it;
  \item $A(\widetilde{y})\sim_{hhp}^{fr} P\{\widetilde{y}/\widetilde{x}\}$. It is sufficient to prove the relation $R=\{(A(\widetilde{y}), P\{\widetilde{y}/\widetilde{x}\})\}\cup \textbf{Id}$ is a FR strongly hhp-bisimulation, we omit it.
\end{enumerate}
\end{proof}

\begin{theorem}[Restriction Laws for FR strongly pomset bisimilarity]
The restriction laws for FR strongly pomset bisimilarity are as follows.

\begin{enumerate}
  \item $(y)P\sim_p^{fr} P$, if $y\notin fn(P)$;
  \item $(y)(z)P\sim_p^{fr} (z)(y)P$;
  \item $(y)(P+Q)\sim_p^{fr} (y)P+(y)Q$;
  \item $(y)\alpha.P\sim_p^{fr} \alpha.(y)P$ if $y\notin n(\alpha)$;
  \item $(y)\alpha.P\sim_p^{fr} \textbf{nil}$ if $y$ is the subject of $\alpha$.
\end{enumerate}
\end{theorem}

\begin{proof}
\begin{enumerate}
  \item $(y)P\sim_p^{fr} P$, if $y\notin fn(P)$. It is sufficient to prove the relation $R=\{((y)P, P)\}\cup \textbf{Id}$, if $y\notin fn(P)$, is a FR strongly pomset bisimulation, we omit it;
  \item $(y)(z)P\sim_p^{fr} (z)(y)P$. It is sufficient to prove the relation $R=\{((y)(z)P, (z)(y)P)\}\cup \textbf{Id}$ is a FR strongly pomset bisimulation, we omit it;
  \item $(y)(P+Q)\sim_p^{fr} (y)P+(y)Q$. It is sufficient to prove the relation $R=\{((y)(P+Q), (y)P+(y)Q)\}\cup \textbf{Id}$ is a FR strongly pomset bisimulation, we omit it;
  \item $(y)\alpha.P\sim_p^{fr} \alpha.(y)P$ if $y\notin n(\alpha)$. It is sufficient to prove the relation $R=\{((y)\alpha.P, \alpha.(y)P)\}\cup \textbf{Id}$, if $y\notin n(\alpha)$, is a FR strongly pomset bisimulation, we omit it;
  \item $(y)\alpha.P\sim_p^{fr} \textbf{nil}$ if $y$ is the subject of $\alpha$. It is sufficient to prove the relation $R=\{((y)\alpha.P, \textbf{nil})\}\cup \textbf{Id}$, if $y$ is the subject of $\alpha$, is a FR strongly pomset bisimulation, we omit it.
\end{enumerate}
\end{proof}

\begin{theorem}[Restriction Laws for FR strongly step bisimilarity]
The restriction laws for FR strongly step bisimilarity are as follows.

\begin{enumerate}
  \item $(y)P\sim_s^{fr} P$, if $y\notin fn(P)$;
  \item $(y)(z)P\sim_s^{fr} (z)(y)P$;
  \item $(y)(P+Q)\sim_s^{fr} (y)P+(y)Q$;
  \item $(y)\alpha.P\sim_s^{fr} \alpha.(y)P$ if $y\notin n(\alpha)$;
  \item $(y)\alpha.P\sim_s^{fr} \textbf{nil}$ if $y$ is the subject of $\alpha$.
\end{enumerate}
\end{theorem}

\begin{proof}
\begin{enumerate}
  \item $(y)P\sim_s^{fr} P$, if $y\notin fn(P)$. It is sufficient to prove the relation $R=\{((y)P, P)\}\cup \textbf{Id}$, if $y\notin fn(P)$, is a FR strongly step bisimulation, we omit it;
  \item $(y)(z)P\sim_s^{fr} (z)(y)P$. It is sufficient to prove the relation $R=\{((y)(z)P, (z)(y)P)\}\cup \textbf{Id}$ is a FR strongly step bisimulation, we omit it;
  \item $(y)(P+Q)\sim_s^{fr} (y)P+(y)Q$. It is sufficient to prove the relation $R=\{((y)(P+Q), (y)P+(y)Q)\}\cup \textbf{Id}$ is a FR strongly step bisimulation, we omit it;
  \item $(y)\alpha.P\sim_s^{fr} \alpha.(y)P$ if $y\notin n(\alpha)$. It is sufficient to prove the relation $R=\{((y)\alpha.P, \alpha.(y)P)\}\cup \textbf{Id}$, if $y\notin n(\alpha)$, is a FR strongly step bisimulation, we omit it;
  \item $(y)\alpha.P\sim_s^{fr} \textbf{nil}$ if $y$ is the subject of $\alpha$. It is sufficient to prove the relation $R=\{((y)\alpha.P, \textbf{nil})\}\cup \textbf{Id}$, if $y$ is the subject of $\alpha$, is a FR strongly step bisimulation, we omit it.
\end{enumerate}
\end{proof}

\begin{theorem}[Restriction Laws for FR strongly hp-bisimilarity]
The restriction laws for FR strongly hp-bisimilarity are as follows.

\begin{enumerate}
  \item $(y)P\sim_{hp}^{fr} P$, if $y\notin fn(P)$;
  \item $(y)(z)P\sim_{hp}^{fr} (z)(y)P$;
  \item $(y)(P+Q)\sim_{hp}^{fr} (y)P+(y)Q$;
  \item $(y)\alpha.P\sim_{hp}^{fr} \alpha.(y)P$ if $y\notin n(\alpha)$;
  \item $(y)\alpha.P\sim_{hp}^{fr} \textbf{nil}$ if $y$ is the subject of $\alpha$.
\end{enumerate}
\end{theorem}

\begin{proof}
\begin{enumerate}
  \item $(y)P\sim_{hp}^{fr} P$, if $y\notin fn(P)$. It is sufficient to prove the relation $R=\{((y)P, P)\}\cup \textbf{Id}$, if $y\notin fn(P)$, is a FR strongly hp-bisimulation, we omit it;
  \item $(y)(z)P\sim_{hp}^{fr} (z)(y)P$. It is sufficient to prove the relation $R=\{((y)(z)P, (z)(y)P)\}\cup \textbf{Id}$ is a FR strongly hp-bisimulation, we omit it;
  \item $(y)(P+Q)\sim_{hp}^{fr} (y)P+(y)Q$. It is sufficient to prove the relation $R=\{((y)(P+Q), (y)P+(y)Q)\}\cup \textbf{Id}$ is a FR strongly hp-bisimulation, we omit it;
  \item $(y)\alpha.P\sim_{hp}^{fr} \alpha.(y)P$ if $y\notin n(\alpha)$. It is sufficient to prove the relation $R=\{((y)\alpha.P, \alpha.(y)P)\}\cup \textbf{Id}$, if $y\notin n(\alpha)$, is a FR strongly hp-bisimulation, we omit it;
  \item $(y)\alpha.P\sim_{hp}^{fr} \textbf{nil}$ if $y$ is the subject of $\alpha$. It is sufficient to prove the relation $R=\{((y)\alpha.P, \textbf{nil})\}\cup \textbf{Id}$, if $y$ is the subject of $\alpha$, is a FR strongly hp-bisimulation, we omit it.
\end{enumerate}
\end{proof}

\begin{theorem}[Restriction Laws for FR strongly hhp-bisimilarity]
The restriction laws for FR strongly hhp-bisimilarity are as follows.

\begin{enumerate}
  \item $(y)P\sim_{hhp}^{fr} P$, if $y\notin fn(P)$;
  \item $(y)(z)P\sim_{hhp}^{fr} (z)(y)P$;
  \item $(y)(P+Q)\sim_{hhp}^{fr} (y)P+(y)Q$;
  \item $(y)\alpha.P\sim_{hhp}^{fr} \alpha.(y)P$ if $y\notin n(\alpha)$;
  \item $(y)\alpha.P\sim_{hhp}^{fr} \textbf{nil}$ if $y$ is the subject of $\alpha$.
\end{enumerate}
\end{theorem}

\begin{proof}
\begin{enumerate}
  \item $(y)P\sim_{hhp}^{fr} P$, if $y\notin fn(P)$. It is sufficient to prove the relation $R=\{((y)P, P)\}\cup \textbf{Id}$, if $y\notin fn(P)$, is a FR strongly hhp-bisimulation, we omit it;
  \item $(y)(z)P\sim_{hhp}^{fr} (z)(y)P$. It is sufficient to prove the relation $R=\{((y)(z)P, (z)(y)P)\}\cup \textbf{Id}$ is a FR strongly hhp-bisimulation, we omit it;
  \item $(y)(P+Q)\sim_{hhp}^{fr} (y)P+(y)Q$. It is sufficient to prove the relation $R=\{((y)(P+Q), (y)P+(y)Q)\}\cup \textbf{Id}$ is a FR strongly hhp-bisimulation, we omit it;
  \item $(y)\alpha.P\sim_{hhp}^{fr} \alpha.(y)P$ if $y\notin n(\alpha)$. It is sufficient to prove the relation $R=\{((y)\alpha.P, \alpha.(y)P)\}\cup \textbf{Id}$, if $y\notin n(\alpha)$, is a FR strongly hhp-bisimulation, we omit it;
  \item $(y)\alpha.P\sim_{hhp}^{fr} \textbf{nil}$ if $y$ is the subject of $\alpha$. It is sufficient to prove the relation $R=\{((y)\alpha.P, \textbf{nil})\}\cup \textbf{Id}$, if $y$ is the subject of $\alpha$, is a FR strongly hhp-bisimulation, we omit it.
\end{enumerate}
\end{proof}

\begin{theorem}[Parallel laws for FR strongly pomset bisimilarity]
The parallel laws for FR strongly pomset bisimilarity are as follows.

\begin{enumerate}
  \item $P\parallel \textbf{nil}\sim_p^{fr} P$;
  \item $P_1\parallel P_2\sim_p^{fr} P_2\parallel P_1$;
  \item $(P_1\parallel P_2)\parallel P_3\sim_p^{fr} P_1\parallel (P_2\parallel P_3)$;
  \item $(y)(P_1\parallel P_2)\sim_p^{fr} (y)P_1\parallel (y)P_2$, if $y\notin fn(P_1)\cap fn(P_2)$.
\end{enumerate}
\end{theorem}

\begin{proof}
\begin{enumerate}
  \item $P\parallel \textbf{nil}\sim_p^{fr} P$. It is sufficient to prove the relation $R=\{(P\parallel \textbf{nil}, P)\}\cup \textbf{Id}$ is a FR strongly pomset bisimulation, we omit it;
  \item $P_1\parallel P_2\sim_p^{fr} P_2\parallel P_1$. It is sufficient to prove the relation $R=\{(P_1\parallel P_2, P_2\parallel P_1)\}\cup \textbf{Id}$ is a FR strongly pomset bisimulation, we omit it;
  \item $(P_1\parallel P_2)\parallel P_3\sim_p^{fr} P_1\parallel (P_2\parallel P_3)$. It is sufficient to prove the relation $R=\{((P_1\parallel P_2)\parallel P_3, P_1\parallel (P_2\parallel P_3))\}\cup \textbf{Id}$ is a FR strongly pomset bisimulation, we omit it;
  \item $(y)(P_1\parallel P_2)\sim_p^{fr} (y)P_1\parallel (y)P_2$, if $y\notin fn(P_1)\cap fn(P_2)$. It is sufficient to prove the relation $R=\{((y)(P_1\parallel P_2), (y)P_1\parallel (y)P_2)\}\cup \textbf{Id}$, if $y\notin fn(P_1)\cap fn(P_2)$, is a FR strongly pomset bisimulation, we omit it.
\end{enumerate}
\end{proof}

\begin{theorem}[Parallel laws for FR strongly step bisimilarity]
The parallel laws for FR strongly step bisimilarity are as follows.

\begin{enumerate}
  \item $P\parallel \textbf{nil}\sim_s^{fr} P$;
  \item $P_1\parallel P_2\sim_s^{fr} P_2\parallel P_1$;
  \item $(P_1\parallel P_2)\parallel P_3\sim_s^{fr} P_1\parallel (P_2\parallel P_3)$;
  \item $(y)(P_1\parallel P_2)\sim_s^{fr} (y)P_1\parallel (y)P_2$, if $y\notin fn(P_1)\cap fn(P_2)$.
\end{enumerate}
\end{theorem}

\begin{proof}
\begin{enumerate}
  \item $P\parallel \textbf{nil}\sim_s^{fr} P$. It is sufficient to prove the relation $R=\{(P\parallel \textbf{nil}, P)\}\cup \textbf{Id}$ is a FR strongly step bisimulation, we omit it;
  \item $P_1\parallel P_2\sim_s^{fr} P_2\parallel P_1$. It is sufficient to prove the relation $R=\{(P_1\parallel P_2, P_2\parallel P_1)\}\cup \textbf{Id}$ is a FR strongly step bisimulation, we omit it;
  \item $(P_1\parallel P_2)\parallel P_3\sim_s^{fr} P_1\parallel (P_2\parallel P_3)$. It is sufficient to prove the relation $R=\{((P_1\parallel P_2)\parallel P_3, P_1\parallel (P_2\parallel P_3))\}\cup \textbf{Id}$ is a FR strongly step bisimulation, we omit it;
  \item $(y)(P_1\parallel P_2)\sim_s^{fr} (y)P_1\parallel (y)P_2$, if $y\notin fn(P_1)\cap fn(P_2)$. It is sufficient to prove the relation $R=\{((y)(P_1\parallel P_2), (y)P_1\parallel (y)P_2)\}\cup \textbf{Id}$, if $y\notin fn(P_1)\cap fn(P_2)$, is a FR strongly step bisimulation, we omit it.
\end{enumerate}
\end{proof}

\begin{theorem}[Parallel laws for FR strongly hp-bisimilarity]
The parallel laws for FR strongly hp-bisimilarity are as follows.

\begin{enumerate}
  \item $P\parallel \textbf{nil}\sim_{hp}^{fr} P$;
  \item $P_1\parallel P_2\sim_{hp}^{fr} P_2\parallel P_1$;
  \item $(P_1\parallel P_2)\parallel P_3\sim_{hp}^{fr} P_1\parallel (P_2\parallel P_3)$;
  \item $(y)(P_1\parallel P_2)\sim_{hp}^{fr} (y)P_1\parallel (y)P_2$, if $y\notin fn(P_1)\cap fn(P_2)$.
\end{enumerate}
\end{theorem}

\begin{proof}
\begin{enumerate}
  \item $P\parallel \textbf{nil}\sim_{hp}^{fr} P$. It is sufficient to prove the relation $R=\{(P\parallel \textbf{nil}, P)\}\cup \textbf{Id}$ is a FR strongly hp-bisimulation, we omit it;
  \item $P_1\parallel P_2\sim_{hp}^{fr} P_2\parallel P_1$. It is sufficient to prove the relation $R=\{(P_1\parallel P_2, P_2\parallel P_1)\}\cup \textbf{Id}$ is a FR strongly hp-bisimulation, we omit it;
  \item $(P_1\parallel P_2)\parallel P_3\sim_{hp}^{fr} P_1\parallel (P_2\parallel P_3)$. It is sufficient to prove the relation $R=\{((P_1\parallel P_2)\parallel P_3, P_1\parallel (P_2\parallel P_3))\}\cup \textbf{Id}$ is a FR strongly hp-bisimulation, we omit it;
  \item $(y)(P_1\parallel P_2)\sim_{hp}^{fr} (y)P_1\parallel (y)P_2$, if $y\notin fn(P_1)\cap fn(P_2)$. It is sufficient to prove the relation $R=\{((y)(P_1\parallel P_2), (y)P_1\parallel (y)P_2)\}\cup \textbf{Id}$, if $y\notin fn(P_1)\cap fn(P_2)$, is a FR strongly hp-bisimulation, we omit it.
\end{enumerate}
\end{proof}

\begin{theorem}[Parallel laws for FR strongly hhp-bisimilarity]
The parallel laws for FR strongly hhp-bisimilarity are as follows.

\begin{enumerate}
  \item $P\parallel \textbf{nil}\sim_{hhp}^{fr} P$;
  \item $P_1\parallel P_2\sim_{hhp}^{fr} P_2\parallel P_1$;
  \item $(P_1\parallel P_2)\parallel P_3\sim_{hhp}^{fr} P_1\parallel (P_2\parallel P_3)$;
  \item $(y)(P_1\parallel P_2)\sim_{hhp}^{fr} (y)P_1\parallel (y)P_2$, if $y\notin fn(P_1)\cap fn(P_2)$.
\end{enumerate}
\end{theorem}

\begin{proof}
\begin{enumerate}
  \item $P\parallel \textbf{nil}\sim_{hhp}^{fr} P$. It is sufficient to prove the relation $R=\{(P\parallel \textbf{nil}, P)\}\cup \textbf{Id}$ is a FR strongly hhp-bisimulation, we omit it;
  \item $P_1\parallel P_2\sim_{hhp}^{fr} P_2\parallel P_1$. It is sufficient to prove the relation $R=\{(P_1\parallel P_2, P_2\parallel P_1)\}\cup \textbf{Id}$ is a FR strongly hhp-bisimulation, we omit it;
  \item $(P_1\parallel P_2)\parallel P_3\sim_{hhp}^{fr} P_1\parallel (P_2\parallel P_3)$. It is sufficient to prove the relation $R=\{((P_1\parallel P_2)\parallel P_3, P_1\parallel (P_2\parallel P_3))\}\cup \textbf{Id}$ is a FR strongly hhp-bisimulation, we omit it;
  \item $(y)(P_1\parallel P_2)\sim_{hhp}^{fr} (y)P_1\parallel (y)P_2$, if $y\notin fn(P_1)\cap fn(P_2)$. It is sufficient to prove the relation $R=\{((y)(P_1\parallel P_2), (y)P_1\parallel (y)P_2)\}\cup \textbf{Id}$, if $y\notin fn(P_1)\cap fn(P_2)$, is a FR strongly hhp-bisimulation, we omit it.
\end{enumerate}
\end{proof}

\begin{theorem}[Expansion law for truly concurrent bisimilarities]
Let $P\equiv\sum_i \alpha_{i}.P_{i}$ and $Q\equiv\sum_j\beta_{j}.Q_{j}$, where $bn(\alpha_{i})\cap fn(Q)=\emptyset$ for all $i$, and
  $bn(\beta_{j})\cap fn(P)=\emptyset$ for all $j$. Then,

\begin{enumerate}
  \item $P\parallel Q\sim_p^{fr} \sum_i\sum_j (\alpha_{i}\parallel \beta_{j}).(P_{i}\parallel Q_{j})+\sum_{\alpha_{i} \textrm{ comp }\beta_{j}}\tau.R_{ij}$;
  \item $P\parallel Q\sim_s^{fr} \sum_i\sum_j (\alpha_{i}\parallel \beta_{j}).(P_{i}\parallel Q_{j})+\sum_{\alpha_{i} \textrm{ comp }\beta_{j}}\tau.R_{ij}$;
  \item $P\parallel Q\sim_{hp}^{fr} \sum_i\sum_j (\alpha_{i}\parallel \beta_{j}).(P_{i}\parallel Q_{j})+\sum_{\alpha_{i} \textrm{ comp }\beta_{j}}\tau.R_{ij}$;
  \item $P\parallel Q\nsim_{phhp} \sum_i\sum_j (\alpha_{i}\parallel \beta_{j}).(P_{i}\parallel Q_{j})+\sum_{\alpha_{i} \textrm{ comp }\beta_{j}}\tau.R_{ij}$.
\end{enumerate}

Where $\alpha_i$ comp $\beta_j$ and $R_{ij}$ are defined as follows:
\begin{enumerate}
  \item $\alpha_{i}$ is $\overline{x}u$ and $\beta_{j}$ is $x(v)$, then $R_{ij}=P_{i}\parallel Q_{j}\{u/v\}$;
  \item $\alpha_{i}$ is $\overline{x}(u)$ and $\beta_{j}$ is $x(v)$, then $R_{ij}=(w)(P_{i}\{w/u\}\parallel Q_{j}\{w/v\})$, if $w\notin fn((u)P_{i})\cup fn((v)Q_{j})$;
  \item $\alpha_{i}$ is $x(v)$ and $\beta_{j}$ is $\overline{x}u$, then $R_{ij}=P_{i}\{u/v\}\parallel Q_{j}$;
  \item $\alpha_{i}$ is $x(v)$ and $\beta_{j}$ is $\overline{x}(u)$, then $R_{ij}=(w)(P_{i}\{w/v\}\parallel Q_{j}\{w/u\})$, if $w\notin fn((v)P_{i})\cup fn((u)Q_{j})$.
\end{enumerate}

Let $P\equiv\sum_i P_{i}.\alpha_{i}[m]$ and $Q\equiv\sum_l Q_{j}.\beta_{j}[m]$, where $bn(\alpha_{i}[m])\cap fn(Q)=\emptyset$ for all $i$, and
  $bn(\beta_{j}[m])\cap fn(P)=\emptyset$ for all $j$. Then,

\begin{enumerate}
  \item $P\parallel Q\sim_p^{fr} \sum_i\sum_j(P_{i}\parallel Q_{j}).(\alpha_{i}[m]\parallel \beta_{j}[m])+\sum_{\alpha_{i} \textrm{ comp }\beta_{j}} R_{ij}.\tau$;
  \item $P\parallel Q\sim_s^{fr} \sum_i\sum_j(P_{i}\parallel Q_{j}).(\alpha_{i}[m]\parallel \beta_{j}[m])+\sum_{\alpha_{i} \textrm{ comp }\beta_{j}} R_{ij}.\tau$;
  \item $P\parallel Q\sim_{hp}^{fr} \sum_i\sum_j(P_{i}\parallel Q_{j}).(\alpha_{i}[m]\parallel \beta_{j}[m])+\sum_{\alpha_{i} \textrm{ comp }\beta_{j}} R_{ij}.\tau$;
  \item $P\parallel Q\nsim_{phhp} \sum_i\sum_j(P_{i}\parallel Q_{j}).(\alpha_{i}[m]\parallel \beta_{j}[m])+\sum_{\alpha_{i} \textrm{ comp }\beta_{j}} R_{ij}.\tau$.
\end{enumerate}

Where $\alpha_i$ comp $\beta_j$ and $R_{ij}$ are defined as follows:
\begin{enumerate}
  \item $\alpha_{i}[m]$ is $\overline{x}u$ and $\beta_{j}[m]$ is $x(v)$, then $R_{ij}=P_{i}\parallel Q_{j}\{u/v\}$;
  \item $\alpha_{i}[m]$ is $\overline{x}(u)$ and $\beta_{j}[m]$ is $x(v)$, then $R_{ij}=(w)(P_{i}\{w/u\}\parallel Q_{j}\{w/v\})$, if $w\notin fn((u)P_{i})\cup fn((v)Q_{j})$;
  \item $\alpha_{i}[m]$ is $x(v)$ and $\beta_{j}[m]$ is $\overline{x}u$, then $R_{ij}=P_{i}\{u/v\}\parallel Q_{j}$;
  \item $\alpha_{i}[m]$ is $x(v)$ and $\beta_{j}[m]$ is $\overline{x}(u)$, then $R_{ij}=(w)(P_{i}\{w/v\}\parallel Q_{j}\{w/u\})$, if $w\notin fn((v)P_{i})\cup fn((u)Q_{j})$.
\end{enumerate}
\end{theorem}

\begin{proof}
According to the definition of FR strongly truly concurrent bisimulations, we can easily prove the above equations, and we omit the proof.
\end{proof}

\begin{theorem}[Equivalence and congruence for FR strongly pomset bisimilarity]
We can enjoy the full congruence modulo FR strongly pomset bisimilarity.

\begin{enumerate}
  \item $\sim_p^{fr}$ is an equivalence relation;
  \item If $P\sim_p^{fr} Q$ then
  \begin{enumerate}
    \item $\alpha.P\sim_p^{f} \alpha.Q$, $\alpha$ is a free action;
    \item $P.\alpha[m]\sim_p^{r}Q.\alpha[m]$, $\alpha[m]$ is a free action;
    \item $P+R\sim_p^{fr} Q+R$;
    \item $P\parallel R\sim_p^{fr} Q\parallel R$;
    \item $(w)P\sim_p^{fr} (w)Q$;
    \item $x(y).P\sim_p^{f} x(y).Q$;
    \item $P.x(y)[m]\sim_p^{r}Q.x(y)[m]$.
  \end{enumerate}
\end{enumerate}
\end{theorem}

\begin{proof}
\begin{enumerate}
  \item $\sim_p^{fr}$ is an equivalence relation, it is obvious;
  \item If $P\sim_p^{fr} Q$ then
  \begin{enumerate}
    \item $\alpha.P\sim_p^{f} \alpha.Q$, $\alpha$ is a free action. It is sufficient to prove the relation $R=\{(\alpha.P, \alpha.Q)\}\cup \textbf{Id}$ is a F strongly pomset bisimulation, we omit it;
    \item $P.\alpha[m]\sim_p^{r}Q.\alpha[m]$, $\alpha[m]$ is a free action. It is sufficient to prove the relation $R=\{(P.\alpha[m], Q.\alpha[m])\}\cup \textbf{Id}$ is a R strongly pomset bisimulation, we omit it;
    \item $P+R\sim_p^{fr} Q+R$. It is sufficient to prove the relation $R=\{(P+R, Q+R)\}\cup \textbf{Id}$ is a FR strongly pomset bisimulation, we omit it;
    \item $P\parallel R\sim_p^{fr} Q\parallel R$. It is sufficient to prove the relation $R=\{(P\parallel R, Q\parallel R)\}\cup \textbf{Id}$ is a FR strongly pomset bisimulation, we omit it;
    \item $(w)P\sim_p^{fr} (w)Q$. It is sufficient to prove the relation $R=\{((w)P, (w)Q)\}\cup \textbf{Id}$ is a FR strongly pomset bisimulation, we omit it;
    \item $x(y).P\sim_p^{f} x(y).Q$. It is sufficient to prove the relation $R=\{(x(y).P, x(y).Q)\}\cup \textbf{Id}$ is a F strongly pomset bisimulation, we omit it;
    \item $P.x(y)[m]\sim_p^{r}Q.x(y)[m]$. It is sufficient to prove the relation $R=\{(P.x(y)[m], Q.x(y)[m])\}\cup \textbf{Id}$ is a R strongly pomset bisimulation, we omit it.
  \end{enumerate}
\end{enumerate}
\end{proof}

\begin{theorem}[Equivalence and congruence for FR strongly step bisimilarity]
We can enjoy the full congruence modulo FR strongly step bisimilarity.

\begin{enumerate}
  \item $\sim_s^{fr}$ is an equivalence relation;
  \item If $P\sim_s^{fr} Q$ then
  \begin{enumerate}
    \item $\alpha.P\sim_s^{f} \alpha.Q$, $\alpha$ is a free action;
    \item $P.\alpha[m]\sim_s^{r}Q.\alpha[m]$, $\alpha[m]$ is a free action;
    \item $P+R\sim_s^{fr} Q+R$;
    \item $P\parallel R\sim_s^{fr} Q\parallel R$;
    \item $(w)P\sim_s^{fr} (w)Q$;
    \item $x(y).P\sim_s^{f} x(y).Q$;
    \item $P.x(y)[m]\sim_s^{r}Q.x(y)[m]$.
  \end{enumerate}
\end{enumerate}
\end{theorem}

\begin{proof}
\begin{enumerate}
  \item $\sim_s^{fr}$ is an equivalence relation, it is obvious;
  \item If $P\sim_s^{fr} Q$ then
  \begin{enumerate}
    \item $\alpha.P\sim_s^{f} \alpha.Q$, $\alpha$ is a free action. It is sufficient to prove the relation $R=\{(\alpha.P, \alpha.Q)\}\cup \textbf{Id}$ is a F strongly step bisimulation, we omit it;
    \item $P.\alpha[m]\sim_s^{r}Q.\alpha[m]$, $\alpha[m]$ is a free action. It is sufficient to prove the relation $R=\{(P.\alpha[m], Q.\alpha[m])\}\cup \textbf{Id}$ is a R strongly step bisimulation, we omit it;
    \item $P+R\sim_s^{fr} Q+R$. It is sufficient to prove the relation $R=\{(P+R, Q+R)\}\cup \textbf{Id}$ is a FR strongly step bisimulation, we omit it;
    \item $P\parallel R\sim_s^{fr} Q\parallel R$. It is sufficient to prove the relation $R=\{(P\parallel R, Q\parallel R)\}\cup \textbf{Id}$ is a FR strongly step bisimulation, we omit it;
    \item $(w)P\sim_s^{fr} (w)Q$. It is sufficient to prove the relation $R=\{((w)P, (w)Q)\}\cup \textbf{Id}$ is a FR strongly step bisimulation, we omit it;
    \item $x(y).P\sim_s^{f} x(y).Q$. It is sufficient to prove the relation $R=\{(x(y).P, x(y).Q)\}\cup \textbf{Id}$ is a F strongly step bisimulation, we omit it;
    \item $P.x(y)[m]\sim_s^{r}Q.x(y)[m]$. It is sufficient to prove the relation $R=\{(P.x(y)[m], Q.x(y)[m])\}\cup \textbf{Id}$ is a R strongly step bisimulation, we omit it.
  \end{enumerate}
\end{enumerate}
\end{proof}

\begin{theorem}[Equivalence and congruence for FR strongly hp-bisimilarity]
We can enjoy the full congruence modulo FR strongly hp-bisimilarity.

\begin{enumerate}
  \item $\sim_{hp}^{fr}$ is an equivalence relation;
  \item If $P\sim_{hp}^{fr} Q$ then
  \begin{enumerate}
    \item $\alpha.P\sim_{hp}^{f} \alpha.Q$, $\alpha$ is a free action;
    \item $P.\alpha[m]\sim_{hp}^{r}Q.\alpha[m]$, $\alpha[m]$ is a free action;
    \item $P+R\sim_{hp}^{fr} Q+R$;
    \item $P\parallel R\sim_{hp}^{fr} Q\parallel R$;
    \item $(w)P\sim_{hp}^{fr} (w)Q$;
    \item $x(y).P\sim_{hp}^{f} x(y).Q$;
    \item $P.x(y)[m]\sim_{hp}^{r}Q.x(y)[m]$.
  \end{enumerate}
\end{enumerate}
\end{theorem}

\begin{proof}
\begin{enumerate}
  \item $\sim_{hp}^{fr}$ is an equivalence relation, it is obvious;
  \item If $P\sim_{hp}^{fr} Q$ then
  \begin{enumerate}
    \item $\alpha.P\sim_{hp}^{f} \alpha.Q$, $\alpha$ is a free action. It is sufficient to prove the relation $R=\{(\alpha.P, \alpha.Q)\}\cup \textbf{Id}$ is a F strongly hp-bisimulation, we omit it;
    \item $P.\alpha[m]\sim_{hp}^{r}Q.\alpha[m]$, $\alpha[m]$ is a free action. It is sufficient to prove the relation $R=\{(P.\alpha[m], Q.\alpha[m])\}\cup \textbf{Id}$ is a R strongly hp-bisimulation, we omit it;
    \item $P+R\sim_{hp}^{fr} Q+R$. It is sufficient to prove the relation $R=\{(P+R, Q+R)\}\cup \textbf{Id}$ is a FR strongly hp-bisimulation, we omit it;
    \item $P\parallel R\sim_{hp}^{fr} Q\parallel R$. It is sufficient to prove the relation $R=\{(P\parallel R, Q\parallel R)\}\cup \textbf{Id}$ is a FR strongly hp-bisimulation, we omit it;
    \item $(w)P\sim_{hp}^{fr} (w)Q$. It is sufficient to prove the relation $R=\{((w)P, (w)Q)\}\cup \textbf{Id}$ is a FR strongly hp-bisimulation, we omit it;
    \item $x(y).P\sim_{hp}^{f} x(y).Q$. It is sufficient to prove the relation $R=\{(x(y).P, x(y).Q)\}\cup \textbf{Id}$ is a F strongly hp-bisimulation, we omit it;
    \item $P.x(y)[m]\sim_{hp}^{r}Q.x(y)[m]$. It is sufficient to prove the relation $R=\{(P.x(y)[m], Q.x(y)[m])\}\cup \textbf{Id}$ is a R strongly hp-bisimulation, we omit it.
  \end{enumerate}
\end{enumerate}
\end{proof}

\begin{theorem}[Equivalence and congruence for FR strongly hhp-bisimilarity]
We can enjoy the full congruence modulo FR strongly hhp-bisimilarity.

\begin{enumerate}
  \item $\sim_{hhp}^{fr}$ is an equivalence relation;
  \item If $P\sim_{hhp}^{fr} Q$ then
  \begin{enumerate}
    \item $\alpha.P\sim_{hhp}^{f} \alpha.Q$, $\alpha$ is a free action;
    \item $P.\alpha[m]\sim_{hhp}^{r}Q.\alpha[m]$, $\alpha[m]$ is a free action;
    \item $P+R\sim_{hhp}^{fr} Q+R$;
    \item $P\parallel R\sim_{hhp}^{fr} Q\parallel R$;
    \item $(w)P\sim_{hhp}^{fr} (w)Q$;
    \item $x(y).P\sim_{hhp}^{f} x(y).Q$;
    \item $P.x(y)[m]\sim_{hhp}^{r}Q.x(y)[m]$.
  \end{enumerate}
\end{enumerate}
\end{theorem}

\begin{proof}
\begin{enumerate}
  \item $\sim_{hhp}^{fr}$ is an equivalence relation, it is obvious;
  \item If $P\sim_{hhp}^{fr} Q$ then
  \begin{enumerate}
    \item $\alpha.P\sim_{hhp}^{f} \alpha.Q$, $\alpha$ is a free action. It is sufficient to prove the relation $R=\{(\alpha.P, \alpha.Q)\}\cup \textbf{Id}$ is a F strongly hhp-bisimulation, we omit it;
    \item $P.\alpha[m]\sim_{hhp}^{r}Q.\alpha[m]$, $\alpha[m]$ is a free action. It is sufficient to prove the relation $R=\{(P.\alpha[m], Q.\alpha[m])\}\cup \textbf{Id}$ is a R strongly hhp-bisimulation, we omit it;
    \item $P+R\sim_{hhp}^{fr} Q+R$. It is sufficient to prove the relation $R=\{(P+R, Q+R)\}\cup \textbf{Id}$ is a FR strongly hhp-bisimulation, we omit it;
    \item $P\parallel R\sim_{hhp}^{fr} Q\parallel R$. It is sufficient to prove the relation $R=\{(P\parallel R, Q\parallel R)\}\cup \textbf{Id}$ is a FR strongly hhp-bisimulation, we omit it;
    \item $(w)P\sim_{hhp}^{fr} (w)Q$. It is sufficient to prove the relation $R=\{((w)P, (w)Q)\}\cup \textbf{Id}$ is a FR strongly hhp-bisimulation, we omit it;
    \item $x(y).P\sim_{hhp}^{f} x(y).Q$. It is sufficient to prove the relation $R=\{(x(y).P, x(y).Q)\}\cup \textbf{Id}$ is a F strongly hhp-bisimulation, we omit it;
    \item $P.x(y)[m]\sim_{hhp}^{r}Q.x(y)[m]$. It is sufficient to prove the relation $R=\{(P.x(y)[m], Q.x(y)[m])\}\cup \textbf{Id}$ is a R strongly hhp-bisimulation, we omit it.
  \end{enumerate}
\end{enumerate}
\end{proof}

\subsubsection{Recursion}

\begin{definition}
Let $X$ have arity $n$, and let $\widetilde{x}=x_1,\cdots,x_n$ be distinct names, and $fn(P)\subseteq\{x_1,\cdots,x_n\}$. The replacement of $X(\widetilde{x})$ by $P$ in $E$, written
$E\{X(\widetilde{x}):=P\}$, means the result of replacing each subterm $X(\widetilde{y})$ in $E$ by $P\{\widetilde{y}/\widetilde{x}\}$.
\end{definition}

\begin{definition}
Let $E$ and $F$ be two process expressions containing only $X_1,\cdots,X_m$ with associated name sequences $\widetilde{x}_1,\cdots,\widetilde{x}_m$. Then,
\begin{enumerate}
  \item $E\sim_p^{fr} F$ means $E(\widetilde{P})\sim_p^{fr} F(\widetilde{P})$;
  \item $E\sim_s^{fr} F$ means $E(\widetilde{P})\sim_s^{fr} F(\widetilde{P})$;
  \item $E\sim_{hp}^{fr} F$ means $E(\widetilde{P})\sim_{hp}^{fr} F(\widetilde{P})$;
  \item $E\sim_{hhp}^{fr} F$ means $E(\widetilde{P})\sim_{hhp}^{fr} F(\widetilde{P})$;
\end{enumerate}

for all $\widetilde{P}$ such that $fn(P_i)\subseteq \widetilde{x}_i$ for each $i$.
\end{definition}

\begin{definition}
A term or identifier is weakly guarded in $P$ if it lies within some subterm $\alpha.Q$ or $Q.\alpha[m]$ or $(\alpha_1\parallel\cdots\parallel \alpha_n).Q$ or
$Q.(\alpha_1[m]\parallel\cdots\parallel \alpha_n[m])$ of $P$.
\end{definition}

\begin{theorem}
Assume that $\widetilde{E}$ and $\widetilde{F}$ are expressions containing only $X_i$ with $\widetilde{x}_i$, and $\widetilde{A}$ and $\widetilde{B}$ are identifiers with $A_i$, $B_i$. Then, for all $i$,
\begin{enumerate}
  \item $E_i\sim_s^{fr} F_i$, $A_i(\widetilde{x}_i)\overset{\text{def}}{=}E_i(\widetilde{A})$, $B_i(\widetilde{x}_i)\overset{\text{def}}{=}F_i(\widetilde{B})$, then
  $A_i(\widetilde{x}_i)\sim_s^{fr} B_i(\widetilde{x}_i)$;
  \item $E_i\sim_p^{fr} F_i$, $A_i(\widetilde{x}_i)\overset{\text{def}}{=}E_i(\widetilde{A})$, $B_i(\widetilde{x}_i)\overset{\text{def}}{=}F_i(\widetilde{B})$, then
  $A_i(\widetilde{x}_i)\sim_p^{fr} B_i(\widetilde{x}_i)$;
  \item $E_i\sim_{hp}^{fr} F_i$, $A_i(\widetilde{x}_i)\overset{\text{def}}{=}E_i(\widetilde{A})$, $B_i(\widetilde{x}_i)\overset{\text{def}}{=}F_i(\widetilde{B})$, then
  $A_i(\widetilde{x}_i)\sim_{hp}^{fr} B_i(\widetilde{x}_i)$;
  \item $E_i\sim_{hhp}^{fr} F_i$, $A_i(\widetilde{x}_i)\overset{\text{def}}{=}E_i(\widetilde{A})$, $B_i(\widetilde{x}_i)\overset{\text{def}}{=}F_i(\widetilde{B})$, then
  $A_i(\widetilde{x}_i)\sim_{hhp}^{fr} B_i(\widetilde{x}_i)$.
\end{enumerate}
\end{theorem}

\begin{proof}
\begin{enumerate}
  \item $E_i\sim_s^{fr} F_i$, $A_i(\widetilde{x}_i)\overset{\text{def}}{=}E_i(\widetilde{A})$, $B_i(\widetilde{x}_i)\overset{\text{def}}{=}F_i(\widetilde{B})$, then
  $A_i(\widetilde{x}_i)\sim_s^{fr} B_i(\widetilde{x}_i)$.

      We will consider the case $I=\{1\}$ with loss of generality, and show the following relation $R$ is a FR strongly step bisimulation.

      $$R=\{(G(A),G(B)):G\textrm{ has only identifier }X\}.$$

      By choosing $G\equiv X(\widetilde{y})$, it follows that $A(\widetilde{y})\sim_s^{fr} B(\widetilde{y})$. It is sufficient to prove the following:
      \begin{enumerate}
        \item If $G(A)\xrightarrow{\{\alpha_1,\cdots,\alpha_n\}}P'$, where $\alpha_i(1\leq i\leq n)$ is a free action or bound output action with
        $bn(\alpha_1)\cap\cdots\cap bn(\alpha_n)\cap n(G(A),G(B))=\emptyset$, then $G(B)\xrightarrow{\{\alpha_1,\cdots,\alpha_n\}}Q''$ such that $P'\sim_s^{fr} Q''$;
        \item If $G(A)\xrightarrow{x(y)}P'$ with $x\notin n(G(A),G(B))$, then $G(B)\xrightarrow{x(y)}Q''$, such that for all $u$,
        $P'\{u/y\}\sim_s^{fr} Q''\{u/y\}$;
        \item If $G(A)\xtworightarrow{\{\alpha_1[m],\cdots,\alpha_n[m]\}}P'$, where $\alpha_i[m](1\leq i\leq n)$ is a free action or bound output action with
        $bn(\alpha_1[m])\cap\cdots\cap bn(\alpha_n[m])\cap n(G(A),G(B))=\emptyset$, then $G(B)\xtworightarrow{\{\alpha_1[m],\cdots,\alpha_n[m]\}}Q''$ such that $P'\sim_s^{fr} Q''$;
        \item If $G(A)\xtworightarrow{x(y)[m]}P'$ with $x\notin n(G(A),G(B))$, then $G(B)\xtworightarrow{x(y)[m]}Q''$, such that for all $u$,
        $P'\{u/y\}\sim_s^{fr} Q''\{u/y\}$.
      \end{enumerate}

      To prove the above properties, it is sufficient to induct on the depth of inference and quite routine, we omit it.
  \item $E_i\sim_p^{fr} F_i$, $A_i(\widetilde{x}_i)\overset{\text{def}}{=}E_i(\widetilde{A})$, $B_i(\widetilde{x}_i)\overset{\text{def}}{=}F_i(\widetilde{B})$, then
  $A_i(\widetilde{x}_i)\sim_p^{fr} B_i(\widetilde{x}_i)$. It can be proven similarly to the above case.
  \item $E_i\sim_{hp}^{fr} F_i$, $A_i(\widetilde{x}_i)\overset{\text{def}}{=}E_i(\widetilde{A})$, $B_i(\widetilde{x}_i)\overset{\text{def}}{=}F_i(\widetilde{B})$, then
  $A_i(\widetilde{x}_i)\sim_{hp}^{fr} B_i(\widetilde{x}_i)$. It can be proven similarly to the above case.
  \item $E_i\sim_{hhp}^{fr} F_i$, $A_i(\widetilde{x}_i)\overset{\text{def}}{=}E_i(\widetilde{A})$, $B_i(\widetilde{x}_i)\overset{\text{def}}{=}F_i(\widetilde{B})$, then
  $A_i(\widetilde{x}_i)\sim_{hhp}^{fr} B_i(\widetilde{x}_i)$. It can be proven similarly to the above case.
\end{enumerate}
\end{proof}

\begin{theorem}[Unique solution of equations]
Assume $\widetilde{E}$ are expressions containing only $X_i$ with $\widetilde{x}_i$, and each $X_i$ is weakly guarded in each $E_j$. Assume that $\widetilde{P}$ and $\widetilde{Q}$ are
processes such that $fn(P_i)\subseteq \widetilde{x}_i$ and $fn(Q_i)\subseteq \widetilde{x}_i$. Then, for all $i$,
\begin{enumerate}
  \item if $P_i\sim_p^{fr} E_i(\widetilde{P})$, $Q_i\sim_p^{fr} E_i(\widetilde{Q})$, then $P_i\sim_p^{fr} Q_i$;
  \item if $P_i\sim_s^{fr} E_i(\widetilde{P})$, $Q_i\sim_s^{fr} E_i(\widetilde{Q})$, then $P_i\sim_s^{fr} Q_i$;
  \item if $P_i\sim_{hp}^{fr} E_i(\widetilde{P})$, $Q_i\sim_{hp}^{fr} E_i(\widetilde{Q})$, then $P_i\sim_{hp}^{fr} Q_i$;
  \item if $P_i\sim_{hhp}^{fr} E_i(\widetilde{P})$, $Q_i\sim_{hhp}^{fr} E_i(\widetilde{Q})$, then $P_i\sim_{hhp}^{fr} Q_i$.
\end{enumerate}
\end{theorem}

\begin{proof}
\begin{enumerate}
  \item It is similar to the proof of unique solution of equations for FR strongly pomset bisimulation in CTC, please refer to \cite{CTC2} for details, we omit it;
  \item It is similar to the proof of unique solution of equations for FR strongly step bisimulation in CTC, please refer to \cite{CTC2} for details, we omit it;
  \item It is similar to the proof of unique solution of equations for FR strongly hp-bisimulation in CTC, please refer to \cite{CTC2} for details, we omit it;
  \item It is similar to the proof of unique solution of equations for FR strongly hhp-bisimulation in CTC, please refer to \cite{CTC2} for details, we omit it.
\end{enumerate}
\end{proof}

\subsection{Algebraic Theory}\label{a4}

\begin{definition}[STC]
The theory \textbf{STC} is consisted of the following axioms and inference rules:

\begin{enumerate}
  \item Alpha-conversion $\textbf{A}$.
  \[\textrm{if } P\equiv Q, \textrm{ then } P=Q\]
  \item Congruence $\textbf{C}$. If $P=Q$, then,
  \[\tau.P=\tau.Q\quad \overline{x}y.P=\overline{x}y.Q\quad P.\overline{x}y[m]=Q.\overline{x}y[m]\]
  \[P+R=Q+R\quad P\parallel R=Q\parallel R\]
  \[(x)P=(x)Q\quad x(y).P=x(y).Q\quad P.x(y)[m]=Q.x(y)[m]\]
  \item Summation $\textbf{S}$.
  \[\textbf{S0}\quad P+\textbf{nil}=P\]
  \[\textbf{S1}\quad P+P=P\]
  \[\textbf{S2}\quad P+Q=Q+P\]
  \[\textbf{S3}\quad P+(Q+R)=(P+Q)+R\]
  \item Restriction $\textbf{R}$.
  \[\textbf{R0}\quad (x)P=P\quad \textrm{ if }x\notin fn(P)\]
  \[\textbf{R1}\quad (x)(y)P=(y)(x)P\]
  \[\textbf{R2}\quad (x)(P+Q)=(x)P+(x)Q\]
  \[\textbf{R3}\quad (x)\alpha.P=\alpha.(x)P\quad \textrm{ if }x\notin n(\alpha)\]
  \[\textbf{R4}\quad (x)\alpha.P=\textbf{nil}\quad \textrm{ if }x\textrm{is the subject of }\alpha\]
  \item Expansion $\textbf{E}$.
  Let $P\equiv\sum_i \alpha_{i}.P_{i}$ and $Q\equiv\sum_j\beta_{j}.Q_{j}$, where $bn(\alpha_{i})\cap fn(Q)=\emptyset$ for all $i$, and
  $bn(\beta_{j})\cap fn(P)=\emptyset$ for all $j$. Then,

\begin{enumerate}
  \item $P\parallel Q\sim_p^{fr} \sum_i\sum_j (\alpha_{i}\parallel \beta_{j}).(P_{i}\parallel Q_{j})+\sum_{\alpha_{i} \textrm{ comp }\beta_{j}}\tau.R_{ij}$;
  \item $P\parallel Q\sim_s^{fr} \sum_i\sum_j (\alpha_{i}\parallel \beta_{j}).(P_{i}\parallel Q_{j})+\sum_{\alpha_{i} \textrm{ comp }\beta_{j}}\tau.R_{ij}$;
  \item $P\parallel Q\sim_{hp}^{fr} \sum_i\sum_j (\alpha_{i}\parallel \beta_{j}).(P_{i}\parallel Q_{j})+\sum_{\alpha_{i} \textrm{ comp }\beta_{j}}\tau.R_{ij}$;
  \item $P\parallel Q\nsim_{phhp} \sum_i\sum_j (\alpha_{i}\parallel \beta_{j}).(P_{i}\parallel Q_{j})+\sum_{\alpha_{i} \textrm{ comp }\beta_{j}}\tau.R_{ij}$.
\end{enumerate}

Where $\alpha_i$ comp $\beta_j$ and $R_{ij}$ are defined as follows:
\begin{enumerate}
  \item $\alpha_{i}$ is $\overline{x}u$ and $\beta_{j}$ is $x(v)$, then $R_{ij}=P_{i}\parallel Q_{j}\{u/v\}$;
  \item $\alpha_{i}$ is $\overline{x}(u)$ and $\beta_{j}$ is $x(v)$, then $R_{ij}=(w)(P_{i}\{w/u\}\parallel Q_{j}\{w/v\})$, if $w\notin fn((u)P_{i})\cup fn((v)Q_{j})$;
  \item $\alpha_{i}$ is $x(v)$ and $\beta_{j}$ is $\overline{x}u$, then $R_{ij}=P_{i}\{u/v\}\parallel Q_{j}$;
  \item $\alpha_{i}$ is $x(v)$ and $\beta_{j}$ is $\overline{x}(u)$, then $R_{ij}=(w)(P_{i}\{w/v\}\parallel Q_{j}\{w/u\})$, if $w\notin fn((v)P_{i})\cup fn((u)Q_{j})$.
\end{enumerate}

Let $P\equiv\sum_i P_{i}.\alpha_{i}[m]$ and $Q\equiv\sum_l Q_{j}.\beta_{j}[m]$, where $bn(\alpha_{i}[m])\cap fn(Q)=\emptyset$ for all $i$, and
  $bn(\beta_{j}[m])\cap fn(P)=\emptyset$ for all $j$. Then,

\begin{enumerate}
  \item $P\parallel Q\sim_p^{fr} \sum_i\sum_j(P_{i}\parallel Q_{j}).(\alpha_{i}[m]\parallel \beta_{j}[m])+\sum_{\alpha_{i} \textrm{ comp }\beta_{j}} R_{ij}.\tau$;
  \item $P\parallel Q\sim_s^{fr} \sum_i\sum_j(P_{i}\parallel Q_{j}).(\alpha_{i}[m]\parallel \beta_{j}[m])+\sum_{\alpha_{i} \textrm{ comp }\beta_{j}} R_{ij}.\tau$;
  \item $P\parallel Q\sim_{hp}^{fr} \sum_i\sum_j(P_{i}\parallel Q_{j}).(\alpha_{i}[m]\parallel \beta_{j}[m])+\sum_{\alpha_{i} \textrm{ comp }\beta_{j}} R_{ij}.\tau$;
  \item $P\parallel Q\nsim_{phhp} \sum_i\sum_j(P_{i}\parallel Q_{j}).(\alpha_{i}[m]\parallel \beta_{j}[m])+\sum_{\alpha_{i} \textrm{ comp }\beta_{j}} R_{ij}.\tau$.
\end{enumerate}

Where $\alpha_i$ comp $\beta_j$ and $R_{ij}$ are defined as follows:
\begin{enumerate}
  \item $\alpha_{i}[m]$ is $\overline{x}u$ and $\beta_{j}[m]$ is $x(v)$, then $R_{ij}=P_{i}\parallel Q_{j}\{u/v\}$;
  \item $\alpha_{i}[m]$ is $\overline{x}(u)$ and $\beta_{j}[m]$ is $x(v)$, then $R_{ij}=(w)(P_{i}\{w/u\}\parallel Q_{j}\{w/v\})$, if $w\notin fn((u)P_{i})\cup fn((v)Q_{j})$;
  \item $\alpha_{i}[m]$ is $x(v)$ and $\beta_{j}[m]$ is $\overline{x}u$, then $R_{ij}=P_{i}\{u/v\}\parallel Q_{j}$;
  \item $\alpha_{i}[m]$ is $x(v)$ and $\beta_{j}[m]$ is $\overline{x}(u)$, then $R_{ij}=(w)(P_{i}\{w/v\}\parallel Q_{j}\{w/u\})$, if $w\notin fn((v)P_{i})\cup fn((u)Q_{j})$.
\end{enumerate}
  \item Identifier $\textbf{I}$.
  \[\textrm{If }A(\widetilde{x})\overset{\text{def}}{=}P,\textrm{ then }A(\widetilde{y})= P\{\widetilde{y}/\widetilde{x}\}.\]
\end{enumerate}
\end{definition}

\begin{theorem}[Soundness]
If $\textbf{STC}\vdash P=Q$ then
\begin{enumerate}
  \item $P\sim_p^{fr} Q$;
  \item $P\sim_s^{fr} Q$;
  \item $P\sim_{hp}^{fr} Q$;
  \item $P\sim_{hhp}^{fr} Q$.
\end{enumerate}
\end{theorem}

\begin{proof}
The soundness of these laws modulo strongly truly concurrent bisimilarities is already proven in Section \ref{s4}.
\end{proof}

\begin{definition}
The agent identifier $A$ is weakly guardedly defined if every agent identifier is weakly guarded in the right-hand side of the definition of $A$.
\end{definition}

\begin{definition}[Head normal form]
A Process $P$ is in head normal form if it is a sum of the prefixes:

$$P\equiv \sum_i(\alpha_{i1}\parallel\cdots\parallel\alpha_{in}).P_{i}\quad P\equiv \sum_i P_{i}.(\alpha_{i1}[m]\parallel\cdots\parallel\alpha_{in}[m])$$
\end{definition}

\begin{proposition}
If every agent identifier is weakly guardedly defined, then for any process $P$, there is a head normal form $H$ such that

$$\textbf{STC}\vdash P=H.$$
\end{proposition}

\begin{proof}
It is sufficient to induct on the structure of $P$ and quite obvious.
\end{proof}

\begin{theorem}[Completeness]
For all processes $P$ and $Q$,
\begin{enumerate}
  \item if $P\sim_p^{fr} Q$, then $\textbf{STC}\vdash P=Q$;
  \item if $P\sim_s^{fr} Q$, then $\textbf{STC}\vdash P=Q$;
  \item if $P\sim_{hp}^{fr} Q$, then $\textbf{STC}\vdash P=Q$.
\end{enumerate}
\end{theorem}

\begin{proof}
\begin{enumerate}
  \item if $P\sim_s^{fr} Q$, then $\textbf{STC}\vdash P=Q$.

  For the forward transition case.

Since $P$ and $Q$ all have head normal forms, let $P\equiv\sum_{i=1}^k\alpha_{i}.P_{i}$ and $Q\equiv\sum_{i=1}^k\beta_{i}.Q_{i}$. Then the depth of
$P$, denoted as $d(P)=0$, if $k=0$; $d(P)=1+max\{d(P_{i})\}$ for $1\leq j,i\leq k$. The depth $d(Q)$ can be defined similarly.

It is sufficient to induct on $d=d(P)+d(Q)$. When $d=0$, $P\equiv\textbf{nil}$ and $Q\equiv\textbf{nil}$, $P=Q$, as desired.

Suppose $d>0$.

\begin{itemize}
  \item If $(\alpha_1\parallel\cdots\parallel\alpha_n).M$ with $\alpha_{i}(1\leq i\leq n)$ free actions is a summand of $P$, then
  $P\xrightarrow{\{\alpha_1,\cdots,\alpha_n\}}M$.
  Since $Q$ is in head normal form and has a summand $(\alpha_1\parallel\cdots\parallel\alpha_n).N$ such that $M\sim_s^{fr} N$, by the induction hypothesis $\textbf{STC}\vdash M=N$,
  $\textbf{STC}\vdash (\alpha_1\parallel\cdots\parallel\alpha_n).M= (\alpha_1\parallel\cdots\parallel\alpha_n).N$;
  \item If $x(y).M$ is a summand of $P$, then for $z\notin n(P, Q)$, $P\xrightarrow{x(z)}M'\equiv M\{z/y\}$. Since $Q$ is in head normal form and has a summand
  $x(w).N$ such that for all $v$, $M'\{v/z\}\sim_s^{fr} N'\{v/z\}$ where $N'\equiv N\{z/w\}$, by the induction hypothesis $\textbf{STC}\vdash M'\{v/z\}=N'\{v/z\}$, by the axioms
  $\textbf{C}$ and $\textbf{A}$, $\textbf{STC}\vdash x(y).M=x(w).N$;
  \item If $\overline{x}(y).M$ is a summand of $P$, then for $z\notin n(P,Q)$, $P\xrightarrow{\overline{x}(z)}M'\equiv M\{z/y\}$. Since $Q$ is in head normal form and
  has a summand $\overline{x}(w).N$ such that $M'\sim_s^{fr} N'$ where $N'\equiv N\{z/w\}$, by the induction hypothesis $\textbf{STC}\vdash M'=N'$, by the axioms
  $\textbf{A}$ and $\textbf{C}$, $\textbf{STC}\vdash \overline{x}(y).M= \overline{x}(w).N$.
\end{itemize}

For the reverse transition case, it can be proven similarly, and we omit it.

  \item if $P\sim_p^{fr} Q$, then $\textbf{STC}\vdash P=Q$. It can be proven similarly to the above case.
  \item if $P\sim_{hp}^{fr} Q$, then $\textbf{STC}\vdash P=Q$. It can be proven similarly to the above case.
\end{enumerate}
\end{proof}

\newpage\section{$\pi_{tc}$ with Probabilism and Reversibility}\label{pitcpr}

In this chapter, we design $\pi_{tc}$ with probabilism and reversibility. This chapter is organized as follows. In section \ref{os5}, we introduce the truly concurrent operational semantics. Then, we introduce
the syntax and operational semantics, laws modulo strongly truly concurrent bisimulations, and algebraic theory of $\pi_{tc}$ with probabilism and reversibility in section \ref{sos5},
\ref{s5} and \ref{a5} respectively.

\subsection{Operational Semantics}\label{os5}

Firstly, in this section, we introduce concepts of FR (strongly) probabilistic truly concurrent bisimilarities, including FR probabilistic pomset bisimilarity, FR probabilistic step
bisimilarity, FR probabilistic history-preserving (hp-)bisimilarity and FR probabilistic hereditary history-preserving (hhp-)bisimilarity. In contrast to traditional FR probabilistic truly
concurrent bisimilarities in section \ref{bg}, these versions in $\pi_{ptc}$ must take care of actions with bound objects. Note that, these FR probabilistic truly concurrent bisimilarities
are defined as late bisimilarities, but not early bisimilarities, as defined in $\pi$-calculus \cite{PI1} \cite{PI2}. Note that, here, a PES $\mathcal{E}$ is deemed as a process.

\begin{definition}[Prime event structure with silent event]\label{PES}
Let $\Lambda$ be a fixed set of labels, ranged over $a,b,c,\cdots$ and $\tau$. A ($\Lambda$-labelled) prime event structure with silent event $\tau$ is a tuple
$\mathcal{E}= \langle\mathbb{E}, \leq, \sharp, \lambda\rangle$, where $\mathbb{E}$ is a denumerable set of events, including the silent event $\tau$. Let
$\hat{\mathbb{E}}=\mathbb{E}\backslash\{\tau\}$, exactly excluding $\tau$, it is obvious that $\hat{\tau^*}=\epsilon$, where $\epsilon$ is the empty event.
Let $\lambda:\mathbb{E}\rightarrow\Lambda$ be a labelling function and let $\lambda(\tau)=\tau$. And $\leq$, $\sharp$ are binary relations on $\mathbb{E}$,
called causality and conflict respectively, such that:

\begin{enumerate}
  \item $\leq$ is a partial order and $\lceil e \rceil = \{e'\in \mathbb{E}|e'\leq e\}$ is finite for all $e\in \mathbb{E}$. It is easy to see that
  $e\leq\tau^*\leq e'=e\leq\tau\leq\cdots\leq\tau\leq e'$, then $e\leq e'$.
  \item $\sharp$ is irreflexive, symmetric and hereditary with respect to $\leq$, that is, for all $e,e',e''\in \mathbb{E}$, if $e\sharp e'\leq e''$, then $e\sharp e''$.
\end{enumerate}

Then, the concepts of consistency and concurrency can be drawn from the above definition:

\begin{enumerate}
  \item $e,e'\in \mathbb{E}$ are consistent, denoted as $e\frown e'$, if $\neg(e\sharp e')$. A subset $X\subseteq \mathbb{E}$ is called consistent, if $e\frown e'$ for all
  $e,e'\in X$.
  \item $e,e'\in \mathbb{E}$ are concurrent, denoted as $e\parallel e'$, if $\neg(e\leq e')$, $\neg(e'\leq e)$, and $\neg(e\sharp e')$.
\end{enumerate}
\end{definition}

\begin{definition}[Configuration]
Let $\mathcal{E}$ be a PES. A (finite) configuration in $\mathcal{E}$ is a (finite) consistent subset of events $C\subseteq \mathcal{E}$, closed with respect to causality
(i.e. $\lceil C\rceil=C$). The set of finite configurations of $\mathcal{E}$ is denoted by $\mathcal{C}(\mathcal{E})$. We let $\hat{C}=C\backslash\{\tau\}$.
\end{definition}

A consistent subset of $X\subseteq \mathbb{E}$ of events can be seen as a pomset. Given $X, Y\subseteq \mathbb{E}$, $\hat{X}\sim \hat{Y}$ if $\hat{X}$ and $\hat{Y}$ are
isomorphic as pomsets. In the following of the paper, we say $C_1\sim C_2$, we mean $\hat{C_1}\sim\hat{C_2}$.

\begin{definition}[Pomset transitions and step]
Let $\mathcal{E}$ be a PES and let $C\in\mathcal{C}(\mathcal{E})$, and $\emptyset\neq X\subseteq \mathbb{E}$, if $C\cap X=\emptyset$ and $C'=C\cup X\in\mathcal{C}(\mathcal{E})$, then $C\xrightarrow{X} C'$ is called a pomset transition from $C$ to $C'$. When the events in $X$ are pairwise concurrent, we say that $C\xrightarrow{X}C'$ is a step.
\end{definition}

\begin{definition}[FR pomset transitions and step]
Let $\mathcal{E}$ be a PES and let $C\in\mathcal{C}(\mathcal{E})$, and $\emptyset\neq X\subseteq \mathbb{E}$, if $C\cap X=\emptyset$ and $C'=C\cup X\in\mathcal{C}(\mathcal{E})$, then
$ C\xrightarrow{X}  C'$ is called a forward pomset transition from $ C$ to $ C'$ and
$ C'\xtworightarrow{X[\mathcal{K}]}  C$ is called a reverse pomset transition from $ C'$ to $ C$. When the events in
$X$ and $X[\mathcal{K}]$ are pairwise
concurrent, we say that $ C\xrightarrow{X} C'$ is a forward step and $ C'\xrightarrow{X[\mathcal{K}]} C$ is a reverse step.
It is obvious that $\rightarrow^*\xrightarrow{X}\rightarrow^*=\xrightarrow{X}$ and
$\rightarrow^*\xrightarrow{e}\rightarrow^*=\xrightarrow{e}$ for any $e\in\mathbb{E}$ and $X\subseteq\mathbb{E}$.
\end{definition}

\begin{definition}[Probabilistic transitions]
Let $\mathcal{E}$ be a PES and let $C\in\mathcal{C}(\mathcal{E})$, the transition $ C\xrsquigarrow{\pi}  C^{\pi}$ is called a probabilistic
transition
from $ C$ to $ C^{\pi}$.
\end{definition}

A probability distribution function (PDF) $\mu$ is a map $\mu:\mathcal{C}\times\mathcal{C}\rightarrow[0,1]$ and $\mu^*$ is the cumulative probability distribution function (cPDF).

\begin{definition}[FR strongly probabilistic pomset, step bisimilarity]
Let $\mathcal{E}_1$, $\mathcal{E}_2$ be PESs. A FR strongly probabilistic pomset bisimulation is a relation $R\subseteq\mathcal{C}(\mathcal{E}_1)\times\mathcal{C}(\mathcal{E}_2)$,
such that (1) if $( C_1, C_2)\in R$, and $ C_1\xrightarrow{X_1} C_1'$ (with $\mathcal{E}_1\xrightarrow{X_1}\mathcal{E}_1'$) then $ C_2\xrightarrow{X_2} C_2'$ (with
$\mathcal{E}_2\xrightarrow{X_2}\mathcal{E}_2'$), with $X_1\subseteq \mathbb{E}_1$, $X_2\subseteq \mathbb{E}_2$, $X_1\sim X_2$ and $( C_1', C_2')\in R$:
\begin{enumerate}
  \item for each fresh action $\alpha\in X_1$, if $ C_1''\xrightarrow{\alpha} C_1'''$ (with $\mathcal{E}_1''\xrightarrow{\alpha}\mathcal{E}_1'''$),
  then for some $C_2''$ and $ C_2'''$, $ C_2''\xrightarrow{\alpha} C_2'''$ (with
  $\mathcal{E}_2''\xrightarrow{\alpha}\mathcal{E}_2'''$), such that if $( C_1'', C_2'')\in R$ then $( C_1''', C_2''')\in R$;
  \item for each $x(y)\in X_1$ with ($y\notin n(\mathcal{E}_1, \mathcal{E}_2)$), if $ C_1''\xrightarrow{x(y)} C_1'''$ (with
  $\mathcal{E}_1''\xrightarrow{x(y)}\mathcal{E}_1'''\{w/y\}$) for all $w$, then for some $C_2''$ and $C_2'''$, $ C_2''\xrightarrow{x(y)} C_2'''$
  (with $\mathcal{E}_2''\xrightarrow{x(y)}\mathcal{E}_2'''\{w/y\}$) for all $w$, such that if $( C_1'', C_2'')\in R$ then $( C_1''', C_2''')\in R$;
  \item for each two $x_1(y),x_2(y)\in X_1$ with ($y\notin n(\mathcal{E}_1, \mathcal{E}_2)$), if $ C_1''\xrightarrow{\{x_1(y),x_2(y)\}} C_1'''$
  (with $\mathcal{E}_1''\xrightarrow{\{x_1(y),x_2(y)\}}\mathcal{E}_1'''\{w/y\}$) for all $w$, then for some $C_2''$ and $C_2'''$,
  $ C_2''\xrightarrow{\{x_1(y),x_2(y)\}} C_2'''$ (with $\mathcal{E}_2''\xrightarrow{\{x_1(y),x_2(y)\}}\mathcal{E}_2'''\{w/y\}$) for all $w$, such
  that if $( C_1'', C_2'')\in R$ then $( C_1''', C_2''')\in R$;
  \item for each $\overline{x}(y)\in X_1$ with $y\notin n(\mathcal{E}_1, \mathcal{E}_2)$, if $ C_1''\xrightarrow{\overline{x}(y)} C_1'''$
  (with $\mathcal{E}_1''\xrightarrow{\overline{x}(y)}\mathcal{E}_1'''$), then for some $C_2''$ and $C_2'''$, $ C_2''\xrightarrow{\overline{x}(y)} C_2'''$
  (with $\mathcal{E}_2''\xrightarrow{\overline{x}(y)}\mathcal{E}_2'''$), such that if $( C_1'', C_2'')\in R$ then $( C_1''', C_2''')\in R$.
\end{enumerate}
 and vice-versa; (2) if $( C_1, C_2)\in R$, and $ C_1\xtworightarrow{X_1[\mathcal{K}_1]} C_1'$ (with $\mathcal{E}_1\xtworightarrow{X_1[\mathcal{K}_1]}\mathcal{E}_1'$) then $ C_2\xtworightarrow{X_2[\mathcal{K}_2]} C_2'$ (with
$\mathcal{E}_2\xtworightarrow{X_2[\mathcal{K}_2]}\mathcal{E}_2'$), with $X_1\subseteq \mathbb{E}_1$, $X_2\subseteq \mathbb{E}_2$, $X_1\sim X_2$ and $( C_1', C_2')\in R$:
\begin{enumerate}
  \item for each fresh action $\alpha\in X_1$, if $ C_1''\xtworightarrow{\alpha[m]} C_1'''$ (with $\mathcal{E}_1''\xtworightarrow{\alpha[m]}\mathcal{E}_1'''$),
  then for some $C_2''$ and $ C_2'''$, $ C_2''\xtworightarrow{\alpha[m]} C_2'''$ (with
  $\mathcal{E}_2''\xtworightarrow{\alpha[m]}\mathcal{E}_2'''$), such that if $( C_1'', C_2'')\in R$ then $( C_1''', C_2''')\in R$;
  \item for each $x(y)\in X_1$ with ($y\notin n(\mathcal{E}_1, \mathcal{E}_2)$), if $ C_1''\xtworightarrow{x(y)[m]} C_1'''$ (with
  $\mathcal{E}_1''\xtworightarrow{x(y)[m]}\mathcal{E}_1'''\{w/y\}$) for all $w$, then for some $C_2''$ and $C_2'''$, $ C_2''\xtworightarrow{x(y)[m]} C_2'''$
  (with $\mathcal{E}_2''\xtworightarrow{x(y)[m]}\mathcal{E}_2'''\{w/y\}$) for all $w$, such that if $( C_1'', C_2'')\in R$ then $( C_1''', C_2''')\in R$;
  \item for each two $x_1(y),x_2(y)\in X_1$ with ($y\notin n(\mathcal{E}_1, \mathcal{E}_2)$), if $ C_1''\xtworightarrow{\{x_1(y)[m],x_2(y)[n]\}} C_1'''$
  (with $\mathcal{E}_1''\xtworightarrow{\{x_1(y)[m],x_2(y)[n]\}}\mathcal{E}_1'''\{w/y\}$) for all $w$, then for some $C_2''$ and $C_2'''$,
  $ C_2''\xtworightarrow{\{x_1(y)[m],x_2(y)[n]\}} C_2'''$ (with $\mathcal{E}_2''\xtworightarrow{\{x_1(y)[m],x_2(y)[n]\}}\mathcal{E}_2'''\{w/y\}$) for all $w$, such
  that if $( C_1'', C_2'')\in R$ then $( C_1''', C_2''')\in R$;
  \item for each $\overline{x}(y)\in X_1$ with $y\notin n(\mathcal{E}_1, \mathcal{E}_2)$, if $ C_1''\xtworightarrow{\overline{x}(y)[m]} C_1'''$
  (with $\mathcal{E}_1''\xtworightarrow{\overline{x}(y)[m]}\mathcal{E}_1'''$), then for some $C_2''$ and $C_2'''$, $ C_2''\xtworightarrow{\overline{x}(y)[m]} C_2'''$
  (with $\mathcal{E}_2''\xtworightarrow{\overline{x}(y)[m]}\mathcal{E}_2'''$), such that if $( C_1'', C_2'')\in R$ then $( C_1''', C_2''')\in R$.
\end{enumerate}
 and vice-versa;(3) if $( C_1, C_2)\in R$, and $ C_1\xrsquigarrow{\pi} C_1^{\pi}$ then
 $ C_2\xrsquigarrow{\pi} C_2^{\pi}$ and $( C_1^{\pi}, C_2^{\pi})\in R$, and vice-versa; (4) if $( C_1, C_2)\in R$,
then $\mu(C_1,C)=\mu(C_2,C)$ for each $C\in\mathcal{C}(\mathcal{E})/R$; (5) $[\surd]_R=\{\surd\}$.

We say that $\mathcal{E}_1$, $\mathcal{E}_2$ are FR strongly probabilistic pomset bisimilar, written $\mathcal{E}_1\sim_{pp}^{fr}\mathcal{E}_2$, if there exists a FR strongly probabilistic pomset
bisimulation $R$, such that $(\emptyset,\emptyset)\in R$. By replacing FR probabilistic pomset transitions with steps, we can get the definition of FR strongly probabilistic step bisimulation.
When PESs $\mathcal{E}_1$ and $\mathcal{E}_2$ are FR strongly probabilistic step bisimilar, we write $\mathcal{E}_1\sim_{ps}^{fr}\mathcal{E}_2$.
\end{definition}

\begin{definition}[Posetal product]
Given two PESs $\mathcal{E}_1$, $\mathcal{E}_2$, the posetal product of their configurations, denoted
$\mathcal{C}(\mathcal{E}_1)\overline{\times}\mathcal{C}(\mathcal{E}_2)$, is defined as

$$\{( C_1,f, C_2)|C_1\in\mathcal{C}(\mathcal{E}_1),C_2\in\mathcal{C}(\mathcal{E}_2),f:C_1\rightarrow C_2 \textrm{ isomorphism}\}.$$

A subset $R\subseteq\mathcal{C}(\mathcal{E}_1)\overline{\times}\mathcal{C}(\mathcal{E}_2)$ is called a posetal relation. We say that $R$ is downward
closed when for any
$( C_1,f, C_2),( C_1',f', C_2')\in \mathcal{C}(\mathcal{E}_1)\overline{\times}\mathcal{C}(\mathcal{E}_2)$,
if $( C_1,f, C_2)\subseteq ( C_1',f', C_2')$ pointwise and $( C_1',f', C_2')\in R$,
then $( C_1,f, C_2)\in R$.

For $f:X_1\rightarrow X_2$, we define $f[x_1\mapsto x_2]:X_1\cup\{x_1\}\rightarrow X_2\cup\{x_2\}$, $z\in X_1\cup\{x_1\}$,(1)$f[x_1\mapsto x_2](z)=
x_2$,if $z=x_1$;(2)$f[x_1\mapsto x_2](z)=f(z)$, otherwise. Where $X_1\subseteq \mathbb{E}_1$, $X_2\subseteq \mathbb{E}_2$, $x_1\in \mathbb{E}_1$, $x_2\in \mathbb{E}_2$.
\end{definition}

\begin{definition}[FR strongly probabilistic (hereditary) history-preserving bisimilarity]
A FR strongly probabilistic history-preserving (hp-) bisimulation is a posetal relation $R\subseteq\mathcal{C}(\mathcal{E}_1)\overline{\times}\mathcal{C}(\mathcal{E}_2)$ such that
(1) if $( C_1,f, C_2)\in R$, and
\begin{enumerate}
  \item for $e_1=\alpha$ a fresh action, if $ C_1\xrightarrow{\alpha} C_1'$ (with $\mathcal{E}_1\xrightarrow{\alpha}\mathcal{E}_1'$), then for some
  $C_2'$ and $e_2=\alpha$, $ C_2\xrightarrow{\alpha} C_2'$ (with $\mathcal{E}_2\xrightarrow{\alpha}\mathcal{E}_2'$), such that
  $( C_1',f[e_1\mapsto e_2], C_2')\in R$;
  \item for $e_1=x(y)$ with ($y\notin n(\mathcal{E}_1, \mathcal{E}_2)$), if $ C_1\xrightarrow{x(y)} C_1'$ (with
  $\mathcal{E}_1\xrightarrow{x(y)}\mathcal{E}_1'\{w/y\}$) for all $w$, then for some $C_2'$ and $e_2=x(y)$, $ C_2\xrightarrow{x(y)} C_2'$ (with
  $\mathcal{E}_2\xrightarrow{x(y)}\mathcal{E}_2'\{w/y\}$) for all $w$, such that $( C_1',f[e_1\mapsto e_2], C_2')\in R$;
  \item for $e_1=\overline{x}(y)$ with $y\notin n(\mathcal{E}_1, \mathcal{E}_2)$, if $ C_1\xrightarrow{\overline{x}(y)} C_1'$ (with
  $\mathcal{E}_1\xrightarrow{\overline{x}(y)}\mathcal{E}_1'$), then for some $C_2'$ and $e_2=\overline{x}(y)$, $ C_2\xrightarrow{\overline{x}(y)} C_2'$
  (with $\mathcal{E}_2\xrightarrow{\overline{x}(y)}\mathcal{E}_2'$), such that $( C_1',f[e_1\mapsto e_2], C_2')\in R$.
\end{enumerate}
and vice-versa; (2) if $( C_1,f, C_2)\in R$, and
\begin{enumerate}
  \item for $e_1=\alpha$ a fresh action, if $ C_1\xtworightarrow{\alpha[m]} C_1'$ (with $\mathcal{E}_1\xtworightarrow{\alpha[m]}\mathcal{E}_1'$), then for some
  $C_2'$ and $e_2=\alpha$, $ C_2\xtworightarrow{\alpha[m]} C_2'$ (with $\mathcal{E}_2\xtworightarrow{\alpha[m]}\mathcal{E}_2'$), such that
  $( C_1',f[e_1\mapsto e_2], C_2')\in R$;
  \item for $e_1=x(y)$ with ($y\notin n(\mathcal{E}_1, \mathcal{E}_2)$), if $ C_1\xtworightarrow{x(y)[m]} C_1'$ (with
  $\mathcal{E}_1\xtworightarrow{x(y)[m]}\mathcal{E}_1'\{w/y\}$) for all $w$, then for some $C_2'$ and $e_2=x(y)$, $ C_2\xtworightarrow{x(y)[m]} C_2'$ (with
  $\mathcal{E}_2\xtworightarrow{x(y)[m]}\mathcal{E}_2'\{w/y\}$) for all $w$, such that $( C_1',f[e_1\mapsto e_2], C_2')\in R$;
  \item for $e_1=\overline{x}(y)$ with $y\notin n(\mathcal{E}_1, \mathcal{E}_2)$, if $ C_1\xtworightarrow{\overline{x}(y)[m]} C_1'$ (with
  $\mathcal{E}_1\xtworightarrow{\overline{x}(y)[m]}\mathcal{E}_1'$), then for some $C_2'$ and $e_2=\overline{x}(y)$, $ C_2\xtworightarrow{\overline{x}(y)[m]} C_2'$
  (with $\mathcal{E}_2\xtworightarrow{\overline{x}(y)[m]}\mathcal{E}_2'$), such that $( C_1',f[e_1\mapsto e_2], C_2')\in R$.
\end{enumerate}
and vice-versa; (3) if $( C_1,f, C_2)\in R$, and $ C_1\xrsquigarrow{\pi} C_1^{\pi}$ then
$ C_2\xrsquigarrow{\pi} C_2^{\pi}$ and $( C_1^{\pi},f, C_2^{\pi})\in R$, and vice-versa; (4) if
$( C_1,f, C_2)\in R$, then $\mu(C_1,C)=\mu(C_2,C)$ for each $C\in\mathcal{C}(\mathcal{E})/R$; (5) $[\surd]_R=\{\surd\}$. $\mathcal{E}_1,\mathcal{E}_2$
are FR strongly probabilistic history-preserving (hp-)bisimilar and are written $\mathcal{E}_1\sim_{php}^{fr}\mathcal{E}_2$ if there exists a FR strongly probabilistic hp-bisimulation
$R$ such that $(\emptyset,\emptyset,\emptyset)\in R$.

A FR strongly probabilistic hereditary history-preserving (hhp-)bisimulation is a downward closed FR strongly probabilistic hp-bisimulation. $\mathcal{E}_1,\mathcal{E}_2$ are FR
strongly probabilistic hereditary history-preserving (hhp-)bisimilar and are written $\mathcal{E}_1\sim_{phhp}^{fr}\mathcal{E}_2$.
\end{definition}

\subsection{Syntax and Operational Semantics}\label{sos5}

We assume an infinite set $\mathcal{N}$ of (action or event) names, and use $a,b,c,\cdots$ to range over $\mathcal{N}$, use $x,y,z,w,u,v$ as meta-variables over names. We denote by
$\overline{\mathcal{N}}$ the set of co-names and let $\overline{a},\overline{b},\overline{c},\cdots$ range over $\overline{\mathcal{N}}$. Then we set
$\mathcal{L}=\mathcal{N}\cup\overline{\mathcal{N}}$ as the set of labels, and use $l,\overline{l}$ to range over $\mathcal{L}$. We extend complementation to $\mathcal{L}$ such that
$\overline{\overline{a}}=a$. Let $\tau$ denote the silent step (internal action or event) and define $Act=\mathcal{L}\cup\{\tau\}$ to be the set of actions, $\alpha,\beta$ range over
$Act$. And $K,L$ are used to stand for subsets of $\mathcal{L}$ and $\overline{L}$ is used for the set of complements of labels in $L$.

Further, we introduce a set $\mathcal{X}$ of process variables, and a set $\mathcal{K}$ of process constants, and let $X,Y,\cdots$ range over $\mathcal{X}$, and $A,B,\cdots$ range over
$\mathcal{K}$. For each process constant $A$, a nonnegative arity $ar(A)$ is assigned to it. Let $\widetilde{x}=x_1,\cdots,x_{ar(A)}$ be a tuple of distinct name variables, then
$A(\widetilde{x})$ is called a process constant. $\widetilde{X}$ is a tuple of distinct process variables, and also $E,F,\cdots$ range over the recursive expressions. We write
$\mathcal{P}$ for the set of processes. Sometimes, we use $I,J$ to stand for an indexing set, and we write $E_i:i\in I$ for a family of expressions indexed by $I$. $Id_D$ is the
identity function or relation over set $D$. The symbol $\equiv_{\alpha}$ denotes equality under standard alpha-convertibility, note that the subscript $\alpha$ has no relation to the
action $\alpha$.

\subsubsection{Syntax}

We use the Prefix $.$ to model the causality relation $\leq$ in true concurrency, the Summation $+$ to model the conflict relation $\sharp$, and $\boxplus_{\pi}$ to model the probabilistic
conflict relation $\sharp_{\pi}$ in probabilistic true concurrency, and the Composition $\parallel$ to explicitly model concurrent relation in true concurrency. And we follow the
conventions of process algebra.

\begin{definition}[Syntax]\label{syntax5}
A truly concurrent process $\pi_{tc}$ with reversibility and probabilism is defined inductively by the following formation rules:

\begin{enumerate}
  \item $A(\widetilde{x})\in\mathcal{P}$;
  \item $\textbf{nil}\in\mathcal{P}$;
  \item if $P\in\mathcal{P}$, then the Prefix $\tau.P\in\mathcal{P}$, for $\tau\in Act$ is the silent action;
  \item if $P\in\mathcal{P}$, then the Output $\overline{x}y.P\in\mathcal{P}$, for $x,y\in Act$;
  \item if $P\in\mathcal{P}$, then the Output $P.\overline{x}y[m]\in\mathcal{P}$, for $x,y\in Act$;
  \item if $P\in\mathcal{P}$, then the Input $x(y).P\in\mathcal{P}$, for $x,y\in Act$;
  \item if $P\in\mathcal{P}$, then the Input $P.x(y)[m]\in\mathcal{P}$, for $x,y\in Act$;
  \item if $P\in\mathcal{P}$, then the Restriction $(x)P\in\mathcal{P}$, for $x\in Act$;
  \item if $P,Q\in\mathcal{P}$, then the Summation $P+Q\in\mathcal{P}$;
  \item if $P,Q\in\mathcal{P}$, then the Summation $P\boxplus_{\pi}Q\in\mathcal{P}$;
  \item if $P,Q\in\mathcal{P}$, then the Composition $P\parallel Q\in\mathcal{P}$;
\end{enumerate}

The standard BNF grammar of syntax of $\pi_{tc}$ with reversibility and probabilism can be summarized as follows:

$$P::=A(\widetilde{x})|\textbf{nil}|\tau.P| \overline{x}y.P | x(y).P|\overline{x}y[m].P | x(y)[m].P | (x)P  | P+P| P\boxplus_{\pi}P | P\parallel P.$$
\end{definition}

In $\overline{x}y$, $x(y)$ and $\overline{x}(y)$, $x$ is called the subject, $y$ is called the object and it may be free or bound.

\begin{definition}[Free variables]
The free names of a process $P$, $fn(P)$, are defined as follows.

\begin{enumerate}
  \item $fn(A(\widetilde{x}))\subseteq\{\widetilde{x}\}$;
  \item $fn(\textbf{nil})=\emptyset$;
  \item $fn(\tau.P)=fn(P)$;
  \item $fn(\overline{x}y.P)=fn(P)\cup\{x\}\cup\{y\}$;
  \item $fn(\overline{x}y[m].P)=fn(P)\cup\{x\}\cup\{y\}$;
  \item $fn(x(y).P)=fn(P)\cup\{x\}-\{y\}$;
  \item $fn(x(y)[m].P)=fn(P)\cup\{x\}-\{y\}$;
  \item $fn((x)P)=fn(P)-\{x\}$;
  \item $fn(P+Q)=fn(P)\cup fn(Q)$;
  \item $fn(P\boxplus_{\pi}Q)=fn(P)\cup fn(Q)$;
  \item $fn(P\parallel Q)=fn(P)\cup fn(Q)$.
\end{enumerate}
\end{definition}

\begin{definition}[Bound variables]
Let $n(P)$ be the names of a process $P$, then the bound names $bn(P)=n(P)-fn(P)$.
\end{definition}

For each process constant schema $A(\widetilde{x})$, a defining equation of the form

$$A(\widetilde{x})\overset{\text{def}}{=}P$$

is assumed, where $P$ is a process with $fn(P)\subseteq \{\widetilde{x}\}$.

\begin{definition}[Substitutions]\label{subs5}
A substitution is a function $\sigma:\mathcal{N}\rightarrow\mathcal{N}$. For $x_i\sigma=y_i$ with $1\leq i\leq n$, we write $\{y_1/x_1,\cdots,y_n/x_n\}$ or
$\{\widetilde{y}/\widetilde{x}\}$ for $\sigma$. For a process $P\in\mathcal{P}$, $P\sigma$ is defined inductively as follows:
\begin{enumerate}
  \item if $P$ is a process constant $A(\widetilde{x})=A(x_1,\cdots,x_n)$, then $P\sigma=A(x_1\sigma,\cdots,x_n\sigma)$;
  \item if $P=\textbf{nil}$, then $P\sigma=\textbf{nil}$;
  \item if $P=\tau.P'$, then $P\sigma=\tau.P'\sigma$;
  \item if $P=\overline{x}y.P'$, then $P\sigma=\overline{x\sigma}y\sigma.P'\sigma$;
  \item if $P=\overline{x}y[m].P'$, then $P\sigma=\overline{x\sigma}y\sigma[m].P'\sigma$;
  \item if $P=x(y).P'$, then $P\sigma=x\sigma(y).P'\sigma$;
  \item if $P=x(y)[m].P'$, then $P\sigma=x\sigma(y)[m].P'\sigma$;
  \item if $P=(x)P'$, then $P\sigma=(x\sigma)P'\sigma$;
  \item if $P=P_1+P_2$, then $P\sigma=P_1\sigma+P_2\sigma$;
  \item if $P=P_1\boxplus_{\pi}P_2$, then $P\sigma=P_1\sigma\boxplus_{\pi}P_2\sigma$;
  \item if $P=P_1\parallel P_2$, then $P\sigma=P_1\sigma \parallel P_2\sigma$.
\end{enumerate}
\end{definition}

\subsubsection{Operational Semantics}

The operational semantics is defined by LTSs (labelled transition systems), and it is detailed by the following definition.

\begin{definition}[Semantics]\label{semantics5}
The operational semantics of $\pi_{tc}$ with reversibility and probabilism corresponding to the syntax in Definition \ref{syntax5} is defined by a series of transition rules, named $\textbf{PACT}$, $\textbf{PSUM}$, $\textbf{PBOX-SUM}$,
$\textbf{PIDE}$, $\textbf{PPAR}$, $\textbf{PRES}$ and named $\textbf{ACT}$, $\textbf{SUM}$,
$\textbf{IDE}$, $\textbf{PAR}$, $\textbf{COM}$, $\textbf{CLOSE}$, $\textbf{RES}$, $\textbf{OPEN}$ indicate that the rules are associated respectively with Prefix, Summation, Box-Summation,
Identity, Parallel Composition, Communication, and Restriction in Definition \ref{syntax5}. They are shown in Table \ref{PTRForPITC5} and \ref{TRForPITC5}.

\begin{center}
    \begin{table}
        \[\textbf{PTAU-ACT}\quad \frac{}{\tau.P\rightsquigarrow \breve{\tau}.P}\]

        \[\textbf{POUTPUT-ACT}\quad \frac{}{\overline{x}y.P\rightsquigarrow \breve{\overline{x}y}.P}\]

        \[\textbf{PINPUT-ACT}\quad \frac{}{x(z).P\rightsquigarrow \breve{x(z)}.P}\]

        \[\textbf{PPAR}\quad \frac{P\rightsquigarrow P'\quad Q\rightsquigarrow Q'}{P\parallel Q\rightsquigarrow P'\parallel Q'}\]

        \[\textbf{PSUM}\quad \frac{P\rightsquigarrow P'\quad Q\rightsquigarrow Q'}{P+Q\rightsquigarrow P'+Q'}\]

        \[\textbf{PBOX-SUM}\quad \frac{P\rightsquigarrow P'}{P\boxplus_{\pi}Q\rightsquigarrow P'}\]

        \[\textbf{PIDE}\quad\frac{P\{\widetilde{y}/\widetilde{x}\}\rightsquigarrow P'}{A(\widetilde{y})\rightsquigarrow P'}\quad (A(\widetilde{x})\overset{\text{def}}{=}P)\]

        \[\textbf{PRES}\quad \frac{P\rightsquigarrow P'}{(y)P\rightsquigarrow (y)P'}\quad (y\notin n(\alpha))\]

        \caption{Probabilistic transition rules}
        \label{PTRForPITC5}
    \end{table}
\end{center}

\begin{center}
    \begin{table}
        \[\textbf{TAU-ACT}\quad \frac{}{\breve{\tau}.P\xrightarrow{\tau}P}\]

        \[\textbf{OUTPUT-ACT}\quad \frac{}{\breve{\overline{x}y}.P\xrightarrow{\overline{x}y}P}\]

        \[\textbf{INPUT-ACT}\quad \frac{}{\breve{x(z)}.P\xrightarrow{x(w)}P\{w/z\}}\quad (w\notin fn((z)P))\]

        \[\textbf{PAR}_1\quad \frac{P\xrightarrow{\alpha}P'\quad Q\nrightarrow}{P\parallel Q\xrightarrow{\alpha}P'\parallel Q}\quad (bn(\alpha)\cap fn(Q)=\emptyset)\]

        \[\textbf{PAR}_2\quad \frac{Q\xrightarrow{\alpha}Q'\quad P\nrightarrow}{P\parallel Q\xrightarrow{\alpha}P\parallel Q'}\quad (bn(\alpha)\cap fn(P)=\emptyset)\]

        \[\textbf{PAR}_3\quad \frac{P\xrightarrow{\alpha}P'\quad Q\xrightarrow{\beta}Q'}{P\parallel Q\xrightarrow{\{\alpha,\beta\}}P'\parallel Q'}\] $(\beta\neq\overline{\alpha}, bn(\alpha)\cap bn(\beta)=\emptyset, bn(\alpha)\cap fn(Q)=\emptyset,bn(\beta)\cap fn(P)=\emptyset)$

        \[\textbf{PAR}_4\quad \frac{P\xrightarrow{x_1(z)}P'\quad Q\xrightarrow{x_2(z)}Q'}{P\parallel Q\xrightarrow{\{x_1(w),x_2(w)\}}P'\{w/z\}\parallel Q'\{w/z\}}\quad (w\notin fn((z)P)\cup fn((z)Q))\]

        \[\textbf{COM}\quad \frac{P\xrightarrow{\overline{x}y}P'\quad Q\xrightarrow{x(z)}Q'}{P\parallel Q\xrightarrow{\tau}P'\parallel Q'\{y/z\}}\]

        \[\textbf{CLOSE}\quad \frac{P\xrightarrow{\overline{x}(w)}P'\quad Q\xrightarrow{x(w)}Q'}{P\parallel Q\xrightarrow{\tau}(w)(P'\parallel Q')}\]

%
%
%
%
%
%
%
%
        \caption{Forward action transition rules}
        \label{TRForPITC5}
    \end{table}
\end{center}

\begin{center}
    \begin{table}
%
%
%
%
%
%
%
%
%
        \[\textbf{SUM}_1\quad \frac{P\xrightarrow{\alpha}P'}{P+Q\xrightarrow{\alpha}P'}\]

        \[\textbf{SUM}_2\quad \frac{P\xrightarrow{\{\alpha_1,\cdots,\alpha_n\}}P'}{P+Q\xrightarrow{\{\alpha_1,\cdots,\alpha_n\}}P'}\]

        \[\textbf{IDE}_1\quad\frac{P\{\widetilde{y}/\widetilde{x}\}\xrightarrow{\alpha}P'}{A(\widetilde{y})\xrightarrow{\alpha}P'}\quad (A(\widetilde{x})\overset{\text{def}}{=}P)\]

        \[\textbf{IDE}_2\quad\frac{P\{\widetilde{y}/\widetilde{x}\}\xrightarrow{\{\alpha_1,\cdots,\alpha_n\}}P'} {A(\widetilde{y})\xrightarrow{\{\alpha_1,\cdots,\alpha_n\}}P'}\quad (A(\widetilde{x})\overset{\text{def}}{=}P)\]

        \[\textbf{RES}_1\quad \frac{P\xrightarrow{\alpha}P'}{(y)P\xrightarrow{\alpha}(y)P'}\quad (y\notin n(\alpha))\]

        \[\textbf{RES}_2\quad \frac{P\xrightarrow{\{\alpha_1,\cdots,\alpha_n\}}P'}{(y)P\xrightarrow{\{\alpha_1,\cdots,\alpha_n\}}(y)P'}\quad (y\notin n(\alpha_1)\cup\cdots\cup n(\alpha_n))\]

        \[\textbf{OPEN}_1\quad \frac{P\xrightarrow{\overline{x}y}P'}{(y)P\xrightarrow{\overline{x}(w)}P'\{w/y\}} \quad (y\neq x, w\notin fn((y)P'))\]

        \[\textbf{OPEN}_2\quad \frac{P\xrightarrow{\{\overline{x}_1 y,\cdots,\overline{x}_n y\}}P'}{(y)P\xrightarrow{\{\overline{x}_1(w),\cdots,\overline{x}_n(w)\}}P'\{w/y\}} \quad (y\neq x_1\neq\cdots\neq x_n, w\notin fn((y)P'))\]

        \caption{Forward action transition rules (continuing)}
        \label{TRForPITC52}
    \end{table}
\end{center}

\begin{center}
    \begin{table}
        \[\textbf{RTAU-ACT}\quad \frac{}{\breve{\tau}.P\xtworightarrow{\tau}P}\]

        \[\textbf{ROUTPUT-ACT}\quad \frac{}{\breve{\overline{x}y}[m].P\xtworightarrow{\overline{x}y[m]}P}\]

        \[\textbf{RINPUT-ACT}\quad \frac{}{\breve{x(z)}[m].P\xtworightarrow{x(w)[m]}P\{w/z\}}\quad (w\notin fn((z)P))\]

        \[\textbf{RPAR}_1\quad \frac{P\xtworightarrow{\alpha[m]}P'\quad Q\nrightarrow}{P\parallel Q\xtworightarrow{\alpha[m]}P'\parallel Q}\quad (bn(\alpha)\cap fn(Q)=\emptyset)\]

        \[\textbf{RPAR}_2\quad \frac{Q\xtworightarrow{\alpha[m]}Q'\quad P\nrightarrow}{P\parallel Q\xtworightarrow{\alpha[m]}P\parallel Q'}\quad (bn(\alpha)\cap fn(P)=\emptyset)\]

        \[\textbf{RPAR}_3\quad \frac{P\xtworightarrow{\alpha[m]}P'\quad Q\xtworightarrow{\beta[m]}Q'}{P\parallel Q\xtworightarrow{\{\alpha[m],\beta[m]\}}P'\parallel Q'}\] $(\beta\neq\overline{\alpha}, bn(\alpha)\cap bn(\beta)=\emptyset, bn(\alpha)\cap fn(Q)=\emptyset,bn(\beta)\cap fn(P)=\emptyset)$

        \[\textbf{RPAR}_4\quad \frac{P\xtworightarrow{x_1(z)[m]}P'\quad Q\xtworightarrow{x_2(z)[m]}Q'}{P\parallel Q\xtworightarrow{\{x_1(w)[m],x_2(w)[m]\}}P'\{w/z\}\parallel Q'\{w/z\}}\quad (w\notin fn((z)P)\cup fn((z)Q))\]

        \[\textbf{RCOM}\quad \frac{P\xtworightarrow{\overline{x}y[m]}P'\quad Q\xtworightarrow{x(z)[m]}Q'}{P\parallel Q\xtworightarrow{\tau}P'\parallel Q'\{y/z\}}\]

        \[\textbf{RCLOSE}\quad \frac{P\xtworightarrow{\overline{x}(w)[m]}P'\quad Q\xtworightarrow{x(w)[m]}Q'}{P\parallel Q\xtworightarrow{\tau}(w)(P'\parallel Q')}\]

%
%
%
%
%
%
%
%
        \caption{Reverse action transition rules}
        \label{TRForPITC53}
    \end{table}
\end{center}

\begin{center}
    \begin{table}
%
%
%
%
%
%
%
%
%
        \[\textbf{RSUM}_1\quad \frac{P\xtworightarrow{\alpha[m]}P'}{P+Q\xtworightarrow{\alpha[m]}P'}\]

        \[\textbf{RSUM}_2\quad \frac{P\xtworightarrow{\{\alpha_1[m],\cdots,\alpha_n[m]\}}P'}{P+Q\xtworightarrow{\{\alpha_1[m],\cdots,\alpha_n[m]\}}P'}\]

        \[\textbf{RIDE}_1\quad\frac{P\{\widetilde{y}/\widetilde{x}\}\xtworightarrow{\alpha[m]}P'}{A(\widetilde{y})\xtworightarrow{\alpha[m]}P'}\quad (A(\widetilde{x})\overset{\text{def}}{=}P)\]

        \[\textbf{RIDE}_2\quad\frac{P\{\widetilde{y}/\widetilde{x}\}\xtworightarrow{\{\alpha_1[m],\cdots,\alpha_n[m]\}}P'} {A(\widetilde{y})\xtworightarrow{\{\alpha_1[m],\cdots,\alpha_n[m]\}}P'}\quad (A(\widetilde{x})\overset{\text{def}}{=}P)\]

        \[\textbf{RRES}_1\quad \frac{P\xtworightarrow{\alpha[m]}P'}{(y)P\xtworightarrow{\alpha[m]}(y)P'}\quad (y\notin n(\alpha))\]

        \[\textbf{RRES}_2\quad \frac{P\xtworightarrow{\{\alpha_1[m],\cdots,\alpha_n[m]\}}P'}{(y)P\xtworightarrow{\{\alpha_1[m],\cdots,\alpha_n[m]\}}(y)P'}\quad (y\notin n(\alpha_1)\cup\cdots\cup n(\alpha_n))\]

        \[\textbf{ROPEN}_1\quad \frac{P\xtworightarrow{\overline{x}y[m]}P'}{(y)P\xtworightarrow{\overline{x}(w)[m]}P'\{w/y\}} \quad (y\neq x, w\notin fn((y)P'))\]

        \[\textbf{ROPEN}_2\quad \frac{P\xtworightarrow{\{\overline{x}_1 y[m],\cdots,\overline{x}_n y[m]\}}P'}{(y)P\xtworightarrow{\{\overline{x}_1(w)[m],\cdots,\overline{x}_n(w)[m]\}}P'\{w/y\}} \quad (y\neq x_1\neq\cdots\neq x_n, w\notin fn((y)P'))\]

        \caption{Reverse action transition rules (continuing)}
        \label{TRForPITC54}
    \end{table}
\end{center}
\end{definition}

\subsubsection{Properties of Transitions}

\begin{proposition}
\begin{enumerate}
  \item If $P\xrightarrow{\alpha}P'$ then
  \begin{enumerate}
    \item $fn(\alpha)\subseteq fn(P)$;
    \item $fn(P')\subseteq fn(P)\cup bn(\alpha)$;
  \end{enumerate}
  \item If $P\xrightarrow{\{\alpha_1,\cdots,\alpha_n\}}P'$ then
  \begin{enumerate}
    \item $fn(\alpha_1)\cup\cdots\cup fn(\alpha_n)\subseteq fn(P)$;
    \item $fn(P')\subseteq fn(P)\cup bn(\alpha_1)\cup\cdots\cup bn(\alpha_n)$.
  \end{enumerate}
\end{enumerate}
\end{proposition}

\begin{proof}
By induction on the depth of inference.
\end{proof}

\begin{proposition}
Suppose that $P\xrightarrow{\alpha(y)}P'$, where $\alpha=x$ or $\alpha=\overline{x}$, and $x\notin n(P)$, then there exists some $P''\equiv_{\alpha}P'\{z/y\}$,
$P\xrightarrow{\alpha(z)}P''$.
\end{proposition}

\begin{proof}
By induction on the depth of inference.
\end{proof}

\begin{proposition}
If $P\xrightarrow{\alpha} P'$, $bn(\alpha)\cap fn(P'\sigma)=\emptyset$, and $\sigma\lceil bn(\alpha)=id$, then there exists some $P''\equiv_{\alpha}P'\sigma$,
$P\sigma\xrightarrow{\alpha\sigma}P''$.
\end{proposition}

\begin{proof}
By the definition of substitution (Definition \ref{subs5}) and induction on the depth of inference.
\end{proof}

\begin{proposition}
\begin{enumerate}
  \item If $P\{w/z\}\xrightarrow{\alpha}P'$, where $w\notin fn(P)$ and $bn(\alpha)\cap fn(P,w)=\emptyset$, then there exist some $Q$ and $\beta$ with $Q\{w/z\}\equiv_{\alpha}P'$ and
  $\beta\sigma=\alpha$, $P\xrightarrow{\beta}Q$;
  \item If $P\{w/z\}\xrightarrow{\{\alpha_1,\cdots,\alpha_n\}}P'$, where $w\notin fn(P)$ and $bn(\alpha_1)\cap\cdots\cap bn(\alpha_n)\cap fn(P,w)=\emptyset$, then there exist some $Q$
  and $\beta_1,\cdots,\beta_n$ with $Q\{w/z\}\equiv_{\alpha}P'$ and $\beta_1\sigma=\alpha_1,\cdots,\beta_n\sigma=\alpha_n$, $P\xrightarrow{\{\beta_1,\cdots,\beta_n\}}Q$.
\end{enumerate}

\end{proposition}

\begin{proof}
By the definition of substitution (Definition \ref{subs5}) and induction on the depth of inference.
\end{proof}

\begin{proposition}
\begin{enumerate}
  \item If $P\xtworightarrow{\alpha[m]}P'$ then
  \begin{enumerate}
    \item $fn(\alpha[m])\subseteq fn(P)$;
    \item $fn(P')\subseteq fn(P)\cup bn(\alpha[m])$;
  \end{enumerate}
  \item If $P\xtworightarrow{\{\alpha_1[m],\cdots,\alpha_n[m]\}}P'$ then
  \begin{enumerate}
    \item $fn(\alpha_1[m])\cup\cdots\cup fn(\alpha_n[m])\subseteq fn(P)$;
    \item $fn(P')\subseteq fn(P)\cup bn(\alpha_1[m])\cup\cdots\cup bn(\alpha_n[m])$.
  \end{enumerate}
\end{enumerate}
\end{proposition}

\begin{proof}
By induction on the depth of inference.
\end{proof}

\begin{proposition}
Suppose that $P\xtworightarrow{\alpha(y)[m]}P'$, where $\alpha=x$ or $\alpha=\overline{x}$, and $x\notin n(P)$, then there exists some $P''\equiv_{\alpha}P'\{z/y\}$,
$P\xtworightarrow{\alpha(z)[m]}P''$.
\end{proposition}

\begin{proof}
By induction on the depth of inference.
\end{proof}

\begin{proposition}
If $P\xtworightarrow{\alpha[m]} P'$, $bn(\alpha[m])\cap fn(P'\sigma)=\emptyset$, and $\sigma\lceil bn(\alpha[m])=id$, then there exists some $P''\equiv_{\alpha}P'\sigma$,
$P\sigma\xtworightarrow{\alpha[m]\sigma}P''$.
\end{proposition}

\begin{proof}
By the definition of substitution (Definition \ref{subs5}) and induction on the depth of inference.
\end{proof}

\begin{proposition}
\begin{enumerate}
  \item If $P\{w/z\}\xtworightarrow{\alpha[m]}P'$, where $w\notin fn(P)$ and $bn(\alpha)\cap fn(P,w)=\emptyset$, then there exist some $Q$ and $\beta$ with $Q\{w/z\}\equiv_{\alpha}P'$ and
  $\beta\sigma[m]=\alpha[m]$, $P\xtworightarrow{\beta[m]}Q$;
  \item If $P\{w/z\}\xtworightarrow{\{\alpha_1[m],\cdots,\alpha_n[m]\}}P'$, where $w\notin fn(P)$ and $bn(\alpha_1[m])\cap\cdots\cap bn(\alpha_n[m])\cap fn(P,w)=\emptyset$, then there exist some $Q$
  and $\beta_1[m],\cdots,\beta_n[m]$ with $Q\{w/z\}\equiv_{\alpha}P'$ and $\beta_1\sigma[m]=\alpha_1[m],\cdots,\beta_n\sigma[m]=\alpha_n[m]$, $P\xtworightarrow{\{\beta_1[m],\cdots,\beta_n[m]\}}Q$.
\end{enumerate}

\end{proposition}

\begin{proof}
By the definition of substitution (Definition \ref{subs5}) and induction on the depth of inference.
\end{proof}

\subsection{Strong Bisimilarities}\label{s5}

\subsubsection{Laws and Congruence}

\begin{theorem}
$\equiv_{\alpha}$ are FR strongly probabilistic truly concurrent bisimulations. That is, if $P\equiv_{\alpha}Q$, then,
\begin{enumerate}
  \item $P\sim_{pp}^{fr} Q$;
  \item $P\sim_{ps}^{fr} Q$;
  \item $P\sim_{php}^{fr} Q$;
  \item $P\sim_{phhp}^{fr} Q$.
\end{enumerate}
\end{theorem}

\begin{proof}
By induction on the depth of inference, we can get the following facts:

\begin{enumerate}
  \item If $\alpha$ is a free action and $P\rightsquigarrow\xrightarrow{\alpha}P'$, then equally for some $Q'$ with $P'\equiv_{\alpha}Q'$,
  $Q\rightsquigarrow\xrightarrow{\alpha}Q'$;
  \item If $P\rightsquigarrow\xrightarrow{a(y)}P'$ with $a=x$ or $a=\overline{x}$ and $z\notin n(Q)$, then equally for some $Q'$ with $P'\{z/y\}\equiv_{\alpha}Q'$,
  $Q\rightsquigarrow\xrightarrow{a(z)}Q'$;
  \item If $\alpha[m]$ is a free action and $P\rightsquigarrow\xtworightarrow{\alpha[m]}P'$, then equally for some $Q'$ with $P'\equiv_{\alpha}Q'$,
  $Q\rightsquigarrow\xtworightarrow{\alpha[m]}Q'$;
  \item If $P\rightsquigarrow\xtworightarrow{a(y)[m]}P'$ with $a=x$ or $a=\overline{x}$ and $z\notin n(Q)$, then equally for some $Q'$ with $P'\{z/y\}\equiv_{\alpha}Q'$,
  $Q\rightsquigarrow\xtworightarrow{a(z)[m]}Q'$.
\end{enumerate}

Then, we can get:

\begin{enumerate}
  \item by the definition of FR strongly probabilistic pomset bisimilarity, $P\sim_{pp}^{fr} Q$;
  \item by the definition of FR strongly probabilistic step bisimilarity, $P\sim_{ps}^{fr} Q$;
  \item by the definition of FR strongly probabilistic hp-bisimilarity, $P\sim_{php}^{fr} Q$;
  \item by the definition of FR strongly probabilistic hhp-bisimilarity, $P\sim_{phhp}^{fr} Q$.
\end{enumerate}
\end{proof}

\begin{proposition}[Summation laws for FR strongly probabilistic pomset bisimulation] The Summation laws for FR strongly probabilistic pomset bisimulation are as follows.

\begin{enumerate}
  \item $P+Q\sim_{pp}^{fr} Q+P$;
  \item $P+(Q+R)\sim_{pp}^{fr} (P+Q)+R$;
  \item $P+P\sim_{pp}^{fr} P$;
  \item $P+\textbf{nil}\sim_{pp}^{fr} P$.
\end{enumerate}

\end{proposition}

\begin{proof}
\begin{enumerate}
  \item $P+Q\sim_{pp}^{fr} Q+P$. It is sufficient to prove the relation $R=\{(P+Q, Q+P)\}\cup \textbf{Id}$ is a FR strongly probabilistic pomset bisimulation, we omit it;
  \item $P+(Q+R)\sim_{pp}^{fr} (P+Q)+R$. It is sufficient to prove the relation $R=\{(P+(Q+R), (P+Q)+R)\}\cup \textbf{Id}$ is a FR strongly probabilistic pomset bisimulation, we omit it;
  \item $P+P\sim_{pp}^{fr} P$. It is sufficient to prove the relation $R=\{(P+P, P)\}\cup \textbf{Id}$ is a FR strongly probabilistic pomset bisimulation, we omit it;
  \item $P+\textbf{nil}\sim_{pp}^{fr} P$. It is sufficient to prove the relation $R=\{(P+\textbf{nil}, P)\}\cup \textbf{Id}$ is a FR strongly probabilistic pomset bisimulation, we omit it.
\end{enumerate}
\end{proof}

\begin{proposition}[Summation laws for FR strongly probabilistic step bisimulation] The Summation laws for FR strongly probabilistic step bisimulation are as follows.
\begin{enumerate}
  \item $P+Q\sim_{ps}^{fr} Q+P$;
  \item $P+(Q+R)\sim_{ps}^{fr} (P+Q)+R$;
  \item $P+P\sim_{ps}^{fr} P$;
  \item $P+\textbf{nil}\sim_{ps}^{fr} P$.
\end{enumerate}
\end{proposition}

\begin{proof}
\begin{enumerate}
  \item $P+Q\sim_{ps}^{fr} Q+P$. It is sufficient to prove the relation $R=\{(P+Q, Q+P)\}\cup \textbf{Id}$ is a FR strongly probabilistic step bisimulation, we omit it;
  \item $P+(Q+R)\sim_{ps}^{fr} (P+Q)+R$. It is sufficient to prove the relation $R=\{(P+(Q+R), (P+Q)+R)\}\cup \textbf{Id}$ is a FR strongly probabilistic step bisimulation, we omit it;
  \item $P+P\sim_{ps}^{fr} P$. It is sufficient to prove the relation $R=\{(P+P, P)\}\cup \textbf{Id}$ is a FR strongly probabilistic step bisimulation, we omit it;
  \item $P+\textbf{nil}\sim_{ps}^{fr} P$. It is sufficient to prove the relation $R=\{(P+\textbf{nil}, P)\}\cup \textbf{Id}$ is a FR strongly probabilistic step bisimulation, we omit it.
\end{enumerate}
\end{proof}

\begin{proposition}[Summation laws for FR strongly probabilistic hp-bisimulation] The Summation laws for FR strongly probabilistic hp-bisimulation are as follows.
\begin{enumerate}
  \item $P+Q\sim_{php}^{fr} Q+P$;
  \item $P+(Q+R)\sim_{php}^{fr} (P+Q)+R$;
  \item $P+P\sim_{php}^{fr} P$;
  \item $P+\textbf{nil}\sim_{php}^{fr} P$.
\end{enumerate}
\end{proposition}

\begin{proof}
\begin{enumerate}
  \item $P+Q\sim_{php}^{fr} Q+P$. It is sufficient to prove the relation $R=\{(P+Q, Q+P)\}\cup \textbf{Id}$ is a FR strongly probabilistic hp-bisimulation, we omit it;
  \item $P+(Q+R)\sim_{php}^{fr} (P+Q)+R$. It is sufficient to prove the relation $R=\{(P+(Q+R), (P+Q)+R)\}\cup \textbf{Id}$ is a FR strongly probabilistic hp-bisimulation, we omit it;
  \item $P+P\sim_{php}^{fr} P$. It is sufficient to prove the relation $R=\{(P+P, P)\}\cup \textbf{Id}$ is a FR strongly probabilistic hp-bisimulation, we omit it;
  \item $P+\textbf{nil}\sim_{php}^{fr} P$. It is sufficient to prove the relation $R=\{(P+\textbf{nil}, P)\}\cup \textbf{Id}$ is a FR strongly probabilistic hp-bisimulation, we omit it.
\end{enumerate}
\end{proof}

\begin{proposition}[Summation laws for FR strongly probabilistic hhp-bisimulation] The Summation laws for FR strongly probabilistic hhp-bisimulation are as follows.
\begin{enumerate}
  \item $P+Q\sim_{phhp}^{fr} Q+P$;
  \item $P+(Q+R)\sim_{phhp}^{fr} (P+Q)+R$;
  \item $P+P\sim_{phhp}^{fr} P$;
  \item $P+\textbf{nil}\sim_{phhp}^{fr} P$.
\end{enumerate}
\end{proposition}

\begin{proof}
\begin{enumerate}
  \item $P+Q\sim_{phhp}^{fr} Q+P$. It is sufficient to prove the relation $R=\{(P+Q, Q+P)\}\cup \textbf{Id}$ is a FR strongly probabilistic hhp-bisimulation, we omit it;
  \item $P+(Q+R)\sim_{phhp}^{fr} (P+Q)+R$. It is sufficient to prove the relation $R=\{(P+(Q+R), (P+Q)+R)\}\cup \textbf{Id}$ is a FR strongly probabilistic hhp-bisimulation, we omit it;
  \item $P+P\sim_{phhp}^{fr} P$. It is sufficient to prove the relation $R=\{(P+P, P)\}\cup \textbf{Id}$ is a FR strongly probabilistic hhp-bisimulation, we omit it;
  \item $P+\textbf{nil}\sim_{phhp}^{fr} P$. It is sufficient to prove the relation $R=\{(P+\textbf{nil}, P)\}\cup \textbf{Id}$ is a FR strongly probabilistic hhp-bisimulation, we omit it.
\end{enumerate}
\end{proof}

\begin{proposition}[Box-Summation laws for FR strongly probabilistic pomset bisimulation]
The Box-Summation laws for FR strongly probabilistic pomset bisimulation are as follows.

\begin{enumerate}
  \item $P\boxplus_{\pi} Q\sim_{pp}^{fr} Q\boxplus_{1-\pi} P$;
  \item $P\boxplus_{\pi}(Q\boxplus_{\rho} R)\sim_{pp}^{fr} (P\boxplus_{\frac{\pi}{\pi+\rho-\pi\rho}}Q)\boxplus_{\pi+\rho-\pi\rho} R$;
  \item $P\boxplus_{\pi}P\sim_{pp}^{fr} P$;
  \item $P\boxplus_{\pi}\textbf{nil}\sim_{pp}^{fr} P$.
\end{enumerate}
\end{proposition}

\begin{proof}
\begin{enumerate}
  \item $P\boxplus_{\pi} Q\sim_{pp}^{fr} Q\boxplus_{1-\pi} P$. It is sufficient to prove the relation $R=\{(P\boxplus_{\pi} Q, Q\boxplus_{1-\pi} P)\}\cup \textbf{Id}$ is a FR strongly probabilistic pomset bisimulation, we omit it;
  \item $P\boxplus_{\pi}(Q\boxplus_{\rho} R)\sim_{pp}^{fr} (P\boxplus_{\frac{\pi}{\pi+\rho-\pi\rho}}Q)\boxplus_{\pi+\rho-\pi\rho} R$. It is sufficient to prove the relation $R=\{(P\boxplus_{\pi}(Q\boxplus_{\rho} R), (P\boxplus_{\frac{\pi}{\pi+\rho-\pi\rho}}Q)\boxplus_{\pi+\rho-\pi\rho} R)\}\cup \textbf{Id}$ is a FR strongly probabilistic pomset bisimulation, we omit it;
  \item $P\boxplus_{\pi}P\sim_{pp}^{fr} P$. It is sufficient to prove the relation $R=\{(P\boxplus_{\pi}P, P)\}\cup \textbf{Id}$ is a FR strongly probabilistic pomset bisimulation, we omit it;
  \item $P\boxplus_{\pi}\textbf{nil}\sim_{pp}^{fr} P$. It is sufficient to prove the relation $R=\{(P\boxplus_{\pi}\textbf{nil}, P)\}\cup \textbf{Id}$ is a FR strongly probabilistic pomset bisimulation, we omit it.
\end{enumerate}
\end{proof}

\begin{proposition}[Box-Summation laws for FR strongly probabilistic step bisimulation]
The Box-Summation laws for FR strongly probabilistic step bisimulation are as follows.

\begin{enumerate}
  \item $P\boxplus_{\pi} Q\sim_{ps}^{fr} Q\boxplus_{1-\pi} P$;
  \item $P\boxplus_{\pi}(Q\boxplus_{\rho} R)\sim_{ps}^{fr} (P\boxplus_{\frac{\pi}{\pi+\rho-\pi\rho}}Q)\boxplus_{\pi+\rho-\pi\rho} R$;
  \item $P\boxplus_{\pi}P\sim_{ps}^{fr} P$;
  \item $P\boxplus_{\pi}\textbf{nil}\sim_{ps}^{fr} P$.
\end{enumerate}
\end{proposition}

\begin{proof}
\begin{enumerate}
  \item $P\boxplus_{\pi} Q\sim_{ps}^{fr} Q\boxplus_{1-\pi} P$. It is sufficient to prove the relation $R=\{(P\boxplus_{\pi} Q, Q\boxplus_{1-\pi} P)\}\cup \textbf{Id}$ is a FR strongly probabilistic step bisimulation, we omit it;
  \item $P\boxplus_{\pi}(Q\boxplus_{\rho} R)\sim_{ps}^{fr} (P\boxplus_{\frac{\pi}{\pi+\rho-\pi\rho}}Q)\boxplus_{\pi+\rho-\pi\rho} R$. It is sufficient to prove the relation $R=\{(P\boxplus_{\pi}(Q\boxplus_{\rho} R), (P\boxplus_{\frac{\pi}{\pi+\rho-\pi\rho}}Q)\boxplus_{\pi+\rho-\pi\rho} R)\}\cup \textbf{Id}$ is a FR strongly probabilistic step bisimulation, we omit it;
  \item $P\boxplus_{\pi}P\sim_{ps}^{fr} P$. It is sufficient to prove the relation $R=\{(P\boxplus_{\pi}P, P)\}\cup \textbf{Id}$ is a FR strongly probabilistic step bisimulation, we omit it;
  \item $P\boxplus_{\pi}\textbf{nil}\sim_{ps}^{fr} P$. It is sufficient to prove the relation $R=\{(P\boxplus_{\pi}\textbf{nil}, P)\}\cup \textbf{Id}$ is a FR strongly probabilistic step bisimulation, we omit it.
\end{enumerate}
\end{proof}

\begin{proposition}[Box-Summation laws for FR strongly probabilistic hp-bisimulation]
The Box-Summation laws for FR strongly probabilistic hp-bisimulation are as follows.

\begin{enumerate}
  \item $P\boxplus_{\pi} Q\sim_{php}^{fr} Q\boxplus_{1-\pi} P$;
  \item $P\boxplus_{\pi}(Q\boxplus_{\rho} R)\sim_{php}^{fr} (P\boxplus_{\frac{\pi}{\pi+\rho-\pi\rho}}Q)\boxplus_{\pi+\rho-\pi\rho} R$;
  \item $P\boxplus_{\pi}P\sim_{php}^{fr} P$;
  \item $P\boxplus_{\pi}\textbf{nil}\sim_{php}^{fr} P$.
\end{enumerate}
\end{proposition}

\begin{proof}
\begin{enumerate}
  \item $P\boxplus_{\pi} Q\sim_{php}^{fr} Q\boxplus_{1-\pi} P$. It is sufficient to prove the relation $R=\{(P\boxplus_{\pi} Q, Q\boxplus_{1-\pi} P)\}\cup \textbf{Id}$ is a FR strongly probabilistic hp-bisimulation, we omit it;
  \item $P\boxplus_{\pi}(Q\boxplus_{\rho} R)\sim_{php}^{fr} (P\boxplus_{\frac{\pi}{\pi+\rho-\pi\rho}}Q)\boxplus_{\pi+\rho-\pi\rho} R$. It is sufficient to prove the relation $R=\{(P\boxplus_{\pi}(Q\boxplus_{\rho} R), (P\boxplus_{\frac{\pi}{\pi+\rho-\pi\rho}}Q)\boxplus_{\pi+\rho-\pi\rho} R)\}\cup \textbf{Id}$ is a FR strongly probabilistic hp-bisimulation, we omit it;
  \item $P\boxplus_{\pi}P\sim_{php}^{fr} P$. It is sufficient to prove the relation $R=\{(P\boxplus_{\pi}P, P)\}\cup \textbf{Id}$ is a FR strongly probabilistic hp-bisimulation, we omit it;
  \item $P\boxplus_{\pi}\textbf{nil}\sim_{php}^{fr} P$. It is sufficient to prove the relation $R=\{(P\boxplus_{\pi}\textbf{nil}, P)\}\cup \textbf{Id}$ is a FR strongly probabilistic hp-bisimulation, we omit it.
\end{enumerate}
\end{proof}

\begin{proposition}[Box-Summation laws for FR strongly probabilistic hhp-bisimulation]
The Box-Summation laws for FR strongly probabilistic hhp-bisimulation are as follows.

\begin{enumerate}
  \item $P\boxplus_{\pi} Q\sim_{phhp}^{fr} Q\boxplus_{1-\pi} P$;
  \item $P\boxplus_{\pi}(Q\boxplus_{\rho} R)\sim_{phhp}^{fr} (P\boxplus_{\frac{\pi}{\pi+\rho-\pi\rho}}Q)\boxplus_{\pi+\rho-\pi\rho} R$;
  \item $P\boxplus_{\pi}P\sim_{phhp}^{fr} P$;
  \item $P\boxplus_{\pi}\textbf{nil}\sim_{phhp}^{fr} P$.
\end{enumerate}
\end{proposition}

\begin{proof}
\begin{enumerate}
  \item $P\boxplus_{\pi} Q\sim_{phhp}^{fr} Q\boxplus_{1-\pi} P$. It is sufficient to prove the relation $R=\{(P\boxplus_{\pi} Q, Q\boxplus_{1-\pi} P)\}\cup \textbf{Id}$ is a FR strongly probabilistic hhp-bisimulation, we omit it;
  \item $P\boxplus_{\pi}(Q\boxplus_{\rho} R)\sim_{phhp}^{fr} (P\boxplus_{\frac{\pi}{\pi+\rho-\pi\rho}}Q)\boxplus_{\pi+\rho-\pi\rho} R$. It is sufficient to prove the relation $R=\{(P\boxplus_{\pi}(Q\boxplus_{\rho} R), (P\boxplus_{\frac{\pi}{\pi+\rho-\pi\rho}}Q)\boxplus_{\pi+\rho-\pi\rho} R)\}\cup \textbf{Id}$ is a FR strongly probabilistic hhp-bisimulation, we omit it;
  \item $P\boxplus_{\pi}P\sim_{phhp}^{fr} P$. It is sufficient to prove the relation $R=\{(P\boxplus_{\pi}P, P)\}\cup \textbf{Id}$ is a FR strongly probabilistic hhp-bisimulation, we omit it;
  \item $P\boxplus_{\pi}\textbf{nil}\sim_{phhp}^{fr} P$. It is sufficient to prove the relation $R=\{(P\boxplus_{\pi}\textbf{nil}, P)\}\cup \textbf{Id}$ is a FR strongly probabilistic hhp-bisimulation, we omit it.
\end{enumerate}
\end{proof}

\begin{theorem}[Identity law for FR strongly probabilistic truly concurrent bisimilarities]
If $A(\widetilde{x})\overset{\text{def}}{=}P$, then

\begin{enumerate}
  \item $A(\widetilde{y})\sim_{pp}^{fr} P\{\widetilde{y}/\widetilde{x}\}$;
  \item $A(\widetilde{y})\sim_{ps}^{fr} P\{\widetilde{y}/\widetilde{x}\}$;
  \item $A(\widetilde{y})\sim_{php}^{fr} P\{\widetilde{y}/\widetilde{x}\}$;
  \item $A(\widetilde{y})\sim_{phhp}^{fr} P\{\widetilde{y}/\widetilde{x}\}$.
\end{enumerate}
\end{theorem}

\begin{proof}
\begin{enumerate}
  \item $A(\widetilde{y})\sim_{pp}^{fr} P\{\widetilde{y}/\widetilde{x}\}$. It is sufficient to prove the relation $R=\{(A(\widetilde{y}), P\{\widetilde{y}/\widetilde{x}\})\}\cup \textbf{Id}$ is a FR strongly probabilistic pomset bisimulation, we omit it;
  \item $A(\widetilde{y})\sim_{ps}^{fr} P\{\widetilde{y}/\widetilde{x}\}$. It is sufficient to prove the relation $R=\{(A(\widetilde{y}), P\{\widetilde{y}/\widetilde{x}\})\}\cup \textbf{Id}$ is a FR strongly probabilistic step bisimulation, we omit it;
  \item $A(\widetilde{y})\sim_{php}^{fr} P\{\widetilde{y}/\widetilde{x}\}$. It is sufficient to prove the relation $R=\{(A(\widetilde{y}), P\{\widetilde{y}/\widetilde{x}\})\}\cup \textbf{Id}$ is a FR strongly probabilistic hp-bisimulation, we omit it;
  \item $A(\widetilde{y})\sim_{phhp}^{fr} P\{\widetilde{y}/\widetilde{x}\}$. It is sufficient to prove the relation $R=\{(A(\widetilde{y}), P\{\widetilde{y}/\widetilde{x}\})\}\cup \textbf{Id}$ is a FR strongly probabilistic hhp-bisimulation, we omit it.
\end{enumerate}
\end{proof}

\begin{theorem}[Restriction Laws for FR strongly probabilistic pomset bisimilarity]
The restriction laws for FR strongly probabilistic pomset bisimilarity are as follows.

\begin{enumerate}
  \item $(y)P\sim_{pp}^{fr} P$, if $y\notin fn(P)$;
  \item $(y)(z)P\sim_{pp}^{fr} (z)(y)P$;
  \item $(y)(P+Q)\sim_{pp}^{fr} (y)P+(y)Q$;
  \item $(y)(P\boxplus_{\pi}Q)\sim_{pp}^{fr} (y)P\boxplus_{\pi}(y)Q$;
  \item $(y)\alpha.P\sim_{pp}^{fr} \alpha.(y)P$ if $y\notin n(\alpha)$;
  \item $(y)\alpha.P\sim_{pp}^{fr} \textbf{nil}$ if $y$ is the subject of $\alpha$.
\end{enumerate}
\end{theorem}

\begin{proof}
\begin{enumerate}
  \item $(y)P\sim_{pp}^{fr} P$, if $y\notin fn(P)$. It is sufficient to prove the relation $R=\{((y)P, P)\}\cup \textbf{Id}$, if $y\notin fn(P)$, is a FR strongly probabilistic pomset bisimulation, we omit it;
  \item $(y)(z)P\sim_{pp}^{fr} (z)(y)P$. It is sufficient to prove the relation $R=\{((y)(z)P, (z)(y)P)\}\cup \textbf{Id}$ is a FR strongly probabilistic pomset bisimulation, we omit it;
  \item $(y)(P+Q)\sim_{pp}^{fr} (y)P+(y)Q$. It is sufficient to prove the relation $R=\{((y)(P+Q), (y)P+(y)Q)\}\cup \textbf{Id}$ is a FR strongly probabilistic pomset bisimulation, we omit it;
  \item $(y)(P\boxplus_{\pi}Q)\sim_{pp}^{fr} (y)P\boxplus_{\pi}(y)Q$. It is sufficient to prove the relation $R=\{((y)(P\boxplus_{\pi}Q), (y)P\boxplus_{\pi}(y)Q)\}\cup \textbf{Id}$ is a FR strongly probabilistic pomset bisimulation, we omit it;
  \item $(y)\alpha.P\sim_{pp}^{fr} \alpha.(y)P$ if $y\notin n(\alpha)$. It is sufficient to prove the relation $R=\{((y)\alpha.P, \alpha.(y)P)\}\cup \textbf{Id}$, if $y\notin n(\alpha)$, is a FR strongly probabilistic pomset bisimulation, we omit it;
  \item $(y)\alpha.P\sim_{pp}^{fr} \textbf{nil}$ if $y$ is the subject of $\alpha$. It is sufficient to prove the relation $R=\{((y)\alpha.P, \textbf{nil})\}\cup \textbf{Id}$, if $y$ is the subject of $\alpha$, is a FR strongly probabilistic pomset bisimulation, we omit it.
\end{enumerate}
\end{proof}

\begin{theorem}[Restriction Laws for FR strongly probabilistic step bisimilarity]
The restriction laws for FR strongly probabilistic step bisimilarity are as follows.

\begin{enumerate}
  \item $(y)P\sim_{ps}^{fr} P$, if $y\notin fn(P)$;
  \item $(y)(z)P\sim_{ps}^{fr} (z)(y)P$;
  \item $(y)(P+Q)\sim_{ps}^{fr} (y)P+(y)Q$;
  \item $(y)(P\boxplus_{\pi}Q)\sim_{ps}^{fr} (y)P\boxplus_{\pi}(y)Q$;
  \item $(y)\alpha.P\sim_{ps}^{fr} \alpha.(y)P$ if $y\notin n(\alpha)$;
  \item $(y)\alpha.P\sim_{ps}^{fr} \textbf{nil}$ if $y$ is the subject of $\alpha$.
\end{enumerate}
\end{theorem}

\begin{proof}
\begin{enumerate}
  \item $(y)P\sim_{ps}^{fr} P$, if $y\notin fn(P)$. It is sufficient to prove the relation $R=\{((y)P, P)\}\cup \textbf{Id}$, if $y\notin fn(P)$, is a FR strongly probabilistic step bisimulation, we omit it;
  \item $(y)(z)P\sim_{ps}^{fr} (z)(y)P$. It is sufficient to prove the relation $R=\{((y)(z)P, (z)(y)P)\}\cup \textbf{Id}$ is a FR strongly probabilistic step bisimulation, we omit it;
  \item $(y)(P+Q)\sim_{ps}^{fr} (y)P+(y)Q$. It is sufficient to prove the relation $R=\{((y)(P+Q), (y)P+(y)Q)\}\cup \textbf{Id}$ is a FR strongly probabilistic step bisimulation, we omit it;
  \item $(y)(P\boxplus_{\pi}Q)\sim_{ps}^{fr} (y)P\boxplus_{\pi}(y)Q$. It is sufficient to prove the relation $R=\{((y)(P\boxplus_{\pi}Q), (y)P\boxplus_{\pi}(y)Q)\}\cup \textbf{Id}$ is a FR strongly probabilistic step bisimulation, we omit it;
  \item $(y)\alpha.P\sim_{ps}^{fr} \alpha.(y)P$ if $y\notin n(\alpha)$. It is sufficient to prove the relation $R=\{((y)\alpha.P, \alpha.(y)P)\}\cup \textbf{Id}$, if $y\notin n(\alpha)$, is a FR strongly probabilistic step bisimulation, we omit it;
  \item $(y)\alpha.P\sim_{ps}^{fr} \textbf{nil}$ if $y$ is the subject of $\alpha$. It is sufficient to prove the relation $R=\{((y)\alpha.P, \textbf{nil})\}\cup \textbf{Id}$, if $y$ is the subject of $\alpha$, is a FR strongly probabilistic step bisimulation, we omit it.
\end{enumerate}
\end{proof}

\begin{theorem}[Restriction Laws for FR strongly probabilistic hp-bisimilarity]
The restriction laws for FR strongly probabilistic hp-bisimilarity are as follows.

\begin{enumerate}
  \item $(y)P\sim_{php}^{fr} P$, if $y\notin fn(P)$;
  \item $(y)(z)P\sim_{php}^{fr} (z)(y)P$;
  \item $(y)(P+Q)\sim_{php}^{fr} (y)P+(y)Q$;
  \item $(y)(P\boxplus_{\pi}Q)\sim_{php}^{fr} (y)P\boxplus_{\pi}(y)Q$;
  \item $(y)\alpha.P\sim_{php}^{fr} \alpha.(y)P$ if $y\notin n(\alpha)$;
  \item $(y)\alpha.P\sim_{php}^{fr} \textbf{nil}$ if $y$ is the subject of $\alpha$.
\end{enumerate}
\end{theorem}

\begin{proof}
\begin{enumerate}
  \item $(y)P\sim_{php}^{fr} P$, if $y\notin fn(P)$. It is sufficient to prove the relation $R=\{((y)P, P)\}\cup \textbf{Id}$, if $y\notin fn(P)$, is a FR strongly probabilistic hp-bisimulation, we omit it;
  \item $(y)(z)P\sim_{php}^{fr} (z)(y)P$. It is sufficient to prove the relation $R=\{((y)(z)P, (z)(y)P)\}\cup \textbf{Id}$ is a FR strongly probabilistic hp-bisimulation, we omit it;
  \item $(y)(P+Q)\sim_{php}^{fr} (y)P+(y)Q$. It is sufficient to prove the relation $R=\{((y)(P+Q), (y)P+(y)Q)\}\cup \textbf{Id}$ is a FR strongly probabilistic hp-bisimulation, we omit it;
  \item $(y)(P\boxplus_{\pi}Q)\sim_{php}^{fr} (y)P\boxplus_{\pi}(y)Q$. It is sufficient to prove the relation $R=\{((y)(P\boxplus_{\pi}Q), (y)P\boxplus_{\pi}(y)Q)\}\cup \textbf{Id}$ is a FR strongly probabilistic hp-bisimulation, we omit it;
  \item $(y)\alpha.P\sim_{php}^{fr} \alpha.(y)P$ if $y\notin n(\alpha)$. It is sufficient to prove the relation $R=\{((y)\alpha.P, \alpha.(y)P)\}\cup \textbf{Id}$, if $y\notin n(\alpha)$, is a FR strongly probabilistic hp-bisimulation, we omit it;
  \item $(y)\alpha.P\sim_{php}^{fr} \textbf{nil}$ if $y$ is the subject of $\alpha$. It is sufficient to prove the relation $R=\{((y)\alpha.P, \textbf{nil})\}\cup \textbf{Id}$, if $y$ is the subject of $\alpha$, is a FR strongly probabilistic hp-bisimulation, we omit it.
\end{enumerate}
\end{proof}

\begin{theorem}[Restriction Laws for FR strongly probabilistic hhp-bisimilarity]
The restriction laws for FR strongly probabilistic hhp-bisimilarity are as follows.

\begin{enumerate}
  \item $(y)P\sim_{phhp}^{fr} P$, if $y\notin fn(P)$;
  \item $(y)(z)P\sim_{phhp}^{fr} (z)(y)P$;
  \item $(y)(P+Q)\sim_{phhp}^{fr} (y)P+(y)Q$;
  \item $(y)(P\boxplus_{\pi}Q)\sim_{phhp}^{fr} (y)P\boxplus_{\pi}(y)Q$;
  \item $(y)\alpha.P\sim_{phhp}^{fr} \alpha.(y)P$ if $y\notin n(\alpha)$;
  \item $(y)\alpha.P\sim_{phhp}^{fr} \textbf{nil}$ if $y$ is the subject of $\alpha$.
\end{enumerate}
\end{theorem}

\begin{proof}
\begin{enumerate}
  \item $(y)P\sim_{phhp}^{fr} P$, if $y\notin fn(P)$. It is sufficient to prove the relation $R=\{((y)P, P)\}\cup \textbf{Id}$, if $y\notin fn(P)$, is a FR strongly probabilistic hhp-bisimulation, we omit it;
  \item $(y)(z)P\sim_{phhp}^{fr} (z)(y)P$. It is sufficient to prove the relation $R=\{((y)(z)P, (z)(y)P)\}\cup \textbf{Id}$ is a FR strongly probabilistic hhp-bisimulation, we omit it;
  \item $(y)(P+Q)\sim_{phhp}^{fr} (y)P+(y)Q$. It is sufficient to prove the relation $R=\{((y)(P+Q), (y)P+(y)Q)\}\cup \textbf{Id}$ is a FR strongly probabilistic hhp-bisimulation, we omit it;
  \item $(y)(P\boxplus_{\pi}Q)\sim_{phhp}^{fr} (y)P\boxplus_{\pi}(y)Q$. It is sufficient to prove the relation $R=\{((y)(P\boxplus_{\pi}Q), (y)P\boxplus_{\pi}(y)Q)\}\cup \textbf{Id}$ is a FR strongly probabilistic hhp-bisimulation, we omit it;
  \item $(y)\alpha.P\sim_{phhp}^{fr} \alpha.(y)P$ if $y\notin n(\alpha)$. It is sufficient to prove the relation $R=\{((y)\alpha.P, \alpha.(y)P)\}\cup \textbf{Id}$, if $y\notin n(\alpha)$, is a FR strongly probabilistic hhp-bisimulation, we omit it;
  \item $(y)\alpha.P\sim_{phhp}^{fr} \textbf{nil}$ if $y$ is the subject of $\alpha$. It is sufficient to prove the relation $R=\{((y)\alpha.P, \textbf{nil})\}\cup \textbf{Id}$, if $y$ is the subject of $\alpha$, is a FR strongly probabilistic hhp-bisimulation, we omit it.
\end{enumerate}
\end{proof}

\begin{theorem}[Parallel laws for FR strongly probabilistic pomset bisimilarity]
The parallel laws for FR strongly probabilistic pomset bisimilarity are as follows.

\begin{enumerate}
  \item $P\parallel \textbf{nil}\sim_{pp}^{fr} P$;
  \item $P_1\parallel P_2\sim_{pp}^{fr} P_2\parallel P_1$;
  \item $(P_1\parallel P_2)\parallel P_3\sim_{pp}^{fr} P_1\parallel (P_2\parallel P_3)$;
  \item $(y)(P_1\parallel P_2)\sim_{pp}^{fr} (y)P_1\parallel (y)P_2$, if $y\notin fn(P_1)\cap fn(P_2)$.
\end{enumerate}
\end{theorem}

\begin{proof}
\begin{enumerate}
  \item $P\parallel \textbf{nil}\sim_{pp}^{fr} P$. It is sufficient to prove the relation $R=\{(P\parallel \textbf{nil}, P)\}\cup \textbf{Id}$ is a FR strongly probabilistic pomset bisimulation, we omit it;
  \item $P_1\parallel P_2\sim_{pp}^{fr} P_2\parallel P_1$. It is sufficient to prove the relation $R=\{(P_1\parallel P_2, P_2\parallel P_1)\}\cup \textbf{Id}$ is a FR strongly probabilistic pomset bisimulation, we omit it;
  \item $(P_1\parallel P_2)\parallel P_3\sim_{pp}^{fr} P_1\parallel (P_2\parallel P_3)$. It is sufficient to prove the relation $R=\{((P_1\parallel P_2)\parallel P_3, P_1\parallel (P_2\parallel P_3))\}\cup \textbf{Id}$ is a FR strongly probabilistic pomset bisimulation, we omit it;
  \item $(y)(P_1\parallel P_2)\sim_{pp}^{fr} (y)P_1\parallel (y)P_2$, if $y\notin fn(P_1)\cap fn(P_2)$. It is sufficient to prove the relation $R=\{((y)(P_1\parallel P_2), (y)P_1\parallel (y)P_2)\}\cup \textbf{Id}$, if $y\notin fn(P_1)\cap fn(P_2)$, is a FR strongly probabilistic pomset bisimulation, we omit it.
\end{enumerate}
\end{proof}

\begin{theorem}[Parallel laws for FR strongly probabilistic step bisimilarity]
The parallel laws for FR strongly probabilistic step bisimilarity are as follows.

\begin{enumerate}
  \item $P\parallel \textbf{nil}\sim_{ps}^{fr} P$;
  \item $P_1\parallel P_2\sim_{ps}^{fr} P_2\parallel P_1$;
  \item $(P_1\parallel P_2)\parallel P_3\sim_{ps}^{fr} P_1\parallel (P_2\parallel P_3)$;
  \item $(y)(P_1\parallel P_2)\sim_{ps}^{fr} (y)P_1\parallel (y)P_2$, if $y\notin fn(P_1)\cap fn(P_2)$.
\end{enumerate}
\end{theorem}

\begin{proof}
\begin{enumerate}
  \item $P\parallel \textbf{nil}\sim_{ps}^{fr} P$. It is sufficient to prove the relation $R=\{(P\parallel \textbf{nil}, P)\}\cup \textbf{Id}$ is a FR strongly probabilistic step bisimulation, we omit it;
  \item $P_1\parallel P_2\sim_{ps}^{fr} P_2\parallel P_1$. It is sufficient to prove the relation $R=\{(P_1\parallel P_2, P_2\parallel P_1)\}\cup \textbf{Id}$ is a FR strongly probabilistic step bisimulation, we omit it;
  \item $(P_1\parallel P_2)\parallel P_3\sim_{ps}^{fr} P_1\parallel (P_2\parallel P_3)$. It is sufficient to prove the relation $R=\{((P_1\parallel P_2)\parallel P_3, P_1\parallel (P_2\parallel P_3))\}\cup \textbf{Id}$ is a FR strongly probabilistic step bisimulation, we omit it;
  \item $(y)(P_1\parallel P_2)\sim_{ps}^{fr} (y)P_1\parallel (y)P_2$, if $y\notin fn(P_1)\cap fn(P_2)$. It is sufficient to prove the relation $R=\{((y)(P_1\parallel P_2), (y)P_1\parallel (y)P_2)\}\cup \textbf{Id}$, if $y\notin fn(P_1)\cap fn(P_2)$, is a FR strongly probabilistic step bisimulation, we omit it.
\end{enumerate}
\end{proof}

\begin{theorem}[Parallel laws for FR strongly probabilistic hp-bisimilarity]
The parallel laws for FR strongly probabilistic hp-bisimilarity are as follows.

\begin{enumerate}
  \item $P\parallel \textbf{nil}\sim_{php}^{fr} P$;
  \item $P_1\parallel P_2\sim_{php}^{fr} P_2\parallel P_1$;
  \item $(P_1\parallel P_2)\parallel P_3\sim_{php}^{fr} P_1\parallel (P_2\parallel P_3)$;
  \item $(y)(P_1\parallel P_2)\sim_{php}^{fr} (y)P_1\parallel (y)P_2$, if $y\notin fn(P_1)\cap fn(P_2)$.
\end{enumerate}
\end{theorem}

\begin{proof}
\begin{enumerate}
  \item $P\parallel \textbf{nil}\sim_{php}^{fr} P$. It is sufficient to prove the relation $R=\{(P\parallel \textbf{nil}, P)\}\cup \textbf{Id}$ is a FR strongly probabilistic hp-bisimulation, we omit it;
  \item $P_1\parallel P_2\sim_{php}^{fr} P_2\parallel P_1$. It is sufficient to prove the relation $R=\{(P_1\parallel P_2, P_2\parallel P_1)\}\cup \textbf{Id}$ is a FR strongly probabilistic hp-bisimulation, we omit it;
  \item $(P_1\parallel P_2)\parallel P_3\sim_{php}^{fr} P_1\parallel (P_2\parallel P_3)$. It is sufficient to prove the relation $R=\{((P_1\parallel P_2)\parallel P_3, P_1\parallel (P_2\parallel P_3))\}\cup \textbf{Id}$ is a FR strongly probabilistic hp-bisimulation, we omit it;
  \item $(y)(P_1\parallel P_2)\sim_{php}^{fr} (y)P_1\parallel (y)P_2$, if $y\notin fn(P_1)\cap fn(P_2)$. It is sufficient to prove the relation $R=\{((y)(P_1\parallel P_2), (y)P_1\parallel (y)P_2)\}\cup \textbf{Id}$, if $y\notin fn(P_1)\cap fn(P_2)$, is a FR strongly probabilistic hp-bisimulation, we omit it.
\end{enumerate}
\end{proof}

\begin{theorem}[Parallel laws for FR strongly probabilistic hhp-bisimilarity]
The parallel laws for FR strongly probabilistic hhp-bisimilarity are as follows.

\begin{enumerate}
  \item $P\parallel \textbf{nil}\sim_{phhp}^{fr} P$;
  \item $P_1\parallel P_2\sim_{phhp}^{fr} P_2\parallel P_1$;
  \item $(P_1\parallel P_2)\parallel P_3\sim_{phhp}^{fr} P_1\parallel (P_2\parallel P_3)$;
  \item $(y)(P_1\parallel P_2)\sim_{phhp}^{fr} (y)P_1\parallel (y)P_2$, if $y\notin fn(P_1)\cap fn(P_2)$.
\end{enumerate}
\end{theorem}

\begin{proof}
\begin{enumerate}
  \item $P\parallel \textbf{nil}\sim_{phhp}^{fr} P$. It is sufficient to prove the relation $R=\{(P\parallel \textbf{nil}, P)\}\cup \textbf{Id}$ is a FR strongly probabilistic hhp-bisimulation, we omit it;
  \item $P_1\parallel P_2\sim_{phhp}^{fr} P_2\parallel P_1$. It is sufficient to prove the relation $R=\{(P_1\parallel P_2, P_2\parallel P_1)\}\cup \textbf{Id}$ is a FR strongly probabilistic hhp-bisimulation, we omit it;
  \item $(P_1\parallel P_2)\parallel P_3\sim_{phhp}^{fr} P_1\parallel (P_2\parallel P_3)$. It is sufficient to prove the relation $R=\{((P_1\parallel P_2)\parallel P_3, P_1\parallel (P_2\parallel P_3))\}\cup \textbf{Id}$ is a FR strongly probabilistic hhp-bisimulation, we omit it;
  \item $(y)(P_1\parallel P_2)\sim_{phhp}^{fr} (y)P_1\parallel (y)P_2$, if $y\notin fn(P_1)\cap fn(P_2)$. It is sufficient to prove the relation $R=\{((y)(P_1\parallel P_2), (y)P_1\parallel (y)P_2)\}\cup \textbf{Id}$, if $y\notin fn(P_1)\cap fn(P_2)$, is a FR strongly probabilistic hhp-bisimulation, we omit it.
\end{enumerate}
\end{proof}

\begin{theorem}[Expansion law for truly concurrent bisimilarities]
Let $P\equiv\boxplus_i\sum_j \alpha_{ij}.P_{ij}$ and $Q\equiv\boxplus_{k}\sum_l\beta_{kl}.Q_{kl}$, where $bn(\alpha_{ij})\cap fn(Q)=\emptyset$ for all $i,j$, and
  $bn(\beta_{kl})\cap fn(P)=\emptyset$ for all $k,l$. Then,

\begin{enumerate}
  \item $P\parallel Q\sim_{pp}^{fr} \boxplus_{i}\boxplus_{k}\sum_j\sum_l (\alpha_{ij}\parallel \beta_{kl}).(P_{ij}\parallel Q_{kl})+\boxplus_i\boxplus_k\sum_{\alpha_{ij} \textrm{ comp }\beta_{kl}}\tau.R_{ijkl}$;
  \item $P\parallel Q\sim_{ps}^{fr} \boxplus_{i}\boxplus_{k}\sum_j\sum_l (\alpha_{ij}\parallel \beta_{kl}).(P_{ij}\parallel Q_{kl})+\boxplus_i\boxplus_k\sum_{\alpha_{ij} \textrm{ comp }\beta_{kl}}\tau.R_{ijkl}$;
  \item $P\parallel Q\sim_{php}^{fr} \boxplus_{i}\boxplus_{k}\sum_j\sum_l (\alpha_{ij}\parallel \beta_{kl}).(P_{ij}\parallel Q_{kl})+\boxplus_i\boxplus_k\sum_{\alpha_{ij} \textrm{ comp }\beta_{kl}}\tau.R_{ijkl}$;
  \item $P\parallel Q\nsim_{phhp} \boxplus_{i}\boxplus_{k}\sum_j\sum_l (\alpha_{ij}\parallel \beta_{kl}).(P_{ij}\parallel Q_{kl})+\boxplus_i\boxplus_k\sum_{\alpha_{ij} \textrm{ comp }\beta_{kl}}\tau.R_{ijkl}$.
\end{enumerate}

Where $\alpha_{ij}$ comp $\beta_{kl}$ and $R_{ijkl}$ are defined as follows:
\begin{enumerate}
  \item $\alpha_{ij}$ is $\overline{x}u$ and $\beta_{kl}$ is $x(v)$, then $R_{ijkl}=P_{ij}\parallel Q_{kl}\{u/v\}$;
  \item $\alpha_{ij}$ is $\overline{x}(u)$ and $\beta_{kl}$ is $x(v)$, then $R_{ijkl}=(w)(P_{ij}\{w/u\}\parallel Q_{kl}\{w/v\})$, if $w\notin fn((u)P_{ij})\cup fn((v)Q_{kl})$;
  \item $\alpha_{ij}$ is $x(v)$ and $\beta_{kl}$ is $\overline{x}u$, then $R_{ijkl}=P_{ij}\{u/v\}\parallel Q_{kl}$;
  \item $\alpha_{ij}$ is $x(v)$ and $\beta_{kl}$ is $\overline{x}(u)$, then $R_{ijkl}=(w)(P_{ij}\{w/v\}\parallel Q_{kl}\{w/u\})$, if $w\notin fn((v)P_{ij})\cup fn((u)Q_{kl})$.
\end{enumerate}

Let $P\equiv\boxplus_i\sum_j P_{ij}.\alpha_{ij}[m]$ and $Q\equiv\boxplus_{k}\sum_l Q_{kl}.\beta_{kl}[m]$, where $bn(\alpha_{ij}[m])\cap fn(Q)=\emptyset$ for all $i,j$, and
  $bn(\beta_{kl}[m])\cap fn(P)=\emptyset$ for all $k,l$. Then,

\begin{enumerate}
  \item $P\parallel Q\sim_{pp}^{fr} \boxplus_{i}\boxplus_{k}\sum_j\sum_l(P_{ij}\parallel Q_{kl}).(\alpha_{ij}[m]\parallel \beta_{kl}[m])+\boxplus_i\boxplus_k\sum_{\alpha_{ij} \textrm{ comp }\beta_{kl}} R_{ijkl}.\tau$;
  \item $P\parallel Q\sim_{ps}^{fr} \boxplus_{i}\boxplus_{k}\sum_j\sum_l(P_{ij}\parallel Q_{kl}).(\alpha_{ij}[m]\parallel \beta_{kl}[m])+\boxplus_i\boxplus_k\sum_{\alpha_{ij} \textrm{ comp }\beta_{kl}} R_{ijkl}.\tau$;
  \item $P\parallel Q\sim_{php}^{fr} \boxplus_{i}\boxplus_{k}\sum_j\sum_l(P_{ij}\parallel Q_{kl}).(\alpha_{ij}[m]\parallel \beta_{kl}[m])+\boxplus_i\boxplus_k\sum_{\alpha_{ij} \textrm{ comp }\beta_{kl}} R_{ijkl}.\tau$;
  \item $P\parallel Q\nsim_{phhp} \boxplus_{i}\boxplus_{k}\sum_j\sum_l(P_{ij}\parallel Q_{kl}).(\alpha_{ij}[m]\parallel \beta_{kl}[m])+\boxplus_i\boxplus_k\sum_{\alpha_{ij} \textrm{ comp }\beta_{kl}} R_{ijkl}.\tau$.
\end{enumerate}

Where $\alpha_{ij}[m]$ comp $\beta_{kl}[m]$ and $R_{ijkl}$ are defined as follows:
\begin{enumerate}
  \item $\alpha_{ij}[m]$ is $\overline{x}u$ and $\beta_{kl}[m]$ is $x(v)$, then $R_{ijkl}=P_{ij}\parallel Q_{kl}\{u/v\}$;
  \item $\alpha_{ij}[m]$ is $\overline{x}(u)$ and $\beta_{kl}[m]$ is $x(v)$, then $R_{ijkl}=(w)(P_{ij}\{w/u\}\parallel Q_{kl}\{w/v\})$, if $w\notin fn((u)P_{ij})\cup fn((v)Q_{kl})$;
  \item $\alpha_{ij}[m]$ is $x(v)$ and $\beta_{kl}[m]$ is $\overline{x}u$, then $R_{ijkl}=P_{ij}\{u/v\}\parallel Q_{kl}$;
  \item $\alpha_{ij}[m]$ is $x(v)$ and $\beta_{kl}[m]$ is $\overline{x}(u)$, then $R_{ijkl}=(w)(P_{ij}\{w/v\}\parallel Q_{kl}\{w/u\})$, if $w\notin fn((v)P_{ij})\cup fn((u)Q_{kl})$.
\end{enumerate}
\end{theorem}

\begin{proof}
According to the definition of FR strongly probabilistic truly concurrent bisimulations, we can easily prove the above equations, and we omit the proof.
\end{proof}

\begin{theorem}[Equivalence and congruence for FR strongly probabilistic pomset bisimilarity]
We can enjoy the full congruence modulo FR strongly probabilistic pomset bisimilarity.

\begin{enumerate}
  \item $\sim_{pp}^{fr}$ is an equivalence relation;
  \item If $P\sim_{pp}^{fr} Q$ then
  \begin{enumerate}
    \item $\alpha.P\sim_{pp}^{f} \alpha.Q$, $\alpha$ is a free action;
    \item $P.\alpha[m]\sim_{pp}^{r}Q.\alpha[m]$, $\alpha[m]$ is a free action;
    \item $P+R\sim_{pp}^{fr} Q+R$;
    \item $P\boxplus_{\pi} R\sim_{pp}^{fr}Q\boxplus_{\pi}R$;
    \item $P\parallel R\sim_{pp}^{fr} Q\parallel R$;
    \item $(w)P\sim_{pp}^{fr} (w)Q$;
    \item $x(y).P\sim_{pp}^{f} x(y).Q$;
    \item $P.x(y)[m]\sim_{pp}^{r}Q.x(y)[m]$.
  \end{enumerate}
\end{enumerate}
\end{theorem}

\begin{proof}
\begin{enumerate}
  \item $\sim_{pp}^{fr}$ is an equivalence relation, it is obvious;
  \item If $P\sim_{pp}^{fr} Q$ then
  \begin{enumerate}
    \item $\alpha.P\sim_{pp}^{f} \alpha.Q$, $\alpha$ is a free action. It is sufficient to prove the relation $R=\{(\alpha.P, \alpha.Q)\}\cup \textbf{Id}$ is a F strongly probabilistic pomset bisimulation, we omit it;
    \item $P.\alpha[m]\sim_{pp}^{r}Q.\alpha[m]$, $\alpha[m]$ is a free action. It is sufficient to prove the relation $R=\{(P.\alpha[m], Q.\alpha[m])\}\cup \textbf{Id}$ is a R strongly probabilistic pomset bisimulation, we omit it;
    \item $P+R\sim_{pp}^{fr} Q+R$. It is sufficient to prove the relation $R=\{(P+R, Q+R)\}\cup \textbf{Id}$ is a FR strongly probabilistic pomset bisimulation, we omit it;
    \item $P\boxplus_{\pi} R\sim_{pp}^{fr}Q\boxplus_{\pi}R$. It is sufficient to prove the relation $R=\{(P\boxplus_{\pi} R, Q\boxplus_{\pi} R)\}\cup \textbf{Id}$ is a FR strongly probabilistic pomset bisimulation, we omit it;
    \item $P\parallel R\sim_{pp}^{fr} Q\parallel R$. It is sufficient to prove the relation $R=\{(P\parallel R, Q\parallel R)\}\cup \textbf{Id}$ is a FR strongly probabilistic pomset bisimulation, we omit it;
    \item $(w)P\sim_{pp}^{fr} (w)Q$. It is sufficient to prove the relation $R=\{((w)P, (w)Q)\}\cup \textbf{Id}$ is a FR strongly probabilistic pomset bisimulation, we omit it;
    \item $x(y).P\sim_{pp}^{f} x(y).Q$. It is sufficient to prove the relation $R=\{(x(y).P, x(y).Q)\}\cup \textbf{Id}$ is a F strongly probabilistic pomset bisimulation, we omit it;
    \item $P.x(y)[m]\sim_{pp}^{r}Q.x(y)[m]$. It is sufficient to prove the relation $R=\{(P.x(y)[m], Q.x(y)[m])\}\cup \textbf{Id}$ is a R strongly probabilistic pomset bisimulation, we omit it.
  \end{enumerate}
\end{enumerate}
\end{proof}

\begin{theorem}[Equivalence and congruence for FR strongly probabilistic step bisimilarity]
We can enjoy the full congruence modulo FR strongly probabilistic step bisimilarity.

\begin{enumerate}
  \item $\sim_{ps}^{fr}$ is an equivalence relation;
  \item If $P\sim_{ps}^{fr} Q$ then
  \begin{enumerate}
    \item $\alpha.P\sim_{ps}^{f} \alpha.Q$, $\alpha$ is a free action;
    \item $P.\alpha[m]\sim_{ps}^{r}Q.\alpha[m]$, $\alpha[m]$ is a free action;
    \item $P+R\sim_{ps}^{fr} Q+R$;
    \item $P\boxplus_{\pi} R\sim_{ps}^{fr}Q\boxplus_{\pi}R$;
    \item $P\parallel R\sim_{ps}^{fr} Q\parallel R$;
    \item $(w)P\sim_{ps}^{fr} (w)Q$;
    \item $x(y).P\sim_{ps}^{f} x(y).Q$;
    \item $P.x(y)[m]\sim_{ps}^{r}Q.x(y)[m]$.
  \end{enumerate}
\end{enumerate}
\end{theorem}

\begin{proof}
\begin{enumerate}
  \item $\sim_{ps}^{fr}$ is an equivalence relation, it is obvious;
  \item If $P\sim_{ps}^{fr} Q$ then
  \begin{enumerate}
    \item $\alpha.P\sim_{ps}^{f} \alpha.Q$, $\alpha$ is a free action. It is sufficient to prove the relation $R=\{(\alpha.P, \alpha.Q)\}\cup \textbf{Id}$ is a F strongly probabilistic step bisimulation, we omit it;
    \item $P.\alpha[m]\sim_{ps}^{r}Q.\alpha[m]$, $\alpha[m]$ is a free action. It is sufficient to prove the relation $R=\{(P.\alpha[m], Q.\alpha[m])\}\cup \textbf{Id}$ is a R strongly probabilistic step bisimulation, we omit it;
    \item $P+R\sim_{ps}^{fr} Q+R$. It is sufficient to prove the relation $R=\{(P+R, Q+R)\}\cup \textbf{Id}$ is a FR strongly probabilistic step bisimulation, we omit it;
    \item $P\boxplus_{\pi} R\sim_{ps}^{fr}Q\boxplus_{\pi}R$. It is sufficient to prove the relation $R=\{(P\boxplus_{\pi} R, Q\boxplus_{\pi} R)\}\cup \textbf{Id}$ is a FR strongly probabilistic step bisimulation, we omit it;
    \item $P\parallel R\sim_{ps}^{fr} Q\parallel R$. It is sufficient to prove the relation $R=\{(P\parallel R, Q\parallel R)\}\cup \textbf{Id}$ is a FR strongly probabilistic step bisimulation, we omit it;
    \item $(w)P\sim_{ps}^{fr} (w)Q$. It is sufficient to prove the relation $R=\{((w)P, (w)Q)\}\cup \textbf{Id}$ is a FR strongly probabilistic step bisimulation, we omit it;
    \item $x(y).P\sim_{ps}^{f} x(y).Q$. It is sufficient to prove the relation $R=\{(x(y).P, x(y).Q)\}\cup \textbf{Id}$ is a F strongly probabilistic step bisimulation, we omit it;
    \item $P.x(y)[m]\sim_{ps}^{r}Q.x(y)[m]$. It is sufficient to prove the relation $R=\{(P.x(y)[m], Q.x(y)[m])\}\cup \textbf{Id}$ is a R strongly probabilistic step bisimulation, we omit it.
  \end{enumerate}
\end{enumerate}
\end{proof}

\begin{theorem}[Equivalence and congruence for FR strongly probabilistic hp-bisimilarity]
We can enjoy the full congruence modulo FR strongly probabilistic hp-bisimilarity.

\begin{enumerate}
  \item $\sim_{php}^{fr}$ is an equivalence relation;
  \item If $P\sim_{php}^{fr} Q$ then
  \begin{enumerate}
    \item $\alpha.P\sim_{php}^{f} \alpha.Q$, $\alpha$ is a free action;
    \item $P.\alpha[m]\sim_{php}^{r}Q.\alpha[m]$, $\alpha[m]$ is a free action;
    \item $P+R\sim_{php}^{fr} Q+R$;
    \item $P\boxplus_{\pi} R\sim_{php}^{fr}Q\boxplus_{\pi}R$;
    \item $P\parallel R\sim_{php}^{fr} Q\parallel R$;
    \item $(w)P\sim_{php}^{fr} (w)Q$;
    \item $x(y).P\sim_{php}^{f} x(y).Q$;
    \item $P.x(y)[m]\sim_{php}^{r}Q.x(y)[m]$.
  \end{enumerate}
\end{enumerate}
\end{theorem}

\begin{proof}
\begin{enumerate}
  \item $\sim_{php}^{fr}$ is an equivalence relation, it is obvious;
  \item If $P\sim_{php}^{fr} Q$ then
  \begin{enumerate}
    \item $\alpha.P\sim_{php}^{f} \alpha.Q$, $\alpha$ is a free action. It is sufficient to prove the relation $R=\{(\alpha.P, \alpha.Q)\}\cup \textbf{Id}$ is a F strongly probabilistic hp-bisimulation, we omit it;
    \item $P.\alpha[m]\sim_{php}^{r}Q.\alpha[m]$, $\alpha[m]$ is a free action. It is sufficient to prove the relation $R=\{(P.\alpha[m], Q.\alpha[m])\}\cup \textbf{Id}$ is a R strongly probabilistic hp-bisimulation, we omit it;
    \item $P+R\sim_{php}^{fr} Q+R$. It is sufficient to prove the relation $R=\{(P+R, Q+R)\}\cup \textbf{Id}$ is a FR strongly probabilistic hp-bisimulation, we omit it;
    \item $P\boxplus_{\pi} R\sim_{php}^{fr}Q\boxplus_{\pi}R$. It is sufficient to prove the relation $R=\{(P\boxplus_{\pi} R, Q\boxplus_{\pi} R)\}\cup \textbf{Id}$ is a FR strongly probabilistic hp-bisimulation, we omit it;
    \item $P\parallel R\sim_{php}^{fr} Q\parallel R$. It is sufficient to prove the relation $R=\{(P\parallel R, Q\parallel R)\}\cup \textbf{Id}$ is a FR strongly probabilistic hp-bisimulation, we omit it;
    \item $(w)P\sim_{php}^{fr} (w)Q$. It is sufficient to prove the relation $R=\{((w)P, (w)Q)\}\cup \textbf{Id}$ is a FR strongly probabilistic hp-bisimulation, we omit it;
    \item $x(y).P\sim_{php}^{f} x(y).Q$. It is sufficient to prove the relation $R=\{(x(y).P, x(y).Q)\}\cup \textbf{Id}$ is a F strongly probabilistic hp-bisimulation, we omit it;
    \item $P.x(y)[m]\sim_{php}^{r}Q.x(y)[m]$. It is sufficient to prove the relation $R=\{(P.x(y)[m], Q.x(y)[m])\}\cup \textbf{Id}$ is a R strongly probabilistic hp-bisimulation, we omit it.
  \end{enumerate}
\end{enumerate}
\end{proof}

\begin{theorem}[Equivalence and congruence for FR strongly probabilistic hhp-bisimilarity]
We can enjoy the full congruence modulo FR strongly probabilistic hhp-bisimilarity.

\begin{enumerate}
  \item $\sim_{phhp}^{fr}$ is an equivalence relation;
  \item If $P\sim_{phhp}^{fr} Q$ then
  \begin{enumerate}
    \item $\alpha.P\sim_{phhp}^{f} \alpha.Q$, $\alpha$ is a free action;
    \item $P.\alpha[m]\sim_{phhp}^{r}Q.\alpha[m]$, $\alpha[m]$ is a free action;
    \item $P+R\sim_{phhp}^{fr} Q+R$;
    \item $P\boxplus_{\pi} R\sim_{phhp}^{fr}Q\boxplus_{\pi}R$;
    \item $P\parallel R\sim_{phhp}^{fr} Q\parallel R$;
    \item $(w)P\sim_{phhp}^{fr} (w)Q$;
    \item $x(y).P\sim_{phhp}^{f} x(y).Q$;
    \item $P.x(y)[m]\sim_{phhp}^{r}Q.x(y)[m]$.
  \end{enumerate}
\end{enumerate}
\end{theorem}

\begin{proof}
\begin{enumerate}
  \item $\sim_{phhp}^{fr}$ is an equivalence relation, it is obvious;
  \item If $P\sim_{phhp}^{fr} Q$ then
  \begin{enumerate}
    \item $\alpha.P\sim_{phhp}^{f} \alpha.Q$, $\alpha$ is a free action. It is sufficient to prove the relation $R=\{(\alpha.P, \alpha.Q)\}\cup \textbf{Id}$ is a F strongly probabilistic hhp-bisimulation, we omit it;
    \item $P.\alpha[m]\sim_{phhp}^{r}Q.\alpha[m]$, $\alpha[m]$ is a free action. It is sufficient to prove the relation $R=\{(P.\alpha[m], Q.\alpha[m])\}\cup \textbf{Id}$ is a R strongly probabilistic hhp-bisimulation, we omit it;
    \item $P+R\sim_{phhp}^{fr} Q+R$. It is sufficient to prove the relation $R=\{(P+R, Q+R)\}\cup \textbf{Id}$ is a FR strongly probabilistic hhp-bisimulation, we omit it;
    \item $P\boxplus_{\pi} R\sim_{phhp}^{fr}Q\boxplus_{\pi}R$. It is sufficient to prove the relation $R=\{(P\boxplus_{\pi} R, Q\boxplus_{\pi} R)\}\cup \textbf{Id}$ is a FR strongly probabilistic hhp-bisimulation, we omit it;
    \item $P\parallel R\sim_{phhp}^{fr} Q\parallel R$. It is sufficient to prove the relation $R=\{(P\parallel R, Q\parallel R)\}\cup \textbf{Id}$ is a FR strongly probabilistic hhp-bisimulation, we omit it;
    \item $(w)P\sim_{phhp}^{fr} (w)Q$. It is sufficient to prove the relation $R=\{((w)P, (w)Q)\}\cup \textbf{Id}$ is a FR strongly probabilistic hhp-bisimulation, we omit it;
    \item $x(y).P\sim_{phhp}^{f} x(y).Q$. It is sufficient to prove the relation $R=\{(x(y).P, x(y).Q)\}\cup \textbf{Id}$ is a F strongly probabilistic hhp-bisimulation, we omit it;
    \item $P.x(y)[m]\sim_{phhp}^{r}Q.x(y)[m]$. It is sufficient to prove the relation $R=\{(P.x(y)[m], Q.x(y)[m])\}\cup \textbf{Id}$ is a R strongly probabilistic hhp-bisimulation, we omit it.
  \end{enumerate}
\end{enumerate}
\end{proof}

\subsubsection{Recursion}

\begin{definition}
Let $X$ have arity $n$, and let $\widetilde{x}=x_1,\cdots,x_n$ be distinct names, and $fn(P)\subseteq\{x_1,\cdots,x_n\}$. The replacement of $X(\widetilde{x})$ by $P$ in $E$, written
$E\{X(\widetilde{x}):=P\}$, means the result of replacing each subterm $X(\widetilde{y})$ in $E$ by $P\{\widetilde{y}/\widetilde{x}\}$.
\end{definition}

\begin{definition}
Let $E$ and $F$ be two process expressions containing only $X_1,\cdots,X_m$ with associated name sequences $\widetilde{x}_1,\cdots,\widetilde{x}_m$. Then,
\begin{enumerate}
  \item $E\sim_{pp}^{fr} F$ means $E(\widetilde{P})\sim_{pp}^{fr} F(\widetilde{P})$;
  \item $E\sim_{ps}^{fr} F$ means $E(\widetilde{P})\sim_{ps}^{fr} F(\widetilde{P})$;
  \item $E\sim_{php}^{fr} F$ means $E(\widetilde{P})\sim_{php}^{fr} F(\widetilde{P})$;
  \item $E\sim_{phhp}^{fr} F$ means $E(\widetilde{P})\sim_{phhp}^{fr} F(\widetilde{P})$;
\end{enumerate}

for all $\widetilde{P}$ such that $fn(P_i)\subseteq \widetilde{x}_i$ for each $i$.
\end{definition}

\begin{definition}
A term or identifier is weakly guarded in $P$ if it lies within some subterm $\alpha.Q$ or $Q.\alpha[m]$ or $(\alpha_1\parallel\cdots\parallel \alpha_n).Q$ or
$Q.(\alpha_1[m]\parallel\cdots\parallel \alpha_n[m])$ of $P$.
\end{definition}

\begin{theorem}
Assume that $\widetilde{E}$ and $\widetilde{F}$ are expressions containing only $X_i$ with $\widetilde{x}_i$, and $\widetilde{A}$ and $\widetilde{B}$ are identifiers with $A_i$, $B_i$. Then, for all $i$,
\begin{enumerate}
  \item $E_i\sim_{ps}^{fr} F_i$, $A_i(\widetilde{x}_i)\overset{\text{def}}{=}E_i(\widetilde{A})$, $B_i(\widetilde{x}_i)\overset{\text{def}}{=}F_i(\widetilde{B})$, then
  $A_i(\widetilde{x}_i)\sim_{ps}^{fr} B_i(\widetilde{x}_i)$;
  \item $E_i\sim_{pp}^{fr} F_i$, $A_i(\widetilde{x}_i)\overset{\text{def}}{=}E_i(\widetilde{A})$, $B_i(\widetilde{x}_i)\overset{\text{def}}{=}F_i(\widetilde{B})$, then
  $A_i(\widetilde{x}_i)\sim_{pp}^{fr} B_i(\widetilde{x}_i)$;
  \item $E_i\sim_{php}^{fr} F_i$, $A_i(\widetilde{x}_i)\overset{\text{def}}{=}E_i(\widetilde{A})$, $B_i(\widetilde{x}_i)\overset{\text{def}}{=}F_i(\widetilde{B})$, then
  $A_i(\widetilde{x}_i)\sim_{php}^{fr} B_i(\widetilde{x}_i)$;
  \item $E_i\sim_{phhp}^{fr} F_i$, $A_i(\widetilde{x}_i)\overset{\text{def}}{=}E_i(\widetilde{A})$, $B_i(\widetilde{x}_i)\overset{\text{def}}{=}F_i(\widetilde{B})$, then
  $A_i(\widetilde{x}_i)\sim_{phhp}^{fr} B_i(\widetilde{x}_i)$.
\end{enumerate}
\end{theorem}

\begin{proof}
\begin{enumerate}
  \item $E_i\sim_{ps}^{fr} F_i$, $A_i(\widetilde{x}_i)\overset{\text{def}}{=}E_i(\widetilde{A})$, $B_i(\widetilde{x}_i)\overset{\text{def}}{=}F_i(\widetilde{B})$, then
  $A_i(\widetilde{x}_i)\sim_{ps}^{fr} B_i(\widetilde{x}_i)$.

      We will consider the case $I=\{1\}$ with loss of generality, and show the following relation $R$ is a FR strongly probabilistic step bisimulation.

      $$R=\{(G(A),G(B)):G\textrm{ has only identifier }X\}.$$

      By choosing $G\equiv X(\widetilde{y})$, it follows that $A(\widetilde{y})\sim_{ps}^{fr} B(\widetilde{y})$. It is sufficient to prove the following:
      \begin{enumerate}
        \item If $G(A)\rightsquigarrow\xrightarrow{\{\alpha_1,\cdots,\alpha_n\}}P'$, where $\alpha_i(1\leq i\leq n)$ is a free action or bound output action with
        $bn(\alpha_1)\cap\cdots\cap bn(\alpha_n)\cap n(G(A),G(B))=\emptyset$, then $G(B)\rightsquigarrow\xrightarrow{\{\alpha_1,\cdots,\alpha_n\}}Q''$ such that $P'\sim_{ps}^{fr} Q''$;
        \item If $G(A)\rightsquigarrow\xrightarrow{x(y)}P'$ with $x\notin n(G(A),G(B))$, then $G(B)\rightsquigarrow\xrightarrow{x(y)}Q''$, such that for all $u$,
        $P'\{u/y\}\sim_{ps}^{fr} Q''\{u/y\}$;
        \item If $G(A)\rightsquigarrow\xtworightarrow{\{\alpha_1[m],\cdots,\alpha_n[m]\}}P'$, where $\alpha_i[m](1\leq i\leq n)$ is a free action or bound output action with
        $bn(\alpha_1[m])\cap\cdots\cap bn(\alpha_n[m])\cap n(G(A),G(B))=\emptyset$, then $G(B)\rightsquigarrow\xtworightarrow{\{\alpha_1[m],\cdots,\alpha_n[m]\}}Q''$ such that $P'\sim_{ps}^{fr} Q''$;
        \item If $G(A)\rightsquigarrow\xtworightarrow{x(y)[m]}P'$ with $x\notin n(G(A),G(B))$, then $G(B)\rightsquigarrow\xtworightarrow{x(y)[m]}Q''$, such that for all $u$,
        $P'\{u/y\}\sim_{ps}^{fr} Q''\{u/y\}$.
      \end{enumerate}

      To prove the above properties, it is sufficient to induct on the depth of inference and quite routine, we omit it.
  \item $E_i\sim_{pp}^{fr} F_i$, $A_i(\widetilde{x}_i)\overset{\text{def}}{=}E_i(\widetilde{A})$, $B_i(\widetilde{x}_i)\overset{\text{def}}{=}F_i(\widetilde{B})$, then
  $A_i(\widetilde{x}_i)\sim_{pp}^{fr} B_i(\widetilde{x}_i)$. It can be proven similarly to the above case.
  \item $E_i\sim_{php}^{fr} F_i$, $A_i(\widetilde{x}_i)\overset{\text{def}}{=}E_i(\widetilde{A})$, $B_i(\widetilde{x}_i)\overset{\text{def}}{=}F_i(\widetilde{B})$, then
  $A_i(\widetilde{x}_i)\sim_{php}^{fr} B_i(\widetilde{x}_i)$. It can be proven similarly to the above case.
  \item $E_i\sim_{phhp}^{fr} F_i$, $A_i(\widetilde{x}_i)\overset{\text{def}}{=}E_i(\widetilde{A})$, $B_i(\widetilde{x}_i)\overset{\text{def}}{=}F_i(\widetilde{B})$, then
  $A_i(\widetilde{x}_i)\sim_{phhp}^{fr} B_i(\widetilde{x}_i)$. It can be proven similarly to the above case.
\end{enumerate}
\end{proof}

\begin{theorem}[Unique solution of equations]
Assume $\widetilde{E}$ are expressions containing only $X_i$ with $\widetilde{x}_i$, and each $X_i$ is weakly guarded in each $E_j$. Assume that $\widetilde{P}$ and $\widetilde{Q}$ are
processes such that $fn(P_i)\subseteq \widetilde{x}_i$ and $fn(Q_i)\subseteq \widetilde{x}_i$. Then, for all $i$,
\begin{enumerate}
  \item if $P_i\sim_{pp}^{fr} E_i(\widetilde{P})$, $Q_i\sim_{pp}^{fr} E_i(\widetilde{Q})$, then $P_i\sim_{pp}^{fr} Q_i$;
  \item if $P_i\sim_{ps}^{fr} E_i(\widetilde{P})$, $Q_i\sim_{ps}^{fr} E_i(\widetilde{Q})$, then $P_i\sim_{ps}^{fr} Q_i$;
  \item if $P_i\sim_{php}^{fr} E_i(\widetilde{P})$, $Q_i\sim_{php}^{fr} E_i(\widetilde{Q})$, then $P_i\sim_{php}^{fr} Q_i$;
  \item if $P_i\sim_{phhp}^{fr} E_i(\widetilde{P})$, $Q_i\sim_{phhp}^{fr} E_i(\widetilde{Q})$, then $P_i\sim_{phhp}^{fr} Q_i$.
\end{enumerate}
\end{theorem}

\begin{proof}
\begin{enumerate}
  \item It is similar to the proof of unique solution of equations for FR strongly probabilistic pomset bisimulation in CTC, please refer to \cite{CTC2} for details, we omit it;
  \item It is similar to the proof of unique solution of equations for FR strongly probabilistic step bisimulation in CTC, please refer to \cite{CTC2} for details, we omit it;
  \item It is similar to the proof of unique solution of equations for FR strongly probabilistic hp-bisimulation in CTC, please refer to \cite{CTC2} for details, we omit it;
  \item It is similar to the proof of unique solution of equations for FR strongly probabilistic hhp-bisimulation in CTC, please refer to \cite{CTC2} for details, we omit it.
\end{enumerate}
\end{proof}

\subsection{Algebraic Theory}\label{a5}

\begin{definition}[STC]
The theory \textbf{STC} is consisted of the following axioms and inference rules:

\begin{enumerate}
  \item Alpha-conversion $\textbf{A}$.
  \[\textrm{if } P\equiv Q, \textrm{ then } P=Q\]
  \item Congruence $\textbf{C}$. If $P=Q$, then,
  \[\tau.P=\tau.Q\quad \overline{x}y.P=\overline{x}y.Q\quad P.\overline{x}y[m]=Q.\overline{x}y[m]\]
  \[P+R=Q+R\quad P\parallel R=Q\parallel R\]
  \[(x)P=(x)Q\quad x(y).P=x(y).Q\quad P.x(y)[m]=Q.x(y)[m]\]
  \item Summation $\textbf{S}$.
  \[\textbf{S0}\quad P+\textbf{nil}=P\]
  \[\textbf{S1}\quad P+P=P\]
  \[\textbf{S2}\quad P+Q=Q+P\]
  \[\textbf{S3}\quad P+(Q+R)=(P+Q)+R\]
  \item Box-Summation $\textbf(BS)$.
  \[\textbf{BS0}\quad P\boxplus_{\pi}\textbf{nil}= P\]
  \[\textbf{BS1}\quad P\boxplus_{\pi}P= P\]
  \[\textbf{BS2}\quad P\boxplus_{\pi} Q= Q\boxplus_{1-\pi} P\]
  \[\textbf{BS3}\quad P\boxplus_{\pi}(Q\boxplus_{\rho} R)= (P\boxplus_{\frac{\pi}{\pi+\rho-\pi\rho}}Q)\boxplus_{\pi+\rho-\pi\rho} R\]
  \item Restriction $\textbf{R}$.
  \[\textbf{R0}\quad (x)P=P\quad \textrm{ if }x\notin fn(P)\]
  \[\textbf{R1}\quad (x)(y)P=(y)(x)P\]
  \[\textbf{R2}\quad (x)(P+Q)=(x)P+(x)Q\]
  \[\textbf{R3}\quad (x)\alpha.P=\alpha.(x)P\quad \textrm{ if }x\notin n(\alpha)\]
  \[\textbf{R4}\quad (x)\alpha.P=\textbf{nil}\quad \textrm{ if }x\textrm{is the subject of }\alpha\]
  \item Expansion $\textbf{E}$.
  Let $P\equiv\boxplus_i\sum_j \alpha_{ij}.P_{ij}$ and $Q\equiv\boxplus_{k}\sum_l\beta_{kl}.Q_{kl}$, where $bn(\alpha_{ij})\cap fn(Q)=\emptyset$ for all $i,j$, and
  $bn(\beta_{kl})\cap fn(P)=\emptyset$ for all $k,l$. Then,

\begin{enumerate}
  \item $P\parallel Q\sim_{pp}^{fr} \boxplus_{i}\boxplus_{k}\sum_j\sum_l (\alpha_{ij}\parallel \beta_{kl}).(P_{ij}\parallel Q_{kl})+\boxplus_i\boxplus_k\sum_{\alpha_{ij} \textrm{ comp }\beta_{kl}}\tau.R_{ijkl}$;
  \item $P\parallel Q\sim_{ps}^{fr} \boxplus_{i}\boxplus_{k}\sum_j\sum_l (\alpha_{ij}\parallel \beta_{kl}).(P_{ij}\parallel Q_{kl})+\boxplus_i\boxplus_k\sum_{\alpha_{ij} \textrm{ comp }\beta_{kl}}\tau.R_{ijkl}$;
  \item $P\parallel Q\sim_{php}^{fr} \boxplus_{i}\boxplus_{k}\sum_j\sum_l (\alpha_{ij}\parallel \beta_{kl}).(P_{ij}\parallel Q_{kl})+\boxplus_i\boxplus_k\sum_{\alpha_{ij} \textrm{ comp }\beta_{kl}}\tau.R_{ijkl}$;
  \item $P\parallel Q\nsim_{phhp} \boxplus_{i}\boxplus_{k}\sum_j\sum_l (\alpha_{ij}\parallel \beta_{kl}).(P_{ij}\parallel Q_{kl})+\boxplus_i\boxplus_k\sum_{\alpha_{ij} \textrm{ comp }\beta_{kl}}\tau.R_{ijkl}$.
\end{enumerate}

Where $\alpha_{ij}$ comp $\beta_{kl}$ and $R_{ijkl}$ are defined as follows:
\begin{enumerate}
  \item $\alpha_{ij}$ is $\overline{x}u$ and $\beta_{kl}$ is $x(v)$, then $R_{ijkl}=P_{ij}\parallel Q_{kl}\{u/v\}$;
  \item $\alpha_{ij}$ is $\overline{x}(u)$ and $\beta_{kl}$ is $x(v)$, then $R_{ijkl}=(w)(P_{ij}\{w/u\}\parallel Q_{kl}\{w/v\})$, if $w\notin fn((u)P_{ij})\cup fn((v)Q_{kl})$;
  \item $\alpha_{ij}$ is $x(v)$ and $\beta_{kl}$ is $\overline{x}u$, then $R_{ijkl}=P_{ij}\{u/v\}\parallel Q_{kl}$;
  \item $\alpha_{ij}$ is $x(v)$ and $\beta_{kl}$ is $\overline{x}(u)$, then $R_{ijkl}=(w)(P_{ij}\{w/v\}\parallel Q_{kl}\{w/u\})$, if $w\notin fn((v)P_{ij})\cup fn((u)Q_{kl})$.
\end{enumerate}

Let $P\equiv\boxplus_i\sum_j P_{ij}.\alpha_{ij}[m]$ and $Q\equiv\boxplus_{k}\sum_l Q_{kl}.\beta_{kl}[m]$, where $bn(\alpha_{ij}[m])\cap fn(Q)=\emptyset$ for all $i,j$, and
  $bn(\beta_{kl}[m])\cap fn(P)=\emptyset$ for all $k,l$. Then,

\begin{enumerate}
  \item $P\parallel Q\sim_{pp}^{fr} \boxplus_{i}\boxplus_{k}\sum_j\sum_l(P_{ij}\parallel Q_{kl}).(\alpha_{ij}[m]\parallel \beta_{kl}[m])+\boxplus_i\boxplus_k\sum_{\alpha_{ij} \textrm{ comp }\beta_{kl}} R_{ijkl}.\tau$;
  \item $P\parallel Q\sim_{ps}^{fr} \boxplus_{i}\boxplus_{k}\sum_j\sum_l(P_{ij}\parallel Q_{kl}).(\alpha_{ij}[m]\parallel \beta_{kl}[m])+\boxplus_i\boxplus_k\sum_{\alpha_{ij} \textrm{ comp }\beta_{kl}} R_{ijkl}.\tau$;
  \item $P\parallel Q\sim_{php}^{fr} \boxplus_{i}\boxplus_{k}\sum_j\sum_l(P_{ij}\parallel Q_{kl}).(\alpha_{ij}[m]\parallel \beta_{kl}[m])+\boxplus_i\boxplus_k\sum_{\alpha_{ij} \textrm{ comp }\beta_{kl}} R_{ijkl}.\tau$;
  \item $P\parallel Q\nsim_{phhp} \boxplus_{i}\boxplus_{k}\sum_j\sum_l(P_{ij}\parallel Q_{kl}).(\alpha_{ij}[m]\parallel \beta_{kl}[m])+\boxplus_i\boxplus_k\sum_{\alpha_{ij} \textrm{ comp }\beta_{kl}} R_{ijkl}.\tau$.
\end{enumerate}

Where $\alpha_{ij}[m]$ comp $\beta_{kl}[m]$ and $R_{ijkl}$ are defined as follows:
\begin{enumerate}
  \item $\alpha_{ij}[m]$ is $\overline{x}u$ and $\beta_{kl}[m]$ is $x(v)$, then $R_{ijkl}=P_{ij}\parallel Q_{kl}\{u/v\}$;
  \item $\alpha_{ij}[m]$ is $\overline{x}(u)$ and $\beta_{kl}[m]$ is $x(v)$, then $R_{ijkl}=(w)(P_{ij}\{w/u\}\parallel Q_{kl}\{w/v\})$, if $w\notin fn((u)P_{ij})\cup fn((v)Q_{kl})$;
  \item $\alpha_{ij}[m]$ is $x(v)$ and $\beta_{kl}[m]$ is $\overline{x}u$, then $R_{ijkl}=P_{ij}\{u/v\}\parallel Q_{kl}$;
  \item $\alpha_{ij}[m]$ is $x(v)$ and $\beta_{kl}[m]$ is $\overline{x}(u)$, then $R_{ijkl}=(w)(P_{ij}\{w/v\}\parallel Q_{kl}\{w/u\})$, if $w\notin fn((v)P_{ij})\cup fn((u)Q_{kl})$.
\end{enumerate}
  \item Identifier $\textbf{I}$.
  \[\textrm{If }A(\widetilde{x})\overset{\text{def}}{=}P,\textrm{ then }A(\widetilde{y})= P\{\widetilde{y}/\widetilde{x}\}.\]
\end{enumerate}
\end{definition}

\begin{theorem}[Soundness]
If $\textbf{STC}\vdash P=Q$ then
\begin{enumerate}
  \item $P\sim_{pp}^{fr} Q$;
  \item $P\sim_{pp}^{fr} Q$;
  \item $P\sim_{php}^{fr} Q$;
  \item $P\sim_{phhp}^{fr} Q$.
\end{enumerate}
\end{theorem}

\begin{proof}
The soundness of these laws modulo strongly truly concurrent bisimilarities is already proven in Section \ref{s5}.
\end{proof}

\begin{definition}
The agent identifier $A$ is weakly guardedly defined if every agent identifier is weakly guarded in the right-hand side of the definition of $A$.
\end{definition}

\begin{definition}[Head normal form]
A Process $P$ is in head normal form if it is a sum of the prefixes:

$$P\equiv \boxplus_{i}\sum_j(\alpha_{ij1}\parallel\cdots\parallel\alpha_{ijn}).P_{ij}\quad P\equiv \boxplus_{i}\sum_jP_{ij}.(\alpha_{ij1}[m]\parallel\cdots\parallel\alpha_{ijn}[m])$$
\end{definition}

\begin{proposition}
If every agent identifier is weakly guardedly defined, then for any process $P$, there is a head normal form $H$ such that

$$\textbf{STC}\vdash P=H.$$
\end{proposition}

\begin{proof}
It is sufficient to induct on the structure of $P$ and quite obvious.
\end{proof}

\begin{theorem}[Completeness]
For all processes $P$ and $Q$,
\begin{enumerate}
  \item if $P\sim_{pp}^{fr} Q$, then $\textbf{STC}\vdash P=Q$;
  \item if $P\sim_{pp}^{fr} Q$, then $\textbf{STC}\vdash P=Q$;
  \item if $P\sim_{php}^{fr} Q$, then $\textbf{STC}\vdash P=Q$.
\end{enumerate}
\end{theorem}

\begin{proof}
\begin{enumerate}
  \item if $P\sim_{pp}^{fr} Q$, then $\textbf{STC}\vdash P=Q$.

  For the forward transition case.

Since $P$ and $Q$ all have head normal forms, let $P\equiv\boxplus_{j=1}^l\sum_{i=1}^k\alpha_{ji}.P_{ji}$ and $Q\equiv\boxplus_{j=1}^l\sum_{i=1}^k\beta_{ji}.Q_{ji}$. Then the depth of
$P$, denoted as $d(P)=0$, if $k=0$; $d(P)=1+max\{d(P_{ji})\}$ for $1\leq j,i\leq k$. The depth $d(Q)$ can be defined similarly.

It is sufficient to induct on $d=d(P)+d(Q)$. When $d=0$, $P\equiv\textbf{nil}$ and $Q\equiv\textbf{nil}$, $P=Q$, as desired.

Suppose $d>0$.

\begin{itemize}
  \item If $(\alpha_1\parallel\cdots\parallel\alpha_n).M$ with $\alpha_{ji}(1\leq j,i\leq n)$ free actions is a summand of $P$, then
  $P\rightsquigarrow\xrightarrow{\{\alpha_1,\cdots,\alpha_n\}}M$.
  Since $Q$ is in head normal form and has a summand $(\alpha_1\parallel\cdots\parallel\alpha_n).N$ such that $M\sim_{pp}^{fr} N$, by the induction hypothesis $\textbf{STC}\vdash M=N$,
  $\textbf{STC}\vdash (\alpha_1\parallel\cdots\parallel\alpha_n).M= (\alpha_1\parallel\cdots\parallel\alpha_n).N$;
  \item If $x(y).M$ is a summand of $P$, then for $z\notin n(P, Q)$, $P\rightsquigarrow\xrightarrow{x(z)}M'\equiv M\{z/y\}$. Since $Q$ is in head normal form and has a summand
  $x(w).N$ such that for all $v$, $M'\{v/z\}\sim_{pp}^{fr} N'\{v/z\}$ where $N'\equiv N\{z/w\}$, by the induction hypothesis $\textbf{STC}\vdash M'\{v/z\}=N'\{v/z\}$, by the axioms
  $\textbf{C}$ and $\textbf{A}$, $\textbf{STC}\vdash x(y).M=x(w).N$;
  \item If $\overline{x}(y).M$ is a summand of $P$, then for $z\notin n(P,Q)$, $P\rightsquigarrow\xrightarrow{\overline{x}(z)}M'\equiv M\{z/y\}$. Since $Q$ is in head normal form and
  has a summand $\overline{x}(w).N$ such that $M'\sim_{pp}^{fr} N'$ where $N'\equiv N\{z/w\}$, by the induction hypothesis $\textbf{STC}\vdash M'=N'$, by the axioms
  $\textbf{A}$ and $\textbf{C}$, $\textbf{STC}\vdash \overline{x}(y).M= \overline{x}(w).N$.
\end{itemize}

For the reverse transition case, it can be proven similarly, and we omit it.

  \item if $P\sim_{pp}^{fr} Q$, then $\textbf{STC}\vdash P=Q$. It can be proven similarly to the above case.
  \item if $P\sim_{php}^{fr} Q$, then $\textbf{STC}\vdash P=Q$. It can be proven similarly to the above case.
\end{enumerate}
\end{proof}

\newpage\section{$\pi_{tc}$ with Probabilism and Guards}\label{pitcpg}

In this chapter, we design $\pi_{tc}$ with probabilism and guards. This chapter is organized as follows. In section \ref{os6}, we introduce the truly concurrent operational semantics. Then, we introduce
the syntax and operational semantics, laws modulo strongly truly concurrent bisimulations, and algebraic theory of $\pi_{tc}$ with probabilism and guards in section \ref{sos6},
\ref{s6} and \ref{a6} respectively.

\subsection{Operational Semantics}\label{os6}

Firstly, in this section, we introduce concepts of (strongly) probabilistic truly concurrent bisimilarities, including probabilistic pomset bisimilarity, probabilistic step
bisimilarity, probabilistic history-preserving (hp-)bisimilarity and probabilistic hereditary history-preserving (hhp-)bisimilarity. In contrast to traditional probabilistic truly
concurrent bisimilarities in section \ref{bg}, these versions in $\pi_{ptc}$ must take care of actions with bound objects. Note that, these probabilistic truly concurrent bisimilarities
are defined as late bisimilarities, but not early bisimilarities, as defined in $\pi$-calculus \cite{PI1} \cite{PI2}. Note that, here, a PES $\mathcal{E}$ is deemed as a process.

\begin{definition}[Prime event structure with silent event and empty event]
Let $\Lambda$ be a fixed set of labels, ranged over $a,b,c,\cdots$ and $\tau,\epsilon$. A ($\Lambda$-labelled) prime event structure with silent event $\tau$ and empty event
$\epsilon$ is a tuple $\mathcal{E}=\langle \mathbb{E}, \leq, \sharp, \lambda\rangle$, where $\mathbb{E}$ is a denumerable set of events, including the silent event $\tau$ and empty
event $\epsilon$. Let $\hat{\mathbb{E}}=\mathbb{E}\backslash\{\tau,\epsilon\}$, exactly excluding $\tau$ and $\epsilon$, it is obvious that $\hat{\tau^*}=\epsilon$. Let
$\lambda:\mathbb{E}\rightarrow\Lambda$ be a labelling function and let $\lambda(\tau)=\tau$ and $\lambda(\epsilon)=\epsilon$. And $\leq$, $\sharp$ are binary relations on $\mathbb{E}$,
called causality and conflict respectively, such that:

\begin{enumerate}
  \item $\leq$ is a partial order and $\lceil e \rceil = \{e'\in \mathbb{E}|e'\leq e\}$ is finite for all $e\in \mathbb{E}$. It is easy to see that
  $e\leq\tau^*\leq e'=e\leq\tau\leq\cdots\leq\tau\leq e'$, then $e\leq e'$.
  \item $\sharp$ is irreflexive, symmetric and hereditary with respect to $\leq$, that is, for all $e,e',e''\in \mathbb{E}$, if $e\sharp e'\leq e''$, then $e\sharp e''$.
\end{enumerate}

Then, the concepts of consistency and concurrency can be drawn from the above definition:

\begin{enumerate}
  \item $e,e'\in \mathbb{E}$ are consistent, denoted as $e\frown e'$, if $\neg(e\sharp e')$. A subset $X\subseteq \mathbb{E}$ is called consistent, if $e\frown e'$ for all
  $e,e'\in X$.
  \item $e,e'\in \mathbb{E}$ are concurrent, denoted as $e\parallel e'$, if $\neg(e\leq e')$, $\neg(e'\leq e)$, and $\neg(e\sharp e')$.
\end{enumerate}
\end{definition}

\begin{definition}[Configuration]
Let $\mathcal{E}$ be a PES. A (finite) configuration in $\mathcal{E}$ is a (finite) consistent subset of events $C\subseteq \mathcal{E}$, closed with respect to causality (i.e.
$\lceil C\rceil=C$), and a data state $s\in S$ with $S$ the set of all data states, denoted $\langle C, s\rangle$. The set of finite configurations of $\mathcal{E}$ is denoted by
$\langle\mathcal{C}(\mathcal{E}), S\rangle$. We let $\hat{C}=C\backslash\{\tau\}\cup\{\epsilon\}$.
\end{definition}

A consistent subset of $X\subseteq \mathbb{E}$ of events can be seen as a pomset. Given $X, Y\subseteq \mathbb{E}$, $\hat{X}\sim \hat{Y}$ if $\hat{X}$ and $\hat{Y}$ are isomorphic as
pomsets. In the following of the paper, we say $C_1\sim C_2$, we mean $\hat{C_1}\sim\hat{C_2}$.

\begin{definition}[Probabilistic transitions]
Let $\mathcal{E}$ be a PES and let $C\in\mathcal{C}(\mathcal{E})$, the transition $\langle C,s\rangle\xrsquigarrow{\pi} \langle C^{\pi},s\rangle$ is called a probabilistic
transition
from $\langle C,s\rangle$ to $\langle C^{\pi},s\rangle$.
\end{definition}

A probability distribution function (PDF) $\mu$ is a map $\mu:\mathcal{C}\times\mathcal{C}\rightarrow[0,1]$ and $\mu^*$ is the cumulative probability distribution function (cPDF).

\begin{definition}[Strongly probabilistic pomset, step bisimilarity]
Let $\mathcal{E}_1$, $\mathcal{E}_2$ be PESs. A strongly probabilistic pomset bisimulation is a relation $R\subseteq\langle\mathcal{C}(\mathcal{E}_1),s\rangle\times\langle\mathcal{C}(\mathcal{E}_2),s\rangle$,
such that (1) if $(\langle C_1,s\rangle,\langle C_2,s\rangle)\in R$, and $\langle C_1,s\rangle\xrightarrow{X_1}\langle C_1',s'\rangle$ (with $\mathcal{E}_1\xrightarrow{X_1}\mathcal{E}_1'$) then $\langle C_2,s\rangle\xrightarrow{X_2}\langle C_2',s'\rangle$ (with
$\mathcal{E}_2\xrightarrow{X_2}\mathcal{E}_2'$), with $X_1\subseteq \mathbb{E}_1$, $X_2\subseteq \mathbb{E}_2$, $X_1\sim X_2$ and $(\langle C_1',s'\rangle,\langle C_2',s'\rangle)\in R$:
\begin{enumerate}
  \item for each fresh action $\alpha\in X_1$, if $\langle C_1'',s''\rangle\xrightarrow{\alpha}\langle C_1''',s'''\rangle$ (with $\mathcal{E}_1''\xrightarrow{\alpha}\mathcal{E}_1'''$),
  then for some $C_2''$ and $\langle C_2''',s'''\rangle$, $\langle C_2'',s''\rangle\xrightarrow{\alpha}\langle C_2''',s'''\rangle$ (with
  $\mathcal{E}_2''\xrightarrow{\alpha}\mathcal{E}_2'''$), such that if $(\langle C_1'',s''\rangle,\langle C_2'',s''\rangle)\in R$ then $(\langle C_1''',s'''\rangle,\langle C_2''',s'''\rangle)\in R$;
  \item for each $x(y)\in X_1$ with ($y\notin n(\mathcal{E}_1, \mathcal{E}_2)$), if $\langle C_1'',s''\rangle\xrightarrow{x(y)}\langle C_1''',s'''\rangle$ (with
  $\mathcal{E}_1''\xrightarrow{x(y)}\mathcal{E}_1'''\{w/y\}$) for all $w$, then for some $C_2''$ and $C_2'''$, $\langle C_2'',s''\rangle\xrightarrow{x(y)}\langle C_2''',s'''\rangle$
  (with $\mathcal{E}_2''\xrightarrow{x(y)}\mathcal{E}_2'''\{w/y\}$) for all $w$, such that if $(\langle C_1'',s''\rangle,\langle C_2'',s''\rangle)\in R$ then $(\langle C_1''',s'''\rangle,\langle C_2''',s'''\rangle)\in R$;
  \item for each two $x_1(y),x_2(y)\in X_1$ with ($y\notin n(\mathcal{E}_1, \mathcal{E}_2)$), if $\langle C_1'',s''\rangle\xrightarrow{\{x_1(y),x_2(y)\}}\langle C_1''',s'''\rangle$
  (with $\mathcal{E}_1''\xrightarrow{\{x_1(y),x_2(y)\}}\mathcal{E}_1'''\{w/y\}$) for all $w$, then for some $C_2''$ and $C_2'''$,
  $\langle C_2'',s''\rangle\xrightarrow{\{x_1(y),x_2(y)\}}\langle C_2''',s'''\rangle$ (with $\mathcal{E}_2''\xrightarrow{\{x_1(y),x_2(y)\}}\mathcal{E}_2'''\{w/y\}$) for all $w$, such
  that if $(\langle C_1'',s''\rangle,\langle C_2'',s''\rangle)\in R$ then $(\langle C_1''',s'''\rangle,\langle C_2''',s'''\rangle)\in R$;
  \item for each $\overline{x}(y)\in X_1$ with $y\notin n(\mathcal{E}_1, \mathcal{E}_2)$, if $\langle C_1'',s''\rangle\xrightarrow{\overline{x}(y)}\langle C_1''',s'''\rangle$
  (with $\mathcal{E}_1''\xrightarrow{\overline{x}(y)}\mathcal{E}_1'''$), then for some $C_2''$ and $C_2'''$, $\langle C_2'',s''\rangle\xrightarrow{\overline{x}(y)}\langle C_2''',s'''\rangle$
  (with $\mathcal{E}_2''\xrightarrow{\overline{x}(y)}\mathcal{E}_2'''$), such that if $(\langle C_1'',s''\rangle,\langle C_2'',s''\rangle)\in R$ then $(\langle C_1''',s'''\rangle,\langle C_2''',s'''\rangle)\in R$.
\end{enumerate}
 and vice-versa; (2) if $(\langle C_1,s\rangle,\langle C_2,s\rangle)\in R$, and $\langle C_1,s\rangle\xrsquigarrow{\pi}\langle C_1^{\pi},s\rangle$ then
 $\langle C_2,s\rangle\xrsquigarrow{\pi}\langle C_2^{\pi},s\rangle$ and $(\langle C_1^{\pi},s\rangle,\langle C_2^{\pi},s\rangle)\in R$, and vice-versa; (3) if $(\langle C_1,s\rangle,\langle C_2,s\rangle)\in R$,
then $\mu(C_1,C)=\mu(C_2,C)$ for each $C\in\mathcal{C}(\mathcal{E})/R$; (4) $[\surd]_R=\{\surd\}$.

We say that $\mathcal{E}_1$, $\mathcal{E}_2$ are strongly probabilistic pomset bisimilar, written $\mathcal{E}_1\sim_{pp}\mathcal{E}_2$, if there exists a strongly probabilistic pomset
bisimulation $R$, such that $(\emptyset,\emptyset)\in R$. By replacing probabilistic pomset transitions with steps, we can get the definition of strongly probabilistic step bisimulation.
When PESs $\mathcal{E}_1$ and $\mathcal{E}_2$ are strongly probabilistic step bisimilar, we write $\mathcal{E}_1\sim_{ps}\mathcal{E}_2$.
\end{definition}

\begin{definition}[Posetal product]
Given two PESs $\mathcal{E}_1$, $\mathcal{E}_2$, the posetal product of their configurations, denoted
$\langle\mathcal{C}(\mathcal{E}_1),S\rangle\overline{\times}\langle\mathcal{C}(\mathcal{E}_2),S\rangle$, is defined as

$$\{(\langle C_1,s\rangle,f,\langle C_2,s\rangle)|C_1\in\mathcal{C}(\mathcal{E}_1),C_2\in\mathcal{C}(\mathcal{E}_2),f:C_1\rightarrow C_2 \textrm{ isomorphism}\}.$$

A subset $R\subseteq\langle\mathcal{C}(\mathcal{E}_1),S\rangle\overline{\times}\langle\mathcal{C}(\mathcal{E}_2),S\rangle$ is called a posetal relation. We say that $R$ is downward
closed when for any
$(\langle C_1,s\rangle,f,\langle C_2,s\rangle),(\langle C_1',s'\rangle,f',\langle C_2',s'\rangle)\in \langle\mathcal{C}(\mathcal{E}_1),S\rangle\overline{\times}\langle\mathcal{C}(\mathcal{E}_2),S\rangle$,
if $(\langle C_1,s\rangle,f,\langle C_2,s\rangle)\subseteq (\langle C_1',s'\rangle,f',\langle C_2',s'\rangle)$ pointwise and $(\langle C_1',s'\rangle,f',\langle C_2',s'\rangle)\in R$,
then $(\langle C_1,s\rangle,f,\langle C_2,s\rangle)\in R$.

For $f:X_1\rightarrow X_2$, we define $f[x_1\mapsto x_2]:X_1\cup\{x_1\}\rightarrow X_2\cup\{x_2\}$, $z\in X_1\cup\{x_1\}$,(1)$f[x_1\mapsto x_2](z)=
x_2$,if $z=x_1$;(2)$f[x_1\mapsto x_2](z)=f(z)$, otherwise. Where $X_1\subseteq \mathbb{E}_1$, $X_2\subseteq \mathbb{E}_2$, $x_1\in \mathbb{E}_1$, $x_2\in \mathbb{E}_2$.
\end{definition}

\begin{definition}[Strongly probabilistic (hereditary) history-preserving bisimilarity]
A strongly probabilistic history-preserving (hp-) bisimulation is a posetal relation $R\subseteq\mathcal{C}(\mathcal{E}_1)\overline{\times}\mathcal{C}(\mathcal{E}_2)$ such that
(1) if $(\langle C_1,s\rangle,f,\langle C_2,s\rangle)\in R$, and
\begin{enumerate}
  \item for $e_1=\alpha$ a fresh action, if $\langle C_1,s\rangle\xrightarrow{\alpha}\langle C_1',s'\rangle$ (with $\mathcal{E}_1\xrightarrow{\alpha}\mathcal{E}_1'$), then for some
  $C_2'$ and $e_2=\alpha$, $\langle C_2,s\rangle\xrightarrow{\alpha}\langle C_2',s'\rangle$ (with $\mathcal{E}_2\xrightarrow{\alpha}\mathcal{E}_2'$), such that
  $(\langle C_1',s'\rangle,f[e_1\mapsto e_2],\langle C_2',s'\rangle)\in R$;
  \item for $e_1=x(y)$ with ($y\notin n(\mathcal{E}_1, \mathcal{E}_2)$), if $\langle C_1,s\rangle\xrightarrow{x(y)}\langle C_1',s'\rangle$ (with
  $\mathcal{E}_1\xrightarrow{x(y)}\mathcal{E}_1'\{w/y\}$) for all $w$, then for some $C_2'$ and $e_2=x(y)$, $\langle C_2,s\rangle\xrightarrow{x(y)}\langle C_2',s'\rangle$ (with
  $\mathcal{E}_2\xrightarrow{x(y)}\mathcal{E}_2'\{w/y\}$) for all $w$, such that $(\langle C_1',s'\rangle,f[e_1\mapsto e_2],\langle C_2',s'\rangle)\in R$;
  \item for $e_1=\overline{x}(y)$ with $y\notin n(\mathcal{E}_1, \mathcal{E}_2)$, if $\langle C_1,s\rangle\xrightarrow{\overline{x}(y)}\langle C_1',s'\rangle$ (with
  $\mathcal{E}_1\xrightarrow{\overline{x}(y)}\mathcal{E}_1'$), then for some $C_2'$ and $e_2=\overline{x}(y)$, $\langle C_2,s\rangle\xrightarrow{\overline{x}(y)}\langle C_2',s'\rangle$
  (with $\mathcal{E}_2\xrightarrow{\overline{x}(y)}\mathcal{E}_2'$), such that $(\langle C_1',s'\rangle,f[e_1\mapsto e_2],\langle C_2',s'\rangle)\in R$.
\end{enumerate}
and vice-versa; (2) if $(\langle C_1,s\rangle,f,\langle C_2,s\rangle)\in R$, and $\langle C_1,s\rangle\xrsquigarrow{\pi}\langle C_1^{\pi},s\rangle$ then
$\langle C_2,s\rangle\xrsquigarrow{\pi}\langle C_2^{\pi},s\rangle$ and $(\langle C_1^{\pi},s\rangle,f,\langle C_2^{\pi},s\rangle)\in R$, and vice-versa; (3) if
$(\langle C_1,s\rangle,f,\langle C_2,s\rangle)\in R$, then $\mu(C_1,C)=\mu(C_2,C)$ for each $C\in\mathcal{C}(\mathcal{E})/R$; (4) $[\surd]_R=\{\surd\}$. $\mathcal{E}_1,\mathcal{E}_2$
are strongly probabilistic history-preserving (hp-)bisimilar and are written $\mathcal{E}_1\sim_{php}\mathcal{E}_2$ if there exists a strongly probabilistic hp-bisimulation
$R$ such that $(\emptyset,\emptyset,\emptyset)\in R$.

A strongly probabilistic hereditary history-preserving (hhp-)bisimulation is a downward closed strongly probabilistic hp-bisimulation. $\mathcal{E}_1,\mathcal{E}_2$ are FR
strongly probabilistic hereditary history-preserving (hhp-)bisimilar and are written $\mathcal{E}_1\sim_{phhp}\mathcal{E}_2$.
\end{definition}

\subsection{Syntax and Operational Semantics}\label{sos6}

We assume an infinite set $\mathcal{N}$ of (action or event) names, and use $a,b,c,\cdots$ to range over $\mathcal{N}$, use $x,y,z,w,u,v$ as meta-variables over names. We denote by
$\overline{\mathcal{N}}$ the set of co-names and let $\overline{a},\overline{b},\overline{c},\cdots$ range over $\overline{\mathcal{N}}$. Then we set
$\mathcal{L}=\mathcal{N}\cup\overline{\mathcal{N}}$ as the set of labels, and use $l,\overline{l}$ to range over $\mathcal{L}$. We extend complementation to $\mathcal{L}$ such that
$\overline{\overline{a}}=a$. Let $\tau$ denote the silent step (internal action or event) and define $Act=\mathcal{L}\cup\{\tau\}$ to be the set of actions, $\alpha,\beta$ range over
$Act$. And $K,L$ are used to stand for subsets of $\mathcal{L}$ and $\overline{L}$ is used for the set of complements of labels in $L$.

Further, we introduce a set $\mathcal{X}$ of process variables, and a set $\mathcal{K}$ of process constants, and let $X,Y,\cdots$ range over $\mathcal{X}$, and $A,B,\cdots$ range over
$\mathcal{K}$. For each process constant $A$, a nonnegative arity $ar(A)$ is assigned to it. Let $\widetilde{x}=x_1,\cdots,x_{ar(A)}$ be a tuple of distinct name variables, then
$A(\widetilde{x})$ is called a process constant. $\widetilde{X}$ is a tuple of distinct process variables, and also $E,F,\cdots$ range over the recursive expressions. We write
$\mathcal{P}$ for the set of processes. Sometimes, we use $I,J$ to stand for an indexing set, and we write $E_i:i\in I$ for a family of expressions indexed by $I$. $Id_D$ is the
identity function or relation over set $D$. The symbol $\equiv_{\alpha}$ denotes equality under standard alpha-convertibility, note that the subscript $\alpha$ has no relation to the
action $\alpha$.

Let $G_{at}$ be the set of atomic guards, $\delta$ be the deadlock constant, and $\epsilon$ be the empty action, and extend $Act$ to $Act\cup\{\epsilon\}\cup\{\delta\}$. We extend
$G_{at}$ to the set of basic guards $G$ with element $\phi,\psi,\cdots$, which is generated by the following formation rules:

$$\phi::=\delta|\epsilon|\neg\phi|\psi\in G_{at}|\phi+\psi|\phi\cdot\psi$$

The predicate $test(\phi,s)$ represents that $\phi$ holds in the state $s$, and $test(\epsilon,s)$ holds and $test(\delta,s)$ does not hold. $effect(e,s)\in S$ denotes $s'$ in
$s\xrightarrow{e}s'$. The predicate weakest precondition $wp(e,\phi)$ denotes that $\forall s,s'\in S, test(\phi,effect(e,s))$ holds.

\subsubsection{Syntax}

We use the Prefix $.$ to model the causality relation $\leq$ in true concurrency, the Summation $+$ to model the conflict relation $\sharp$, and $\boxplus_{\pi}$ to model the probabilistic
conflict relation $\sharp_{\pi}$ in probabilistic true concurrency, and the Composition $\parallel$ to explicitly model concurrent relation in true concurrency. And we follow the
conventions of process algebra.

\begin{definition}[Syntax]\label{syntax6}
A truly concurrent process $\pi_{tc}$ with probabilism and guards is defined inductively by the following formation rules:

\begin{enumerate}
  \item $A(\widetilde{x})\in\mathcal{P}$;
  \item $\phi\in\mathcal{P}$;
  \item $\textbf{nil}\in\mathcal{P}$;
  \item if $P\in\mathcal{P}$, then the Prefix $\tau.P\in\mathcal{P}$, for $\tau\in Act$ is the silent action;
  \item if $P\in\mathcal{P}$, then the Prefix $\phi.P\in\mathcal{P}$, for $\phi\in G_{at}$;
  \item if $P\in\mathcal{P}$, then the Output $\overline{x}y.P\in\mathcal{P}$, for $x,y\in Act$;
  \item if $P\in\mathcal{P}$, then the Input $x(y).P\in\mathcal{P}$, for $x,y\in Act$;
  \item if $P\in\mathcal{P}$, then the Restriction $(x)P\in\mathcal{P}$, for $x\in Act$;
  \item if $P,Q\in\mathcal{P}$, then the Summation $P+Q\in\mathcal{P}$;
  \item if $P,Q\in\mathcal{P}$, then the Summation $P\boxplus_{\pi}Q\in\mathcal{P}$;
  \item if $P,Q\in\mathcal{P}$, then the Composition $P\parallel Q\in\mathcal{P}$;
\end{enumerate}

The standard BNF grammar of syntax of $\pi_{tc}$ with probabilism and guards can be summarized as follows:

$$P::=A(\widetilde{x})|\textbf{nil}|\tau.P| \overline{x}y.P | x(y).P| (x)P  |\phi.P|  P+P| P\boxplus_{\pi}P | P\parallel P.$$
\end{definition}

In $\overline{x}y$, $x(y)$ and $\overline{x}(y)$, $x$ is called the subject, $y$ is called the object and it may be free or bound.

\begin{definition}[Free variables]
The free names of a process $P$, $fn(P)$, are defined as follows.

\begin{enumerate}
  \item $fn(A(\widetilde{x}))\subseteq\{\widetilde{x}\}$;
  \item $fn(\textbf{nil})=\emptyset$;
  \item $fn(\tau.P)=fn(P)$;
  \item $fn(\phi.P)=fn(P)$;
  \item $fn(\overline{x}y.P)=fn(P)\cup\{x\}\cup\{y\}$;
  \item $fn(x(y).P)=fn(P)\cup\{x\}-\{y\}$;
  \item $fn((x)P)=fn(P)-\{x\}$;
  \item $fn(P+Q)=fn(P)\cup fn(Q)$;
  \item $fn(P\boxplus_{\pi}Q)=fn(P)\cup fn(Q)$;
  \item $fn(P\parallel Q)=fn(P)\cup fn(Q)$.
\end{enumerate}
\end{definition}

\begin{definition}[Bound variables]
Let $n(P)$ be the names of a process $P$, then the bound names $bn(P)=n(P)-fn(P)$.
\end{definition}

For each process constant schema $A(\widetilde{x})$, a defining equation of the form

$$A(\widetilde{x})\overset{\text{def}}{=}P$$

is assumed, where $P$ is a process with $fn(P)\subseteq \{\widetilde{x}\}$.

\begin{definition}[Substitutions]\label{subs6}
A substitution is a function $\sigma:\mathcal{N}\rightarrow\mathcal{N}$. For $x_i\sigma=y_i$ with $1\leq i\leq n$, we write $\{y_1/x_1,\cdots,y_n/x_n\}$ or
$\{\widetilde{y}/\widetilde{x}\}$ for $\sigma$. For a process $P\in\mathcal{P}$, $P\sigma$ is defined inductively as follows:
\begin{enumerate}
  \item if $P$ is a process constant $A(\widetilde{x})=A(x_1,\cdots,x_n)$, then $P\sigma=A(x_1\sigma,\cdots,x_n\sigma)$;
  \item if $P=\textbf{nil}$, then $P\sigma=\textbf{nil}$;
  \item if $P=\tau.P'$, then $P\sigma=\tau.P'\sigma$;
  \item if $P=\phi.P'$, then $P\sigma=\phi.P'\sigma$;
  \item if $P=\overline{x}y.P'$, then $P\sigma=\overline{x\sigma}y\sigma.P'\sigma$;
  \item if $P=x(y).P'$, then $P\sigma=x\sigma(y).P'\sigma$;
  \item if $P=(x)P'$, then $P\sigma=(x\sigma)P'\sigma$;
  \item if $P=P_1+P_2$, then $P\sigma=P_1\sigma+P_2\sigma$;
  \item if $P=P_1\boxplus_{\pi}P_2$, then $P\sigma=P_1\sigma\boxplus_{\pi}P_2\sigma$;
  \item if $P=P_1\parallel P_2$, then $P\sigma=P_1\sigma \parallel P_2\sigma$.
\end{enumerate}
\end{definition}

\subsubsection{Operational Semantics}

The operational semantics is defined by LTSs (labelled transition systems), and it is detailed by the following definition.

\begin{definition}[Semantics]\label{semantics6}
The operational semantics of $\pi_{tc}$ with probabilism and guards corresponding to the syntax in Definition \ref{syntax6} is defined by a series of transition rules, named $\textbf{PACT}$, $\textbf{PSUM}$, $\textbf{PBOX-SUM}$,
$\textbf{PIDE}$, $\textbf{PPAR}$, $\textbf{PRES}$ and named $\textbf{ACT}$, $\textbf{SUM}$,
$\textbf{IDE}$, $\textbf{PAR}$, $\textbf{COM}$, $\textbf{CLOSE}$, $\textbf{RES}$, $\textbf{OPEN}$ indicate that the rules are associated respectively with Prefix, Summation, Box-Summation,
Identity, Parallel Composition, Communication, and Restriction in Definition \ref{syntax6}. They are shown in Table \ref{PTRForPITC6} and \ref{TRForPITC6}.

\begin{center}
    \begin{table}
        \[\textbf{PTAU-ACT}\quad \frac{}{\langle \tau.P,s\rangle\rightsquigarrow \langle\breve{\tau}.P,s\rangle}\]

        \[\textbf{POUTPUT-ACT}\quad \frac{}{\langle\overline{x}y.P,s\rangle\rightsquigarrow \langle\breve{\overline{x}y}.P,s\rangle}\]

        \[\textbf{PINPUT-ACT}\quad \frac{}{\langle x(z).P,s\rangle\rightsquigarrow \langle\breve{x(z)}.P,s\rangle}\]

        \[\textbf{PPAR}\quad \frac{\langle P,s\rangle\rightsquigarrow \langle P',s\rangle\quad \langle Q,s\rangle\rightsquigarrow \langle Q',s\rangle}{\langle P\parallel Q,s\rangle\rightsquigarrow \langle P'\parallel Q',s\rangle}\]

        \[\textbf{PSUM}\quad \frac{\langle P,s\rangle\rightsquigarrow \langle P',s\rangle\quad \langle Q,s\rangle\rightsquigarrow \langle Q',s\rangle}{\langle P+Q,s\rangle\rightsquigarrow \langle P'+Q',s\rangle}\]

        \[\textbf{PBOX-SUM}\quad \frac{\langle P,s\rangle\rightsquigarrow \langle P',s\rangle}{\langle P\boxplus_{\pi}Q,s\rangle\rightsquigarrow \langle P',s\rangle}\]

        \[\textbf{PIDE}\quad\frac{\langle P\{\widetilde{y}/\widetilde{x}\},s\rangle\rightsquigarrow \langle P',s\rangle}{\langle A(\widetilde{y}),s\rangle\rightsquigarrow \langle P',s\rangle}\quad (A(\widetilde{x})\overset{\text{def}}{=}P)\]

        \[\textbf{PRES}\quad \frac{\langle P,s\rangle\rightsquigarrow \langle P',s\rangle}{\langle (y)P,s\rangle\rightsquigarrow \langle (y)P',s\rangle}\quad (y\notin n(\alpha))\]

        \caption{Probabilistic transition rules}
        \label{PTRForPITC6}
    \end{table}
\end{center}

\begin{center}
    \begin{table}
        \[\textbf{TAU-ACT}\quad \frac{}{\langle\breve{\tau}.P,s\rangle\xrightarrow{\tau}\langle P,\tau(s)\rangle}\]

        \[\textbf{OUTPUT-ACT}\quad \frac{}{\langle\breve{\overline{x}y}.P,s\rangle\xrightarrow{\overline{x}y}\langle P,s'\rangle}\]

        \[\textbf{INPUT-ACT}\quad \frac{}{\langle\breve{x(z)}.P,s\rangle\xrightarrow{x(w)}\langle P\{w/z\},s'\rangle}\quad (w\notin fn((z)P))\]

        \[\textbf{PAR}_1\quad \frac{\langle P,s\rangle\xrightarrow{\alpha}\langle P',s'\rangle\quad \langle Q,s\rangle\nrightarrow}{\langle P\parallel Q,s\rangle\xrightarrow{\alpha}\langle P'\parallel Q,s'\rangle}\quad (bn(\alpha)\cap fn(Q)=\emptyset)\]

        \[\textbf{PAR}_2\quad \frac{\langle Q,s\rangle\xrightarrow{\alpha}\langle Q',s'\rangle\quad \langle P,s\rangle\nrightarrow}{\langle P\parallel Q,s\rangle\xrightarrow{\alpha}\langle P\parallel Q',s'\rangle}\quad (bn(\alpha)\cap fn(P)=\emptyset)\]

        \[\textbf{PAR}_3\quad \frac{\langle P,s\rangle\xrightarrow{\alpha}\langle P',s'\rangle\quad \langle Q,s\rangle\xrightarrow{\beta}\langle Q',s''\rangle}{\langle P\parallel Q,s\rangle\xrightarrow{\{\alpha,\beta\}}\langle P'\parallel Q',s'\cup s''\rangle}\] $(\beta\neq\overline{\alpha}, bn(\alpha)\cap bn(\beta)=\emptyset, bn(\alpha)\cap fn(Q)=\emptyset,bn(\beta)\cap fn(P)=\emptyset)$

        \[\textbf{PAR}_4\quad \frac{\langle P,s\rangle\xrightarrow{x_1(z)}\langle P',s'\rangle\quad \langle Q,s\rangle\xrightarrow{x_2(z)}\langle Q',s''\rangle}{\langle P\parallel Q,s\rangle\xrightarrow{\{x_1(w),x_2(w)\}}\langle P'\{w/z\}\parallel Q'\{w/z\},s'\cup s''\rangle}\quad (w\notin fn((z)P)\cup fn((z)Q))\]

        \[\textbf{COM}\quad \frac{\langle P,s\rangle\xrightarrow{\overline{x}y}\langle P',s'\rangle\quad \langle Q,s\rangle\xrightarrow{x(z)}\langle Q',s''\rangle}{\langle P\parallel Q,s\rangle\xrightarrow{\tau}\langle P'\parallel Q'\{y/z\},s'\cup s''\rangle}\]

        \[\textbf{CLOSE}\quad \frac{\langle P,s\rangle\xrightarrow{\overline{x}(w)}\langle P',s'\rangle\quad \langle Q,s\rangle\xrightarrow{x(w)}\langle Q',s''\rangle}{\langle P\parallel Q,s\rangle\xrightarrow{\tau}\langle (w)(P'\parallel Q'),s'\cup s''\rangle}\]

        \caption{Action transition rules}
        \label{TRForPITC6}
    \end{table}
\end{center}

\begin{center}
    \begin{table}
        \[\textbf{SUM}_1\quad \frac{\langle P,s\rangle\xrightarrow{\alpha}\langle P',s'\rangle}{\langle P+Q,s\rangle\xrightarrow{\alpha}\langle P',s'\rangle}\]

        \[\textbf{SUM}_2\quad \frac{\langle P,s\rangle\xrightarrow{\{\alpha_1,\cdots,\alpha_n\}}\langle P',s'\rangle}{\langle P+Q,s\rangle\xrightarrow{\{\alpha_1,\cdots,\alpha_n\}}\langle P',s'\rangle}\]

        \[\textbf{IDE}_1\quad\frac{\langle P\{\widetilde{y}/\widetilde{x}\},s\rangle\xrightarrow{\alpha}\langle P',s'\rangle}{\langle A(\widetilde{y}),s\rangle\xrightarrow{\alpha}\langle P',s'\rangle}\quad (A(\widetilde{x})\overset{\text{def}}{=}P)\]

        \[\textbf{IDE}_2\quad\frac{\langle P\{\widetilde{y}/\widetilde{x}\},s\rangle\xrightarrow{\{\alpha_1,\cdots,\alpha_n\}}\langle P',s'\rangle} {\langle A(\widetilde{y}),s\rangle\xrightarrow{\{\alpha_1,\cdots,\alpha_n\}}\langle P',s'\rangle}\quad (A(\widetilde{x})\overset{\text{def}}{=}P)\]

        \[\textbf{RES}_1\quad \frac{\langle P,s\rangle\xrightarrow{\alpha}\langle P',s'\rangle}{\langle (y)P,s\rangle\xrightarrow{\alpha}\langle (y)P',s'\rangle}\quad (y\notin n(\alpha))\]

        \[\textbf{RES}_2\quad \frac{\langle P,s\rangle\xrightarrow{\{\alpha_1,\cdots,\alpha_n\}}\langle P',s'\rangle}{\langle (y)P,s\rangle\xrightarrow{\{\alpha_1,\cdots,\alpha_n\}}\langle (y)P',s'\rangle}\quad (y\notin n(\alpha_1)\cup\cdots\cup n(\alpha_n))\]

        \[\textbf{OPEN}_1\quad \frac{\langle P,s\rangle\xrightarrow{\overline{x}y}\langle P',s'\rangle}{\langle (y)P,s\rangle\xrightarrow{\overline{x}(w)}\langle P'\{w/y\},s'\rangle} \quad (y\neq x, w\notin fn((y)P'))\]

        \[\textbf{OPEN}_2\quad \frac{\langle P,s\rangle\xrightarrow{\{\overline{x}_1 y,\cdots,\overline{x}_n y\}}\langle P',s'\rangle}{\langle(y)P,s\rangle\xrightarrow{\{\overline{x}_1(w),\cdots,\overline{x}_n(w)\}}\langle P'\{w/y\},s'\rangle} \quad (y\neq x_1\neq\cdots\neq x_n, w\notin fn((y)P'))\]

        \caption{Action transition rules (continuing)}
        \label{TRForPITC62}
    \end{table}
\end{center}
\end{definition}

\subsubsection{Properties of Transitions}

\begin{proposition}
\begin{enumerate}
  \item If $\langle P,s\rangle\xrightarrow{\alpha}\langle P',s'\rangle$ then
  \begin{enumerate}
    \item $fn(\alpha)\subseteq fn(P)$;
    \item $fn(P')\subseteq fn(P)\cup bn(\alpha)$;
  \end{enumerate}
  \item If $\langle P,s\rangle\xrightarrow{\{\alpha_1,\cdots,\alpha_n\}}\langle P',s\rangle$ then
  \begin{enumerate}
    \item $fn(\alpha_1)\cup\cdots\cup fn(\alpha_n)\subseteq fn(P)$;
    \item $fn(P')\subseteq fn(P)\cup bn(\alpha_1)\cup\cdots\cup bn(\alpha_n)$.
  \end{enumerate}
\end{enumerate}
\end{proposition}

\begin{proof}
By induction on the depth of inference.
\end{proof}

\begin{proposition}
Suppose that $\langle P,s\rangle\xrightarrow{\alpha(y)}\langle P',s'\rangle$, where $\alpha=x$ or $\alpha=\overline{x}$, and $x\notin n(P)$, then there exists some $P''\equiv_{\alpha}P'\{z/y\}$,
$\langle P,s\rangle\xrightarrow{\alpha(z)}\langle P'',s''\rangle$.
\end{proposition}

\begin{proof}
By induction on the depth of inference.
\end{proof}

\begin{proposition}
If $\langle P,s\rangle\xrightarrow{\alpha} \langle P',s'\rangle$, $bn(\alpha)\cap fn(P'\sigma)=\emptyset$, and $\sigma\lceil bn(\alpha)=id$, then there exists some $P''\equiv_{\alpha}P'\sigma$,
$\langle P,s\rangle\sigma\xrightarrow{\alpha\sigma}\langle P'',s''\rangle$.
\end{proposition}

\begin{proof}
By the definition of substitution (Definition \ref{subs6}) and induction on the depth of inference.
\end{proof}

\begin{proposition}
\begin{enumerate}
  \item If $\langle P\{w/z\},s\rangle\xrightarrow{\alpha}\langle P',s'\rangle$, where $w\notin fn(P)$ and $bn(\alpha)\cap fn(P,w)=\emptyset$, then there exist some $Q$ and $\beta$ with $Q\{w/z\}\equiv_{\alpha}P'$ and
  $\beta\sigma=\alpha$, $\langle P,s\rangle\xrightarrow{\beta}\langle Q,s'\rangle$;
  \item If $\langle P\{w/z\},s\rangle\xrightarrow{\{\alpha_1,\cdots,\alpha_n\}}\langle P',s'\rangle$, where $w\notin fn(P)$ and $bn(\alpha_1)\cap\cdots\cap bn(\alpha_n)\cap fn(P,w)=\emptyset$, then there exist some $Q$
  and $\beta_1,\cdots,\beta_n$ with $Q\{w/z\}\equiv_{\alpha}P'$ and $\beta_1\sigma=\alpha_1,\cdots,\beta_n\sigma=\alpha_n$, $\langle P,s\rangle\xrightarrow{\{\beta_1,\cdots,\beta_n\}}\langle Q,s'\rangle$.
\end{enumerate}

\end{proposition}

\begin{proof}
By the definition of substitution (Definition \ref{subs6}) and induction on the depth of inference.
\end{proof}

\subsection{Strong Bisimilarities}\label{s6}

\subsubsection{Laws and Congruence}

\begin{theorem}
$\equiv_{\alpha}$ are strongly probabilistic truly concurrent bisimulations. That is, if $P\equiv_{\alpha}Q$, then,
\begin{enumerate}
  \item $P\sim_{pp}  Q$;
  \item $P\sim_{ps}  Q$;
  \item $P\sim_{php}  Q$;
  \item $P\sim_{phhp}  Q$.
\end{enumerate}
\end{theorem}

\begin{proof}
By induction on the depth of inference, we can get the following facts:

\begin{enumerate}
  \item If $\alpha$ is a free action and $\langle P,s\rangle\rightsquigarrow\xrightarrow{\alpha}\langle P',s'\rangle$, then equally for some $Q'$ with $P'\equiv_{\alpha}Q'$,
  $\langle Q,s\rangle\rightsquigarrow\xrightarrow{\alpha}\langle Q',s'\rangle$;
  \item If $\langle P,s\rangle\rightsquigarrow\xrightarrow{a(y)}\langle P',s'\rangle$ with $a=x$ or $a=\overline{x}$ and $z\notin n(Q)$, then equally for some $Q'$ with $P'\{z/y\}\equiv_{\alpha}Q'$,
  $\langle Q,s\rangle\rightsquigarrow\xrightarrow{a(z)}\langle Q',s'\rangle$.
\end{enumerate}

Then, we can get:

\begin{enumerate}
  \item by the definition of strongly probabilistic pomset bisimilarity, $P\sim_{pp}  Q$;
  \item by the definition of strongly probabilistic step bisimilarity, $P\sim_{ps}  Q$;
  \item by the definition of strongly probabilistic hp-bisimilarity, $P\sim_{php}  Q$;
  \item by the definition of strongly probabilistic hhp-bisimilarity, $P\sim_{phhp}  Q$.
\end{enumerate}
\end{proof}

\begin{proposition}[Summation laws for strongly probabilistic pomset bisimulation] The Summation laws for strongly probabilistic pomset bisimulation are as follows.

\begin{enumerate}
  \item $P+Q\sim_{pp}  Q+P$;
  \item $P+(Q+R)\sim_{pp}  (P+Q)+R$;
  \item $P+P\sim_{pp}  P$;
  \item $P+\textbf{nil}\sim_{pp}  P$.
\end{enumerate}

\end{proposition}

\begin{proof}
\begin{enumerate}
  \item $P+Q\sim_{pp}  Q+P$. It is sufficient to prove the relation $R=\{(P+Q, Q+P)\}\cup \textbf{Id}$ is a strongly probabilistic pomset bisimulation, we omit it;
  \item $P+(Q+R)\sim_{pp}  (P+Q)+R$. It is sufficient to prove the relation $R=\{(P+(Q+R), (P+Q)+R)\}\cup \textbf{Id}$ is a strongly probabilistic pomset bisimulation, we omit it;
  \item $P+P\sim_{pp}  P$. It is sufficient to prove the relation $R=\{(P+P, P)\}\cup \textbf{Id}$ is a strongly probabilistic pomset bisimulation, we omit it;
  \item $P+\textbf{nil}\sim_{pp}  P$. It is sufficient to prove the relation $R=\{(P+\textbf{nil}, P)\}\cup \textbf{Id}$ is a strongly probabilistic pomset bisimulation, we omit it.
\end{enumerate}
\end{proof}

\begin{proposition}[Summation laws for strongly probabilistic step bisimulation] The Summation laws for strongly probabilistic step bisimulation are as follows.
\begin{enumerate}
  \item $P+Q\sim_{ps}  Q+P$;
  \item $P+(Q+R)\sim_{ps}  (P+Q)+R$;
  \item $P+P\sim_{ps}  P$;
  \item $P+\textbf{nil}\sim_{ps}  P$.
\end{enumerate}
\end{proposition}

\begin{proof}
\begin{enumerate}
  \item $P+Q\sim_{ps}  Q+P$. It is sufficient to prove the relation $R=\{(P+Q, Q+P)\}\cup \textbf{Id}$ is a strongly probabilistic step bisimulation, we omit it;
  \item $P+(Q+R)\sim_{ps}  (P+Q)+R$. It is sufficient to prove the relation $R=\{(P+(Q+R), (P+Q)+R)\}\cup \textbf{Id}$ is a strongly probabilistic step bisimulation, we omit it;
  \item $P+P\sim_{ps}  P$. It is sufficient to prove the relation $R=\{(P+P, P)\}\cup \textbf{Id}$ is a strongly probabilistic step bisimulation, we omit it;
  \item $P+\textbf{nil}\sim_{ps}  P$. It is sufficient to prove the relation $R=\{(P+\textbf{nil}, P)\}\cup \textbf{Id}$ is a strongly probabilistic step bisimulation, we omit it.
\end{enumerate}
\end{proof}

\begin{proposition}[Summation laws for strongly probabilistic hp-bisimulation] The Summation laws for strongly probabilistic hp-bisimulation are as follows.
\begin{enumerate}
  \item $P+Q\sim_{php}  Q+P$;
  \item $P+(Q+R)\sim_{php}  (P+Q)+R$;
  \item $P+P\sim_{php}  P$;
  \item $P+\textbf{nil}\sim_{php}  P$.
\end{enumerate}
\end{proposition}

\begin{proof}
\begin{enumerate}
  \item $P+Q\sim_{php}  Q+P$. It is sufficient to prove the relation $R=\{(P+Q, Q+P)\}\cup \textbf{Id}$ is a strongly probabilistic hp-bisimulation, we omit it;
  \item $P+(Q+R)\sim_{php}  (P+Q)+R$. It is sufficient to prove the relation $R=\{(P+(Q+R), (P+Q)+R)\}\cup \textbf{Id}$ is a strongly probabilistic hp-bisimulation, we omit it;
  \item $P+P\sim_{php}  P$. It is sufficient to prove the relation $R=\{(P+P, P)\}\cup \textbf{Id}$ is a strongly probabilistic hp-bisimulation, we omit it;
  \item $P+\textbf{nil}\sim_{php}  P$. It is sufficient to prove the relation $R=\{(P+\textbf{nil}, P)\}\cup \textbf{Id}$ is a strongly probabilistic hp-bisimulation, we omit it.
\end{enumerate}
\end{proof}

\begin{proposition}[Summation laws for strongly probabilistic hhp-bisimulation] The Summation laws for strongly probabilistic hhp-bisimulation are as follows.
\begin{enumerate}
  \item $P+Q\sim_{phhp}  Q+P$;
  \item $P+(Q+R)\sim_{phhp}  (P+Q)+R$;
  \item $P+P\sim_{phhp}  P$;
  \item $P+\textbf{nil}\sim_{phhp}  P$.
\end{enumerate}
\end{proposition}

\begin{proof}
\begin{enumerate}
  \item $P+Q\sim_{phhp}  Q+P$. It is sufficient to prove the relation $R=\{(P+Q, Q+P)\}\cup \textbf{Id}$ is a strongly probabilistic hhp-bisimulation, we omit it;
  \item $P+(Q+R)\sim_{phhp}  (P+Q)+R$. It is sufficient to prove the relation $R=\{(P+(Q+R), (P+Q)+R)\}\cup \textbf{Id}$ is a strongly probabilistic hhp-bisimulation, we omit it;
  \item $P+P\sim_{phhp}  P$. It is sufficient to prove the relation $R=\{(P+P, P)\}\cup \textbf{Id}$ is a strongly probabilistic hhp-bisimulation, we omit it;
  \item $P+\textbf{nil}\sim_{phhp}  P$. It is sufficient to prove the relation $R=\{(P+\textbf{nil}, P)\}\cup \textbf{Id}$ is a strongly probabilistic hhp-bisimulation, we omit it.
\end{enumerate}
\end{proof}

\begin{proposition}[Box-Summation laws for strongly probabilistic pomset bisimulation]
The Box-Summation laws for strongly probabilistic pomset bisimulation are as follows.

\begin{enumerate}
  \item $P\boxplus_{\pi} Q\sim_{pp}  Q\boxplus_{1-\pi} P$;
  \item $P\boxplus_{\pi}(Q\boxplus_{\rho} R)\sim_{pp}  (P\boxplus_{\frac{\pi}{\pi+\rho-\pi\rho}}Q)\boxplus_{\pi+\rho-\pi\rho} R$;
  \item $P\boxplus_{\pi}P\sim_{pp}  P$;
  \item $P\boxplus_{\pi}\textbf{nil}\sim_{pp}  P$.
\end{enumerate}
\end{proposition}

\begin{proof}
\begin{enumerate}
  \item $P\boxplus_{\pi} Q\sim_{pp}  Q\boxplus_{1-\pi} P$. It is sufficient to prove the relation $R=\{(P\boxplus_{\pi} Q, Q\boxplus_{1-\pi} P)\}\cup \textbf{Id}$ is a strongly probabilistic pomset bisimulation, we omit it;
  \item $P\boxplus_{\pi}(Q\boxplus_{\rho} R)\sim_{pp}  (P\boxplus_{\frac{\pi}{\pi+\rho-\pi\rho}}Q)\boxplus_{\pi+\rho-\pi\rho} R$. It is sufficient to prove the relation $R=\{(P\boxplus_{\pi}(Q\boxplus_{\rho} R), (P\boxplus_{\frac{\pi}{\pi+\rho-\pi\rho}}Q)\boxplus_{\pi+\rho-\pi\rho} R)\}\cup \textbf{Id}$ is a strongly probabilistic pomset bisimulation, we omit it;
  \item $P\boxplus_{\pi}P\sim_{pp}  P$. It is sufficient to prove the relation $R=\{(P\boxplus_{\pi}P, P)\}\cup \textbf{Id}$ is a strongly probabilistic pomset bisimulation, we omit it;
  \item $P\boxplus_{\pi}\textbf{nil}\sim_{pp}  P$. It is sufficient to prove the relation $R=\{(P\boxplus_{\pi}\textbf{nil}, P)\}\cup \textbf{Id}$ is a strongly probabilistic pomset bisimulation, we omit it.
\end{enumerate}
\end{proof}

\begin{proposition}[Box-Summation laws for strongly probabilistic step bisimulation]
The Box-Summation laws for strongly probabilistic step bisimulation are as follows.

\begin{enumerate}
  \item $P\boxplus_{\pi} Q\sim_{ps}  Q\boxplus_{1-\pi} P$;
  \item $P\boxplus_{\pi}(Q\boxplus_{\rho} R)\sim_{ps}  (P\boxplus_{\frac{\pi}{\pi+\rho-\pi\rho}}Q)\boxplus_{\pi+\rho-\pi\rho} R$;
  \item $P\boxplus_{\pi}P\sim_{ps}  P$;
  \item $P\boxplus_{\pi}\textbf{nil}\sim_{ps}  P$.
\end{enumerate}
\end{proposition}

\begin{proof}
\begin{enumerate}
  \item $P\boxplus_{\pi} Q\sim_{ps}  Q\boxplus_{1-\pi} P$. It is sufficient to prove the relation $R=\{(P\boxplus_{\pi} Q, Q\boxplus_{1-\pi} P)\}\cup \textbf{Id}$ is a strongly probabilistic step bisimulation, we omit it;
  \item $P\boxplus_{\pi}(Q\boxplus_{\rho} R)\sim_{ps}  (P\boxplus_{\frac{\pi}{\pi+\rho-\pi\rho}}Q)\boxplus_{\pi+\rho-\pi\rho} R$. It is sufficient to prove the relation $R=\{(P\boxplus_{\pi}(Q\boxplus_{\rho} R), (P\boxplus_{\frac{\pi}{\pi+\rho-\pi\rho}}Q)\boxplus_{\pi+\rho-\pi\rho} R)\}\cup \textbf{Id}$ is a strongly probabilistic step bisimulation, we omit it;
  \item $P\boxplus_{\pi}P\sim_{ps}  P$. It is sufficient to prove the relation $R=\{(P\boxplus_{\pi}P, P)\}\cup \textbf{Id}$ is a strongly probabilistic step bisimulation, we omit it;
  \item $P\boxplus_{\pi}\textbf{nil}\sim_{ps}  P$. It is sufficient to prove the relation $R=\{(P\boxplus_{\pi}\textbf{nil}, P)\}\cup \textbf{Id}$ is a strongly probabilistic step bisimulation, we omit it.
\end{enumerate}
\end{proof}

\begin{proposition}[Box-Summation laws for strongly probabilistic hp-bisimulation]
The Box-Summation laws for strongly probabilistic hp-bisimulation are as follows.

\begin{enumerate}
  \item $P\boxplus_{\pi} Q\sim_{php}  Q\boxplus_{1-\pi} P$;
  \item $P\boxplus_{\pi}(Q\boxplus_{\rho} R)\sim_{php}  (P\boxplus_{\frac{\pi}{\pi+\rho-\pi\rho}}Q)\boxplus_{\pi+\rho-\pi\rho} R$;
  \item $P\boxplus_{\pi}P\sim_{php}  P$;
  \item $P\boxplus_{\pi}\textbf{nil}\sim_{php}  P$.
\end{enumerate}
\end{proposition}

\begin{proof}
\begin{enumerate}
  \item $P\boxplus_{\pi} Q\sim_{php}  Q\boxplus_{1-\pi} P$. It is sufficient to prove the relation $R=\{(P\boxplus_{\pi} Q, Q\boxplus_{1-\pi} P)\}\cup \textbf{Id}$ is a strongly probabilistic hp-bisimulation, we omit it;
  \item $P\boxplus_{\pi}(Q\boxplus_{\rho} R)\sim_{php}  (P\boxplus_{\frac{\pi}{\pi+\rho-\pi\rho}}Q)\boxplus_{\pi+\rho-\pi\rho} R$. It is sufficient to prove the relation $R=\{(P\boxplus_{\pi}(Q\boxplus_{\rho} R), (P\boxplus_{\frac{\pi}{\pi+\rho-\pi\rho}}Q)\boxplus_{\pi+\rho-\pi\rho} R)\}\cup \textbf{Id}$ is a strongly probabilistic hp-bisimulation, we omit it;
  \item $P\boxplus_{\pi}P\sim_{php}  P$. It is sufficient to prove the relation $R=\{(P\boxplus_{\pi}P, P)\}\cup \textbf{Id}$ is a strongly probabilistic hp-bisimulation, we omit it;
  \item $P\boxplus_{\pi}\textbf{nil}\sim_{php}  P$. It is sufficient to prove the relation $R=\{(P\boxplus_{\pi}\textbf{nil}, P)\}\cup \textbf{Id}$ is a strongly probabilistic hp-bisimulation, we omit it.
\end{enumerate}
\end{proof}

\begin{proposition}[Box-Summation laws for strongly probabilistic hhp-bisimulation]
The Box-Summation laws for strongly probabilistic hhp-bisimulation are as follows.

\begin{enumerate}
  \item $P\boxplus_{\pi} Q\sim_{phhp}  Q\boxplus_{1-\pi} P$;
  \item $P\boxplus_{\pi}(Q\boxplus_{\rho} R)\sim_{phhp}  (P\boxplus_{\frac{\pi}{\pi+\rho-\pi\rho}}Q)\boxplus_{\pi+\rho-\pi\rho} R$;
  \item $P\boxplus_{\pi}P\sim_{phhp}  P$;
  \item $P\boxplus_{\pi}\textbf{nil}\sim_{phhp}  P$.
\end{enumerate}
\end{proposition}

\begin{proof}
\begin{enumerate}
  \item $P\boxplus_{\pi} Q\sim_{phhp}  Q\boxplus_{1-\pi} P$. It is sufficient to prove the relation $R=\{(P\boxplus_{\pi} Q, Q\boxplus_{1-\pi} P)\}\cup \textbf{Id}$ is a strongly probabilistic hhp-bisimulation, we omit it;
  \item $P\boxplus_{\pi}(Q\boxplus_{\rho} R)\sim_{phhp}  (P\boxplus_{\frac{\pi}{\pi+\rho-\pi\rho}}Q)\boxplus_{\pi+\rho-\pi\rho} R$. It is sufficient to prove the relation $R=\{(P\boxplus_{\pi}(Q\boxplus_{\rho} R), (P\boxplus_{\frac{\pi}{\pi+\rho-\pi\rho}}Q)\boxplus_{\pi+\rho-\pi\rho} R)\}\cup \textbf{Id}$ is a strongly probabilistic hhp-bisimulation, we omit it;
  \item $P\boxplus_{\pi}P\sim_{phhp}  P$. It is sufficient to prove the relation $R=\{(P\boxplus_{\pi}P, P)\}\cup \textbf{Id}$ is a strongly probabilistic hhp-bisimulation, we omit it;
  \item $P\boxplus_{\pi}\textbf{nil}\sim_{phhp}  P$. It is sufficient to prove the relation $R=\{(P\boxplus_{\pi}\textbf{nil}, P)\}\cup \textbf{Id}$ is a strongly probabilistic hhp-bisimulation, we omit it.
\end{enumerate}
\end{proof}

\begin{theorem}[Identity law for strongly probabilistic truly concurrent bisimilarities]
If $A(\widetilde{x})\overset{\text{def}}{=}P$, then

\begin{enumerate}
  \item $A(\widetilde{y})\sim_{pp}  P\{\widetilde{y}/\widetilde{x}\}$;
  \item $A(\widetilde{y})\sim_{ps}  P\{\widetilde{y}/\widetilde{x}\}$;
  \item $A(\widetilde{y})\sim_{php}  P\{\widetilde{y}/\widetilde{x}\}$;
  \item $A(\widetilde{y})\sim_{phhp}  P\{\widetilde{y}/\widetilde{x}\}$.
\end{enumerate}
\end{theorem}

\begin{proof}
\begin{enumerate}
  \item $A(\widetilde{y})\sim_{pp}  P\{\widetilde{y}/\widetilde{x}\}$. It is sufficient to prove the relation $R=\{(A(\widetilde{y}), P\{\widetilde{y}/\widetilde{x}\})\}\cup \textbf{Id}$ is a strongly probabilistic pomset bisimulation, we omit it;
  \item $A(\widetilde{y})\sim_{ps}  P\{\widetilde{y}/\widetilde{x}\}$. It is sufficient to prove the relation $R=\{(A(\widetilde{y}), P\{\widetilde{y}/\widetilde{x}\})\}\cup \textbf{Id}$ is a strongly probabilistic step bisimulation, we omit it;
  \item $A(\widetilde{y})\sim_{php}  P\{\widetilde{y}/\widetilde{x}\}$. It is sufficient to prove the relation $R=\{(A(\widetilde{y}), P\{\widetilde{y}/\widetilde{x}\})\}\cup \textbf{Id}$ is a strongly probabilistic hp-bisimulation, we omit it;
  \item $A(\widetilde{y})\sim_{phhp}  P\{\widetilde{y}/\widetilde{x}\}$. It is sufficient to prove the relation $R=\{(A(\widetilde{y}), P\{\widetilde{y}/\widetilde{x}\})\}\cup \textbf{Id}$ is a strongly probabilistic hhp-bisimulation, we omit it.
\end{enumerate}
\end{proof}

\begin{theorem}[Restriction Laws for strongly probabilistic pomset bisimilarity]
The restriction laws for strongly probabilistic pomset bisimilarity are as follows.

\begin{enumerate}
  \item $(y)P\sim_{pp}  P$, if $y\notin fn(P)$;
  \item $(y)(z)P\sim_{pp}  (z)(y)P$;
  \item $(y)(P+Q)\sim_{pp}  (y)P+(y)Q$;
  \item $(y)(P\boxplus_{\pi}Q)\sim_{pp}  (y)P\boxplus_{\pi}(y)Q$;
  \item $(y)\alpha.P\sim_{pp}  \alpha.(y)P$ if $y\notin n(\alpha)$;
  \item $(y)\alpha.P\sim_{pp}  \textbf{nil}$ if $y$ is the subject of $\alpha$.
\end{enumerate}
\end{theorem}

\begin{proof}
\begin{enumerate}
  \item $(y)P\sim_{pp}  P$, if $y\notin fn(P)$. It is sufficient to prove the relation $R=\{((y)P, P)\}\cup \textbf{Id}$, if $y\notin fn(P)$, is a strongly probabilistic pomset bisimulation, we omit it;
  \item $(y)(z)P\sim_{pp}  (z)(y)P$. It is sufficient to prove the relation $R=\{((y)(z)P, (z)(y)P)\}\cup \textbf{Id}$ is a strongly probabilistic pomset bisimulation, we omit it;
  \item $(y)(P+Q)\sim_{pp}  (y)P+(y)Q$. It is sufficient to prove the relation $R=\{((y)(P+Q), (y)P+(y)Q)\}\cup \textbf{Id}$ is a strongly probabilistic pomset bisimulation, we omit it;
  \item $(y)(P\boxplus_{\pi}Q)\sim_{pp}  (y)P\boxplus_{\pi}(y)Q$. It is sufficient to prove the relation $R=\{((y)(P\boxplus_{\pi}Q), (y)P\boxplus_{\pi}(y)Q)\}\cup \textbf{Id}$ is a strongly probabilistic pomset bisimulation, we omit it;
  \item $(y)\alpha.P\sim_{pp}  \alpha.(y)P$ if $y\notin n(\alpha)$. It is sufficient to prove the relation $R=\{((y)\alpha.P, \alpha.(y)P)\}\cup \textbf{Id}$, if $y\notin n(\alpha)$, is a strongly probabilistic pomset bisimulation, we omit it;
  \item $(y)\alpha.P\sim_{pp}  \textbf{nil}$ if $y$ is the subject of $\alpha$. It is sufficient to prove the relation $R=\{((y)\alpha.P, \textbf{nil})\}\cup \textbf{Id}$, if $y$ is the subject of $\alpha$, is a strongly probabilistic pomset bisimulation, we omit it.
\end{enumerate}
\end{proof}

\begin{theorem}[Restriction Laws for strongly probabilistic step bisimilarity]
The restriction laws for strongly probabilistic step bisimilarity are as follows.

\begin{enumerate}
  \item $(y)P\sim_{ps}  P$, if $y\notin fn(P)$;
  \item $(y)(z)P\sim_{ps}  (z)(y)P$;
  \item $(y)(P+Q)\sim_{ps}  (y)P+(y)Q$;
  \item $(y)(P\boxplus_{\pi}Q)\sim_{ps}  (y)P\boxplus_{\pi}(y)Q$;
  \item $(y)\alpha.P\sim_{ps}  \alpha.(y)P$ if $y\notin n(\alpha)$;
  \item $(y)\alpha.P\sim_{ps}  \textbf{nil}$ if $y$ is the subject of $\alpha$.
\end{enumerate}
\end{theorem}

\begin{proof}
\begin{enumerate}
  \item $(y)P\sim_{ps}  P$, if $y\notin fn(P)$. It is sufficient to prove the relation $R=\{((y)P, P)\}\cup \textbf{Id}$, if $y\notin fn(P)$, is a strongly probabilistic step bisimulation, we omit it;
  \item $(y)(z)P\sim_{ps}  (z)(y)P$. It is sufficient to prove the relation $R=\{((y)(z)P, (z)(y)P)\}\cup \textbf{Id}$ is a strongly probabilistic step bisimulation, we omit it;
  \item $(y)(P+Q)\sim_{ps}  (y)P+(y)Q$. It is sufficient to prove the relation $R=\{((y)(P+Q), (y)P+(y)Q)\}\cup \textbf{Id}$ is a strongly probabilistic step bisimulation, we omit it;
  \item $(y)(P\boxplus_{\pi}Q)\sim_{ps}  (y)P\boxplus_{\pi}(y)Q$. It is sufficient to prove the relation $R=\{((y)(P\boxplus_{\pi}Q), (y)P\boxplus_{\pi}(y)Q)\}\cup \textbf{Id}$ is a strongly probabilistic step bisimulation, we omit it;
  \item $(y)\alpha.P\sim_{ps}  \alpha.(y)P$ if $y\notin n(\alpha)$. It is sufficient to prove the relation $R=\{((y)\alpha.P, \alpha.(y)P)\}\cup \textbf{Id}$, if $y\notin n(\alpha)$, is a strongly probabilistic step bisimulation, we omit it;
  \item $(y)\alpha.P\sim_{ps}  \textbf{nil}$ if $y$ is the subject of $\alpha$. It is sufficient to prove the relation $R=\{((y)\alpha.P, \textbf{nil})\}\cup \textbf{Id}$, if $y$ is the subject of $\alpha$, is a strongly probabilistic step bisimulation, we omit it.
\end{enumerate}
\end{proof}

\begin{theorem}[Restriction Laws for strongly probabilistic hp-bisimilarity]
The restriction laws for strongly probabilistic hp-bisimilarity are as follows.

\begin{enumerate}
  \item $(y)P\sim_{php}  P$, if $y\notin fn(P)$;
  \item $(y)(z)P\sim_{php}  (z)(y)P$;
  \item $(y)(P+Q)\sim_{php}  (y)P+(y)Q$;
  \item $(y)(P\boxplus_{\pi}Q)\sim_{php}  (y)P\boxplus_{\pi}(y)Q$;
  \item $(y)\alpha.P\sim_{php}  \alpha.(y)P$ if $y\notin n(\alpha)$;
  \item $(y)\alpha.P\sim_{php}  \textbf{nil}$ if $y$ is the subject of $\alpha$.
\end{enumerate}
\end{theorem}

\begin{proof}
\begin{enumerate}
  \item $(y)P\sim_{php}  P$, if $y\notin fn(P)$. It is sufficient to prove the relation $R=\{((y)P, P)\}\cup \textbf{Id}$, if $y\notin fn(P)$, is a strongly probabilistic hp-bisimulation, we omit it;
  \item $(y)(z)P\sim_{php}  (z)(y)P$. It is sufficient to prove the relation $R=\{((y)(z)P, (z)(y)P)\}\cup \textbf{Id}$ is a strongly probabilistic hp-bisimulation, we omit it;
  \item $(y)(P+Q)\sim_{php}  (y)P+(y)Q$. It is sufficient to prove the relation $R=\{((y)(P+Q), (y)P+(y)Q)\}\cup \textbf{Id}$ is a strongly probabilistic hp-bisimulation, we omit it;
  \item $(y)(P\boxplus_{\pi}Q)\sim_{php}  (y)P\boxplus_{\pi}(y)Q$. It is sufficient to prove the relation $R=\{((y)(P\boxplus_{\pi}Q), (y)P\boxplus_{\pi}(y)Q)\}\cup \textbf{Id}$ is a strongly probabilistic hp-bisimulation, we omit it;
  \item $(y)\alpha.P\sim_{php}  \alpha.(y)P$ if $y\notin n(\alpha)$. It is sufficient to prove the relation $R=\{((y)\alpha.P, \alpha.(y)P)\}\cup \textbf{Id}$, if $y\notin n(\alpha)$, is a strongly probabilistic hp-bisimulation, we omit it;
  \item $(y)\alpha.P\sim_{php}  \textbf{nil}$ if $y$ is the subject of $\alpha$. It is sufficient to prove the relation $R=\{((y)\alpha.P, \textbf{nil})\}\cup \textbf{Id}$, if $y$ is the subject of $\alpha$, is a strongly probabilistic hp-bisimulation, we omit it.
\end{enumerate}
\end{proof}

\begin{theorem}[Restriction Laws for strongly probabilistic hhp-bisimilarity]
The restriction laws for strongly probabilistic hhp-bisimilarity are as follows.

\begin{enumerate}
  \item $(y)P\sim_{phhp}  P$, if $y\notin fn(P)$;
  \item $(y)(z)P\sim_{phhp}  (z)(y)P$;
  \item $(y)(P+Q)\sim_{phhp}  (y)P+(y)Q$;
  \item $(y)(P\boxplus_{\pi}Q)\sim_{phhp}  (y)P\boxplus_{\pi}(y)Q$;
  \item $(y)\alpha.P\sim_{phhp}  \alpha.(y)P$ if $y\notin n(\alpha)$;
  \item $(y)\alpha.P\sim_{phhp}  \textbf{nil}$ if $y$ is the subject of $\alpha$.
\end{enumerate}
\end{theorem}

\begin{proof}
\begin{enumerate}
  \item $(y)P\sim_{phhp}  P$, if $y\notin fn(P)$. It is sufficient to prove the relation $R=\{((y)P, P)\}\cup \textbf{Id}$, if $y\notin fn(P)$, is a strongly probabilistic hhp-bisimulation, we omit it;
  \item $(y)(z)P\sim_{phhp}  (z)(y)P$. It is sufficient to prove the relation $R=\{((y)(z)P, (z)(y)P)\}\cup \textbf{Id}$ is a strongly probabilistic hhp-bisimulation, we omit it;
  \item $(y)(P+Q)\sim_{phhp}  (y)P+(y)Q$. It is sufficient to prove the relation $R=\{((y)(P+Q), (y)P+(y)Q)\}\cup \textbf{Id}$ is a strongly probabilistic hhp-bisimulation, we omit it;
  \item $(y)(P\boxplus_{\pi}Q)\sim_{phhp}  (y)P\boxplus_{\pi}(y)Q$. It is sufficient to prove the relation $R=\{((y)(P\boxplus_{\pi}Q), (y)P\boxplus_{\pi}(y)Q)\}\cup \textbf{Id}$ is a strongly probabilistic hhp-bisimulation, we omit it;
  \item $(y)\alpha.P\sim_{phhp}  \alpha.(y)P$ if $y\notin n(\alpha)$. It is sufficient to prove the relation $R=\{((y)\alpha.P, \alpha.(y)P)\}\cup \textbf{Id}$, if $y\notin n(\alpha)$, is a strongly probabilistic hhp-bisimulation, we omit it;
  \item $(y)\alpha.P\sim_{phhp}  \textbf{nil}$ if $y$ is the subject of $\alpha$. It is sufficient to prove the relation $R=\{((y)\alpha.P, \textbf{nil})\}\cup \textbf{Id}$, if $y$ is the subject of $\alpha$, is a strongly probabilistic hhp-bisimulation, we omit it.
\end{enumerate}
\end{proof}

\begin{theorem}[Parallel laws for strongly probabilistic pomset bisimilarity]
The parallel laws for strongly probabilistic pomset bisimilarity are as follows.

\begin{enumerate}
  \item $P\parallel \textbf{nil}\sim_{pp}  P$;
  \item $P_1\parallel P_2\sim_{pp}  P_2\parallel P_1$;
  \item $(P_1\parallel P_2)\parallel P_3\sim_{pp}  P_1\parallel (P_2\parallel P_3)$;
  \item $(y)(P_1\parallel P_2)\sim_{pp}  (y)P_1\parallel (y)P_2$, if $y\notin fn(P_1)\cap fn(P_2)$.
\end{enumerate}
\end{theorem}

\begin{proof}
\begin{enumerate}
  \item $P\parallel \textbf{nil}\sim_{pp}  P$. It is sufficient to prove the relation $R=\{(P\parallel \textbf{nil}, P)\}\cup \textbf{Id}$ is a strongly probabilistic pomset bisimulation, we omit it;
  \item $P_1\parallel P_2\sim_{pp}  P_2\parallel P_1$. It is sufficient to prove the relation $R=\{(P_1\parallel P_2, P_2\parallel P_1)\}\cup \textbf{Id}$ is a strongly probabilistic pomset bisimulation, we omit it;
  \item $(P_1\parallel P_2)\parallel P_3\sim_{pp}  P_1\parallel (P_2\parallel P_3)$. It is sufficient to prove the relation $R=\{((P_1\parallel P_2)\parallel P_3, P_1\parallel (P_2\parallel P_3))\}\cup \textbf{Id}$ is a strongly probabilistic pomset bisimulation, we omit it;
  \item $(y)(P_1\parallel P_2)\sim_{pp}  (y)P_1\parallel (y)P_2$, if $y\notin fn(P_1)\cap fn(P_2)$. It is sufficient to prove the relation $R=\{((y)(P_1\parallel P_2), (y)P_1\parallel (y)P_2)\}\cup \textbf{Id}$, if $y\notin fn(P_1)\cap fn(P_2)$, is a strongly probabilistic pomset bisimulation, we omit it.
\end{enumerate}
\end{proof}

\begin{theorem}[Parallel laws for strongly probabilistic step bisimilarity]
The parallel laws for strongly probabilistic step bisimilarity are as follows.

\begin{enumerate}
  \item $P\parallel \textbf{nil}\sim_{ps}  P$;
  \item $P_1\parallel P_2\sim_{ps}  P_2\parallel P_1$;
  \item $(P_1\parallel P_2)\parallel P_3\sim_{ps}  P_1\parallel (P_2\parallel P_3)$;
  \item $(y)(P_1\parallel P_2)\sim_{ps}  (y)P_1\parallel (y)P_2$, if $y\notin fn(P_1)\cap fn(P_2)$.
\end{enumerate}
\end{theorem}

\begin{proof}
\begin{enumerate}
  \item $P\parallel \textbf{nil}\sim_{ps}  P$. It is sufficient to prove the relation $R=\{(P\parallel \textbf{nil}, P)\}\cup \textbf{Id}$ is a strongly probabilistic step bisimulation, we omit it;
  \item $P_1\parallel P_2\sim_{ps}  P_2\parallel P_1$. It is sufficient to prove the relation $R=\{(P_1\parallel P_2, P_2\parallel P_1)\}\cup \textbf{Id}$ is a strongly probabilistic step bisimulation, we omit it;
  \item $(P_1\parallel P_2)\parallel P_3\sim_{ps}  P_1\parallel (P_2\parallel P_3)$. It is sufficient to prove the relation $R=\{((P_1\parallel P_2)\parallel P_3, P_1\parallel (P_2\parallel P_3))\}\cup \textbf{Id}$ is a strongly probabilistic step bisimulation, we omit it;
  \item $(y)(P_1\parallel P_2)\sim_{ps}  (y)P_1\parallel (y)P_2$, if $y\notin fn(P_1)\cap fn(P_2)$. It is sufficient to prove the relation $R=\{((y)(P_1\parallel P_2), (y)P_1\parallel (y)P_2)\}\cup \textbf{Id}$, if $y\notin fn(P_1)\cap fn(P_2)$, is a strongly probabilistic step bisimulation, we omit it.
\end{enumerate}
\end{proof}

\begin{theorem}[Parallel laws for strongly probabilistic hp-bisimilarity]
The parallel laws for strongly probabilistic hp-bisimilarity are as follows.

\begin{enumerate}
  \item $P\parallel \textbf{nil}\sim_{php}  P$;
  \item $P_1\parallel P_2\sim_{php}  P_2\parallel P_1$;
  \item $(P_1\parallel P_2)\parallel P_3\sim_{php}  P_1\parallel (P_2\parallel P_3)$;
  \item $(y)(P_1\parallel P_2)\sim_{php}  (y)P_1\parallel (y)P_2$, if $y\notin fn(P_1)\cap fn(P_2)$.
\end{enumerate}
\end{theorem}

\begin{proof}
\begin{enumerate}
  \item $P\parallel \textbf{nil}\sim_{php}  P$. It is sufficient to prove the relation $R=\{(P\parallel \textbf{nil}, P)\}\cup \textbf{Id}$ is a strongly probabilistic hp-bisimulation, we omit it;
  \item $P_1\parallel P_2\sim_{php}  P_2\parallel P_1$. It is sufficient to prove the relation $R=\{(P_1\parallel P_2, P_2\parallel P_1)\}\cup \textbf{Id}$ is a strongly probabilistic hp-bisimulation, we omit it;
  \item $(P_1\parallel P_2)\parallel P_3\sim_{php}  P_1\parallel (P_2\parallel P_3)$. It is sufficient to prove the relation $R=\{((P_1\parallel P_2)\parallel P_3, P_1\parallel (P_2\parallel P_3))\}\cup \textbf{Id}$ is a strongly probabilistic hp-bisimulation, we omit it;
  \item $(y)(P_1\parallel P_2)\sim_{php}  (y)P_1\parallel (y)P_2$, if $y\notin fn(P_1)\cap fn(P_2)$. It is sufficient to prove the relation $R=\{((y)(P_1\parallel P_2), (y)P_1\parallel (y)P_2)\}\cup \textbf{Id}$, if $y\notin fn(P_1)\cap fn(P_2)$, is a strongly probabilistic hp-bisimulation, we omit it.
\end{enumerate}
\end{proof}

\begin{theorem}[Parallel laws for strongly probabilistic hhp-bisimilarity]
The parallel laws for strongly probabilistic hhp-bisimilarity are as follows.

\begin{enumerate}
  \item $P\parallel \textbf{nil}\sim_{phhp}  P$;
  \item $P_1\parallel P_2\sim_{phhp}  P_2\parallel P_1$;
  \item $(P_1\parallel P_2)\parallel P_3\sim_{phhp}  P_1\parallel (P_2\parallel P_3)$;
  \item $(y)(P_1\parallel P_2)\sim_{phhp}  (y)P_1\parallel (y)P_2$, if $y\notin fn(P_1)\cap fn(P_2)$.
\end{enumerate}
\end{theorem}

\begin{proof}
\begin{enumerate}
  \item $P\parallel \textbf{nil}\sim_{phhp}  P$. It is sufficient to prove the relation $R=\{(P\parallel \textbf{nil}, P)\}\cup \textbf{Id}$ is a strongly probabilistic hhp-bisimulation, we omit it;
  \item $P_1\parallel P_2\sim_{phhp}  P_2\parallel P_1$. It is sufficient to prove the relation $R=\{(P_1\parallel P_2, P_2\parallel P_1)\}\cup \textbf{Id}$ is a strongly probabilistic hhp-bisimulation, we omit it;
  \item $(P_1\parallel P_2)\parallel P_3\sim_{phhp}  P_1\parallel (P_2\parallel P_3)$. It is sufficient to prove the relation $R=\{((P_1\parallel P_2)\parallel P_3, P_1\parallel (P_2\parallel P_3))\}\cup \textbf{Id}$ is a strongly probabilistic hhp-bisimulation, we omit it;
  \item $(y)(P_1\parallel P_2)\sim_{phhp}  (y)P_1\parallel (y)P_2$, if $y\notin fn(P_1)\cap fn(P_2)$. It is sufficient to prove the relation $R=\{((y)(P_1\parallel P_2), (y)P_1\parallel (y)P_2)\}\cup \textbf{Id}$, if $y\notin fn(P_1)\cap fn(P_2)$, is a strongly probabilistic hhp-bisimulation, we omit it.
\end{enumerate}
\end{proof}

\begin{theorem}[Expansion law for truly concurrent bisimilarities]
Let $P\equiv\boxplus_i\sum_j \alpha_{ij}.P_{ij}$ and $Q\equiv\boxplus_{k}\sum_l\beta_{kl}.Q_{kl}$, where $bn(\alpha_{ij})\cap fn(Q)=\emptyset$ for all $i,j$, and
  $bn(\beta_{kl})\cap fn(P)=\emptyset$ for all $k,l$. Then,

\begin{enumerate}
  \item $P\parallel Q\sim_{pp}  \boxplus_{i}\boxplus_{k}\sum_j\sum_l (\alpha_{ij}\parallel \beta_{kl}).(P_{ij}\parallel Q_{kl})+\boxplus_i\boxplus_k\sum_{\alpha_{ij} \textrm{ comp }\beta_{kl}}\tau.R_{ijkl}$;
  \item $P\parallel Q\sim_{ps}  \boxplus_{i}\boxplus_{k}\sum_j\sum_l (\alpha_{ij}\parallel \beta_{kl}).(P_{ij}\parallel Q_{kl})+\boxplus_i\boxplus_k\sum_{\alpha_{ij} \textrm{ comp }\beta_{kl}}\tau.R_{ijkl}$;
  \item $P\parallel Q\sim_{php}  \boxplus_{i}\boxplus_{k}\sum_j\sum_l (\alpha_{ij}\parallel \beta_{kl}).(P_{ij}\parallel Q_{kl})+\boxplus_i\boxplus_k\sum_{\alpha_{ij} \textrm{ comp }\beta_{kl}}\tau.R_{ijkl}$;
  \item $P\parallel Q\nsim_{phhp} \boxplus_{i}\boxplus_{k}\sum_j\sum_l (\alpha_{ij}\parallel \beta_{kl}).(P_{ij}\parallel Q_{kl})+\boxplus_i\boxplus_k\sum_{\alpha_{ij} \textrm{ comp }\beta_{kl}}\tau.R_{ijkl}$.
\end{enumerate}

Where $\alpha_{ij}$ comp $\beta_{kl}$ and $R_{ijkl}$ are defined as follows:
\begin{enumerate}
  \item $\alpha_{ij}$ is $\overline{x}u$ and $\beta_{kl}$ is $x(v)$, then $R_{ijkl}=P_{ij}\parallel Q_{kl}\{u/v\}$;
  \item $\alpha_{ij}$ is $\overline{x}(u)$ and $\beta_{kl}$ is $x(v)$, then $R_{ijkl}=(w)(P_{ij}\{w/u\}\parallel Q_{kl}\{w/v\})$, if $w\notin fn((u)P_{ij})\cup fn((v)Q_{kl})$;
  \item $\alpha_{ij}$ is $x(v)$ and $\beta_{kl}$ is $\overline{x}u$, then $R_{ijkl}=P_{ij}\{u/v\}\parallel Q_{kl}$;
  \item $\alpha_{ij}$ is $x(v)$ and $\beta_{kl}$ is $\overline{x}(u)$, then $R_{ijkl}=(w)(P_{ij}\{w/v\}\parallel Q_{kl}\{w/u\})$, if $w\notin fn((v)P_{ij})\cup fn((u)Q_{kl})$.
\end{enumerate}
\end{theorem}

\begin{proof}
According to the definition of strongly probabilistic truly concurrent bisimulations, we can easily prove the above equations, and we omit the proof.
\end{proof}

\begin{theorem}[Equivalence and congruence for strongly probabilistic pomset bisimilarity]
We can enjoy the full congruence modulo strongly probabilistic pomset bisimilarity.

\begin{enumerate}
  \item $\sim_{pp} $ is an equivalence relation;
  \item If $P\sim_{pp}  Q$ then
  \begin{enumerate}
    \item $\alpha.P\sim_{pp}  \alpha.Q$, $\alpha$ is a free action;
    \item $\phi.P\sim_{pp}  \phi.Q$;
    \item $P+R\sim_{pp}  Q+R$;
    \item $P\boxplus_{\pi} R\sim_{pp} Q\boxplus_{\pi}R$;
    \item $P\parallel R\sim_{pp}  Q\parallel R$;
    \item $(w)P\sim_{pp}  (w)Q$;
    \item $x(y).P\sim_{pp}  x(y).Q$.
  \end{enumerate}
\end{enumerate}
\end{theorem}

\begin{proof}
\begin{enumerate}
  \item $\sim_{pp} $ is an equivalence relation, it is obvious;
  \item If $P\sim_{pp}  Q$ then
  \begin{enumerate}
    \item $\alpha.P\sim_{pp}  \alpha.Q$, $\alpha$ is a free action. It is sufficient to prove the relation $R=\{(\alpha.P, \alpha.Q)\}\cup \textbf{Id}$ is a strongly probabilistic pomset bisimulation, we omit it;
    \item $\phi.P\sim_{pp}  \phi.Q$. It is sufficient to prove the relation $R=\{(\phi.P, \phi.Q)\}\cup \textbf{Id}$ is a strongly probabilistic pomset bisimulation, we omit it;
    \item $P+R\sim_{pp}  Q+R$. It is sufficient to prove the relation $R=\{(P+R, Q+R)\}\cup \textbf{Id}$ is a strongly probabilistic pomset bisimulation, we omit it;
    \item $P\boxplus_{\pi} R\sim_{pp} Q\boxplus_{\pi}R$. It is sufficient to prove the relation $R=\{(P\boxplus_{\pi} R, Q\boxplus_{\pi} R)\}\cup \textbf{Id}$ is a strongly probabilistic pomset bisimulation, we omit it;
    \item $P\parallel R\sim_{pp}  Q\parallel R$. It is sufficient to prove the relation $R=\{(P\parallel R, Q\parallel R)\}\cup \textbf{Id}$ is a strongly probabilistic pomset bisimulation, we omit it;
    \item $(w)P\sim_{pp}  (w)Q$. It is sufficient to prove the relation $R=\{((w)P, (w)Q)\}\cup \textbf{Id}$ is a strongly probabilistic pomset bisimulation, we omit it;
    \item $x(y).P\sim_{pp}  x(y).Q$. It is sufficient to prove the relation $R=\{(x(y).P, x(y).Q)\}\cup \textbf{Id}$ is a strongly probabilistic pomset bisimulation, we omit it.
  \end{enumerate}
\end{enumerate}
\end{proof}

\begin{theorem}[Equivalence and congruence for strongly probabilistic step bisimilarity]
We can enjoy the full congruence modulo strongly probabilistic step bisimilarity.

\begin{enumerate}
  \item $\sim_{ps} $ is an equivalence relation;
  \item If $P\sim_{ps}  Q$ then
  \begin{enumerate}
    \item $\alpha.P\sim_{ps}  \alpha.Q$, $\alpha$ is a free action;
    \item $\phi.P\sim_{ps}  \phi.Q$;
    \item $P+R\sim_{ps}  Q+R$;
    \item $P\boxplus_{\pi} R\sim_{ps} Q\boxplus_{\pi}R$;
    \item $P\parallel R\sim_{ps}  Q\parallel R$;
    \item $(w)P\sim_{ps}  (w)Q$;
    \item $x(y).P\sim_{ps}  x(y).Q$.
  \end{enumerate}
\end{enumerate}
\end{theorem}

\begin{proof}
\begin{enumerate}
  \item $\sim_{ps} $ is an equivalence relation, it is obvious;
  \item If $P\sim_{ps}  Q$ then
  \begin{enumerate}
    \item $\alpha.P\sim_{ps}  \alpha.Q$, $\alpha$ is a free action. It is sufficient to prove the relation $R=\{(\alpha.P, \alpha.Q)\}\cup \textbf{Id}$ is a strongly probabilistic step bisimulation, we omit it;
    \item $\phi.P\sim_{ps}  \phi.Q$. It is sufficient to prove the relation $R=\{(\phi.P, \phi.Q)\}\cup \textbf{Id}$ is a strongly probabilistic step bisimulation, we omit it;
    \item $P+R\sim_{ps}  Q+R$. It is sufficient to prove the relation $R=\{(P+R, Q+R)\}\cup \textbf{Id}$ is a strongly probabilistic step bisimulation, we omit it;
    \item $P\boxplus_{\pi} R\sim_{ps} Q\boxplus_{\pi}R$. It is sufficient to prove the relation $R=\{(P\boxplus_{\pi} R, Q\boxplus_{\pi} R)\}\cup \textbf{Id}$ is a strongly probabilistic step bisimulation, we omit it;
    \item $P\parallel R\sim_{ps}  Q\parallel R$. It is sufficient to prove the relation $R=\{(P\parallel R, Q\parallel R)\}\cup \textbf{Id}$ is a strongly probabilistic step bisimulation, we omit it;
    \item $(w)P\sim_{ps}  (w)Q$. It is sufficient to prove the relation $R=\{((w)P, (w)Q)\}\cup \textbf{Id}$ is a strongly probabilistic step bisimulation, we omit it;
    \item $x(y).P\sim_{ps}  x(y).Q$. It is sufficient to prove the relation $R=\{(x(y).P, x(y).Q)\}\cup \textbf{Id}$ is a strongly probabilistic step bisimulation, we omit it.
  \end{enumerate}
\end{enumerate}
\end{proof}

\begin{theorem}[Equivalence and congruence for strongly probabilistic hp-bisimilarity]
We can enjoy the full congruence modulo strongly probabilistic hp-bisimilarity.

\begin{enumerate}
  \item $\sim_{php} $ is an equivalence relation;
  \item If $P\sim_{php}  Q$ then
  \begin{enumerate}
    \item $\alpha.P\sim_{php}  \alpha.Q$, $\alpha$ is a free action;
    \item $\phi.P\sim_{php}  \phi.Q$;
    \item $P+R\sim_{php}  Q+R$;
    \item $P\boxplus_{\pi} R\sim_{php} Q\boxplus_{\pi}R$;
    \item $P\parallel R\sim_{php}  Q\parallel R$;
    \item $(w)P\sim_{php}  (w)Q$;
    \item $x(y).P\sim_{php}  x(y).Q$.
  \end{enumerate}
\end{enumerate}
\end{theorem}

\begin{proof}
\begin{enumerate}
  \item $\sim_{php} $ is an equivalence relation, it is obvious;
  \item If $P\sim_{php}  Q$ then
  \begin{enumerate}
    \item $\alpha.P\sim_{php}  \alpha.Q$, $\alpha$ is a free action. It is sufficient to prove the relation $R=\{(\alpha.P, \alpha.Q)\}\cup \textbf{Id}$ is a strongly probabilistic hp-bisimulation, we omit it;
    \item $\phi.P\sim_{php}  \phi.Q$. It is sufficient to prove the relation $R=\{(\phi.P, \phi.Q)\}\cup \textbf{Id}$ is a strongly probabilistic hp-bisimulation, we omit it;
    \item $P+R\sim_{php}  Q+R$. It is sufficient to prove the relation $R=\{(P+R, Q+R)\}\cup \textbf{Id}$ is a strongly probabilistic hp-bisimulation, we omit it;
    \item $P\boxplus_{\pi} R\sim_{php} Q\boxplus_{\pi}R$. It is sufficient to prove the relation $R=\{(P\boxplus_{\pi} R, Q\boxplus_{\pi} R)\}\cup \textbf{Id}$ is a strongly probabilistic hp-bisimulation, we omit it;
    \item $P\parallel R\sim_{php}  Q\parallel R$. It is sufficient to prove the relation $R=\{(P\parallel R, Q\parallel R)\}\cup \textbf{Id}$ is a strongly probabilistic hp-bisimulation, we omit it;
    \item $(w)P\sim_{php}  (w)Q$. It is sufficient to prove the relation $R=\{((w)P, (w)Q)\}\cup \textbf{Id}$ is a strongly probabilistic hp-bisimulation, we omit it;
    \item $x(y).P\sim_{php}  x(y).Q$. It is sufficient to prove the relation $R=\{(x(y).P, x(y).Q)\}\cup \textbf{Id}$ is a strongly probabilistic hp-bisimulation, we omit it.
  \end{enumerate}
\end{enumerate}
\end{proof}

\begin{theorem}[Equivalence and congruence for strongly probabilistic hhp-bisimilarity]
We can enjoy the full congruence modulo strongly probabilistic hhp-bisimilarity.

\begin{enumerate}
  \item $\sim_{phhp} $ is an equivalence relation;
  \item If $P\sim_{phhp}  Q$ then
  \begin{enumerate}
    \item $\alpha.P\sim_{phhp}  \alpha.Q$, $\alpha$ is a free action;
    \item $\phi.P\sim_{phhp}  \phi.Q$;
    \item $P+R\sim_{phhp}  Q+R$;
    \item $P\boxplus_{\pi} R\sim_{phhp} Q\boxplus_{\pi}R$;
    \item $P\parallel R\sim_{phhp}  Q\parallel R$;
    \item $(w)P\sim_{phhp}  (w)Q$;
    \item $x(y).P\sim_{phhp}  x(y).Q$.
  \end{enumerate}
\end{enumerate}
\end{theorem}

\begin{proof}
\begin{enumerate}
  \item $\sim_{phhp} $ is an equivalence relation, it is obvious;
  \item If $P\sim_{phhp}  Q$ then
  \begin{enumerate}
    \item $\alpha.P\sim_{phhp}  \alpha.Q$, $\alpha$ is a free action. It is sufficient to prove the relation $R=\{(\alpha.P, \alpha.Q)\}\cup \textbf{Id}$ is a strongly probabilistic hhp-bisimulation, we omit it;
    \item $\phi.P\sim_{phhp}  \phi.Q$. It is sufficient to prove the relation $R=\{(\phi.P, \phi.Q)\}\cup \textbf{Id}$ is a strongly probabilistic hhp-bisimulation, we omit it;
    \item $P+R\sim_{phhp}  Q+R$. It is sufficient to prove the relation $R=\{(P+R, Q+R)\}\cup \textbf{Id}$ is a strongly probabilistic hhp-bisimulation, we omit it;
    \item $P\boxplus_{\pi} R\sim_{phhp} Q\boxplus_{\pi}R$. It is sufficient to prove the relation $R=\{(P\boxplus_{\pi} R, Q\boxplus_{\pi} R)\}\cup \textbf{Id}$ is a strongly probabilistic hhp-bisimulation, we omit it;
    \item $P\parallel R\sim_{phhp}  Q\parallel R$. It is sufficient to prove the relation $R=\{(P\parallel R, Q\parallel R)\}\cup \textbf{Id}$ is a strongly probabilistic hhp-bisimulation, we omit it;
    \item $(w)P\sim_{phhp}  (w)Q$. It is sufficient to prove the relation $R=\{((w)P, (w)Q)\}\cup \textbf{Id}$ is a strongly probabilistic hhp-bisimulation, we omit it;
    \item $x(y).P\sim_{phhp}  x(y).Q$. It is sufficient to prove the relation $R=\{(x(y).P, x(y).Q)\}\cup \textbf{Id}$ is a strongly probabilistic hhp-bisimulation, we omit it.
  \end{enumerate}
\end{enumerate}
\end{proof}

\subsubsection{Recursion}

\begin{definition}
Let $X$ have arity $n$, and let $\widetilde{x}=x_1,\cdots,x_n$ be distinct names, and $fn(P)\subseteq\{x_1,\cdots,x_n\}$. The replacement of $X(\widetilde{x})$ by $P$ in $E$, written
$E\{X(\widetilde{x}):=P\}$, means the result of replacing each subterm $X(\widetilde{y})$ in $E$ by $P\{\widetilde{y}/\widetilde{x}\}$.
\end{definition}

\begin{definition}
Let $E$ and $F$ be two process expressions containing only $X_1,\cdots,X_m$ with associated name sequences $\widetilde{x}_1,\cdots,\widetilde{x}_m$. Then,
\begin{enumerate}
  \item $E\sim_{pp}  F$ means $E(\widetilde{P})\sim_{pp}  F(\widetilde{P})$;
  \item $E\sim_{ps}  F$ means $E(\widetilde{P})\sim_{ps}  F(\widetilde{P})$;
  \item $E\sim_{php}  F$ means $E(\widetilde{P})\sim_{php}  F(\widetilde{P})$;
  \item $E\sim_{phhp}  F$ means $E(\widetilde{P})\sim_{phhp}  F(\widetilde{P})$;
\end{enumerate}

for all $\widetilde{P}$ such that $fn(P_i)\subseteq \widetilde{x}_i$ for each $i$.
\end{definition}

\begin{definition}
A term or identifier is weakly guarded in $P$ if it lies within some subterm $\alpha.Q$ or $(\alpha_1\parallel\cdots\parallel \alpha_n).Q$ of $P$.
\end{definition}

\begin{theorem}
Assume that $\widetilde{E}$ and $\widetilde{F}$ are expressions containing only $X_i$ with $\widetilde{x}_i$, and $\widetilde{A}$ and $\widetilde{B}$ are identifiers with $A_i$, $B_i$. Then, for all $i$,
\begin{enumerate}
  \item $E_i\sim_{ps}  F_i$, $A_i(\widetilde{x}_i)\overset{\text{def}}{=}E_i(\widetilde{A})$, $B_i(\widetilde{x}_i)\overset{\text{def}}{=}F_i(\widetilde{B})$, then
  $A_i(\widetilde{x}_i)\sim_{ps}  B_i(\widetilde{x}_i)$;
  \item $E_i\sim_{pp}  F_i$, $A_i(\widetilde{x}_i)\overset{\text{def}}{=}E_i(\widetilde{A})$, $B_i(\widetilde{x}_i)\overset{\text{def}}{=}F_i(\widetilde{B})$, then
  $A_i(\widetilde{x}_i)\sim_{pp}  B_i(\widetilde{x}_i)$;
  \item $E_i\sim_{php}  F_i$, $A_i(\widetilde{x}_i)\overset{\text{def}}{=}E_i(\widetilde{A})$, $B_i(\widetilde{x}_i)\overset{\text{def}}{=}F_i(\widetilde{B})$, then
  $A_i(\widetilde{x}_i)\sim_{php}  B_i(\widetilde{x}_i)$;
  \item $E_i\sim_{phhp}  F_i$, $A_i(\widetilde{x}_i)\overset{\text{def}}{=}E_i(\widetilde{A})$, $B_i(\widetilde{x}_i)\overset{\text{def}}{=}F_i(\widetilde{B})$, then
  $A_i(\widetilde{x}_i)\sim_{phhp}  B_i(\widetilde{x}_i)$.
\end{enumerate}
\end{theorem}

\begin{proof}
\begin{enumerate}
  \item $E_i\sim_{ps}  F_i$, $A_i(\widetilde{x}_i)\overset{\text{def}}{=}E_i(\widetilde{A})$, $B_i(\widetilde{x}_i)\overset{\text{def}}{=}F_i(\widetilde{B})$, then
  $A_i(\widetilde{x}_i)\sim_{ps}  B_i(\widetilde{x}_i)$.

      We will consider the case $I=\{1\}$ with loss of generality, and show the following relation $R$ is a strongly probabilistic step bisimulation.

      $$R=\{(G(A),G(B)):G\textrm{ has only identifier }X\}.$$

      By choosing $G\equiv X(\widetilde{y})$, it follows that $A(\widetilde{y})\sim_{ps}  B(\widetilde{y})$. It is sufficient to prove the following:
      \begin{enumerate}
        \item If $\langle G(A),s\rangle\rightsquigarrow\xrightarrow{\{\alpha_1,\cdots,\alpha_n\}}\langle P',s'\rangle$, where $\alpha_i(1\leq i\leq n)$ is a free action or bound output action with
        $bn(\alpha_1)\cap\cdots\cap bn(\alpha_n)\cap n(G(A),G(B))=\emptyset$, then $\langle G(B),s\rangle\rightsquigarrow\xrightarrow{\{\alpha_1,\cdots,\alpha_n\}}\langle Q'',s''\rangle$ such that $P'\sim_{ps}  Q''$;
        \item If $\langle G(A),s\rangle\rightsquigarrow\xrightarrow{x(y)}\langle P',s'\rangle$ with $x\notin n(G(A),G(B))$, then $\langle G(B),s\rangle\rightsquigarrow\xrightarrow{x(y)}\langle Q'',s''\rangle$, such that for all $u$,
        $\langle P',s'\rangle\{u/y\}\sim_{ps}  \langle Q''\{u/y\},s''\rangle$.
      \end{enumerate}

      To prove the above properties, it is sufficient to induct on the depth of inference and quite routine, we omit it.
  \item $E_i\sim_{pp}  F_i$, $A_i(\widetilde{x}_i)\overset{\text{def}}{=}E_i(\widetilde{A})$, $B_i(\widetilde{x}_i)\overset{\text{def}}{=}F_i(\widetilde{B})$, then
  $A_i(\widetilde{x}_i)\sim_{pp}  B_i(\widetilde{x}_i)$. It can be proven similarly to the above case.
  \item $E_i\sim_{php}  F_i$, $A_i(\widetilde{x}_i)\overset{\text{def}}{=}E_i(\widetilde{A})$, $B_i(\widetilde{x}_i)\overset{\text{def}}{=}F_i(\widetilde{B})$, then
  $A_i(\widetilde{x}_i)\sim_{php}  B_i(\widetilde{x}_i)$. It can be proven similarly to the above case.
  \item $E_i\sim_{phhp}  F_i$, $A_i(\widetilde{x}_i)\overset{\text{def}}{=}E_i(\widetilde{A})$, $B_i(\widetilde{x}_i)\overset{\text{def}}{=}F_i(\widetilde{B})$, then
  $A_i(\widetilde{x}_i)\sim_{phhp}  B_i(\widetilde{x}_i)$. It can be proven similarly to the above case.
\end{enumerate}
\end{proof}

\begin{theorem}[Unique solution of equations]
Assume $\widetilde{E}$ are expressions containing only $X_i$ with $\widetilde{x}_i$, and each $X_i$ is weakly guarded in each $E_j$. Assume that $\widetilde{P}$ and $\widetilde{Q}$ are
processes such that $fn(P_i)\subseteq \widetilde{x}_i$ and $fn(Q_i)\subseteq \widetilde{x}_i$. Then, for all $i$,
\begin{enumerate}
  \item if $P_i\sim_{pp}  E_i(\widetilde{P})$, $Q_i\sim_{pp}  E_i(\widetilde{Q})$, then $P_i\sim_{pp}  Q_i$;
  \item if $P_i\sim_{ps}  E_i(\widetilde{P})$, $Q_i\sim_{ps}  E_i(\widetilde{Q})$, then $P_i\sim_{ps}  Q_i$;
  \item if $P_i\sim_{php}  E_i(\widetilde{P})$, $Q_i\sim_{php}  E_i(\widetilde{Q})$, then $P_i\sim_{php}  Q_i$;
  \item if $P_i\sim_{phhp}  E_i(\widetilde{P})$, $Q_i\sim_{phhp}  E_i(\widetilde{Q})$, then $P_i\sim_{phhp}  Q_i$.
\end{enumerate}
\end{theorem}

\begin{proof}
\begin{enumerate}
  \item It is similar to the proof of unique solution of equations for strongly probabilistic pomset bisimulation in CTC, please refer to \cite{CTC2} for details, we omit it;
  \item It is similar to the proof of unique solution of equations for strongly probabilistic step bisimulation in CTC, please refer to \cite{CTC2} for details, we omit it;
  \item It is similar to the proof of unique solution of equations for strongly probabilistic hp-bisimulation in CTC, please refer to \cite{CTC2} for details, we omit it;
  \item It is similar to the proof of unique solution of equations for strongly probabilistic hhp-bisimulation in CTC, please refer to \cite{CTC2} for details, we omit it.
\end{enumerate}
\end{proof}

\subsection{Algebraic Theory}\label{a6}

\begin{definition}[STC]
The theory \textbf{STC} is consisted of the following axioms and inference rules:

\begin{enumerate}
  \item Alpha-conversion $\textbf{A}$.
  \[\textrm{if } P\equiv Q, \textrm{ then } P=Q\]
  \item Congruence $\textbf{C}$. If $P=Q$, then,
  \[\tau.P=\tau.Q\quad \overline{x}y.P=\overline{x}y.Q\]
  \[P+R=Q+R\quad P\parallel R=Q\parallel R\]
  \[(x)P=(x)Q\quad x(y).P=x(y).Q\]
  \item Summation $\textbf{S}$.
  \[\textbf{S0}\quad P+\textbf{nil}=P\]
  \[\textbf{S1}\quad P+P=P\]
  \[\textbf{S2}\quad P+Q=Q+P\]
  \[\textbf{S3}\quad P+(Q+R)=(P+Q)+R\]
  \item Box-Summation $\textbf(BS)$.
  \[\textbf{BS0}\quad P\boxplus_{\pi}\textbf{nil}= P\]
  \[\textbf{BS1}\quad P\boxplus_{\pi}P= P\]
  \[\textbf{BS2}\quad P\boxplus_{\pi} Q= Q\boxplus_{1-\pi} P\]
  \[\textbf{BS3}\quad P\boxplus_{\pi}(Q\boxplus_{\rho} R)= (P\boxplus_{\frac{\pi}{\pi+\rho-\pi\rho}}Q)\boxplus_{\pi+\rho-\pi\rho} R\]
  \item Restriction $\textbf{R}$.
  \[\textbf{R0}\quad (x)P=P\quad \textrm{ if }x\notin fn(P)\]
  \[\textbf{R1}\quad (x)(y)P=(y)(x)P\]
  \[\textbf{R2}\quad (x)(P+Q)=(x)P+(x)Q\]
  \[\textbf{R3}\quad (x)\alpha.P=\alpha.(x)P\quad \textrm{ if }x\notin n(\alpha)\]
  \[\textbf{R4}\quad (x)\alpha.P=\textbf{nil}\quad \textrm{ if }x\textrm{is the subject of }\alpha\]
  \item Expansion $\textbf{E}$.
  Let $P\equiv\boxplus_i\sum_j \alpha_{ij}.P_{ij}$ and $Q\equiv\boxplus_{k}\sum_l\beta_{kl}.Q_{kl}$, where $bn(\alpha_{ij})\cap fn(Q)=\emptyset$ for all $i,j$, and
  $bn(\beta_{kl})\cap fn(P)=\emptyset$ for all $k,l$. Then,

\begin{enumerate}
  \item $P\parallel Q\sim_{pp}  \boxplus_{i}\boxplus_{k}\sum_j\sum_l (\alpha_{ij}\parallel \beta_{kl}).(P_{ij}\parallel Q_{kl})+\boxplus_i\boxplus_k\sum_{\alpha_{ij} \textrm{ comp }\beta_{kl}}\tau.R_{ijkl}$;
  \item $P\parallel Q\sim_{ps}  \boxplus_{i}\boxplus_{k}\sum_j\sum_l (\alpha_{ij}\parallel \beta_{kl}).(P_{ij}\parallel Q_{kl})+\boxplus_i\boxplus_k\sum_{\alpha_{ij} \textrm{ comp }\beta_{kl}}\tau.R_{ijkl}$;
  \item $P\parallel Q\sim_{php}  \boxplus_{i}\boxplus_{k}\sum_j\sum_l (\alpha_{ij}\parallel \beta_{kl}).(P_{ij}\parallel Q_{kl})+\boxplus_i\boxplus_k\sum_{\alpha_{ij} \textrm{ comp }\beta_{kl}}\tau.R_{ijkl}$;
  \item $P\parallel Q\nsim_{phhp} \boxplus_{i}\boxplus_{k}\sum_j\sum_l (\alpha_{ij}\parallel \beta_{kl}).(P_{ij}\parallel Q_{kl})+\boxplus_i\boxplus_k\sum_{\alpha_{ij} \textrm{ comp }\beta_{kl}}\tau.R_{ijkl}$.
\end{enumerate}

Where $\alpha_{ij}$ comp $\beta_{kl}$ and $R_{ijkl}$ are defined as follows:
\begin{enumerate}
  \item $\alpha_{ij}$ is $\overline{x}u$ and $\beta_{kl}$ is $x(v)$, then $R_{ijkl}=P_{ij}\parallel Q_{kl}\{u/v\}$;
  \item $\alpha_{ij}$ is $\overline{x}(u)$ and $\beta_{kl}$ is $x(v)$, then $R_{ijkl}=(w)(P_{ij}\{w/u\}\parallel Q_{kl}\{w/v\})$, if $w\notin fn((u)P_{ij})\cup fn((v)Q_{kl})$;
  \item $\alpha_{ij}$ is $x(v)$ and $\beta_{kl}$ is $\overline{x}u$, then $R_{ijkl}=P_{ij}\{u/v\}\parallel Q_{kl}$;
  \item $\alpha_{ij}$ is $x(v)$ and $\beta_{kl}$ is $\overline{x}(u)$, then $R_{ijkl}=(w)(P_{ij}\{w/v\}\parallel Q_{kl}\{w/u\})$, if $w\notin fn((v)P_{ij})\cup fn((u)Q_{kl})$.
\end{enumerate}
  \item Identifier $\textbf{I}$.
  \[\textrm{If }A(\widetilde{x})\overset{\text{def}}{=}P,\textrm{ then }A(\widetilde{y})= P\{\widetilde{y}/\widetilde{x}\}.\]
\end{enumerate}
\end{definition}

\begin{theorem}[Soundness]
If $\textbf{STC}\vdash P=Q$ then
\begin{enumerate}
  \item $P\sim_{pp}  Q$;
  \item $P\sim_{pp}  Q$;
  \item $P\sim_{php}  Q$;
  \item $P\sim_{phhp}  Q$.
\end{enumerate}
\end{theorem}

\begin{proof}
The soundness of these laws modulo strongly truly concurrent bisimilarities is already proven in Section \ref{s6}.
\end{proof}

\begin{definition}
The agent identifier $A$ is weakly guardedly defined if every agent identifier is weakly guarded in the right-hand side of the definition of $A$.
\end{definition}

\begin{definition}[Head normal form]
A Process $P$ is in head normal form if it is a sum of the prefixes:

$$P\equiv \boxplus_{i}\sum_j(\alpha_{ij1}\parallel\cdots\parallel\alpha_{ijn}).P_{ij}$$
\end{definition}

\begin{proposition}
If every agent identifier is weakly guardedly defined, then for any process $P$, there is a head normal form $H$ such that

$$\textbf{STC}\vdash P=H.$$
\end{proposition}

\begin{proof}
It is sufficient to induct on the structure of $P$ and quite obvious.
\end{proof}

\begin{theorem}[Completeness]
For all processes $P$ and $Q$,
\begin{enumerate}
  \item if $P\sim_{pp}  Q$, then $\textbf{STC}\vdash P=Q$;
  \item if $P\sim_{pp}  Q$, then $\textbf{STC}\vdash P=Q$;
  \item if $P\sim_{php}  Q$, then $\textbf{STC}\vdash P=Q$.
\end{enumerate}
\end{theorem}

\begin{proof}
\begin{enumerate}
  \item if $P\sim_{pp}  Q$, then $\textbf{STC}\vdash P=Q$.

Since $P$ and $Q$ all have head normal forms, let $P\equiv\boxplus_{j=1}^l\sum_{i=1}^k\alpha_{ji}.P_{ji}$ and $Q\equiv\boxplus_{j=1}^l\sum_{i=1}^k\beta_{ji}.Q_{ji}$. Then the depth of
$P$, denoted as $d(P)=0$, if $k=0$; $d(P)=1+max\{d(P_{ji})\}$ for $1\leq j,i\leq k$. The depth $d(Q)$ can be defined similarly.

It is sufficient to induct on $d=d(P)+d(Q)$. When $d=0$, $P\equiv\textbf{nil}$ and $Q\equiv\textbf{nil}$, $P=Q$, as desired.

Suppose $d>0$.

\begin{itemize}
  \item If $(\alpha_1\parallel\cdots\parallel\alpha_n).M$ with $\alpha_{ji}(1\leq j,i\leq n)$ free actions is a summand of $P$, then
  $\langle P,s\rangle\rightsquigarrow\xrightarrow{\{\alpha_1,\cdots,\alpha_n\}}\langle M,s'\rangle$.
  Since $Q$ is in head normal form and has a summand $(\alpha_1\parallel\cdots\parallel\alpha_n).N$ such that $M\sim_{pp}  N$, by the induction hypothesis $\textbf{STC}\vdash M=N$,
  $\textbf{STC}\vdash (\alpha_1\parallel\cdots\parallel\alpha_n).M= (\alpha_1\parallel\cdots\parallel\alpha_n).N$;
  \item If $x(y).M$ is a summand of $P$, then for $z\notin n(P, Q)$, $\langle P,s\rangle\rightsquigarrow\xrightarrow{x(z)}\langle M',s'\rangle\equiv \langle M\{z/y\},s'\rangle$. Since $Q$ is in head normal form and has a summand
  $x(w).N$ such that for all $v$, $M'\{v/z\}\sim_{pp}  N'\{v/z\}$ where $N'\equiv N\{z/w\}$, by the induction hypothesis $\textbf{STC}\vdash M'\{v/z\}=N'\{v/z\}$, by the axioms
  $\textbf{C}$ and $\textbf{A}$, $\textbf{STC}\vdash x(y).M=x(w).N$;
  \item If $\overline{x}(y).M$ is a summand of $P$, then for $z\notin n(P,Q)$, $\langle P,s\rangle\rightsquigarrow\xrightarrow{\overline{x}(z)}\langle M',s'\rangle\equiv \langle M\{z/y\},s'\rangle$. Since $Q$ is in head normal form and
  has a summand $\overline{x}(w).N$ such that $M'\sim_{pp}  N'$ where $N'\equiv N\{z/w\}$, by the induction hypothesis $\textbf{STC}\vdash M'=N'$, by the axioms
  $\textbf{A}$ and $\textbf{C}$, $\textbf{STC}\vdash \overline{x}(y).M= \overline{x}(w).N$.
\end{itemize}
  \item if $P\sim_{pp}  Q$, then $\textbf{STC}\vdash P=Q$. It can be proven similarly to the above case.
  \item if $P\sim_{php}  Q$, then $\textbf{STC}\vdash P=Q$. It can be proven similarly to the above case.
\end{enumerate}
\end{proof}

\newpage\section{$\pi_{tc}$ with Reversibility and Guards}\label{pitcrg}

In this chapter, we design $\pi_{tc}$ with reversibility and guards. This chapter is organized as follows. In section \ref{os7}, we introduce the truly concurrent operational semantics. Then, we introduce
the syntax and operational semantics, laws modulo strongly truly concurrent bisimulations, and algebraic theory of $\pi_{tc}$ with reversibility and guards in section \ref{sos7},
\ref{s7} and \ref{a7} respectively.

\subsection{Operational Semantics}\label{os7}

Firstly, in this section, we introduce concepts of FR (strongly) truly concurrent bisimilarities, including FR pomset bisimilarity, FR step
bisimilarity, FR history-preserving (hp-)bisimilarity and FR hereditary history-preserving (hhp-)bisimilarity. In contrast to traditional FR truly
concurrent bisimilarities in section \ref{bg}, these versions in $\pi_{ptc}$ must take care of actions with bound objects. Note that, these FR truly concurrent bisimilarities
are defined as late bisimilarities, but not early bisimilarities, as defined in $\pi$-calculus \cite{PI1} \cite{PI2}. Note that, here, a PES $\mathcal{E}$ is deemed as a process.

\begin{definition}[Prime event structure with silent event and empty event]
Let $\Lambda$ be a fixed set of labels, ranged over $a,b,c,\cdots$ and $\tau,\epsilon$. A ($\Lambda$-labelled) prime event structure with silent event $\tau$ and empty event
$\epsilon$ is a tuple $\mathcal{E}=\langle \mathbb{E}, \leq, \sharp, \lambda\rangle$, where $\mathbb{E}$ is a denumerable set of events, including the silent event $\tau$ and empty
event $\epsilon$. Let $\hat{\mathbb{E}}=\mathbb{E}\backslash\{\tau,\epsilon\}$, exactly excluding $\tau$ and $\epsilon$, it is obvious that $\hat{\tau^*}=\epsilon$. Let
$\lambda:\mathbb{E}\rightarrow\Lambda$ be a labelling function and let $\lambda(\tau)=\tau$ and $\lambda(\epsilon)=\epsilon$. And $\leq$, $\sharp$ are binary relations on $\mathbb{E}$,
called causality and conflict respectively, such that:

\begin{enumerate}
  \item $\leq$ is a partial order and $\lceil e \rceil = \{e'\in \mathbb{E}|e'\leq e\}$ is finite for all $e\in \mathbb{E}$. It is easy to see that
  $e\leq\tau^*\leq e'=e\leq\tau\leq\cdots\leq\tau\leq e'$, then $e\leq e'$.
  \item $\sharp$ is irreflexive, symmetric and hereditary with respect to $\leq$, that is, for all $e,e',e''\in \mathbb{E}$, if $e\sharp e'\leq e''$, then $e\sharp e''$.
\end{enumerate}

Then, the concepts of consistency and concurrency can be drawn from the above definition:

\begin{enumerate}
  \item $e,e'\in \mathbb{E}$ are consistent, denoted as $e\frown e'$, if $\neg(e\sharp e')$. A subset $X\subseteq \mathbb{E}$ is called consistent, if $e\frown e'$ for all
  $e,e'\in X$.
  \item $e,e'\in \mathbb{E}$ are concurrent, denoted as $e\parallel e'$, if $\neg(e\leq e')$, $\neg(e'\leq e)$, and $\neg(e\sharp e')$.
\end{enumerate}
\end{definition}

\begin{definition}[Configuration]
Let $\mathcal{E}$ be a PES. A (finite) configuration in $\mathcal{E}$ is a (finite) consistent subset of events $C\subseteq \mathcal{E}$, closed with respect to causality (i.e.
$\lceil C\rceil=C$), and a data state $s\in S$ with $S$ the set of all data states, denoted $\langle C, s\rangle$. The set of finite configurations of $\mathcal{E}$ is denoted by
$\langle\mathcal{C}(\mathcal{E}), S\rangle$. We let $\hat{C}=C\backslash\{\tau\}\cup\{\epsilon\}$.
\end{definition}

A consistent subset of $X\subseteq \mathbb{E}$ of events can be seen as a pomset. Given $X, Y\subseteq \mathbb{E}$, $\hat{X}\sim \hat{Y}$ if $\hat{X}$ and $\hat{Y}$ are isomorphic as
pomsets. In the following of the paper, we say $C_1\sim C_2$, we mean $\hat{C_1}\sim\hat{C_2}$.

\begin{definition}[FR pomset transitions and step]
Let $\mathcal{E}$ be a PES and let $C\in\mathcal{C}(\mathcal{E})$, and $\emptyset\neq X\subseteq \mathbb{E}$, if $C\cap X=\emptyset$ and $C'=C\cup X\in\mathcal{C}(\mathcal{E})$, then
$\langle C,s\rangle\xrightarrow{X} \langle C',s'\rangle$ is called a forward pomset transition from $\langle C,s\rangle$ to $\langle C',s'\rangle$ and
$\langle C',s'\rangle\xtworightarrow{X[\mathcal{K}]} \langle C,s\rangle$ is called a reverse pomset transition from $\langle C',s'\rangle$ to $\langle C,s\rangle$. When the events in
$X$ and $X[\mathcal{K}]$ are pairwise
concurrent, we say that $\langle C,s\rangle\xrightarrow{X}\langle C',s'\rangle$ is a forward step and $\langle C',s'\rangle\xrightarrow{X[\mathcal{K}]}\langle C,s\rangle$ is a reverse step.
It is obvious that $\rightarrow^*\xrightarrow{X}\rightarrow^*=\xrightarrow{X}$ and
$\rightarrow^*\xrightarrow{e}\rightarrow^*=\xrightarrow{e}$ for any $e\in\mathbb{E}$ and $X\subseteq\mathbb{E}$.
\end{definition}

\begin{definition}[FR strongly pomset, step bisimilarity]
Let $\mathcal{E}_1$, $\mathcal{E}_2$ be PESs. A FR strongly pomset bisimulation is a relation $R\subseteq\langle\mathcal{C}(\mathcal{E}_1),s\rangle\times\langle\mathcal{C}(\mathcal{E}_2),s\rangle$,
such that (1) if $(\langle C_1,s\rangle,\langle C_2,s\rangle)\in R$, and $\langle C_1,s\rangle\xrightarrow{X_1}\langle C_1',s'\rangle$ (with $\mathcal{E}_1\xrightarrow{X_1}\mathcal{E}_1'$) then $\langle C_2,s\rangle\xrightarrow{X_2}\langle C_2',s'\rangle$ (with
$\mathcal{E}_2\xrightarrow{X_2}\mathcal{E}_2'$), with $X_1\subseteq \mathbb{E}_1$, $X_2\subseteq \mathbb{E}_2$, $X_1\sim X_2$ and $(\langle C_1',s'\rangle,\langle C_2',s'\rangle)\in R$:
\begin{enumerate}
  \item for each fresh action $\alpha\in X_1$, if $\langle C_1'',s''\rangle\xrightarrow{\alpha}\langle C_1''',s'''\rangle$ (with $\mathcal{E}_1''\xrightarrow{\alpha}\mathcal{E}_1'''$),
  then for some $C_2''$ and $\langle C_2''',s'''\rangle$, $\langle C_2'',s''\rangle\xrightarrow{\alpha}\langle C_2''',s'''\rangle$ (with
  $\mathcal{E}_2''\xrightarrow{\alpha}\mathcal{E}_2'''$), such that if $(\langle C_1'',s''\rangle,\langle C_2'',s''\rangle)\in R$ then $(\langle C_1''',s'''\rangle,\langle C_2''',s'''\rangle)\in R$;
  \item for each $x(y)\in X_1$ with ($y\notin n(\mathcal{E}_1, \mathcal{E}_2)$), if $\langle C_1'',s''\rangle\xrightarrow{x(y)}\langle C_1''',s'''\rangle$ (with
  $\mathcal{E}_1''\xrightarrow{x(y)}\mathcal{E}_1'''\{w/y\}$) for all $w$, then for some $C_2''$ and $C_2'''$, $\langle C_2'',s''\rangle\xrightarrow{x(y)}\langle C_2''',s'''\rangle$
  (with $\mathcal{E}_2''\xrightarrow{x(y)}\mathcal{E}_2'''\{w/y\}$) for all $w$, such that if $(\langle C_1'',s''\rangle,\langle C_2'',s''\rangle)\in R$ then $(\langle C_1''',s'''\rangle,\langle C_2''',s'''\rangle)\in R$;
  \item for each two $x_1(y),x_2(y)\in X_1$ with ($y\notin n(\mathcal{E}_1, \mathcal{E}_2)$), if $\langle C_1'',s''\rangle\xrightarrow{\{x_1(y),x_2(y)\}}\langle C_1''',s'''\rangle$
  (with $\mathcal{E}_1''\xrightarrow{\{x_1(y),x_2(y)\}}\mathcal{E}_1'''\{w/y\}$) for all $w$, then for some $C_2''$ and $C_2'''$,
  $\langle C_2'',s''\rangle\xrightarrow{\{x_1(y),x_2(y)\}}\langle C_2''',s'''\rangle$ (with $\mathcal{E}_2''\xrightarrow{\{x_1(y),x_2(y)\}}\mathcal{E}_2'''\{w/y\}$) for all $w$, such
  that if $(\langle C_1'',s''\rangle,\langle C_2'',s''\rangle)\in R$ then $(\langle C_1''',s'''\rangle,\langle C_2''',s'''\rangle)\in R$;
  \item for each $\overline{x}(y)\in X_1$ with $y\notin n(\mathcal{E}_1, \mathcal{E}_2)$, if $\langle C_1'',s''\rangle\xrightarrow{\overline{x}(y)}\langle C_1''',s'''\rangle$
  (with $\mathcal{E}_1''\xrightarrow{\overline{x}(y)}\mathcal{E}_1'''$), then for some $C_2''$ and $C_2'''$, $\langle C_2'',s''\rangle\xrightarrow{\overline{x}(y)}\langle C_2''',s'''\rangle$
  (with $\mathcal{E}_2''\xrightarrow{\overline{x}(y)}\mathcal{E}_2'''$), such that if $(\langle C_1'',s''\rangle,\langle C_2'',s''\rangle)\in R$ then $(\langle C_1''',s'''\rangle,\langle C_2''',s'''\rangle)\in R$.
\end{enumerate}
 and vice-versa; (2) if $(\langle C_1,s\rangle,\langle C_2,s\rangle)\in R$, and $\langle C_1,s\rangle\xtworightarrow{X_1[\mathcal{K}_1]}\langle C_1',s'\rangle$ (with $\mathcal{E}_1\xtworightarrow{X_1[\mathcal{K}_1]}\mathcal{E}_1'$) then $\langle C_2,s\rangle\xtworightarrow{X_2[\mathcal{K}_2]}\langle C_2',s'\rangle$ (with
$\mathcal{E}_2\xtworightarrow{X_2[\mathcal{K}_2]}\mathcal{E}_2'$), with $X_1\subseteq \mathbb{E}_1$, $X_2\subseteq \mathbb{E}_2$, $X_1\sim X_2$ and $(\langle C_1',s'\rangle,\langle C_2',s'\rangle)\in R$:
\begin{enumerate}
  \item for each fresh action $\alpha\in X_1$, if $\langle C_1'',s''\rangle\xtworightarrow{\alpha[m]}\langle C_1''',s'''\rangle$ (with $\mathcal{E}_1''\xtworightarrow{\alpha[m]}\mathcal{E}_1'''$),
  then for some $C_2''$ and $\langle C_2''',s'''\rangle$, $\langle C_2'',s''\rangle\xtworightarrow{\alpha[m]}\langle C_2''',s'''\rangle$ (with
  $\mathcal{E}_2''\xtworightarrow{\alpha[m]}\mathcal{E}_2'''$), such that if $(\langle C_1'',s''\rangle,\langle C_2'',s''\rangle)\in R$ then $(\langle C_1''',s'''\rangle,\langle C_2''',s'''\rangle)\in R$;
  \item for each $x(y)\in X_1$ with ($y\notin n(\mathcal{E}_1, \mathcal{E}_2)$), if $\langle C_1'',s''\rangle\xtworightarrow{x(y)[m]}\langle C_1''',s'''\rangle$ (with
  $\mathcal{E}_1''\xtworightarrow{x(y)[m]}\mathcal{E}_1'''\{w/y\}$) for all $w$, then for some $C_2''$ and $C_2'''$, $\langle C_2'',s''\rangle\xtworightarrow{x(y)[m]}\langle C_2''',s'''\rangle$
  (with $\mathcal{E}_2''\xtworightarrow{x(y)[m]}\mathcal{E}_2'''\{w/y\}$) for all $w$, such that if $(\langle C_1'',s''\rangle,\langle C_2'',s''\rangle)\in R$ then $(\langle C_1''',s'''\rangle,\langle C_2''',s'''\rangle)\in R$;
  \item for each two $x_1(y),x_2(y)\in X_1$ with ($y\notin n(\mathcal{E}_1, \mathcal{E}_2)$), if $\langle C_1'',s''\rangle\xtworightarrow{\{x_1(y)[m],x_2(y)[n]\}}\langle C_1''',s'''\rangle$
  (with $\mathcal{E}_1''\xtworightarrow{\{x_1(y)[m],x_2(y)[n]\}}\mathcal{E}_1'''\{w/y\}$) for all $w$, then for some $C_2''$ and $C_2'''$,
  $\langle C_2'',s''\rangle\xtworightarrow{\{x_1(y)[m],x_2(y)[n]\}}\langle C_2''',s'''\rangle$ (with $\mathcal{E}_2''\xtworightarrow{\{x_1(y)[m],x_2(y)[n]\}}\mathcal{E}_2'''\{w/y\}$) for all $w$, such
  that if $(\langle C_1'',s''\rangle,\langle C_2'',s''\rangle)\in R$ then $(\langle C_1''',s'''\rangle,\langle C_2''',s'''\rangle)\in R$;
  \item for each $\overline{x}(y)\in X_1$ with $y\notin n(\mathcal{E}_1, \mathcal{E}_2)$, if $\langle C_1'',s''\rangle\xtworightarrow{\overline{x}(y)[m]}\langle C_1''',s'''\rangle$
  (with $\mathcal{E}_1''\xtworightarrow{\overline{x}(y)[m]}\mathcal{E}_1'''$), then for some $C_2''$ and $C_2'''$, $\langle C_2'',s''\rangle\xtworightarrow{\overline{x}(y)[m]}\langle C_2''',s'''\rangle$
  (with $\mathcal{E}_2''\xtworightarrow{\overline{x}(y)[m]}\mathcal{E}_2'''$), such that if $(\langle C_1'',s''\rangle,\langle C_2'',s''\rangle)\in R$ then $(\langle C_1''',s'''\rangle,\langle C_2''',s'''\rangle)\in R$.
\end{enumerate}
 and vice-versa.

We say that $\mathcal{E}_1$, $\mathcal{E}_2$ are FR strongly pomset bisimilar, written $\mathcal{E}_1\sim_p^{fr}\mathcal{E}_2$, if there exists a FR strongly pomset
bisimulation $R$, such that $(\emptyset,\emptyset)\in R$. By replacing FR pomset transitions with steps, we can get the definition of FR strongly step bisimulation.
When PESs $\mathcal{E}_1$ and $\mathcal{E}_2$ are FR strongly step bisimilar, we write $\mathcal{E}_1\sim_s^{fr}\mathcal{E}_2$.
\end{definition}

\begin{definition}[Posetal product]
Given two PESs $\mathcal{E}_1$, $\mathcal{E}_2$, the posetal product of their configurations, denoted
$\langle\mathcal{C}(\mathcal{E}_1),S\rangle\overline{\times}\langle\mathcal{C}(\mathcal{E}_2),S\rangle$, is defined as

$$\{(\langle C_1,s\rangle,f,\langle C_2,s\rangle)|C_1\in\mathcal{C}(\mathcal{E}_1),C_2\in\mathcal{C}(\mathcal{E}_2),f:C_1\rightarrow C_2 \textrm{ isomorphism}\}.$$

A subset $R\subseteq\langle\mathcal{C}(\mathcal{E}_1),S\rangle\overline{\times}\langle\mathcal{C}(\mathcal{E}_2),S\rangle$ is called a posetal relation. We say that $R$ is downward
closed when for any
$(\langle C_1,s\rangle,f,\langle C_2,s\rangle),(\langle C_1',s'\rangle,f',\langle C_2',s'\rangle)\in \langle\mathcal{C}(\mathcal{E}_1),S\rangle\overline{\times}\langle\mathcal{C}(\mathcal{E}_2),S\rangle$,
if $(\langle C_1,s\rangle,f,\langle C_2,s\rangle)\subseteq (\langle C_1',s'\rangle,f',\langle C_2',s'\rangle)$ pointwise and $(\langle C_1',s'\rangle,f',\langle C_2',s'\rangle)\in R$,
then $(\langle C_1,s\rangle,f,\langle C_2,s\rangle)\in R$.

For $f:X_1\rightarrow X_2$, we define $f[x_1\mapsto x_2]:X_1\cup\{x_1\}\rightarrow X_2\cup\{x_2\}$, $z\in X_1\cup\{x_1\}$,(1)$f[x_1\mapsto x_2](z)=
x_2$,if $z=x_1$;(2)$f[x_1\mapsto x_2](z)=f(z)$, otherwise. Where $X_1\subseteq \mathbb{E}_1$, $X_2\subseteq \mathbb{E}_2$, $x_1\in \mathbb{E}_1$, $x_2\in \mathbb{E}_2$.
\end{definition}

\begin{definition}[FR strongly (hereditary) history-preserving bisimilarity]
A FR strongly history-preserving (hp-) bisimulation is a posetal relation $R\subseteq\mathcal{C}(\mathcal{E}_1)\overline{\times}\mathcal{C}(\mathcal{E}_2)$ such that
(1) if $(\langle C_1,s\rangle,f,\langle C_2,s\rangle)\in R$, and
\begin{enumerate}
  \item for $e_1=\alpha$ a fresh action, if $\langle C_1,s\rangle\xrightarrow{\alpha}\langle C_1',s'\rangle$ (with $\mathcal{E}_1\xrightarrow{\alpha}\mathcal{E}_1'$), then for some
  $C_2'$ and $e_2=\alpha$, $\langle C_2,s\rangle\xrightarrow{\alpha}\langle C_2',s'\rangle$ (with $\mathcal{E}_2\xrightarrow{\alpha}\mathcal{E}_2'$), such that
  $(\langle C_1',s'\rangle,f[e_1\mapsto e_2],\langle C_2',s'\rangle)\in R$;
  \item for $e_1=x(y)$ with ($y\notin n(\mathcal{E}_1, \mathcal{E}_2)$), if $\langle C_1,s\rangle\xrightarrow{x(y)}\langle C_1',s'\rangle$ (with
  $\mathcal{E}_1\xrightarrow{x(y)}\mathcal{E}_1'\{w/y\}$) for all $w$, then for some $C_2'$ and $e_2=x(y)$, $\langle C_2,s\rangle\xrightarrow{x(y)}\langle C_2',s'\rangle$ (with
  $\mathcal{E}_2\xrightarrow{x(y)}\mathcal{E}_2'\{w/y\}$) for all $w$, such that $(\langle C_1',s'\rangle,f[e_1\mapsto e_2],\langle C_2',s'\rangle)\in R$;
  \item for $e_1=\overline{x}(y)$ with $y\notin n(\mathcal{E}_1, \mathcal{E}_2)$, if $\langle C_1,s\rangle\xrightarrow{\overline{x}(y)}\langle C_1',s'\rangle$ (with
  $\mathcal{E}_1\xrightarrow{\overline{x}(y)}\mathcal{E}_1'$), then for some $C_2'$ and $e_2=\overline{x}(y)$, $\langle C_2,s\rangle\xrightarrow{\overline{x}(y)}\langle C_2',s'\rangle$
  (with $\mathcal{E}_2\xrightarrow{\overline{x}(y)}\mathcal{E}_2'$), such that $(\langle C_1',s'\rangle,f[e_1\mapsto e_2],\langle C_2',s'\rangle)\in R$.
\end{enumerate}
and vice-versa; (2) if $(\langle C_1,s\rangle,f,\langle C_2,s\rangle)\in R$, and
\begin{enumerate}
  \item for $e_1=\alpha$ a fresh action, if $\langle C_1,s\rangle\xtworightarrow{\alpha[m]}\langle C_1',s'\rangle$ (with $\mathcal{E}_1\xtworightarrow{\alpha[m]}\mathcal{E}_1'$), then for some
  $C_2'$ and $e_2=\alpha$, $\langle C_2,s\rangle\xtworightarrow{\alpha[m]}\langle C_2',s'\rangle$ (with $\mathcal{E}_2\xtworightarrow{\alpha[m]}\mathcal{E}_2'$), such that
  $(\langle C_1',s'\rangle,f[e_1\mapsto e_2],\langle C_2',s'\rangle)\in R$;
  \item for $e_1=x(y)$ with ($y\notin n(\mathcal{E}_1, \mathcal{E}_2)$), if $\langle C_1,s\rangle\xtworightarrow{x(y)[m]}\langle C_1',s'\rangle$ (with
  $\mathcal{E}_1\xtworightarrow{x(y)[m]}\mathcal{E}_1'\{w/y\}$) for all $w$, then for some $C_2'$ and $e_2=x(y)$, $\langle C_2,s\rangle\xtworightarrow{x(y)[m]}\langle C_2',s'\rangle$ (with
  $\mathcal{E}_2\xtworightarrow{x(y)[m]}\mathcal{E}_2'\{w/y\}$) for all $w$, such that $(\langle C_1',s'\rangle,f[e_1\mapsto e_2],\langle C_2',s'\rangle)\in R$;
  \item for $e_1=\overline{x}(y)$ with $y\notin n(\mathcal{E}_1, \mathcal{E}_2)$, if $\langle C_1,s\rangle\xtworightarrow{\overline{x}(y)[m]}\langle C_1',s'\rangle$ (with
  $\mathcal{E}_1\xtworightarrow{\overline{x}(y)[m]}\mathcal{E}_1'$), then for some $C_2'$ and $e_2=\overline{x}(y)$, $\langle C_2,s\rangle\xtworightarrow{\overline{x}(y)[m]}\langle C_2',s'\rangle$
  (with $\mathcal{E}_2\xtworightarrow{\overline{x}(y)[m]}\mathcal{E}_2'$), such that $(\langle C_1',s'\rangle,f[e_1\mapsto e_2],\langle C_2',s'\rangle)\in R$.
\end{enumerate}
and vice-versa. $\mathcal{E}_1,\mathcal{E}_2$
are FR strongly history-preserving (hp-)bisimilar and are written $\mathcal{E}_1\sim_{hp}^{fr}\mathcal{E}_2$ if there exists a FR strongly hp-bisimulation
$R$ such that $(\emptyset,\emptyset,\emptyset)\in R$.

A FR strongly hereditary history-preserving (hhp-)bisimulation is a downward closed FR strongly hp-bisimulation. $\mathcal{E}_1,\mathcal{E}_2$ are FR
strongly hereditary history-preserving (hhp-)bisimilar and are written $\mathcal{E}_1\sim_{hhp}^{fr}\mathcal{E}_2$.
\end{definition}

\subsection{Syntax and Operational Semantics}\label{sos7}

We assume an infinite set $\mathcal{N}$ of (action or event) names, and use $a,b,c,\cdots$ to range over $\mathcal{N}$, use $x,y,z,w,u,v$ as meta-variables over names. We denote by
$\overline{\mathcal{N}}$ the set of co-names and let $\overline{a},\overline{b},\overline{c},\cdots$ range over $\overline{\mathcal{N}}$. Then we set
$\mathcal{L}=\mathcal{N}\cup\overline{\mathcal{N}}$ as the set of labels, and use $l,\overline{l}$ to range over $\mathcal{L}$. We extend complementation to $\mathcal{L}$ such that
$\overline{\overline{a}}=a$. Let $\tau$ denote the silent step (internal action or event) and define $Act=\mathcal{L}\cup\{\tau\}$ to be the set of actions, $\alpha,\beta$ range over
$Act$. And $K,L$ are used to stand for subsets of $\mathcal{L}$ and $\overline{L}$ is used for the set of complements of labels in $L$.

Further, we introduce a set $\mathcal{X}$ of process variables, and a set $\mathcal{K}$ of process constants, and let $X,Y,\cdots$ range over $\mathcal{X}$, and $A,B,\cdots$ range over
$\mathcal{K}$. For each process constant $A$, a nonnegative arity $ar(A)$ is assigned to it. Let $\widetilde{x}=x_1,\cdots,x_{ar(A)}$ be a tuple of distinct name variables, then
$A(\widetilde{x})$ is called a process constant. $\widetilde{X}$ is a tuple of distinct process variables, and also $E,F,\cdots$ range over the recursive expressions. We write
$\mathcal{P}$ for the set of processes. Sometimes, we use $I,J$ to stand for an indexing set, and we write $E_i:i\in I$ for a family of expressions indexed by $I$. $Id_D$ is the
identity function or relation over set $D$. The symbol $\equiv_{\alpha}$ denotes equality under standard alpha-convertibility, note that the subscript $\alpha$ has no relation to the
action $\alpha$.

Let $G_{at}$ be the set of atomic guards, $\delta$ be the deadlock constant, and $\epsilon$ be the empty action, and extend $Act$ to $Act\cup\{\epsilon\}\cup\{\delta\}$. We extend
$G_{at}$ to the set of basic guards $G$ with element $\phi,\psi,\cdots$, which is generated by the following formation rules:

$$\phi::=\delta|\epsilon|\neg\phi|\psi\in G_{at}|\phi+\psi|\phi\cdot\psi$$

The predicate $test(\phi,s)$ represents that $\phi$ holds in the state $s$, and $test(\epsilon,s)$ holds and $test(\delta,s)$ does not hold. $effect(e,s)\in S$ denotes $s'$ in
$s\xrightarrow{e}s'$. The predicate weakest precondition $wp(e,\phi)$ denotes that $\forall s,s'\in S, test(\phi,effect(e,s))$ holds.

\subsubsection{Syntax}

We use the Prefix $.$ to model the causality relation $\leq$ in true concurrency, the Summation $+$ to model the conflict relation $\sharp$ in true concurrency, and the Composition $\parallel$ to explicitly model concurrent relation in true concurrency. And we follow the
conventions of process algebra.

\begin{definition}[Syntax]\label{syntax7}
A truly concurrent process $\pi_{tc}$ with reversibility and guards is defined inductively by the following formation rules:

\begin{enumerate}
  \item $A(\widetilde{x})\in\mathcal{P}$;
  \item $\phi\in\mathcal{P}$;
  \item $\textbf{nil}\in\mathcal{P}$;
  \item if $P\in\mathcal{P}$, then the Prefix $\tau.P\in\mathcal{P}$, for $\tau\in Act$ is the silent action;
  \item if $P\in\mathcal{P}$, then the Prefix $\phi.P\in\mathcal{P}$, for $\phi\in G_{at}$;
  \item if $P\in\mathcal{P}$, then the Output $\overline{x}y.P\in\mathcal{P}$, for $x,y\in Act$;
  \item if $P\in\mathcal{P}$, then the Output $P.\overline{x}y[m]\in\mathcal{P}$, for $x,y\in Act$;
  \item if $P\in\mathcal{P}$, then the Input $x(y).P\in\mathcal{P}$, for $x,y\in Act$;
  \item if $P\in\mathcal{P}$, then the Input $P.x(y)[m]\in\mathcal{P}$, for $x,y\in Act$;
  \item if $P\in\mathcal{P}$, then the Restriction $(x)P\in\mathcal{P}$, for $x\in Act$;
  \item if $P,Q\in\mathcal{P}$, then the Summation $P+Q\in\mathcal{P}$;
  \item if $P,Q\in\mathcal{P}$, then the Composition $P\parallel Q\in\mathcal{P}$;
\end{enumerate}

The standard BNF grammar of syntax of $\pi_{tc}$ with reversibility and guards can be summarized as follows:

$$P::=A(\widetilde{x})|\textbf{nil}|\tau.P| \overline{x}y.P | x(y).P|\overline{x}y[m].P | x(y)[m].P | (x)P  |\phi.P|  P+P| P\parallel P.$$
\end{definition}

In $\overline{x}y$, $x(y)$ and $\overline{x}(y)$, $x$ is called the subject, $y$ is called the object and it may be free or bound.

\begin{definition}[Free variables]
The free names of a process $P$, $fn(P)$, are defined as follows.

\begin{enumerate}
  \item $fn(A(\widetilde{x}))\subseteq\{\widetilde{x}\}$;
  \item $fn(\textbf{nil})=\emptyset$;
  \item $fn(\tau.P)=fn(P)$;
  \item $fn(\phi.P)=fn(P)$;
  \item $fn(\overline{x}y.P)=fn(P)\cup\{x\}\cup\{y\}$;
  \item $fn(\overline{x}y[m].P)=fn(P)\cup\{x\}\cup\{y\}$;
  \item $fn(x(y).P)=fn(P)\cup\{x\}-\{y\}$;
  \item $fn(x(y)[m].P)=fn(P)\cup\{x\}-\{y\}$;
  \item $fn((x)P)=fn(P)-\{x\}$;
  \item $fn(P+Q)=fn(P)\cup fn(Q)$;
  \item $fn(P\parallel Q)=fn(P)\cup fn(Q)$.
\end{enumerate}
\end{definition}

\begin{definition}[Bound variables]
Let $n(P)$ be the names of a process $P$, then the bound names $bn(P)=n(P)-fn(P)$.
\end{definition}

For each process constant schema $A(\widetilde{x})$, a defining equation of the form

$$A(\widetilde{x})\overset{\text{def}}{=}P$$

is assumed, where $P$ is a process with $fn(P)\subseteq \{\widetilde{x}\}$.

\begin{definition}[Substitutions]\label{subs7}
A substitution is a function $\sigma:\mathcal{N}\rightarrow\mathcal{N}$. For $x_i\sigma=y_i$ with $1\leq i\leq n$, we write $\{y_1/x_1,\cdots,y_n/x_n\}$ or
$\{\widetilde{y}/\widetilde{x}\}$ for $\sigma$. For a process $P\in\mathcal{P}$, $P\sigma$ is defined inductively as follows:
\begin{enumerate}
  \item if $P$ is a process constant $A(\widetilde{x})=A(x_1,\cdots,x_n)$, then $P\sigma=A(x_1\sigma,\cdots,x_n\sigma)$;
  \item if $P=\textbf{nil}$, then $P\sigma=\textbf{nil}$;
  \item if $P=\tau.P'$, then $P\sigma=\tau.P'\sigma$;
  \item if $P=\phi.P'$, then $P\sigma=\phi.P'\sigma$;
  \item if $P=\overline{x}y.P'$, then $P\sigma=\overline{x\sigma}y\sigma.P'\sigma$;
  \item if $P=\overline{x}y[m].P'$, then $P\sigma=\overline{x\sigma}y\sigma[m].P'\sigma$;
  \item if $P=x(y).P'$, then $P\sigma=x\sigma(y).P'\sigma$;
  \item if $P=x(y)[m].P'$, then $P\sigma=x\sigma(y)[m].P'\sigma$;
  \item if $P=(x)P'$, then $P\sigma=(x\sigma)P'\sigma$;
  \item if $P=P_1+P_2$, then $P\sigma=P_1\sigma+P_2\sigma$;
  \item if $P=P_1\parallel P_2$, then $P\sigma=P_1\sigma \parallel P_2\sigma$.
\end{enumerate}
\end{definition}

\subsubsection{Operational Semantics}

The operational semantics is defined by LTSs (labelled transition systems), and it is detailed by the following definition.

\begin{definition}[Semantics]\label{semantics7}
The operational semantics of $\pi_{tc}$ with reversibility and guards corresponding to the syntax in Definition \ref{syntax7} is defined by a series of transition rules, named $\textbf{ACT}$, $\textbf{SUM}$,
$\textbf{IDE}$, $\textbf{PAR}$, $\textbf{COM}$, $\textbf{CLOSE}$, $\textbf{RES}$, $\textbf{OPEN}$ indicate that the rules are associated respectively with Prefix, Summation, 
Identity, Parallel Composition, Communication, and Restriction in Definition \ref{syntax7}. They are shown in \ref{TRForPITC7}.

\begin{center}
    \begin{table}
        \[\textbf{TAU-ACT}\quad \frac{}{\langle\tau.P,s\rangle\xrightarrow{\tau}\langle P,\tau(s)\rangle}\]

        \[\textbf{OUTPUT-ACT}\quad \frac{}{\langle\overline{x}y.P,s\rangle\xrightarrow{\overline{x}y}\langle P,s'\rangle}\]

        \[\textbf{INPUT-ACT}\quad \frac{}{\langle x(z).P,s\rangle\xrightarrow{x(w)}\langle P\{w/z\},s'\rangle}\quad (w\notin fn((z)P))\]

        \[\textbf{PAR}_1\quad \frac{\langle P,s\rangle\xrightarrow{\alpha}\langle P',s'\rangle\quad \langle Q,s\rangle\nrightarrow}{\langle P\parallel Q,s\rangle\xrightarrow{\alpha}\langle P'\parallel Q,s'\rangle}\quad (bn(\alpha)\cap fn(Q)=\emptyset)\]

        \[\textbf{PAR}_2\quad \frac{\langle Q,s\rangle\xrightarrow{\alpha}\langle Q',s'\rangle\quad \langle P,s\rangle\nrightarrow}{\langle P\parallel Q,s\rangle\xrightarrow{\alpha}\langle P\parallel Q',s'\rangle}\quad (bn(\alpha)\cap fn(P)=\emptyset)\]

        \[\textbf{PAR}_3\quad \frac{\langle P,s\rangle\xrightarrow{\alpha}\langle P',s'\rangle\quad \langle Q,s\rangle\xrightarrow{\beta}\langle Q',s''\rangle}{\langle P\parallel Q,s\rangle\xrightarrow{\{\alpha,\beta\}}\langle P'\parallel Q',s'\cup s''\rangle}\] $(\beta\neq\overline{\alpha}, bn(\alpha)\cap bn(\beta)=\emptyset, bn(\alpha)\cap fn(Q)=\emptyset,bn(\beta)\cap fn(P)=\emptyset)$

        \[\textbf{PAR}_4\quad \frac{\langle P,s\rangle\xrightarrow{x_1(z)}\langle P',s'\rangle\quad \langle Q,s\rangle\xrightarrow{x_2(z)}\langle Q',s''\rangle}{\langle P\parallel Q,s\rangle\xrightarrow{\{x_1(w),x_2(w)\}}\langle P'\{w/z\}\parallel Q'\{w/z\},s'\cup s''\rangle}\quad (w\notin fn((z)P)\cup fn((z)Q))\]

        \[\textbf{COM}\quad \frac{\langle P,s\rangle\xrightarrow{\overline{x}y}\langle P',s'\rangle\quad \langle Q,s\rangle\xrightarrow{x(z)}\langle Q',s''\rangle}{\langle P\parallel Q,s\rangle\xrightarrow{\tau}\langle P'\parallel Q'\{y/z\},s'\cup s''\rangle}\]

        \[\textbf{CLOSE}\quad \frac{\langle P,s\rangle\xrightarrow{\overline{x}(w)}\langle P',s'\rangle\quad \langle Q,s\rangle\xrightarrow{x(w)}\langle Q',s''\rangle}{\langle P\parallel Q,s\rangle\xrightarrow{\tau}\langle (w)(P'\parallel Q'),s'\cup s''\rangle}\]

        \caption{Forward transition rules}
        \label{TRForPITC7}
    \end{table}
\end{center}

\begin{center}
    \begin{table}
        \[\textbf{SUM}_1\quad \frac{\langle P,s\rangle\xrightarrow{\alpha}\langle P',s'\rangle}{\langle P+Q,s\rangle\xrightarrow{\alpha}\langle P',s'\rangle}\]

        \[\textbf{SUM}_2\quad \frac{\langle P,s\rangle\xrightarrow{\{\alpha_1,\cdots,\alpha_n\}}\langle P',s'\rangle}{\langle P+Q,s\rangle\xrightarrow{\{\alpha_1,\cdots,\alpha_n\}}\langle P',s'\rangle}\]

        \[\textbf{IDE}_1\quad\frac{\langle P\{\widetilde{y}/\widetilde{x}\},s\rangle\xrightarrow{\alpha}\langle P',s'\rangle}{\langle A(\widetilde{y}),s\rangle\xrightarrow{\alpha}\langle P',s'\rangle}\quad (A(\widetilde{x})\overset{\text{def}}{=}P)\]

        \[\textbf{IDE}_2\quad\frac{\langle P\{\widetilde{y}/\widetilde{x}\},s\rangle\xrightarrow{\{\alpha_1,\cdots,\alpha_n\}}\langle P',s'\rangle} {\langle A(\widetilde{y}),s\rangle\xrightarrow{\{\alpha_1,\cdots,\alpha_n\}}\langle P',s'\rangle}\quad (A(\widetilde{x})\overset{\text{def}}{=}P)\]

        \[\textbf{RES}_1\quad \frac{\langle P,s\rangle\xrightarrow{\alpha}\langle P',s'\rangle}{\langle (y)P,s\rangle\xrightarrow{\alpha}\langle (y)P',s'\rangle}\quad (y\notin n(\alpha))\]

        \[\textbf{RES}_2\quad \frac{\langle P,s\rangle\xrightarrow{\{\alpha_1,\cdots,\alpha_n\}}\langle P',s'\rangle}{\langle (y)P,s\rangle\xrightarrow{\{\alpha_1,\cdots,\alpha_n\}}\langle (y)P',s'\rangle}\quad (y\notin n(\alpha_1)\cup\cdots\cup n(\alpha_n))\]

        \[\textbf{OPEN}_1\quad \frac{\langle P,s\rangle\xrightarrow{\overline{x}y}\langle P',s'\rangle}{\langle (y)P,s\rangle\xrightarrow{\overline{x}(w)}\langle P'\{w/y\},s'\rangle} \quad (y\neq x, w\notin fn((y)P'))\]

        \[\textbf{OPEN}_2\quad \frac{\langle P,s\rangle\xrightarrow{\{\overline{x}_1 y,\cdots,\overline{x}_n y\}}\langle P',s'\rangle}{\langle(y)P,s\rangle\xrightarrow{\{\overline{x}_1(w),\cdots,\overline{x}_n(w)\}}\langle P'\{w/y\},s'\rangle} \quad (y\neq x_1\neq\cdots\neq x_n, w\notin fn((y)P'))\]

        \caption{Forward transition rules (continuing)}
        \label{TRForPITC72}
    \end{table}
\end{center}

\begin{center}
    \begin{table}
        \[\textbf{RTAU-ACT}\quad \frac{}{\langle\tau.P,s\rangle\xtworightarrow{\tau}\langle P,\tau(s)\rangle}\]

        \[\textbf{ROUTPUT-ACT}\quad \frac{}{\langle\overline{x}y[m].P,s\rangle\xtworightarrow{\overline{x}y[m]}\langle P,s'\rangle}\]

        \[\textbf{RINPUT-ACT}\quad \frac{}{\langle x(z)[m].P,s\rangle\xtworightarrow{x(w)[m]}\langle P\{w/z\},s'\rangle}\quad (w\notin fn((z)P))\]

        \[\textbf{RPAR}_1\quad \frac{\langle P,s\rangle\xtworightarrow{\alpha[m]}\langle P',s'\rangle\quad \langle Q,s\rangle\nrightarrow}{\langle P\parallel Q,s\rangle\xtworightarrow{\alpha[m]}\langle P'\parallel Q,s'\rangle}\quad (bn(\alpha)\cap fn(Q)=\emptyset)\]

        \[\textbf{RPAR}_2\quad \frac{\langle Q,s\rangle\xtworightarrow{\alpha[m]}\langle Q',s'\rangle\quad \langle P,s\rangle\nrightarrow}{\langle P\parallel Q,s\rangle\xtworightarrow{\alpha[m]}\langle P\parallel Q',s'\rangle}\quad (bn(\alpha)\cap fn(P)=\emptyset)\]

        \[\textbf{RPAR}_3\quad \frac{\langle P,s\rangle\xtworightarrow{\alpha[m]}\langle P',s'\rangle\quad \langle Q,s\rangle\xtworightarrow{\beta[m]}\langle Q',s''\rangle}{\langle P\parallel Q,s\rangle\xtworightarrow{\{\alpha[m],\beta[m]\}}\langle P'\parallel Q',s'\cup s''\rangle}\] $(\beta\neq\overline{\alpha}, bn(\alpha)\cap bn(\beta)=\emptyset, bn(\alpha)\cap fn(Q)=\emptyset,bn(\beta)\cap fn(P)=\emptyset)$

        \[\textbf{RPAR}_4\quad \frac{\langle P,s\rangle\xtworightarrow{x_1(z)[m]}\langle P',s'\rangle\quad \langle Q,s\rangle\xtworightarrow{x_2(z)[m]}\langle Q',s''\rangle}{\langle P\parallel Q,s\rangle\xtworightarrow{\{x_1(w)[m],x_2(w)[m]\}}\langle P'\{w/z\}\parallel Q'\{w/z\},s'\cup s''\rangle}\quad (w\notin fn((z)P)\cup fn((z)Q))\]

        \[\textbf{RCOM}\quad \frac{\langle P,s\rangle\xtworightarrow{\overline{x}y[m]}\langle P',s'\rangle\quad \langle Q,s\rangle\xtworightarrow{x(z)[m]}\langle Q',s''\rangle}{\langle P\parallel Q,s\rangle\xtworightarrow{\tau}\langle P'\parallel Q'\{y/z\},s'\cup s''\rangle}\]

        \[\textbf{RCLOSE}\quad \frac{\langle P,s\rangle\xtworightarrow{\overline{x}(w)[m]}\langle P',s'\rangle\quad \langle Q,s\rangle\xtworightarrow{x(w)[m]}\langle Q',s''\rangle}{\langle P\parallel Q,s\rangle\xtworightarrow{\tau}\langle (w)(P'\parallel Q'),s'\cup s''\rangle}\]

        \caption{Reverse transition rules}
        \label{TRForPITC73}
    \end{table}
\end{center}

\begin{center}
    \begin{table}
        \[\textbf{RSUM}_1\quad \frac{\langle P,s\rangle\xtworightarrow{\alpha[m]}\langle P',s'\rangle}{\langle P+Q,s\rangle\xtworightarrow{\alpha[m]}\langle P',s'\rangle}\]

        \[\textbf{RSUM}_2\quad \frac{\langle P,s\rangle\xtworightarrow{\{\alpha_1[m],\cdots,\alpha_n[m]\}}\langle P',s'\rangle}{\langle P+Q,s\rangle\xtworightarrow{\{\alpha_1[m],\cdots,\alpha_n[m]\}}\langle P',s'\rangle}\]

        \[\textbf{RIDE}_1\quad\frac{\langle P\{\widetilde{y}/\widetilde{x}\},s\rangle\xtworightarrow{\alpha[m]}\langle P',s'\rangle}{\langle A(\widetilde{y}),s\rangle\xtworightarrow{\alpha[m]}\langle P',s'\rangle}\quad (A(\widetilde{x})\overset{\text{def}}{=}P)\]

        \[\textbf{RIDE}_2\quad\frac{\langle P\{\widetilde{y}/\widetilde{x}\},s\rangle\xtworightarrow{\{\alpha_1[m],\cdots,\alpha_n[m]\}}\langle P',s'\rangle} {\langle A(\widetilde{y}),s\rangle\xtworightarrow{\{\alpha_1[m],\cdots,\alpha_n[m]\}}\langle P',s'\rangle}\quad (A(\widetilde{x})\overset{\text{def}}{=}P)\]

        \[\textbf{RRES}_1\quad \frac{\langle P,s\rangle\xtworightarrow{\alpha[m]}\langle P',s'\rangle}{\langle (y)P,s\rangle\xtworightarrow{\alpha[m]}\langle (y)P',s'\rangle}\quad (y\notin n(\alpha))\]

        \[\textbf{RRES}_2\quad \frac{\langle P,s\rangle\xtworightarrow{\{\alpha_1[m],\cdots,\alpha_n[m]\}}\langle P',s'\rangle}{\langle (y)P,s\rangle\xtworightarrow{\{\alpha_1[m],\cdots,\alpha_n[m]\}}\langle (y)P',s'\rangle}\quad (y\notin n(\alpha_1)\cup\cdots\cup n(\alpha_n))\]

        \[\textbf{ROPEN}_1\quad \frac{\langle P,s\rangle\xtworightarrow{\overline{x}y[m]}\langle P',s'\rangle}{\langle (y)P,s\rangle\xtworightarrow{\overline{x}(w)[m]}\langle P'\{w/y\},s'\rangle} \quad (y\neq x, w\notin fn((y)P'))\]

        \[\textbf{ROPEN}_2\quad \frac{\langle P,s\rangle\xtworightarrow{\{\overline{x}_1 y[m],\cdots,\overline{x}_n y[m]\}}\langle P',s'\rangle}{\langle(y)P,s\rangle\xtworightarrow{\{\overline{x}_1(w)[m],\cdots,\overline{x}_n(w)[m]\}}\langle P'\{w/y\},s'\rangle} \quad (y\neq x_1\neq\cdots\neq x_n, w\notin fn((y)P'))\]

        \caption{Reverse transition rules (continuing)}
        \label{TRForPITC74}
    \end{table}
\end{center}
\end{definition}

\subsubsection{Properties of Transitions}

\begin{proposition}
\begin{enumerate}
  \item If $\langle P,s\rangle\xrightarrow{\alpha}\langle P',s'\rangle$ then
  \begin{enumerate}
    \item $fn(\alpha)\subseteq fn(P)$;
    \item $fn(P')\subseteq fn(P)\cup bn(\alpha)$;
  \end{enumerate}
  \item If $\langle P,s\rangle\xrightarrow{\{\alpha_1,\cdots,\alpha_n\}}\langle P',s\rangle$ then
  \begin{enumerate}
    \item $fn(\alpha_1)\cup\cdots\cup fn(\alpha_n)\subseteq fn(P)$;
    \item $fn(P')\subseteq fn(P)\cup bn(\alpha_1)\cup\cdots\cup bn(\alpha_n)$.
  \end{enumerate}
\end{enumerate}
\end{proposition}

\begin{proof}
By induction on the depth of inference.
\end{proof}

\begin{proposition}
Suppose that $\langle P,s\rangle\xrightarrow{\alpha(y)}\langle P',s'\rangle$, where $\alpha=x$ or $\alpha=\overline{x}$, and $x\notin n(P)$, then there exists some $P''\equiv_{\alpha}P'\{z/y\}$,
$\langle P,s\rangle\xrightarrow{\alpha(z)}\langle P'',s''\rangle$.
\end{proposition}

\begin{proof}
By induction on the depth of inference.
\end{proof}

\begin{proposition}
If $\langle P,s\rangle\xrightarrow{\alpha} \langle P',s'\rangle$, $bn(\alpha)\cap fn(P'\sigma)=\emptyset$, and $\sigma\lceil bn(\alpha)=id$, then there exists some $P''\equiv_{\alpha}P'\sigma$,
$\langle P,s\rangle\sigma\xrightarrow{\alpha\sigma}\langle P'',s''\rangle$.
\end{proposition}

\begin{proof}
By the definition of substitution (Definition \ref{subs7}) and induction on the depth of inference.
\end{proof}

\begin{proposition}
\begin{enumerate}
  \item If $\langle P\{w/z\},s\rangle\xrightarrow{\alpha}\langle P',s'\rangle$, where $w\notin fn(P)$ and $bn(\alpha)\cap fn(P,w)=\emptyset$, then there exist some $Q$ and $\beta$ with $Q\{w/z\}\equiv_{\alpha}P'$ and
  $\beta\sigma=\alpha$, $\langle P,s\rangle\xrightarrow{\beta}\langle Q,s'\rangle$;
  \item If $\langle P\{w/z\},s\rangle\xrightarrow{\{\alpha_1,\cdots,\alpha_n\}}\langle P',s'\rangle$, where $w\notin fn(P)$ and $bn(\alpha_1)\cap\cdots\cap bn(\alpha_n)\cap fn(P,w)=\emptyset$, then there exist some $Q$
  and $\beta_1,\cdots,\beta_n$ with $Q\{w/z\}\equiv_{\alpha}P'$ and $\beta_1\sigma=\alpha_1,\cdots,\beta_n\sigma=\alpha_n$, $\langle P,s\rangle\xrightarrow{\{\beta_1,\cdots,\beta_n\}}\langle Q,s'\rangle$.
\end{enumerate}

\end{proposition}

\begin{proof}
By the definition of substitution (Definition \ref{subs7}) and induction on the depth of inference.
\end{proof}

\begin{proposition}
\begin{enumerate}
  \item If $\langle P,s\rangle\xtworightarrow{\alpha[m]}\langle P',s'\rangle$ then
  \begin{enumerate}
    \item $fn(\alpha[m])\subseteq fn(P)$;
    \item $fn(P')\subseteq fn(P)\cup bn(\alpha[m])$;
  \end{enumerate}
  \item If $\langle P,s\rangle\xtworightarrow{\{\alpha_1[m],\cdots,\alpha_n[m]\}}\langle P',s\rangle$ then
  \begin{enumerate}
    \item $fn(\alpha_1[m])\cup\cdots\cup fn(\alpha_n[m])\subseteq fn(P)$;
    \item $fn(P')\subseteq fn(P)\cup bn(\alpha_1[m])\cup\cdots\cup bn(\alpha_n[m])$.
  \end{enumerate}
\end{enumerate}
\end{proposition}

\begin{proof}
By induction on the depth of inference.
\end{proof}

\begin{proposition}
Suppose that $\langle P,s\rangle\xtworightarrow{\alpha(y)[m]}\langle P',s'\rangle$, where $\alpha=x$ or $\alpha=\overline{x}$, and $x\notin n(P)$, then there exists some $P''\equiv_{\alpha}P'\{z/y\}$,
$\langle P,s\rangle\xtworightarrow{\alpha(z)[m]}\langle P'',s''\rangle$.
\end{proposition}

\begin{proof}
By induction on the depth of inference.
\end{proof}

\begin{proposition}
If $\langle P,s\rangle\xtworightarrow{\alpha[m]} \langle P',s'\rangle$, $bn(\alpha[m])\cap fn(P'\sigma)=\emptyset$, and $\sigma\lceil bn(\alpha[m])=id$, then there exists some $P''\equiv_{\alpha}P'\sigma$,
$\langle P,s\rangle\sigma\xtworightarrow{\alpha[m]\sigma}\langle P'',s''\rangle$.
\end{proposition}

\begin{proof}
By the definition of substitution (Definition \ref{subs7}) and induction on the depth of inference.
\end{proof}

\begin{proposition}
\begin{enumerate}
  \item If $\langle P\{w/z\},s\rangle\xtworightarrow{\alpha[m]}\langle P',s'\rangle$, where $w\notin fn(P)$ and $bn(\alpha)\cap fn(P,w)=\emptyset$, then there exist some $Q$ and $\beta$ with $Q\{w/z\}\equiv_{\alpha}P'$ and
  $\beta\sigma[m]=\alpha[m]$, $\langle P,s\rangle\xtworightarrow{\beta[m]}\langle Q,s'\rangle$;
  \item If $\langle P\{w/z\},s\rangle\xtworightarrow{\{\alpha_1[m],\cdots,\alpha_n[m]\}}\langle P',s'\rangle$, where $w\notin fn(P)$ and $bn(\alpha_1[m])\cap\cdots\cap bn(\alpha_n[m])\cap fn(P,w)=\emptyset$, then there exist some $Q$
  and $\beta_1[m],\cdots,\beta_n[m]$ with $Q\{w/z\}\equiv_{\alpha}P'$ and $\beta_1\sigma[m]=\alpha_1[m],\cdots,\beta_n\sigma[m]=\alpha_n[m]$, $\langle P,s\rangle\xtworightarrow{\{\beta_1[m],\cdots,\beta_n[m]\}}\langle Q,s'\rangle$.
\end{enumerate}

\end{proposition}

\begin{proof}
By the definition of substitution (Definition \ref{subs7}) and induction on the depth of inference.
\end{proof}

\subsection{Strong Bisimilarities}\label{s7}

\subsubsection{Laws and Congruence}

\begin{theorem}
$\equiv_{\alpha}$ are FR strongly truly concurrent bisimulations. That is, if $P\equiv_{\alpha}Q$, then,
\begin{enumerate}
  \item $P\sim_p^{fr} Q$;
  \item $P\sim_s^{fr} Q$;
  \item $P\sim_{hp}^{fr} Q$;
  \item $P\sim_{hhp}^{fr} Q$.
\end{enumerate}
\end{theorem}

\begin{proof}
By induction on the depth of inference, we can get the following facts:

\begin{enumerate}
  \item If $\alpha$ is a free action and $\langle P,s\rangle\xrightarrow{\alpha}\langle P',s'\rangle$, then equally for some $Q'$ with $P'\equiv_{\alpha}Q'$,
  $\langle Q,s\rangle\xrightarrow{\alpha}\langle Q',s'\rangle$;
  \item If $\langle P,s\rangle\xrightarrow{a(y)}\langle P',s'\rangle$ with $a=x$ or $a=\overline{x}$ and $z\notin n(Q)$, then equally for some $Q'$ with $P'\{z/y\}\equiv_{\alpha}Q'$,
  $\langle Q,s\rangle\xrightarrow{a(z)}\langle Q',s'\rangle$;
  \item If $\alpha[m]$ is a free action and $\langle P,s\rangle\xtworightarrow{\alpha[m]}\langle P',s'\rangle$, then equally for some $Q'$ with $P'\equiv_{\alpha}Q'$,
  $\langle Q,s\rangle\xtworightarrow{\alpha[m]}\langle Q',s'\rangle$;
  \item If $\langle P,s\rangle\xtworightarrow{a(y)[m]}\langle P',s'\rangle$ with $a=x$ or $a=\overline{x}$ and $z\notin n(Q)$, then equally for some $Q'$ with $P'\{z/y\}\equiv_{\alpha}Q'$,
  $\langle Q,s\rangle\xtworightarrow{a(z)[m]}\langle Q',s'\rangle$.
\end{enumerate}

Then, we can get:

\begin{enumerate}
  \item by the definition of FR strongly pomset bisimilarity, $P\sim_p^{fr} Q$;
  \item by the definition of FR strongly step bisimilarity, $P\sim_s^{fr} Q$;
  \item by the definition of FR strongly hp-bisimilarity, $P\sim_{hp}^{fr} Q$;
  \item by the definition of FR strongly hhp-bisimilarity, $P\sim_{hhp}^{fr} Q$.
\end{enumerate}
\end{proof}

\begin{proposition}[Summation laws for FR strongly pomset bisimulation] The Summation laws for FR strongly pomset bisimulation are as follows.

\begin{enumerate}
  \item $P+Q\sim_p^{fr} Q+P$;
  \item $P+(Q+R)\sim_p^{fr} (P+Q)+R$;
  \item $P+P\sim_p^{fr} P$;
  \item $P+\textbf{nil}\sim_p^{fr} P$.
\end{enumerate}

\end{proposition}

\begin{proof}
\begin{enumerate}
  \item $P+Q\sim_p^{fr} Q+P$. It is sufficient to prove the relation $R=\{(P+Q, Q+P)\}\cup \textbf{Id}$ is a FR strongly pomset bisimulation, we omit it;
  \item $P+(Q+R)\sim_p^{fr} (P+Q)+R$. It is sufficient to prove the relation $R=\{(P+(Q+R), (P+Q)+R)\}\cup \textbf{Id}$ is a FR strongly pomset bisimulation, we omit it;
  \item $P+P\sim_p^{fr} P$. It is sufficient to prove the relation $R=\{(P+P, P)\}\cup \textbf{Id}$ is a FR strongly pomset bisimulation, we omit it;
  \item $P+\textbf{nil}\sim_p^{fr} P$. It is sufficient to prove the relation $R=\{(P+\textbf{nil}, P)\}\cup \textbf{Id}$ is a FR strongly pomset bisimulation, we omit it.
\end{enumerate}
\end{proof}

\begin{proposition}[Summation laws for FR strongly step bisimulation] The Summation laws for FR strongly step bisimulation are as follows.
\begin{enumerate}
  \item $P+Q\sim_s^{fr} Q+P$;
  \item $P+(Q+R)\sim_s^{fr} (P+Q)+R$;
  \item $P+P\sim_s^{fr} P$;
  \item $P+\textbf{nil}\sim_s^{fr} P$.
\end{enumerate}
\end{proposition}

\begin{proof}
\begin{enumerate}
  \item $P+Q\sim_s^{fr} Q+P$. It is sufficient to prove the relation $R=\{(P+Q, Q+P)\}\cup \textbf{Id}$ is a FR strongly step bisimulation, we omit it;
  \item $P+(Q+R)\sim_s^{fr} (P+Q)+R$. It is sufficient to prove the relation $R=\{(P+(Q+R), (P+Q)+R)\}\cup \textbf{Id}$ is a FR strongly step bisimulation, we omit it;
  \item $P+P\sim_s^{fr} P$. It is sufficient to prove the relation $R=\{(P+P, P)\}\cup \textbf{Id}$ is a FR strongly step bisimulation, we omit it;
  \item $P+\textbf{nil}\sim_s^{fr} P$. It is sufficient to prove the relation $R=\{(P+\textbf{nil}, P)\}\cup \textbf{Id}$ is a FR strongly step bisimulation, we omit it.
\end{enumerate}
\end{proof}

\begin{proposition}[Summation laws for FR strongly hp-bisimulation] The Summation laws for FR strongly hp-bisimulation are as follows.
\begin{enumerate}
  \item $P+Q\sim_{hp}^{fr} Q+P$;
  \item $P+(Q+R)\sim_{hp}^{fr} (P+Q)+R$;
  \item $P+P\sim_{hp}^{fr} P$;
  \item $P+\textbf{nil}\sim_{hp}^{fr} P$.
\end{enumerate}
\end{proposition}

\begin{proof}
\begin{enumerate}
  \item $P+Q\sim_{hp}^{fr} Q+P$. It is sufficient to prove the relation $R=\{(P+Q, Q+P)\}\cup \textbf{Id}$ is a FR strongly hp-bisimulation, we omit it;
  \item $P+(Q+R)\sim_{hp}^{fr} (P+Q)+R$. It is sufficient to prove the relation $R=\{(P+(Q+R), (P+Q)+R)\}\cup \textbf{Id}$ is a FR strongly hp-bisimulation, we omit it;
  \item $P+P\sim_{hp}^{fr} P$. It is sufficient to prove the relation $R=\{(P+P, P)\}\cup \textbf{Id}$ is a FR strongly hp-bisimulation, we omit it;
  \item $P+\textbf{nil}\sim_{hp}^{fr} P$. It is sufficient to prove the relation $R=\{(P+\textbf{nil}, P)\}\cup \textbf{Id}$ is a FR strongly hp-bisimulation, we omit it.
\end{enumerate}
\end{proof}

\begin{proposition}[Summation laws for FR strongly hhp-bisimulation] The Summation laws for FR strongly hhp-bisimulation are as follows.
\begin{enumerate}
  \item $P+Q\sim_{hhp}^{fr} Q+P$;
  \item $P+(Q+R)\sim_{hhp}^{fr} (P+Q)+R$;
  \item $P+P\sim_{hhp}^{fr} P$;
  \item $P+\textbf{nil}\sim_{hhp}^{fr} P$.
\end{enumerate}
\end{proposition}

\begin{proof}
\begin{enumerate}
  \item $P+Q\sim_{hhp}^{fr} Q+P$. It is sufficient to prove the relation $R=\{(P+Q, Q+P)\}\cup \textbf{Id}$ is a FR strongly hhp-bisimulation, we omit it;
  \item $P+(Q+R)\sim_{hhp}^{fr} (P+Q)+R$. It is sufficient to prove the relation $R=\{(P+(Q+R), (P+Q)+R)\}\cup \textbf{Id}$ is a FR strongly hhp-bisimulation, we omit it;
  \item $P+P\sim_{hhp}^{fr} P$. It is sufficient to prove the relation $R=\{(P+P, P)\}\cup \textbf{Id}$ is a FR strongly hhp-bisimulation, we omit it;
  \item $P+\textbf{nil}\sim_{hhp}^{fr} P$. It is sufficient to prove the relation $R=\{(P+\textbf{nil}, P)\}\cup \textbf{Id}$ is a FR strongly hhp-bisimulation, we omit it.
\end{enumerate}
\end{proof}

\begin{theorem}[Identity law for FR strongly truly concurrent bisimilarities]
If $A(\widetilde{x})\overset{\text{def}}{=}P$, then

\begin{enumerate}
  \item $A(\widetilde{y})\sim_p^{fr} P\{\widetilde{y}/\widetilde{x}\}$;
  \item $A(\widetilde{y})\sim_s^{fr} P\{\widetilde{y}/\widetilde{x}\}$;
  \item $A(\widetilde{y})\sim_{hp}^{fr} P\{\widetilde{y}/\widetilde{x}\}$;
  \item $A(\widetilde{y})\sim_{hhp}^{fr} P\{\widetilde{y}/\widetilde{x}\}$.
\end{enumerate}
\end{theorem}

\begin{proof}
\begin{enumerate}
  \item $A(\widetilde{y})\sim_p^{fr} P\{\widetilde{y}/\widetilde{x}\}$. It is sufficient to prove the relation $R=\{(A(\widetilde{y}), P\{\widetilde{y}/\widetilde{x}\})\}\cup \textbf{Id}$ is a FR strongly pomset bisimulation, we omit it;
  \item $A(\widetilde{y})\sim_s^{fr} P\{\widetilde{y}/\widetilde{x}\}$. It is sufficient to prove the relation $R=\{(A(\widetilde{y}), P\{\widetilde{y}/\widetilde{x}\})\}\cup \textbf{Id}$ is a FR strongly step bisimulation, we omit it;
  \item $A(\widetilde{y})\sim_{hp}^{fr} P\{\widetilde{y}/\widetilde{x}\}$. It is sufficient to prove the relation $R=\{(A(\widetilde{y}), P\{\widetilde{y}/\widetilde{x}\})\}\cup \textbf{Id}$ is a FR strongly hp-bisimulation, we omit it;
  \item $A(\widetilde{y})\sim_{hhp}^{fr} P\{\widetilde{y}/\widetilde{x}\}$. It is sufficient to prove the relation $R=\{(A(\widetilde{y}), P\{\widetilde{y}/\widetilde{x}\})\}\cup \textbf{Id}$ is a FR strongly hhp-bisimulation, we omit it.
\end{enumerate}
\end{proof}

\begin{theorem}[Restriction Laws for FR strongly pomset bisimilarity]
The restriction laws for FR strongly pomset bisimilarity are as follows.

\begin{enumerate}
  \item $(y)P\sim_p^{fr} P$, if $y\notin fn(P)$;
  \item $(y)(z)P\sim_p^{fr} (z)(y)P$;
  \item $(y)(P+Q)\sim_p^{fr} (y)P+(y)Q$;
  \item $(y)\alpha.P\sim_p^{fr} \alpha.(y)P$ if $y\notin n(\alpha)$;
  \item $(y)\alpha.P\sim_p^{fr} \textbf{nil}$ if $y$ is the subject of $\alpha$.
\end{enumerate}
\end{theorem}

\begin{proof}
\begin{enumerate}
  \item $(y)P\sim_p^{fr} P$, if $y\notin fn(P)$. It is sufficient to prove the relation $R=\{((y)P, P)\}\cup \textbf{Id}$, if $y\notin fn(P)$, is a FR strongly pomset bisimulation, we omit it;
  \item $(y)(z)P\sim_p^{fr} (z)(y)P$. It is sufficient to prove the relation $R=\{((y)(z)P, (z)(y)P)\}\cup \textbf{Id}$ is a FR strongly pomset bisimulation, we omit it;
  \item $(y)(P+Q)\sim_p^{fr} (y)P+(y)Q$. It is sufficient to prove the relation $R=\{((y)(P+Q), (y)P+(y)Q)\}\cup \textbf{Id}$ is a FR strongly pomset bisimulation, we omit it;
  \item $(y)\alpha.P\sim_p^{fr} \alpha.(y)P$ if $y\notin n(\alpha)$. It is sufficient to prove the relation $R=\{((y)\alpha.P, \alpha.(y)P)\}\cup \textbf{Id}$, if $y\notin n(\alpha)$, is a FR strongly pomset bisimulation, we omit it;
  \item $(y)\alpha.P\sim_p^{fr} \textbf{nil}$ if $y$ is the subject of $\alpha$. It is sufficient to prove the relation $R=\{((y)\alpha.P, \textbf{nil})\}\cup \textbf{Id}$, if $y$ is the subject of $\alpha$, is a FR strongly pomset bisimulation, we omit it.
\end{enumerate}
\end{proof}

\begin{theorem}[Restriction Laws for FR strongly step bisimilarity]
The restriction laws for FR strongly step bisimilarity are as follows.

\begin{enumerate}
  \item $(y)P\sim_s^{fr} P$, if $y\notin fn(P)$;
  \item $(y)(z)P\sim_s^{fr} (z)(y)P$;
  \item $(y)(P+Q)\sim_s^{fr} (y)P+(y)Q$;
  \item $(y)\alpha.P\sim_s^{fr} \alpha.(y)P$ if $y\notin n(\alpha)$;
  \item $(y)\alpha.P\sim_s^{fr} \textbf{nil}$ if $y$ is the subject of $\alpha$.
\end{enumerate}
\end{theorem}

\begin{proof}
\begin{enumerate}
  \item $(y)P\sim_s^{fr} P$, if $y\notin fn(P)$. It is sufficient to prove the relation $R=\{((y)P, P)\}\cup \textbf{Id}$, if $y\notin fn(P)$, is a FR strongly step bisimulation, we omit it;
  \item $(y)(z)P\sim_s^{fr} (z)(y)P$. It is sufficient to prove the relation $R=\{((y)(z)P, (z)(y)P)\}\cup \textbf{Id}$ is a FR strongly step bisimulation, we omit it;
  \item $(y)(P+Q)\sim_s^{fr} (y)P+(y)Q$. It is sufficient to prove the relation $R=\{((y)(P+Q), (y)P+(y)Q)\}\cup \textbf{Id}$ is a FR strongly step bisimulation, we omit it;
  \item $(y)\alpha.P\sim_s^{fr} \alpha.(y)P$ if $y\notin n(\alpha)$. It is sufficient to prove the relation $R=\{((y)\alpha.P, \alpha.(y)P)\}\cup \textbf{Id}$, if $y\notin n(\alpha)$, is a FR strongly step bisimulation, we omit it;
  \item $(y)\alpha.P\sim_s^{fr} \textbf{nil}$ if $y$ is the subject of $\alpha$. It is sufficient to prove the relation $R=\{((y)\alpha.P, \textbf{nil})\}\cup \textbf{Id}$, if $y$ is the subject of $\alpha$, is a FR strongly step bisimulation, we omit it.
\end{enumerate}
\end{proof}

\begin{theorem}[Restriction Laws for FR strongly hp-bisimilarity]
The restriction laws for FR strongly hp-bisimilarity are as follows.

\begin{enumerate}
  \item $(y)P\sim_{hp}^{fr} P$, if $y\notin fn(P)$;
  \item $(y)(z)P\sim_{hp}^{fr} (z)(y)P$;
  \item $(y)(P+Q)\sim_{hp}^{fr} (y)P+(y)Q$;
  \item $(y)\alpha.P\sim_{hp}^{fr} \alpha.(y)P$ if $y\notin n(\alpha)$;
  \item $(y)\alpha.P\sim_{hp}^{fr} \textbf{nil}$ if $y$ is the subject of $\alpha$.
\end{enumerate}
\end{theorem}

\begin{proof}
\begin{enumerate}
  \item $(y)P\sim_{hp}^{fr} P$, if $y\notin fn(P)$. It is sufficient to prove the relation $R=\{((y)P, P)\}\cup \textbf{Id}$, if $y\notin fn(P)$, is a FR strongly hp-bisimulation, we omit it;
  \item $(y)(z)P\sim_{hp}^{fr} (z)(y)P$. It is sufficient to prove the relation $R=\{((y)(z)P, (z)(y)P)\}\cup \textbf{Id}$ is a FR strongly hp-bisimulation, we omit it;
  \item $(y)(P+Q)\sim_{hp}^{fr} (y)P+(y)Q$. It is sufficient to prove the relation $R=\{((y)(P+Q), (y)P+(y)Q)\}\cup \textbf{Id}$ is a FR strongly hp-bisimulation, we omit it;
  \item $(y)\alpha.P\sim_{hp}^{fr} \alpha.(y)P$ if $y\notin n(\alpha)$. It is sufficient to prove the relation $R=\{((y)\alpha.P, \alpha.(y)P)\}\cup \textbf{Id}$, if $y\notin n(\alpha)$, is a FR strongly hp-bisimulation, we omit it;
  \item $(y)\alpha.P\sim_{hp}^{fr} \textbf{nil}$ if $y$ is the subject of $\alpha$. It is sufficient to prove the relation $R=\{((y)\alpha.P, \textbf{nil})\}\cup \textbf{Id}$, if $y$ is the subject of $\alpha$, is a FR strongly hp-bisimulation, we omit it.
\end{enumerate}
\end{proof}

\begin{theorem}[Restriction Laws for FR strongly hhp-bisimilarity]
The restriction laws for FR strongly hhp-bisimilarity are as follows.

\begin{enumerate}
  \item $(y)P\sim_{hhp}^{fr} P$, if $y\notin fn(P)$;
  \item $(y)(z)P\sim_{hhp}^{fr} (z)(y)P$;
  \item $(y)(P+Q)\sim_{hhp}^{fr} (y)P+(y)Q$;
  \item $(y)\alpha.P\sim_{hhp}^{fr} \alpha.(y)P$ if $y\notin n(\alpha)$;
  \item $(y)\alpha.P\sim_{hhp}^{fr} \textbf{nil}$ if $y$ is the subject of $\alpha$.
\end{enumerate}
\end{theorem}

\begin{proof}
\begin{enumerate}
  \item $(y)P\sim_{hhp}^{fr} P$, if $y\notin fn(P)$. It is sufficient to prove the relation $R=\{((y)P, P)\}\cup \textbf{Id}$, if $y\notin fn(P)$, is a FR strongly hhp-bisimulation, we omit it;
  \item $(y)(z)P\sim_{hhp}^{fr} (z)(y)P$. It is sufficient to prove the relation $R=\{((y)(z)P, (z)(y)P)\}\cup \textbf{Id}$ is a FR strongly hhp-bisimulation, we omit it;
  \item $(y)(P+Q)\sim_{hhp}^{fr} (y)P+(y)Q$. It is sufficient to prove the relation $R=\{((y)(P+Q), (y)P+(y)Q)\}\cup \textbf{Id}$ is a FR strongly hhp-bisimulation, we omit it;
  \item $(y)\alpha.P\sim_{hhp}^{fr} \alpha.(y)P$ if $y\notin n(\alpha)$. It is sufficient to prove the relation $R=\{((y)\alpha.P, \alpha.(y)P)\}\cup \textbf{Id}$, if $y\notin n(\alpha)$, is a FR strongly hhp-bisimulation, we omit it;
  \item $(y)\alpha.P\sim_{hhp}^{fr} \textbf{nil}$ if $y$ is the subject of $\alpha$. It is sufficient to prove the relation $R=\{((y)\alpha.P, \textbf{nil})\}\cup \textbf{Id}$, if $y$ is the subject of $\alpha$, is a FR strongly hhp-bisimulation, we omit it.
\end{enumerate}
\end{proof}

\begin{theorem}[Parallel laws for FR strongly pomset bisimilarity]
The parallel laws for FR strongly pomset bisimilarity are as follows.

\begin{enumerate}
  \item $P\parallel \textbf{nil}\sim_p^{fr} P$;
  \item $P_1\parallel P_2\sim_p^{fr} P_2\parallel P_1$;
  \item $(P_1\parallel P_2)\parallel P_3\sim_p^{fr} P_1\parallel (P_2\parallel P_3)$;
  \item $(y)(P_1\parallel P_2)\sim_p^{fr} (y)P_1\parallel (y)P_2$, if $y\notin fn(P_1)\cap fn(P_2)$.
\end{enumerate}
\end{theorem}

\begin{proof}
\begin{enumerate}
  \item $P\parallel \textbf{nil}\sim_p^{fr} P$. It is sufficient to prove the relation $R=\{(P\parallel \textbf{nil}, P)\}\cup \textbf{Id}$ is a FR strongly pomset bisimulation, we omit it;
  \item $P_1\parallel P_2\sim_p^{fr} P_2\parallel P_1$. It is sufficient to prove the relation $R=\{(P_1\parallel P_2, P_2\parallel P_1)\}\cup \textbf{Id}$ is a FR strongly pomset bisimulation, we omit it;
  \item $(P_1\parallel P_2)\parallel P_3\sim_p^{fr} P_1\parallel (P_2\parallel P_3)$. It is sufficient to prove the relation $R=\{((P_1\parallel P_2)\parallel P_3, P_1\parallel (P_2\parallel P_3))\}\cup \textbf{Id}$ is a FR strongly pomset bisimulation, we omit it;
  \item $(y)(P_1\parallel P_2)\sim_p^{fr} (y)P_1\parallel (y)P_2$, if $y\notin fn(P_1)\cap fn(P_2)$. It is sufficient to prove the relation $R=\{((y)(P_1\parallel P_2), (y)P_1\parallel (y)P_2)\}\cup \textbf{Id}$, if $y\notin fn(P_1)\cap fn(P_2)$, is a FR strongly pomset bisimulation, we omit it.
\end{enumerate}
\end{proof}

\begin{theorem}[Parallel laws for FR strongly step bisimilarity]
The parallel laws for FR strongly step bisimilarity are as follows.

\begin{enumerate}
  \item $P\parallel \textbf{nil}\sim_s^{fr} P$;
  \item $P_1\parallel P_2\sim_s^{fr} P_2\parallel P_1$;
  \item $(P_1\parallel P_2)\parallel P_3\sim_s^{fr} P_1\parallel (P_2\parallel P_3)$;
  \item $(y)(P_1\parallel P_2)\sim_s^{fr} (y)P_1\parallel (y)P_2$, if $y\notin fn(P_1)\cap fn(P_2)$.
\end{enumerate}
\end{theorem}

\begin{proof}
\begin{enumerate}
  \item $P\parallel \textbf{nil}\sim_s^{fr} P$. It is sufficient to prove the relation $R=\{(P\parallel \textbf{nil}, P)\}\cup \textbf{Id}$ is a FR strongly step bisimulation, we omit it;
  \item $P_1\parallel P_2\sim_s^{fr} P_2\parallel P_1$. It is sufficient to prove the relation $R=\{(P_1\parallel P_2, P_2\parallel P_1)\}\cup \textbf{Id}$ is a FR strongly step bisimulation, we omit it;
  \item $(P_1\parallel P_2)\parallel P_3\sim_s^{fr} P_1\parallel (P_2\parallel P_3)$. It is sufficient to prove the relation $R=\{((P_1\parallel P_2)\parallel P_3, P_1\parallel (P_2\parallel P_3))\}\cup \textbf{Id}$ is a FR strongly step bisimulation, we omit it;
  \item $(y)(P_1\parallel P_2)\sim_s^{fr} (y)P_1\parallel (y)P_2$, if $y\notin fn(P_1)\cap fn(P_2)$. It is sufficient to prove the relation $R=\{((y)(P_1\parallel P_2), (y)P_1\parallel (y)P_2)\}\cup \textbf{Id}$, if $y\notin fn(P_1)\cap fn(P_2)$, is a FR strongly step bisimulation, we omit it.
\end{enumerate}
\end{proof}

\begin{theorem}[Parallel laws for FR strongly hp-bisimilarity]
The parallel laws for FR strongly hp-bisimilarity are as follows.

\begin{enumerate}
  \item $P\parallel \textbf{nil}\sim_{hp}^{fr} P$;
  \item $P_1\parallel P_2\sim_{hp}^{fr} P_2\parallel P_1$;
  \item $(P_1\parallel P_2)\parallel P_3\sim_{hp}^{fr} P_1\parallel (P_2\parallel P_3)$;
  \item $(y)(P_1\parallel P_2)\sim_{hp}^{fr} (y)P_1\parallel (y)P_2$, if $y\notin fn(P_1)\cap fn(P_2)$.
\end{enumerate}
\end{theorem}

\begin{proof}
\begin{enumerate}
  \item $P\parallel \textbf{nil}\sim_{hp}^{fr} P$. It is sufficient to prove the relation $R=\{(P\parallel \textbf{nil}, P)\}\cup \textbf{Id}$ is a FR strongly hp-bisimulation, we omit it;
  \item $P_1\parallel P_2\sim_{hp}^{fr} P_2\parallel P_1$. It is sufficient to prove the relation $R=\{(P_1\parallel P_2, P_2\parallel P_1)\}\cup \textbf{Id}$ is a FR strongly hp-bisimulation, we omit it;
  \item $(P_1\parallel P_2)\parallel P_3\sim_{hp}^{fr} P_1\parallel (P_2\parallel P_3)$. It is sufficient to prove the relation $R=\{((P_1\parallel P_2)\parallel P_3, P_1\parallel (P_2\parallel P_3))\}\cup \textbf{Id}$ is a FR strongly hp-bisimulation, we omit it;
  \item $(y)(P_1\parallel P_2)\sim_{hp}^{fr} (y)P_1\parallel (y)P_2$, if $y\notin fn(P_1)\cap fn(P_2)$. It is sufficient to prove the relation $R=\{((y)(P_1\parallel P_2), (y)P_1\parallel (y)P_2)\}\cup \textbf{Id}$, if $y\notin fn(P_1)\cap fn(P_2)$, is a FR strongly hp-bisimulation, we omit it.
\end{enumerate}
\end{proof}

\begin{theorem}[Parallel laws for FR strongly hhp-bisimilarity]
The parallel laws for FR strongly hhp-bisimilarity are as follows.

\begin{enumerate}
  \item $P\parallel \textbf{nil}\sim_{hhp}^{fr} P$;
  \item $P_1\parallel P_2\sim_{hhp}^{fr} P_2\parallel P_1$;
  \item $(P_1\parallel P_2)\parallel P_3\sim_{hhp}^{fr} P_1\parallel (P_2\parallel P_3)$;
  \item $(y)(P_1\parallel P_2)\sim_{hhp}^{fr} (y)P_1\parallel (y)P_2$, if $y\notin fn(P_1)\cap fn(P_2)$.
\end{enumerate}
\end{theorem}

\begin{proof}
\begin{enumerate}
  \item $P\parallel \textbf{nil}\sim_{hhp}^{fr} P$. It is sufficient to prove the relation $R=\{(P\parallel \textbf{nil}, P)\}\cup \textbf{Id}$ is a FR strongly hhp-bisimulation, we omit it;
  \item $P_1\parallel P_2\sim_{hhp}^{fr} P_2\parallel P_1$. It is sufficient to prove the relation $R=\{(P_1\parallel P_2, P_2\parallel P_1)\}\cup \textbf{Id}$ is a FR strongly hhp-bisimulation, we omit it;
  \item $(P_1\parallel P_2)\parallel P_3\sim_{hhp}^{fr} P_1\parallel (P_2\parallel P_3)$. It is sufficient to prove the relation $R=\{((P_1\parallel P_2)\parallel P_3, P_1\parallel (P_2\parallel P_3))\}\cup \textbf{Id}$ is a FR strongly hhp-bisimulation, we omit it;
  \item $(y)(P_1\parallel P_2)\sim_{hhp}^{fr} (y)P_1\parallel (y)P_2$, if $y\notin fn(P_1)\cap fn(P_2)$. It is sufficient to prove the relation $R=\{((y)(P_1\parallel P_2), (y)P_1\parallel (y)P_2)\}\cup \textbf{Id}$, if $y\notin fn(P_1)\cap fn(P_2)$, is a FR strongly hhp-bisimulation, we omit it.
\end{enumerate}
\end{proof}

\begin{theorem}[Expansion law for truly concurrent bisimilarities]
Let $P\equiv\sum_i \alpha_{i}.P_{i}$ and $Q\equiv\sum_j\beta_{j}.Q_{j}$, where $bn(\alpha_{i})\cap fn(Q)=\emptyset$ for all $i$, and
  $bn(\beta_{j})\cap fn(P)=\emptyset$ for all $j$. Then,

\begin{enumerate}
  \item $P\parallel Q\sim_p^{fr} \sum_i\sum_j (\alpha_{i}\parallel \beta_{j}).(P_{i}\parallel Q_{j})+\sum_{\alpha_{i} \textrm{ comp }\beta_{j}}\tau.R_{ij}$;
  \item $P\parallel Q\sim_s^{fr} \sum_i\sum_j (\alpha_{i}\parallel \beta_{j}).(P_{i}\parallel Q_{j})+\sum_{\alpha_{i} \textrm{ comp }\beta_{j}}\tau.R_{ij}$;
  \item $P\parallel Q\sim_{hp}^{fr} \sum_i\sum_j (\alpha_{i}\parallel \beta_{j}).(P_{i}\parallel Q_{j})+\sum_{\alpha_{i} \textrm{ comp }\beta_{j}}\tau.R_{ij}$;
  \item $P\parallel Q\nsim_{phhp} \sum_i\sum_j (\alpha_{i}\parallel \beta_{j}).(P_{i}\parallel Q_{j})+\sum_{\alpha_{i} \textrm{ comp }\beta_{j}}\tau.R_{ij}$.
\end{enumerate}

Where $\alpha_i$ comp $\beta_j$ and $R_{ij}$ are defined as follows:
\begin{enumerate}
  \item $\alpha_{i}$ is $\overline{x}u$ and $\beta_{j}$ is $x(v)$, then $R_{ij}=P_{i}\parallel Q_{j}\{u/v\}$;
  \item $\alpha_{i}$ is $\overline{x}(u)$ and $\beta_{j}$ is $x(v)$, then $R_{ij}=(w)(P_{i}\{w/u\}\parallel Q_{j}\{w/v\})$, if $w\notin fn((u)P_{i})\cup fn((v)Q_{j})$;
  \item $\alpha_{i}$ is $x(v)$ and $\beta_{j}$ is $\overline{x}u$, then $R_{ij}=P_{i}\{u/v\}\parallel Q_{j}$;
  \item $\alpha_{i}$ is $x(v)$ and $\beta_{j}$ is $\overline{x}(u)$, then $R_{ij}=(w)(P_{i}\{w/v\}\parallel Q_{j}\{w/u\})$, if $w\notin fn((v)P_{i})\cup fn((u)Q_{j})$.
\end{enumerate}

Let $P\equiv\sum_i P_{i}.\alpha_{i}[m]$ and $Q\equiv\sum_l Q_{j}.\beta_{j}[m]$, where $bn(\alpha_{i}[m])\cap fn(Q)=\emptyset$ for all $i$, and
  $bn(\beta_{j}[m])\cap fn(P)=\emptyset$ for all $j$. Then,

\begin{enumerate}
  \item $P\parallel Q\sim_p^{fr} \sum_i\sum_j(P_{i}\parallel Q_{j}).(\alpha_{i}[m]\parallel \beta_{j}[m])+\sum_{\alpha_{i} \textrm{ comp }\beta_{j}} R_{ij}.\tau$;
  \item $P\parallel Q\sim_s^{fr} \sum_i\sum_j(P_{i}\parallel Q_{j}).(\alpha_{i}[m]\parallel \beta_{j}[m])+\sum_{\alpha_{i} \textrm{ comp }\beta_{j}} R_{ij}.\tau$;
  \item $P\parallel Q\sim_{hp}^{fr} \sum_i\sum_j(P_{i}\parallel Q_{j}).(\alpha_{i}[m]\parallel \beta_{j}[m])+\sum_{\alpha_{i} \textrm{ comp }\beta_{j}} R_{ij}.\tau$;
  \item $P\parallel Q\nsim_{phhp} \sum_i\sum_j(P_{i}\parallel Q_{j}).(\alpha_{i}[m]\parallel \beta_{j}[m])+\sum_{\alpha_{i} \textrm{ comp }\beta_{j}} R_{ij}.\tau$.
\end{enumerate}

Where $\alpha_i$ comp $\beta_j$ and $R_{ij}$ are defined as follows:
\begin{enumerate}
  \item $\alpha_{i}[m]$ is $\overline{x}u$ and $\beta_{j}[m]$ is $x(v)$, then $R_{ij}=P_{i}\parallel Q_{j}\{u/v\}$;
  \item $\alpha_{i}[m]$ is $\overline{x}(u)$ and $\beta_{j}[m]$ is $x(v)$, then $R_{ij}=(w)(P_{i}\{w/u\}\parallel Q_{j}\{w/v\})$, if $w\notin fn((u)P_{i})\cup fn((v)Q_{j})$;
  \item $\alpha_{i}[m]$ is $x(v)$ and $\beta_{j}[m]$ is $\overline{x}u$, then $R_{ij}=P_{i}\{u/v\}\parallel Q_{j}$;
  \item $\alpha_{i}[m]$ is $x(v)$ and $\beta_{j}[m]$ is $\overline{x}(u)$, then $R_{ij}=(w)(P_{i}\{w/v\}\parallel Q_{j}\{w/u\})$, if $w\notin fn((v)P_{i})\cup fn((u)Q_{j})$.
\end{enumerate}
\end{theorem}

\begin{proof}
According to the definition of FR strongly truly concurrent bisimulations, we can easily prove the above equations, and we omit the proof.
\end{proof}

\begin{theorem}[Equivalence and congruence for FR strongly pomset bisimilarity]
We can enjoy the full congruence modulo FR strongly pomset bisimilarity.

\begin{enumerate}
  \item $\sim_p^{fr}$ is an equivalence relation;
  \item If $P\sim_p^{fr} Q$ then
  \begin{enumerate}
    \item $\alpha.P\sim_p^{f} \alpha.Q$, $\alpha$ is a free action;
    \item $P.\alpha[m]\sim_p^{r}Q.\alpha[m]$, $\alpha[m]$ is a free action;
    \item $\phi.P\sim_p^{f} \phi.Q$;
    \item $P.\phi\sim_p^{r}Q.\phi$;
    \item $P+R\sim_p^{fr} Q+R$;
    \item $P\parallel R\sim_p^{fr} Q\parallel R$;
    \item $(w)P\sim_p^{fr} (w)Q$;
    \item $x(y).P\sim_p^{f} x(y).Q$;
    \item $P.x(y)[m]\sim_p^{r}Q.x(y)[m]$.
  \end{enumerate}
\end{enumerate}
\end{theorem}

\begin{proof}
\begin{enumerate}
  \item $\sim_p^{fr}$ is an equivalence relation, it is obvious;
  \item If $P\sim_p^{fr} Q$ then
  \begin{enumerate}
    \item $\alpha.P\sim_p^{f} \alpha.Q$, $\alpha$ is a free action. It is sufficient to prove the relation $R=\{(\alpha.P, \alpha.Q)\}\cup \textbf{Id}$ is a F strongly pomset bisimulation, we omit it;
    \item $P.\alpha[m]\sim_p^{r}Q.\alpha[m]$, $\alpha[m]$ is a free action. It is sufficient to prove the relation $R=\{(P.\alpha[m], Q.\alpha[m])\}\cup \textbf{Id}$ is a R strongly pomset bisimulation, we omit it;
    \item $\phi.P\sim_p^{f} \phi.Q$. It is sufficient to prove the relation $R=\{(\phi.P, \phi.Q)\}\cup \textbf{Id}$ is a F strongly pomset bisimulation, we omit it;
    \item $P.\phi\sim_p^{r}Q.\phi$. It is sufficient to prove the relation $R=\{(P.\phi, Q.\phi)\}\cup \textbf{Id}$ is a R strongly pomset bisimulation, we omit it;
    \item $P+R\sim_p^{fr} Q+R$. It is sufficient to prove the relation $R=\{(P+R, Q+R)\}\cup \textbf{Id}$ is a FR strongly pomset bisimulation, we omit it;
    \item $P\parallel R\sim_p^{fr} Q\parallel R$. It is sufficient to prove the relation $R=\{(P\parallel R, Q\parallel R)\}\cup \textbf{Id}$ is a FR strongly pomset bisimulation, we omit it;
    \item $(w)P\sim_p^{fr} (w)Q$. It is sufficient to prove the relation $R=\{((w)P, (w)Q)\}\cup \textbf{Id}$ is a FR strongly pomset bisimulation, we omit it;
    \item $x(y).P\sim_p^{f} x(y).Q$. It is sufficient to prove the relation $R=\{(x(y).P, x(y).Q)\}\cup \textbf{Id}$ is a F strongly pomset bisimulation, we omit it;
    \item $P.x(y)[m]\sim_p^{r}Q.x(y)[m]$. It is sufficient to prove the relation $R=\{(P.x(y)[m], Q.x(y)[m])\}\cup \textbf{Id}$ is a R strongly pomset bisimulation, we omit it.
  \end{enumerate}
\end{enumerate}
\end{proof}

\begin{theorem}[Equivalence and congruence for FR strongly step bisimilarity]
We can enjoy the full congruence modulo FR strongly step bisimilarity.

\begin{enumerate}
  \item $\sim_s^{fr}$ is an equivalence relation;
  \item If $P\sim_s^{fr} Q$ then
  \begin{enumerate}
    \item $\alpha.P\sim_s^{f} \alpha.Q$, $\alpha$ is a free action;
    \item $P.\alpha[m]\sim_s^{r}Q.\alpha[m]$, $\alpha[m]$ is a free action;
    \item $\phi.P\sim_s^{f} \phi.Q$;
    \item $P.\phi\sim_s^{r}Q.\phi$;
    \item $P+R\sim_s^{fr} Q+R$;
    \item $P\parallel R\sim_s^{fr} Q\parallel R$;
    \item $(w)P\sim_s^{fr} (w)Q$;
    \item $x(y).P\sim_s^{f} x(y).Q$;
    \item $P.x(y)[m]\sim_s^{r}Q.x(y)[m]$.
  \end{enumerate}
\end{enumerate}
\end{theorem}

\begin{proof}
\begin{enumerate}
  \item $\sim_s^{fr}$ is an equivalence relation, it is obvious;
  \item If $P\sim_s^{fr} Q$ then
  \begin{enumerate}
    \item $\alpha.P\sim_s^{f} \alpha.Q$, $\alpha$ is a free action. It is sufficient to prove the relation $R=\{(\alpha.P, \alpha.Q)\}\cup \textbf{Id}$ is a F strongly step bisimulation, we omit it;
    \item $P.\alpha[m]\sim_s^{r}Q.\alpha[m]$, $\alpha[m]$ is a free action. It is sufficient to prove the relation $R=\{(P.\alpha[m], Q.\alpha[m])\}\cup \textbf{Id}$ is a R strongly step bisimulation, we omit it;
    \item $\phi.P\sim_s^{f} \phi.Q$. It is sufficient to prove the relation $R=\{(\phi.P, \phi.Q)\}\cup \textbf{Id}$ is a F strongly step bisimulation, we omit it;
    \item $P.\phi\sim_s^{r}Q.\phi$. It is sufficient to prove the relation $R=\{(P.\phi, Q.\phi)\}\cup \textbf{Id}$ is a R strongly step bisimulation, we omit it;
    \item $P+R\sim_s^{fr} Q+R$. It is sufficient to prove the relation $R=\{(P+R, Q+R)\}\cup \textbf{Id}$ is a FR strongly step bisimulation, we omit it;
    \item $P\parallel R\sim_s^{fr} Q\parallel R$. It is sufficient to prove the relation $R=\{(P\parallel R, Q\parallel R)\}\cup \textbf{Id}$ is a FR strongly step bisimulation, we omit it;
    \item $(w)P\sim_s^{fr} (w)Q$. It is sufficient to prove the relation $R=\{((w)P, (w)Q)\}\cup \textbf{Id}$ is a FR strongly step bisimulation, we omit it;
    \item $x(y).P\sim_s^{f} x(y).Q$. It is sufficient to prove the relation $R=\{(x(y).P, x(y).Q)\}\cup \textbf{Id}$ is a F strongly step bisimulation, we omit it;
    \item $P.x(y)[m]\sim_s^{r}Q.x(y)[m]$. It is sufficient to prove the relation $R=\{(P.x(y)[m], Q.x(y)[m])\}\cup \textbf{Id}$ is a R strongly step bisimulation, we omit it.
  \end{enumerate}
\end{enumerate}
\end{proof}

\begin{theorem}[Equivalence and congruence for FR strongly hp-bisimilarity]
We can enjoy the full congruence modulo FR strongly hp-bisimilarity.

\begin{enumerate}
  \item $\sim_{hp}^{fr}$ is an equivalence relation;
  \item If $P\sim_{hp}^{fr} Q$ then
  \begin{enumerate}
    \item $\alpha.P\sim_{hp}^{f} \alpha.Q$, $\alpha$ is a free action;
    \item $P.\alpha[m]\sim_{hp}^{r}Q.\alpha[m]$, $\alpha[m]$ is a free action;
    \item $\phi.P\sim_{hp}^{f} \phi.Q$;
    \item $P.\phi\sim_{hp}^{r}Q.\phi$;
    \item $P+R\sim_{hp}^{fr} Q+R$;
    \item $P\parallel R\sim_{hp}^{fr} Q\parallel R$;
    \item $(w)P\sim_{hp}^{fr} (w)Q$;
    \item $x(y).P\sim_{hp}^{f} x(y).Q$;
    \item $P.x(y)[m]\sim_{hp}^{r}Q.x(y)[m]$.
  \end{enumerate}
\end{enumerate}
\end{theorem}

\begin{proof}
\begin{enumerate}
  \item $\sim_{hp}^{fr}$ is an equivalence relation, it is obvious;
  \item If $P\sim_{hp}^{fr} Q$ then
  \begin{enumerate}
    \item $\alpha.P\sim_{hp}^{f} \alpha.Q$, $\alpha$ is a free action. It is sufficient to prove the relation $R=\{(\alpha.P, \alpha.Q)\}\cup \textbf{Id}$ is a F strongly hp-bisimulation, we omit it;
    \item $P.\alpha[m]\sim_{hp}^{r}Q.\alpha[m]$, $\alpha[m]$ is a free action. It is sufficient to prove the relation $R=\{(P.\alpha[m], Q.\alpha[m])\}\cup \textbf{Id}$ is a R strongly hp-bisimulation, we omit it;
    \item $\phi.P\sim_{hp}^{f} \phi.Q$. It is sufficient to prove the relation $R=\{(\phi.P, \phi.Q)\}\cup \textbf{Id}$ is a F strongly hp-bisimulation, we omit it;
    \item $P.\phi\sim_{hp}^{r}Q.\phi$. It is sufficient to prove the relation $R=\{(P.\phi, Q.\phi)\}\cup \textbf{Id}$ is a R strongly hp-bisimulation, we omit it;
    \item $P+R\sim_{hp}^{fr} Q+R$. It is sufficient to prove the relation $R=\{(P+R, Q+R)\}\cup \textbf{Id}$ is a FR strongly hp-bisimulation, we omit it;
    \item $P\parallel R\sim_{hp}^{fr} Q\parallel R$. It is sufficient to prove the relation $R=\{(P\parallel R, Q\parallel R)\}\cup \textbf{Id}$ is a FR strongly hp-bisimulation, we omit it;
    \item $(w)P\sim_{hp}^{fr} (w)Q$. It is sufficient to prove the relation $R=\{((w)P, (w)Q)\}\cup \textbf{Id}$ is a FR strongly hp-bisimulation, we omit it;
    \item $x(y).P\sim_{hp}^{f} x(y).Q$. It is sufficient to prove the relation $R=\{(x(y).P, x(y).Q)\}\cup \textbf{Id}$ is a F strongly hp-bisimulation, we omit it;
    \item $P.x(y)[m]\sim_{hp}^{r}Q.x(y)[m]$. It is sufficient to prove the relation $R=\{(P.x(y)[m], Q.x(y)[m])\}\cup \textbf{Id}$ is a R strongly hp-bisimulation, we omit it.
  \end{enumerate}
\end{enumerate}
\end{proof}

\begin{theorem}[Equivalence and congruence for FR strongly hhp-bisimilarity]
We can enjoy the full congruence modulo FR strongly hhp-bisimilarity.

\begin{enumerate}
  \item $\sim_{hhp}^{fr}$ is an equivalence relation;
  \item If $P\sim_{hhp}^{fr} Q$ then
  \begin{enumerate}
    \item $\alpha.P\sim_{hhp}^{f} \alpha.Q$, $\alpha$ is a free action;
    \item $P.\alpha[m]\sim_{hhp}^{r}Q.\alpha[m]$, $\alpha[m]$ is a free action;
    \item $\phi.P\sim_{hhp}^{f} \phi.Q$;
    \item $P.\phi\sim_{hhp}^{r}Q.\phi$;
    \item $P+R\sim_{hhp}^{fr} Q+R$;
    \item $P\parallel R\sim_{hhp}^{fr} Q\parallel R$;
    \item $(w)P\sim_{hhp}^{fr} (w)Q$;
    \item $x(y).P\sim_{hhp}^{f} x(y).Q$;
    \item $P.x(y)[m]\sim_{hhp}^{r}Q.x(y)[m]$.
  \end{enumerate}
\end{enumerate}
\end{theorem}

\begin{proof}
\begin{enumerate}
  \item $\sim_{hhp}^{fr}$ is an equivalence relation, it is obvious;
  \item If $P\sim_{hhp}^{fr} Q$ then
  \begin{enumerate}
    \item $\alpha.P\sim_{hhp}^{f} \alpha.Q$, $\alpha$ is a free action. It is sufficient to prove the relation $R=\{(\alpha.P, \alpha.Q)\}\cup \textbf{Id}$ is a F strongly hhp-bisimulation, we omit it;
    \item $P.\alpha[m]\sim_{hhp}^{r}Q.\alpha[m]$, $\alpha[m]$ is a free action. It is sufficient to prove the relation $R=\{(P.\alpha[m], Q.\alpha[m])\}\cup \textbf{Id}$ is a R strongly hhp-bisimulation, we omit it;
    \item $\phi.P\sim_{hhp}^{f} \phi.Q$. It is sufficient to prove the relation $R=\{(\phi.P, \phi.Q)\}\cup \textbf{Id}$ is a F strongly hhp-bisimulation, we omit it;
    \item $P.\phi\sim_{hhp}^{r}Q.\phi$. It is sufficient to prove the relation $R=\{(P.\phi, Q.\phi)\}\cup \textbf{Id}$ is a R strongly hhp-bisimulation, we omit it;
    \item $P+R\sim_{hhp}^{fr} Q+R$. It is sufficient to prove the relation $R=\{(P+R, Q+R)\}\cup \textbf{Id}$ is a FR strongly hhp-bisimulation, we omit it;
    \item $P\parallel R\sim_{hhp}^{fr} Q\parallel R$. It is sufficient to prove the relation $R=\{(P\parallel R, Q\parallel R)\}\cup \textbf{Id}$ is a FR strongly hhp-bisimulation, we omit it;
    \item $(w)P\sim_{hhp}^{fr} (w)Q$. It is sufficient to prove the relation $R=\{((w)P, (w)Q)\}\cup \textbf{Id}$ is a FR strongly hhp-bisimulation, we omit it;
    \item $x(y).P\sim_{hhp}^{f} x(y).Q$. It is sufficient to prove the relation $R=\{(x(y).P, x(y).Q)\}\cup \textbf{Id}$ is a F strongly hhp-bisimulation, we omit it;
    \item $P.x(y)[m]\sim_{hhp}^{r}Q.x(y)[m]$. It is sufficient to prove the relation $R=\{(P.x(y)[m], Q.x(y)[m])\}\cup \textbf{Id}$ is a R strongly hhp-bisimulation, we omit it.
  \end{enumerate}
\end{enumerate}
\end{proof}

\subsubsection{Recursion}

\begin{definition}
Let $X$ have arity $n$, and let $\widetilde{x}=x_1,\cdots,x_n$ be distinct names, and $fn(P)\subseteq\{x_1,\cdots,x_n\}$. The replacement of $X(\widetilde{x})$ by $P$ in $E$, written
$E\{X(\widetilde{x}):=P\}$, means the result of replacing each subterm $X(\widetilde{y})$ in $E$ by $P\{\widetilde{y}/\widetilde{x}\}$.
\end{definition}

\begin{definition}
Let $E$ and $F$ be two process expressions containing only $X_1,\cdots,X_m$ with associated name sequences $\widetilde{x}_1,\cdots,\widetilde{x}_m$. Then,
\begin{enumerate}
  \item $E\sim_p^{fr} F$ means $E(\widetilde{P})\sim_p^{fr} F(\widetilde{P})$;
  \item $E\sim_s^{fr} F$ means $E(\widetilde{P})\sim_s^{fr} F(\widetilde{P})$;
  \item $E\sim_{hp}^{fr} F$ means $E(\widetilde{P})\sim_{hp}^{fr} F(\widetilde{P})$;
  \item $E\sim_{hhp}^{fr} F$ means $E(\widetilde{P})\sim_{hhp}^{fr} F(\widetilde{P})$;
\end{enumerate}

for all $\widetilde{P}$ such that $fn(P_i)\subseteq \widetilde{x}_i$ for each $i$.
\end{definition}

\begin{definition}
A term or identifier is weakly guarded in $P$ if it lies within some subterm $\alpha.Q$ or $Q.\alpha[m]$ or $(\alpha_1\parallel\cdots\parallel \alpha_n).Q$ or
$Q.(\alpha_1[m]\parallel\cdots\parallel \alpha_n[m])$ of $P$.
\end{definition}

\begin{theorem}
Assume that $\widetilde{E}$ and $\widetilde{F}$ are expressions containing only $X_i$ with $\widetilde{x}_i$, and $\widetilde{A}$ and $\widetilde{B}$ are identifiers with $A_i$, $B_i$. Then, for all $i$,
\begin{enumerate}
  \item $E_i\sim_s^{fr} F_i$, $A_i(\widetilde{x}_i)\overset{\text{def}}{=}E_i(\widetilde{A})$, $B_i(\widetilde{x}_i)\overset{\text{def}}{=}F_i(\widetilde{B})$, then
  $A_i(\widetilde{x}_i)\sim_s^{fr} B_i(\widetilde{x}_i)$;
  \item $E_i\sim_p^{fr} F_i$, $A_i(\widetilde{x}_i)\overset{\text{def}}{=}E_i(\widetilde{A})$, $B_i(\widetilde{x}_i)\overset{\text{def}}{=}F_i(\widetilde{B})$, then
  $A_i(\widetilde{x}_i)\sim_p^{fr} B_i(\widetilde{x}_i)$;
  \item $E_i\sim_{hp}^{fr} F_i$, $A_i(\widetilde{x}_i)\overset{\text{def}}{=}E_i(\widetilde{A})$, $B_i(\widetilde{x}_i)\overset{\text{def}}{=}F_i(\widetilde{B})$, then
  $A_i(\widetilde{x}_i)\sim_{hp}^{fr} B_i(\widetilde{x}_i)$;
  \item $E_i\sim_{hhp}^{fr} F_i$, $A_i(\widetilde{x}_i)\overset{\text{def}}{=}E_i(\widetilde{A})$, $B_i(\widetilde{x}_i)\overset{\text{def}}{=}F_i(\widetilde{B})$, then
  $A_i(\widetilde{x}_i)\sim_{hhp}^{fr} B_i(\widetilde{x}_i)$.
\end{enumerate}
\end{theorem}

\begin{proof}
\begin{enumerate}
  \item $E_i\sim_s^{fr} F_i$, $A_i(\widetilde{x}_i)\overset{\text{def}}{=}E_i(\widetilde{A})$, $B_i(\widetilde{x}_i)\overset{\text{def}}{=}F_i(\widetilde{B})$, then
  $A_i(\widetilde{x}_i)\sim_s^{fr} B_i(\widetilde{x}_i)$.

      We will consider the case $I=\{1\}$ with loss of generality, and show the following relation $R$ is a FR strongly step bisimulation.

      $$R=\{(G(A),G(B)):G\textrm{ has only identifier }X\}.$$

      By choosing $G\equiv X(\widetilde{y})$, it follows that $A(\widetilde{y})\sim_s^{fr} B(\widetilde{y})$. It is sufficient to prove the following:
      \begin{enumerate}
        \item If $\langle G(A),s\rangle\xrightarrow{\{\alpha_1,\cdots,\alpha_n\}}\langle P',s'\rangle$, where $\alpha_i(1\leq i\leq n)$ is a free action or bound output action with
        $bn(\alpha_1)\cap\cdots\cap bn(\alpha_n)\cap n(G(A),G(B))=\emptyset$, then $\langle G(B),s\rangle\xrightarrow{\{\alpha_1,\cdots,\alpha_n\}}\langle Q'',s''\rangle$ such that $P'\sim_s^{fr} Q''$;
        \item If $\langle G(A),s\rangle\xrightarrow{x(y)}\langle P',s'\rangle$ with $x\notin n(G(A),G(B))$, then $\langle G(B),s\rangle\xrightarrow{x(y)}\langle Q'',s''\rangle$, such that for all $u$,
        $\langle P',s'\rangle\{u/y\}\sim_s^{fr} \langle Q''\{u/y\},s''\rangle$;
        \item If $\langle G(A),s\rangle\xtworightarrow{\{\alpha_1[m],\cdots,\alpha_n[m]\}}\langle P',s'\rangle$, where $\alpha_i[m](1\leq i\leq n)$ is a free action or bound output action with
        $bn(\alpha_1[m])\cap\cdots\cap bn(\alpha_n[m])\cap n(G(A),G(B))=\emptyset$, then $\langle G(B),s\rangle\xtworightarrow{\{\alpha_1[m],\cdots,\alpha_n[m]\}}\langle Q'',s''\rangle$ such that $P'\sim_s^{fr} Q''$;
        \item If $\langle G(A),s\rangle\xtworightarrow{x(y)[m]}\langle P',s'\rangle$ with $x\notin n(G(A),G(B))$, then $\langle G(B),s\rangle\xtworightarrow{x(y)[m]}\langle Q'',s''\rangle$, such that for all $u$,
        $P'\{u/y\}\sim_s^{fr} Q''\{u/y\}$.
      \end{enumerate}

      To prove the above properties, it is sufficient to induct on the depth of inference and quite routine, we omit it.
  \item $E_i\sim_p^{fr} F_i$, $A_i(\widetilde{x}_i)\overset{\text{def}}{=}E_i(\widetilde{A})$, $B_i(\widetilde{x}_i)\overset{\text{def}}{=}F_i(\widetilde{B})$, then
  $A_i(\widetilde{x}_i)\sim_p^{fr} B_i(\widetilde{x}_i)$. It can be proven similarly to the above case.
  \item $E_i\sim_{hp}^{fr} F_i$, $A_i(\widetilde{x}_i)\overset{\text{def}}{=}E_i(\widetilde{A})$, $B_i(\widetilde{x}_i)\overset{\text{def}}{=}F_i(\widetilde{B})$, then
  $A_i(\widetilde{x}_i)\sim_{hp}^{fr} B_i(\widetilde{x}_i)$. It can be proven similarly to the above case.
  \item $E_i\sim_{hhp}^{fr} F_i$, $A_i(\widetilde{x}_i)\overset{\text{def}}{=}E_i(\widetilde{A})$, $B_i(\widetilde{x}_i)\overset{\text{def}}{=}F_i(\widetilde{B})$, then
  $A_i(\widetilde{x}_i)\sim_{hhp}^{fr} B_i(\widetilde{x}_i)$. It can be proven similarly to the above case.
\end{enumerate}
\end{proof}

\begin{theorem}[Unique solution of equations]
Assume $\widetilde{E}$ are expressions containing only $X_i$ with $\widetilde{x}_i$, and each $X_i$ is weakly guarded in each $E_j$. Assume that $\widetilde{P}$ and $\widetilde{Q}$ are
processes such that $fn(P_i)\subseteq \widetilde{x}_i$ and $fn(Q_i)\subseteq \widetilde{x}_i$. Then, for all $i$,
\begin{enumerate}
  \item if $P_i\sim_p^{fr} E_i(\widetilde{P})$, $Q_i\sim_p^{fr} E_i(\widetilde{Q})$, then $P_i\sim_p^{fr} Q_i$;
  \item if $P_i\sim_s^{fr} E_i(\widetilde{P})$, $Q_i\sim_s^{fr} E_i(\widetilde{Q})$, then $P_i\sim_s^{fr} Q_i$;
  \item if $P_i\sim_{hp}^{fr} E_i(\widetilde{P})$, $Q_i\sim_{hp}^{fr} E_i(\widetilde{Q})$, then $P_i\sim_{hp}^{fr} Q_i$;
  \item if $P_i\sim_{hhp}^{fr} E_i(\widetilde{P})$, $Q_i\sim_{hhp}^{fr} E_i(\widetilde{Q})$, then $P_i\sim_{hhp}^{fr} Q_i$.
\end{enumerate}
\end{theorem}

\begin{proof}
\begin{enumerate}
  \item It is similar to the proof of unique solution of equations for FR strongly pomset bisimulation in CTC, please refer to \cite{CTC2} for details, we omit it;
  \item It is similar to the proof of unique solution of equations for FR strongly step bisimulation in CTC, please refer to \cite{CTC2} for details, we omit it;
  \item It is similar to the proof of unique solution of equations for FR strongly hp-bisimulation in CTC, please refer to \cite{CTC2} for details, we omit it;
  \item It is similar to the proof of unique solution of equations for FR strongly hhp-bisimulation in CTC, please refer to \cite{CTC2} for details, we omit it.
\end{enumerate}
\end{proof}

\subsection{Algebraic Theory}\label{a7}

\begin{definition}[STC]
The theory \textbf{STC} is consisted of the following axioms and inference rules:

\begin{enumerate}
  \item Alpha-conversion $\textbf{A}$.
  \[\textrm{if } P\equiv Q, \textrm{ then } P=Q\]
  \item Congruence $\textbf{C}$. If $P=Q$, then,
  \[\tau.P=\tau.Q\quad \overline{x}y.P=\overline{x}y.Q\quad P.\overline{x}y[m]=Q.\overline{x}y[m]\]
  \[P+R=Q+R\quad P\parallel R=Q\parallel R\]
  \[(x)P=(x)Q\quad x(y).P=x(y).Q\quad P.x(y)[m]=Q.x(y)[m]\]
  \item Summation $\textbf{S}$.
  \[\textbf{S0}\quad P+\textbf{nil}=P\]
  \[\textbf{S1}\quad P+P=P\]
  \[\textbf{S2}\quad P+Q=Q+P\]
  \[\textbf{S3}\quad P+(Q+R)=(P+Q)+R\]
  \item Restriction $\textbf{R}$.
  \[\textbf{R0}\quad (x)P=P\quad \textrm{ if }x\notin fn(P)\]
  \[\textbf{R1}\quad (x)(y)P=(y)(x)P\]
  \[\textbf{R2}\quad (x)(P+Q)=(x)P+(x)Q\]
  \[\textbf{R3}\quad (x)\alpha.P=\alpha.(x)P\quad \textrm{ if }x\notin n(\alpha)\]
  \[\textbf{R4}\quad (x)\alpha.P=\textbf{nil}\quad \textrm{ if }x\textrm{is the subject of }\alpha\]
  \item Expansion $\textbf{E}$.
  Let $P\equiv\sum_i \alpha_{i}.P_{i}$ and $Q\equiv\sum_j\beta_{j}.Q_{j}$, where $bn(\alpha_{i})\cap fn(Q)=\emptyset$ for all $i$, and
  $bn(\beta_{j})\cap fn(P)=\emptyset$ for all $j$. Then,

\begin{enumerate}
  \item $P\parallel Q\sim_p^{fr} \sum_i\sum_j (\alpha_{i}\parallel \beta_{j}).(P_{i}\parallel Q_{j})+\sum_{\alpha_{i} \textrm{ comp }\beta_{j}}\tau.R_{ij}$;
  \item $P\parallel Q\sim_s^{fr} \sum_i\sum_j (\alpha_{i}\parallel \beta_{j}).(P_{i}\parallel Q_{j})+\sum_{\alpha_{i} \textrm{ comp }\beta_{j}}\tau.R_{ij}$;
  \item $P\parallel Q\sim_{hp}^{fr} \sum_i\sum_j (\alpha_{i}\parallel \beta_{j}).(P_{i}\parallel Q_{j})+\sum_{\alpha_{i} \textrm{ comp }\beta_{j}}\tau.R_{ij}$;
  \item $P\parallel Q\nsim_{phhp} \sum_i\sum_j (\alpha_{i}\parallel \beta_{j}).(P_{i}\parallel Q_{j})+\sum_{\alpha_{i} \textrm{ comp }\beta_{j}}\tau.R_{ij}$.
\end{enumerate}

Where $\alpha_i$ comp $\beta_j$ and $R_{ij}$ are defined as follows:
\begin{enumerate}
  \item $\alpha_{i}$ is $\overline{x}u$ and $\beta_{j}$ is $x(v)$, then $R_{ij}=P_{i}\parallel Q_{j}\{u/v\}$;
  \item $\alpha_{i}$ is $\overline{x}(u)$ and $\beta_{j}$ is $x(v)$, then $R_{ij}=(w)(P_{i}\{w/u\}\parallel Q_{j}\{w/v\})$, if $w\notin fn((u)P_{i})\cup fn((v)Q_{j})$;
  \item $\alpha_{i}$ is $x(v)$ and $\beta_{j}$ is $\overline{x}u$, then $R_{ij}=P_{i}\{u/v\}\parallel Q_{j}$;
  \item $\alpha_{i}$ is $x(v)$ and $\beta_{j}$ is $\overline{x}(u)$, then $R_{ij}=(w)(P_{i}\{w/v\}\parallel Q_{j}\{w/u\})$, if $w\notin fn((v)P_{i})\cup fn((u)Q_{j})$.
\end{enumerate}

Let $P\equiv\sum_i P_{i}.\alpha_{i}[m]$ and $Q\equiv\sum_l Q_{j}.\beta_{j}[m]$, where $bn(\alpha_{i}[m])\cap fn(Q)=\emptyset$ for all $i$, and
  $bn(\beta_{j}[m])\cap fn(P)=\emptyset$ for all $j$. Then,

\begin{enumerate}
  \item $P\parallel Q\sim_p^{fr} \sum_i\sum_j(P_{i}\parallel Q_{j}).(\alpha_{i}[m]\parallel \beta_{j}[m])+\sum_{\alpha_{i} \textrm{ comp }\beta_{j}} R_{ij}.\tau$;
  \item $P\parallel Q\sim_s^{fr} \sum_i\sum_j(P_{i}\parallel Q_{j}).(\alpha_{i}[m]\parallel \beta_{j}[m])+\sum_{\alpha_{i} \textrm{ comp }\beta_{j}} R_{ij}.\tau$;
  \item $P\parallel Q\sim_{hp}^{fr} \sum_i\sum_j(P_{i}\parallel Q_{j}).(\alpha_{i}[m]\parallel \beta_{j}[m])+\sum_{\alpha_{i} \textrm{ comp }\beta_{j}} R_{ij}.\tau$;
  \item $P\parallel Q\nsim_{phhp} \sum_i\sum_j(P_{i}\parallel Q_{j}).(\alpha_{i}[m]\parallel \beta_{j}[m])+\sum_{\alpha_{i} \textrm{ comp }\beta_{j}} R_{ij}.\tau$.
\end{enumerate}

Where $\alpha_i$ comp $\beta_j$ and $R_{ij}$ are defined as follows:
\begin{enumerate}
  \item $\alpha_{i}[m]$ is $\overline{x}u$ and $\beta_{j}[m]$ is $x(v)$, then $R_{ij}=P_{i}\parallel Q_{j}\{u/v\}$;
  \item $\alpha_{i}[m]$ is $\overline{x}(u)$ and $\beta_{j}[m]$ is $x(v)$, then $R_{ij}=(w)(P_{i}\{w/u\}\parallel Q_{j}\{w/v\})$, if $w\notin fn((u)P_{i})\cup fn((v)Q_{j})$;
  \item $\alpha_{i}[m]$ is $x(v)$ and $\beta_{j}[m]$ is $\overline{x}u$, then $R_{ij}=P_{i}\{u/v\}\parallel Q_{j}$;
  \item $\alpha_{i}[m]$ is $x(v)$ and $\beta_{j}[m]$ is $\overline{x}(u)$, then $R_{ij}=(w)(P_{i}\{w/v\}\parallel Q_{j}\{w/u\})$, if $w\notin fn((v)P_{i})\cup fn((u)Q_{j})$.
\end{enumerate}
  \item Identifier $\textbf{I}$.
  \[\textrm{If }A(\widetilde{x})\overset{\text{def}}{=}P,\textrm{ then }A(\widetilde{y})= P\{\widetilde{y}/\widetilde{x}\}.\]
\end{enumerate}
\end{definition}

\begin{theorem}[Soundness]
If $\textbf{STC}\vdash P=Q$ then
\begin{enumerate}
  \item $P\sim_p^{fr} Q$;
  \item $P\sim_s^{fr} Q$;
  \item $P\sim_{hp}^{fr} Q$;
  \item $P\sim_{hhp}^{fr} Q$.
\end{enumerate}
\end{theorem}

\begin{proof}
The soundness of these laws modulo strongly truly concurrent bisimilarities is already proven in Section \ref{s7}.
\end{proof}

\begin{definition}
The agent identifier $A$ is weakly guardedly defined if every agent identifier is weakly guarded in the right-hand side of the definition of $A$.
\end{definition}

\begin{definition}[Head normal form]
A Process $P$ is in head normal form if it is a sum of the prefixes:

$$P\equiv \sum_i(\alpha_{i1}\parallel\cdots\parallel\alpha_{in}).P_{i}\quad P\equiv \sum_i P_{i}.(\alpha_{i1}[m]\parallel\cdots\parallel\alpha_{in}[m])$$
\end{definition}

\begin{proposition}
If every agent identifier is weakly guardedly defined, then for any process $P$, there is a head normal form $H$ such that

$$\textbf{STC}\vdash P=H.$$
\end{proposition}

\begin{proof}
It is sufficient to induct on the structure of $P$ and quite obvious.
\end{proof}

\begin{theorem}[Completeness]
For all processes $P$ and $Q$,
\begin{enumerate}
  \item if $P\sim_p^{fr} Q$, then $\textbf{STC}\vdash P=Q$;
  \item if $P\sim_s^{fr} Q$, then $\textbf{STC}\vdash P=Q$;
  \item if $P\sim_{hp}^{fr} Q$, then $\textbf{STC}\vdash P=Q$.
\end{enumerate}
\end{theorem}

\begin{proof}
\begin{enumerate}
  \item if $P\sim_s^{fr} Q$, then $\textbf{STC}\vdash P=Q$.

  For the forward transition case.

Since $P$ and $Q$ all have head normal forms, let $P\equiv\sum_{i=1}^k\alpha_{i}.P_{i}$ and $Q\equiv\sum_{i=1}^k\beta_{i}.Q_{i}$. Then the depth of
$P$, denoted as $d(P)=0$, if $k=0$; $d(P)=1+max\{d(P_{i})\}$ for $1\leq j,i\leq k$. The depth $d(Q)$ can be defined similarly.

It is sufficient to induct on $d=d(P)+d(Q)$. When $d=0$, $P\equiv\textbf{nil}$ and $Q\equiv\textbf{nil}$, $P=Q$, as desired.

Suppose $d>0$.

\begin{itemize}
  \item If $(\alpha_1\parallel\cdots\parallel\alpha_n).M$ with $\alpha_{i}(1\leq i\leq n)$ free actions is a summand of $P$, then
  $\langle P,s\rangle\xrightarrow{\{\alpha_1,\cdots,\alpha_n\}}\langle M,s'\rangle$.
  Since $Q$ is in head normal form and has a summand $(\alpha_1\parallel\cdots\parallel\alpha_n).N$ such that $M\sim_s^{fr} N$, by the induction hypothesis $\textbf{STC}\vdash M=N$,
  $\textbf{STC}\vdash (\alpha_1\parallel\cdots\parallel\alpha_n).M= (\alpha_1\parallel\cdots\parallel\alpha_n).N$;
  \item If $x(y).M$ is a summand of $P$, then for $z\notin n(P, Q)$, $\langle P,s\rangle\xrightarrow{x(z)}\langle M',s'\rangle\equiv \langle M\{z/y\},s'\rangle$. Since $Q$ is in head normal form and has a summand
  $x(w).N$ such that for all $v$, $M'\{v/z\}\sim_s^{fr} N'\{v/z\}$ where $N'\equiv N\{z/w\}$, by the induction hypothesis $\textbf{STC}\vdash M'\{v/z\}=N'\{v/z\}$, by the axioms
  $\textbf{C}$ and $\textbf{A}$, $\textbf{STC}\vdash x(y).M=x(w).N$;
  \item If $\overline{x}(y).M$ is a summand of $P$, then for $z\notin n(P,Q)$, $\langle P,s\rangle\xrightarrow{\overline{x}(z)}\langle M',s'\rangle\equiv \langle M\{z/y\},s'\rangle$. Since $Q$ is in head normal form and
  has a summand $\overline{x}(w).N$ such that $M'\sim_s^{fr} N'$ where $N'\equiv N\{z/w\}$, by the induction hypothesis $\textbf{STC}\vdash M'=N'$, by the axioms
  $\textbf{A}$ and $\textbf{C}$, $\textbf{STC}\vdash \overline{x}(y).M= \overline{x}(w).N$.
\end{itemize}

For the reverse transition case, it can be proven similarly, and we omit it.

  \item if $P\sim_p^{fr} Q$, then $\textbf{STC}\vdash P=Q$. It can be proven similarly to the above case.
  \item if $P\sim_{hp}^{fr} Q$, then $\textbf{STC}\vdash P=Q$. It can be proven similarly to the above case.
\end{enumerate}
\end{proof}

\newpage\section{Putting All the Things into a Whole}\label{pa}

In this chapter, we design $\pi_{tc}$ with reversibility, probabilism and guards all together. This chapter is organized as follows. In section \ref{os8}, we introduce the truly concurrent operational semantics. Then, we introduce
the syntax and operational semantics, laws modulo strongly truly concurrent bisimulations, and algebraic theory of $\pi_{tc}$ with reversibility, probabilism and guards in section \ref{sos8},
\ref{s8} and \ref{a8} respectively.

\subsection{Operational Semantics}\label{os8}

Firstly, in this section, we introduce concepts of FR (strongly) probabilistic truly concurrent bisimilarities, including FR probabilistic pomset bisimilarity, FR probabilistic step
bisimilarity, FR probabilistic history-preserving (hp-)bisimilarity and FR probabilistic hereditary history-preserving (hhp-)bisimilarity. In contrast to traditional FR probabilistic truly
concurrent bisimilarities in section \ref{bg}, these versions in $\pi_{ptc}$ must take care of actions with bound objects. Note that, these FR probabilistic truly concurrent bisimilarities
are defined as late bisimilarities, but not early bisimilarities, as defined in $\pi$-calculus \cite{PI1} \cite{PI2}. Note that, here, a PES $\mathcal{E}$ is deemed as a process.

\begin{definition}[Prime event structure with silent event and empty event]
Let $\Lambda$ be a fixed set of labels, ranged over $a,b,c,\cdots$ and $\tau,\epsilon$. A ($\Lambda$-labelled) prime event structure with silent event $\tau$ and empty event
$\epsilon$ is a tuple $\mathcal{E}=\langle \mathbb{E}, \leq, \sharp, \lambda\rangle$, where $\mathbb{E}$ is a denumerable set of events, including the silent event $\tau$ and empty
event $\epsilon$. Let $\hat{\mathbb{E}}=\mathbb{E}\backslash\{\tau,\epsilon\}$, exactly excluding $\tau$ and $\epsilon$, it is obvious that $\hat{\tau^*}=\epsilon$. Let
$\lambda:\mathbb{E}\rightarrow\Lambda$ be a labelling function and let $\lambda(\tau)=\tau$ and $\lambda(\epsilon)=\epsilon$. And $\leq$, $\sharp$ are binary relations on $\mathbb{E}$,
called causality and conflict respectively, such that:

\begin{enumerate}
  \item $\leq$ is a partial order and $\lceil e \rceil = \{e'\in \mathbb{E}|e'\leq e\}$ is finite for all $e\in \mathbb{E}$. It is easy to see that
  $e\leq\tau^*\leq e'=e\leq\tau\leq\cdots\leq\tau\leq e'$, then $e\leq e'$.
  \item $\sharp$ is irreflexive, symmetric and hereditary with respect to $\leq$, that is, for all $e,e',e''\in \mathbb{E}$, if $e\sharp e'\leq e''$, then $e\sharp e''$.
\end{enumerate}

Then, the concepts of consistency and concurrency can be drawn from the above definition:

\begin{enumerate}
  \item $e,e'\in \mathbb{E}$ are consistent, denoted as $e\frown e'$, if $\neg(e\sharp e')$. A subset $X\subseteq \mathbb{E}$ is called consistent, if $e\frown e'$ for all
  $e,e'\in X$.
  \item $e,e'\in \mathbb{E}$ are concurrent, denoted as $e\parallel e'$, if $\neg(e\leq e')$, $\neg(e'\leq e)$, and $\neg(e\sharp e')$.
\end{enumerate}
\end{definition}

\begin{definition}[Configuration]
Let $\mathcal{E}$ be a PES. A (finite) configuration in $\mathcal{E}$ is a (finite) consistent subset of events $C\subseteq \mathcal{E}$, closed with respect to causality (i.e.
$\lceil C\rceil=C$), and a data state $s\in S$ with $S$ the set of all data states, denoted $\langle C, s\rangle$. The set of finite configurations of $\mathcal{E}$ is denoted by
$\langle\mathcal{C}(\mathcal{E}), S\rangle$. We let $\hat{C}=C\backslash\{\tau\}\cup\{\epsilon\}$.
\end{definition}

A consistent subset of $X\subseteq \mathbb{E}$ of events can be seen as a pomset. Given $X, Y\subseteq \mathbb{E}$, $\hat{X}\sim \hat{Y}$ if $\hat{X}$ and $\hat{Y}$ are isomorphic as
pomsets. In the following of the paper, we say $C_1\sim C_2$, we mean $\hat{C_1}\sim\hat{C_2}$.

\begin{definition}[FR pomset transitions and step]
Let $\mathcal{E}$ be a PES and let $C\in\mathcal{C}(\mathcal{E})$, and $\emptyset\neq X\subseteq \mathbb{E}$, if $C\cap X=\emptyset$ and $C'=C\cup X\in\mathcal{C}(\mathcal{E})$, then
$\langle C,s\rangle\xrightarrow{X} \langle C',s'\rangle$ is called a forward pomset transition from $\langle C,s\rangle$ to $\langle C',s'\rangle$ and
$\langle C',s'\rangle\xtworightarrow{X[\mathcal{K}]} \langle C,s\rangle$ is called a reverse pomset transition from $\langle C',s'\rangle$ to $\langle C,s\rangle$. When the events in
$X$ and $X[\mathcal{K}]$ are pairwise
concurrent, we say that $\langle C,s\rangle\xrightarrow{X}\langle C',s'\rangle$ is a forward step and $\langle C',s'\rangle\xrightarrow{X[\mathcal{K}]}\langle C,s\rangle$ is a reverse step.
It is obvious that $\rightarrow^*\xrightarrow{X}\rightarrow^*=\xrightarrow{X}$ and
$\rightarrow^*\xrightarrow{e}\rightarrow^*=\xrightarrow{e}$ for any $e\in\mathbb{E}$ and $X\subseteq\mathbb{E}$.
\end{definition}

\begin{definition}[Probabilistic transitions]
Let $\mathcal{E}$ be a PES and let $C\in\mathcal{C}(\mathcal{E})$, the transition $\langle C,s\rangle\xrsquigarrow{\pi} \langle C^{\pi},s\rangle$ is called a probabilistic
transition
from $\langle C,s\rangle$ to $\langle C^{\pi},s\rangle$.
\end{definition}

A probability distribution function (PDF) $\mu$ is a map $\mu:\mathcal{C}\times\mathcal{C}\rightarrow[0,1]$ and $\mu^*$ is the cumulative probability distribution function (cPDF).

\begin{definition}[FR strongly probabilistic pomset, step bisimilarity]
Let $\mathcal{E}_1$, $\mathcal{E}_2$ be PESs. A FR strongly probabilistic pomset bisimulation is a relation $R\subseteq\langle\mathcal{C}(\mathcal{E}_1),s\rangle\times\langle\mathcal{C}(\mathcal{E}_2),s\rangle$, 
such that (1) if $(\langle C_1,s\rangle,\langle C_2,s\rangle)\in R$, and $\langle C_1,s\rangle\xrightarrow{X_1}\langle C_1',s'\rangle$ (with $\mathcal{E}_1\xrightarrow{X_1}\mathcal{E}_1'$) then $\langle C_2,s\rangle\xrightarrow{X_2}\langle C_2',s'\rangle$ (with
$\mathcal{E}_2\xrightarrow{X_2}\mathcal{E}_2'$), with $X_1\subseteq \mathbb{E}_1$, $X_2\subseteq \mathbb{E}_2$, $X_1\sim X_2$ and $(\langle C_1',s'\rangle,\langle C_2',s'\rangle)\in R$:
\begin{enumerate}
  \item for each fresh action $\alpha\in X_1$, if $\langle C_1'',s''\rangle\xrightarrow{\alpha}\langle C_1''',s'''\rangle$ (with $\mathcal{E}_1''\xrightarrow{\alpha}\mathcal{E}_1'''$), 
  then for some $C_2''$ and $\langle C_2''',s'''\rangle$, $\langle C_2'',s''\rangle\xrightarrow{\alpha}\langle C_2''',s'''\rangle$ (with 
  $\mathcal{E}_2''\xrightarrow{\alpha}\mathcal{E}_2'''$), such that if $(\langle C_1'',s''\rangle,\langle C_2'',s''\rangle)\in R$ then $(\langle C_1''',s'''\rangle,\langle C_2''',s'''\rangle)\in R$;
  \item for each $x(y)\in X_1$ with ($y\notin n(\mathcal{E}_1, \mathcal{E}_2)$), if $\langle C_1'',s''\rangle\xrightarrow{x(y)}\langle C_1''',s'''\rangle$ (with 
  $\mathcal{E}_1''\xrightarrow{x(y)}\mathcal{E}_1'''\{w/y\}$) for all $w$, then for some $C_2''$ and $C_2'''$, $\langle C_2'',s''\rangle\xrightarrow{x(y)}\langle C_2''',s'''\rangle$ 
  (with $\mathcal{E}_2''\xrightarrow{x(y)}\mathcal{E}_2'''\{w/y\}$) for all $w$, such that if $(\langle C_1'',s''\rangle,\langle C_2'',s''\rangle)\in R$ then $(\langle C_1''',s'''\rangle,\langle C_2''',s'''\rangle)\in R$;
  \item for each two $x_1(y),x_2(y)\in X_1$ with ($y\notin n(\mathcal{E}_1, \mathcal{E}_2)$), if $\langle C_1'',s''\rangle\xrightarrow{\{x_1(y),x_2(y)\}}\langle C_1''',s'''\rangle$ 
  (with $\mathcal{E}_1''\xrightarrow{\{x_1(y),x_2(y)\}}\mathcal{E}_1'''\{w/y\}$) for all $w$, then for some $C_2''$ and $C_2'''$, 
  $\langle C_2'',s''\rangle\xrightarrow{\{x_1(y),x_2(y)\}}\langle C_2''',s'''\rangle$ (with $\mathcal{E}_2''\xrightarrow{\{x_1(y),x_2(y)\}}\mathcal{E}_2'''\{w/y\}$) for all $w$, such 
  that if $(\langle C_1'',s''\rangle,\langle C_2'',s''\rangle)\in R$ then $(\langle C_1''',s'''\rangle,\langle C_2''',s'''\rangle)\in R$;
  \item for each $\overline{x}(y)\in X_1$ with $y\notin n(\mathcal{E}_1, \mathcal{E}_2)$, if $\langle C_1'',s''\rangle\xrightarrow{\overline{x}(y)}\langle C_1''',s'''\rangle$ 
  (with $\mathcal{E}_1''\xrightarrow{\overline{x}(y)}\mathcal{E}_1'''$), then for some $C_2''$ and $C_2'''$, $\langle C_2'',s''\rangle\xrightarrow{\overline{x}(y)}\langle C_2''',s'''\rangle$ 
  (with $\mathcal{E}_2''\xrightarrow{\overline{x}(y)}\mathcal{E}_2'''$), such that if $(\langle C_1'',s''\rangle,\langle C_2'',s''\rangle)\in R$ then $(\langle C_1''',s'''\rangle,\langle C_2''',s'''\rangle)\in R$.
\end{enumerate}
 and vice-versa; (2) if $(\langle C_1,s\rangle,\langle C_2,s\rangle)\in R$, and $\langle C_1,s\rangle\xtworightarrow{X_1[\mathcal{K}_1]}\langle C_1',s'\rangle$ (with $\mathcal{E}_1\xtworightarrow{X_1[\mathcal{K}_1]}\mathcal{E}_1'$) then $\langle C_2,s\rangle\xtworightarrow{X_2[\mathcal{K}_2]}\langle C_2',s'\rangle$ (with
$\mathcal{E}_2\xtworightarrow{X_2[\mathcal{K}_2]}\mathcal{E}_2'$), with $X_1\subseteq \mathbb{E}_1$, $X_2\subseteq \mathbb{E}_2$, $X_1\sim X_2$ and $(\langle C_1',s'\rangle,\langle C_2',s'\rangle)\in R$:
\begin{enumerate}
  \item for each fresh action $\alpha\in X_1$, if $\langle C_1'',s''\rangle\xtworightarrow{\alpha[m]}\langle C_1''',s'''\rangle$ (with $\mathcal{E}_1''\xtworightarrow{\alpha[m]}\mathcal{E}_1'''$),
  then for some $C_2''$ and $\langle C_2''',s'''\rangle$, $\langle C_2'',s''\rangle\xtworightarrow{\alpha[m]}\langle C_2''',s'''\rangle$ (with
  $\mathcal{E}_2''\xtworightarrow{\alpha[m]}\mathcal{E}_2'''$), such that if $(\langle C_1'',s''\rangle,\langle C_2'',s''\rangle)\in R$ then $(\langle C_1''',s'''\rangle,\langle C_2''',s'''\rangle)\in R$;
  \item for each $x(y)\in X_1$ with ($y\notin n(\mathcal{E}_1, \mathcal{E}_2)$), if $\langle C_1'',s''\rangle\xtworightarrow{x(y)[m]}\langle C_1''',s'''\rangle$ (with
  $\mathcal{E}_1''\xtworightarrow{x(y)[m]}\mathcal{E}_1'''\{w/y\}$) for all $w$, then for some $C_2''$ and $C_2'''$, $\langle C_2'',s''\rangle\xtworightarrow{x(y)[m]}\langle C_2''',s'''\rangle$
  (with $\mathcal{E}_2''\xtworightarrow{x(y)[m]}\mathcal{E}_2'''\{w/y\}$) for all $w$, such that if $(\langle C_1'',s''\rangle,\langle C_2'',s''\rangle)\in R$ then $(\langle C_1''',s'''\rangle,\langle C_2''',s'''\rangle)\in R$;
  \item for each two $x_1(y),x_2(y)\in X_1$ with ($y\notin n(\mathcal{E}_1, \mathcal{E}_2)$), if $\langle C_1'',s''\rangle\xtworightarrow{\{x_1(y)[m],x_2(y)[n]\}}\langle C_1''',s'''\rangle$
  (with $\mathcal{E}_1''\xtworightarrow{\{x_1(y)[m],x_2(y)[n]\}}\mathcal{E}_1'''\{w/y\}$) for all $w$, then for some $C_2''$ and $C_2'''$,
  $\langle C_2'',s''\rangle\xtworightarrow{\{x_1(y)[m],x_2(y)[n]\}}\langle C_2''',s'''\rangle$ (with $\mathcal{E}_2''\xtworightarrow{\{x_1(y)[m],x_2(y)[n]\}}\mathcal{E}_2'''\{w/y\}$) for all $w$, such
  that if $(\langle C_1'',s''\rangle,\langle C_2'',s''\rangle)\in R$ then $(\langle C_1''',s'''\rangle,\langle C_2''',s'''\rangle)\in R$;
  \item for each $\overline{x}(y)\in X_1$ with $y\notin n(\mathcal{E}_1, \mathcal{E}_2)$, if $\langle C_1'',s''\rangle\xtworightarrow{\overline{x}(y)[m]}\langle C_1''',s'''\rangle$
  (with $\mathcal{E}_1''\xtworightarrow{\overline{x}(y)[m]}\mathcal{E}_1'''$), then for some $C_2''$ and $C_2'''$, $\langle C_2'',s''\rangle\xtworightarrow{\overline{x}(y)[m]}\langle C_2''',s'''\rangle$
  (with $\mathcal{E}_2''\xtworightarrow{\overline{x}(y)[m]}\mathcal{E}_2'''$), such that if $(\langle C_1'',s''\rangle,\langle C_2'',s''\rangle)\in R$ then $(\langle C_1''',s'''\rangle,\langle C_2''',s'''\rangle)\in R$.
\end{enumerate}
 and vice-versa;(3) if $(\langle C_1,s\rangle,\langle C_2,s\rangle)\in R$, and $\langle C_1,s\rangle\xrsquigarrow{\pi}\langle C_1^{\pi},s\rangle$ then 
 $\langle C_2,s\rangle\xrsquigarrow{\pi}\langle C_2^{\pi},s\rangle$ and $(\langle C_1^{\pi},s\rangle,\langle C_2^{\pi},s\rangle)\in R$, and vice-versa; (4) if $(\langle C_1,s\rangle,\langle C_2,s\rangle)\in R$,
then $\mu(C_1,C)=\mu(C_2,C)$ for each $C\in\mathcal{C}(\mathcal{E})/R$; (5) $[\surd]_R=\{\surd\}$.

We say that $\mathcal{E}_1$, $\mathcal{E}_2$ are FR strongly probabilistic pomset bisimilar, written $\mathcal{E}_1\sim_{pp}^{fr}\mathcal{E}_2$, if there exists a FR strongly probabilistic pomset
bisimulation $R$, such that $(\emptyset,\emptyset)\in R$. By replacing FR probabilistic pomset transitions with steps, we can get the definition of FR strongly probabilistic step bisimulation.
When PESs $\mathcal{E}_1$ and $\mathcal{E}_2$ are FR strongly probabilistic step bisimilar, we write $\mathcal{E}_1\sim_{ps}^{fr}\mathcal{E}_2$.
\end{definition}

\begin{definition}[Posetal product]
Given two PESs $\mathcal{E}_1$, $\mathcal{E}_2$, the posetal product of their configurations, denoted
$\langle\mathcal{C}(\mathcal{E}_1),S\rangle\overline{\times}\langle\mathcal{C}(\mathcal{E}_2),S\rangle$, is defined as

$$\{(\langle C_1,s\rangle,f,\langle C_2,s\rangle)|C_1\in\mathcal{C}(\mathcal{E}_1),C_2\in\mathcal{C}(\mathcal{E}_2),f:C_1\rightarrow C_2 \textrm{ isomorphism}\}.$$

A subset $R\subseteq\langle\mathcal{C}(\mathcal{E}_1),S\rangle\overline{\times}\langle\mathcal{C}(\mathcal{E}_2),S\rangle$ is called a posetal relation. We say that $R$ is downward
closed when for any
$(\langle C_1,s\rangle,f,\langle C_2,s\rangle),(\langle C_1',s'\rangle,f',\langle C_2',s'\rangle)\in \langle\mathcal{C}(\mathcal{E}_1),S\rangle\overline{\times}\langle\mathcal{C}(\mathcal{E}_2),S\rangle$,
if $(\langle C_1,s\rangle,f,\langle C_2,s\rangle)\subseteq (\langle C_1',s'\rangle,f',\langle C_2',s'\rangle)$ pointwise and $(\langle C_1',s'\rangle,f',\langle C_2',s'\rangle)\in R$,
then $(\langle C_1,s\rangle,f,\langle C_2,s\rangle)\in R$.

For $f:X_1\rightarrow X_2$, we define $f[x_1\mapsto x_2]:X_1\cup\{x_1\}\rightarrow X_2\cup\{x_2\}$, $z\in X_1\cup\{x_1\}$,(1)$f[x_1\mapsto x_2](z)=
x_2$,if $z=x_1$;(2)$f[x_1\mapsto x_2](z)=f(z)$, otherwise. Where $X_1\subseteq \mathbb{E}_1$, $X_2\subseteq \mathbb{E}_2$, $x_1\in \mathbb{E}_1$, $x_2\in \mathbb{E}_2$.
\end{definition}

\begin{definition}[FR strongly probabilistic (hereditary) history-preserving bisimilarity]
A FR strongly probabilistic history-preserving (hp-) bisimulation is a posetal relation $R\subseteq\mathcal{C}(\mathcal{E}_1)\overline{\times}\mathcal{C}(\mathcal{E}_2)$ such that 
(1) if $(\langle C_1,s\rangle,f,\langle C_2,s\rangle)\in R$, and
\begin{enumerate}
  \item for $e_1=\alpha$ a fresh action, if $\langle C_1,s\rangle\xrightarrow{\alpha}\langle C_1',s'\rangle$ (with $\mathcal{E}_1\xrightarrow{\alpha}\mathcal{E}_1'$), then for some 
  $C_2'$ and $e_2=\alpha$, $\langle C_2,s\rangle\xrightarrow{\alpha}\langle C_2',s'\rangle$ (with $\mathcal{E}_2\xrightarrow{\alpha}\mathcal{E}_2'$), such that 
  $(\langle C_1',s'\rangle,f[e_1\mapsto e_2],\langle C_2',s'\rangle)\in R$;
  \item for $e_1=x(y)$ with ($y\notin n(\mathcal{E}_1, \mathcal{E}_2)$), if $\langle C_1,s\rangle\xrightarrow{x(y)}\langle C_1',s'\rangle$ (with 
  $\mathcal{E}_1\xrightarrow{x(y)}\mathcal{E}_1'\{w/y\}$) for all $w$, then for some $C_2'$ and $e_2=x(y)$, $\langle C_2,s\rangle\xrightarrow{x(y)}\langle C_2',s'\rangle$ (with 
  $\mathcal{E}_2\xrightarrow{x(y)}\mathcal{E}_2'\{w/y\}$) for all $w$, such that $(\langle C_1',s'\rangle,f[e_1\mapsto e_2],\langle C_2',s'\rangle)\in R$;
  \item for $e_1=\overline{x}(y)$ with $y\notin n(\mathcal{E}_1, \mathcal{E}_2)$, if $\langle C_1,s\rangle\xrightarrow{\overline{x}(y)}\langle C_1',s'\rangle$ (with 
  $\mathcal{E}_1\xrightarrow{\overline{x}(y)}\mathcal{E}_1'$), then for some $C_2'$ and $e_2=\overline{x}(y)$, $\langle C_2,s\rangle\xrightarrow{\overline{x}(y)}\langle C_2',s'\rangle$ 
  (with $\mathcal{E}_2\xrightarrow{\overline{x}(y)}\mathcal{E}_2'$), such that $(\langle C_1',s'\rangle,f[e_1\mapsto e_2],\langle C_2',s'\rangle)\in R$.
\end{enumerate}
and vice-versa; (2) if $(\langle C_1,s\rangle,f,\langle C_2,s\rangle)\in R$, and
\begin{enumerate}
  \item for $e_1=\alpha$ a fresh action, if $\langle C_1,s\rangle\xtworightarrow{\alpha[m]}\langle C_1',s'\rangle$ (with $\mathcal{E}_1\xtworightarrow{\alpha[m]}\mathcal{E}_1'$), then for some
  $C_2'$ and $e_2=\alpha$, $\langle C_2,s\rangle\xtworightarrow{\alpha[m]}\langle C_2',s'\rangle$ (with $\mathcal{E}_2\xtworightarrow{\alpha[m]}\mathcal{E}_2'$), such that
  $(\langle C_1',s'\rangle,f[e_1\mapsto e_2],\langle C_2',s'\rangle)\in R$;
  \item for $e_1=x(y)$ with ($y\notin n(\mathcal{E}_1, \mathcal{E}_2)$), if $\langle C_1,s\rangle\xtworightarrow{x(y)[m]}\langle C_1',s'\rangle$ (with
  $\mathcal{E}_1\xtworightarrow{x(y)[m]}\mathcal{E}_1'\{w/y\}$) for all $w$, then for some $C_2'$ and $e_2=x(y)$, $\langle C_2,s\rangle\xtworightarrow{x(y)[m]}\langle C_2',s'\rangle$ (with
  $\mathcal{E}_2\xtworightarrow{x(y)[m]}\mathcal{E}_2'\{w/y\}$) for all $w$, such that $(\langle C_1',s'\rangle,f[e_1\mapsto e_2],\langle C_2',s'\rangle)\in R$;
  \item for $e_1=\overline{x}(y)$ with $y\notin n(\mathcal{E}_1, \mathcal{E}_2)$, if $\langle C_1,s\rangle\xtworightarrow{\overline{x}(y)[m]}\langle C_1',s'\rangle$ (with
  $\mathcal{E}_1\xtworightarrow{\overline{x}(y)[m]}\mathcal{E}_1'$), then for some $C_2'$ and $e_2=\overline{x}(y)$, $\langle C_2,s\rangle\xtworightarrow{\overline{x}(y)[m]}\langle C_2',s'\rangle$
  (with $\mathcal{E}_2\xtworightarrow{\overline{x}(y)[m]}\mathcal{E}_2'$), such that $(\langle C_1',s'\rangle,f[e_1\mapsto e_2],\langle C_2',s'\rangle)\in R$.
\end{enumerate}
and vice-versa; (3) if $(\langle C_1,s\rangle,f,\langle C_2,s\rangle)\in R$, and $\langle C_1,s\rangle\xrsquigarrow{\pi}\langle C_1^{\pi},s\rangle$ then 
$\langle C_2,s\rangle\xrsquigarrow{\pi}\langle C_2^{\pi},s\rangle$ and $(\langle C_1^{\pi},s\rangle,f,\langle C_2^{\pi},s\rangle)\in R$, and vice-versa; (4) if
$(\langle C_1,s\rangle,f,\langle C_2,s\rangle)\in R$, then $\mu(C_1,C)=\mu(C_2,C)$ for each $C\in\mathcal{C}(\mathcal{E})/R$; (5) $[\surd]_R=\{\surd\}$. $\mathcal{E}_1,\mathcal{E}_2$ 
are FR strongly probabilistic history-preserving (hp-)bisimilar and are written $\mathcal{E}_1\sim_{php}^{fr}\mathcal{E}_2$ if there exists a FR strongly probabilistic hp-bisimulation 
$R$ such that $(\emptyset,\emptyset,\emptyset)\in R$.

A FR strongly probabilistic hereditary history-preserving (hhp-)bisimulation is a downward closed FR strongly probabilistic hp-bisimulation. $\mathcal{E}_1,\mathcal{E}_2$ are FR 
strongly probabilistic hereditary history-preserving (hhp-)bisimilar and are written $\mathcal{E}_1\sim_{phhp}^{fr}\mathcal{E}_2$.
\end{definition}

\subsection{Syntax and Operational Semantics}\label{sos8}

We assume an infinite set $\mathcal{N}$ of (action or event) names, and use $a,b,c,\cdots$ to range over $\mathcal{N}$, use $x,y,z,w,u,v$ as meta-variables over names. We denote by
$\overline{\mathcal{N}}$ the set of co-names and let $\overline{a},\overline{b},\overline{c},\cdots$ range over $\overline{\mathcal{N}}$. Then we set
$\mathcal{L}=\mathcal{N}\cup\overline{\mathcal{N}}$ as the set of labels, and use $l,\overline{l}$ to range over $\mathcal{L}$. We extend complementation to $\mathcal{L}$ such that
$\overline{\overline{a}}=a$. Let $\tau$ denote the silent step (internal action or event) and define $Act=\mathcal{L}\cup\{\tau\}$ to be the set of actions, $\alpha,\beta$ range over
$Act$. And $K,L$ are used to stand for subsets of $\mathcal{L}$ and $\overline{L}$ is used for the set of complements of labels in $L$.

Further, we introduce a set $\mathcal{X}$ of process variables, and a set $\mathcal{K}$ of process constants, and let $X,Y,\cdots$ range over $\mathcal{X}$, and $A,B,\cdots$ range over
$\mathcal{K}$. For each process constant $A$, a nonnegative arity $ar(A)$ is assigned to it. Let $\widetilde{x}=x_1,\cdots,x_{ar(A)}$ be a tuple of distinct name variables, then
$A(\widetilde{x})$ is called a process constant. $\widetilde{X}$ is a tuple of distinct process variables, and also $E,F,\cdots$ range over the recursive expressions. We write
$\mathcal{P}$ for the set of processes. Sometimes, we use $I,J$ to stand for an indexing set, and we write $E_i:i\in I$ for a family of expressions indexed by $I$. $Id_D$ is the
identity function or relation over set $D$. The symbol $\equiv_{\alpha}$ denotes equality under standard alpha-convertibility, note that the subscript $\alpha$ has no relation to the
action $\alpha$.

Let $G_{at}$ be the set of atomic guards, $\delta$ be the deadlock constant, and $\epsilon$ be the empty action, and extend $Act$ to $Act\cup\{\epsilon\}\cup\{\delta\}$. We extend
$G_{at}$ to the set of basic guards $G$ with element $\phi,\psi,\cdots$, which is generated by the following formation rules:

$$\phi::=\delta|\epsilon|\neg\phi|\psi\in G_{at}|\phi+\psi|\phi\cdot\psi$$

The predicate $test(\phi,s)$ represents that $\phi$ holds in the state $s$, and $test(\epsilon,s)$ holds and $test(\delta,s)$ does not hold. $effect(e,s)\in S$ denotes $s'$ in
$s\xrightarrow{e}s'$. The predicate weakest precondition $wp(e,\phi)$ denotes that $\forall s,s'\in S, test(\phi,effect(e,s))$ holds.

\subsubsection{Syntax}

We use the Prefix $.$ to model the causality relation $\leq$ in true concurrency, the Summation $+$ to model the conflict relation $\sharp$, and $\boxplus_{\pi}$ to model the probabilistic
conflict relation $\sharp_{\pi}$ in probabilistic true concurrency, and the Composition $\parallel$ to explicitly model concurrent relation in true concurrency. And we follow the
conventions of process algebra.

\begin{definition}[Syntax]\label{syntax8}
A truly concurrent process $\pi_{tc}$ with reversibility, probabilism and guards is defined inductively by the following formation rules:

\begin{enumerate}
  \item $A(\widetilde{x})\in\mathcal{P}$;
  \item $\phi\in\mathcal{P}$;
  \item $\textbf{nil}\in\mathcal{P}$;
  \item if $P\in\mathcal{P}$, then the Prefix $\tau.P\in\mathcal{P}$, for $\tau\in Act$ is the silent action;
  \item if $P\in\mathcal{P}$, then the Prefix $\phi.P\in\mathcal{P}$, for $\phi\in G_{at}$;
  \item if $P\in\mathcal{P}$, then the Output $\overline{x}y.P\in\mathcal{P}$, for $x,y\in Act$;
  \item if $P\in\mathcal{P}$, then the Output $P.\overline{x}y[m]\in\mathcal{P}$, for $x,y\in Act$;
  \item if $P\in\mathcal{P}$, then the Input $x(y).P\in\mathcal{P}$, for $x,y\in Act$;
  \item if $P\in\mathcal{P}$, then the Input $P.x(y)[m]\in\mathcal{P}$, for $x,y\in Act$;
  \item if $P\in\mathcal{P}$, then the Restriction $(x)P\in\mathcal{P}$, for $x\in Act$;
  \item if $P,Q\in\mathcal{P}$, then the Summation $P+Q\in\mathcal{P}$;
  \item if $P,Q\in\mathcal{P}$, then the Summation $P\boxplus_{\pi}Q\in\mathcal{P}$;
  \item if $P,Q\in\mathcal{P}$, then the Composition $P\parallel Q\in\mathcal{P}$;
\end{enumerate}

The standard BNF grammar of syntax of $\pi_{tc}$ with reversibility, probabilism and guards can be summarized as follows:

$$P::=A(\widetilde{x})|\textbf{nil}|\tau.P| \overline{x}y.P | x(y).P|\overline{x}y[m].P | x(y)[m].P | (x)P  |\phi.P|  P+P| P\boxplus_{\pi}P | P\parallel P.$$
\end{definition}

In $\overline{x}y$, $x(y)$ and $\overline{x}(y)$, $x$ is called the subject, $y$ is called the object and it may be free or bound.

\begin{definition}[Free variables]
The free names of a process $P$, $fn(P)$, are defined as follows.

\begin{enumerate}
  \item $fn(A(\widetilde{x}))\subseteq\{\widetilde{x}\}$;
  \item $fn(\textbf{nil})=\emptyset$;
  \item $fn(\tau.P)=fn(P)$;
  \item $fn(\phi.P)=fn(P)$;
  \item $fn(\overline{x}y.P)=fn(P)\cup\{x\}\cup\{y\}$;
  \item $fn(\overline{x}y[m].P)=fn(P)\cup\{x\}\cup\{y\}$;
  \item $fn(x(y).P)=fn(P)\cup\{x\}-\{y\}$;
  \item $fn(x(y)[m].P)=fn(P)\cup\{x\}-\{y\}$;
  \item $fn((x)P)=fn(P)-\{x\}$;
  \item $fn(P+Q)=fn(P)\cup fn(Q)$;
  \item $fn(P\boxplus_{\pi}Q)=fn(P)\cup fn(Q)$;
  \item $fn(P\parallel Q)=fn(P)\cup fn(Q)$.
\end{enumerate}
\end{definition}

\begin{definition}[Bound variables]
Let $n(P)$ be the names of a process $P$, then the bound names $bn(P)=n(P)-fn(P)$.
\end{definition}

For each process constant schema $A(\widetilde{x})$, a defining equation of the form

$$A(\widetilde{x})\overset{\text{def}}{=}P$$

is assumed, where $P$ is a process with $fn(P)\subseteq \{\widetilde{x}\}$.

\begin{definition}[Substitutions]\label{subs8}
A substitution is a function $\sigma:\mathcal{N}\rightarrow\mathcal{N}$. For $x_i\sigma=y_i$ with $1\leq i\leq n$, we write $\{y_1/x_1,\cdots,y_n/x_n\}$ or
$\{\widetilde{y}/\widetilde{x}\}$ for $\sigma$. For a process $P\in\mathcal{P}$, $P\sigma$ is defined inductively as follows:
\begin{enumerate}
  \item if $P$ is a process constant $A(\widetilde{x})=A(x_1,\cdots,x_n)$, then $P\sigma=A(x_1\sigma,\cdots,x_n\sigma)$;
  \item if $P=\textbf{nil}$, then $P\sigma=\textbf{nil}$;
  \item if $P=\tau.P'$, then $P\sigma=\tau.P'\sigma$;
  \item if $P=\phi.P'$, then $P\sigma=\phi.P'\sigma$;
  \item if $P=\overline{x}y.P'$, then $P\sigma=\overline{x\sigma}y\sigma.P'\sigma$;
  \item if $P=\overline{x}y[m].P'$, then $P\sigma=\overline{x\sigma}y\sigma[m].P'\sigma$;
  \item if $P=x(y).P'$, then $P\sigma=x\sigma(y).P'\sigma$;
  \item if $P=x(y)[m].P'$, then $P\sigma=x\sigma(y)[m].P'\sigma$;
  \item if $P=(x)P'$, then $P\sigma=(x\sigma)P'\sigma$;
  \item if $P=P_1+P_2$, then $P\sigma=P_1\sigma+P_2\sigma$;
  \item if $P=P_1\boxplus_{\pi}P_2$, then $P\sigma=P_1\sigma\boxplus_{\pi}P_2\sigma$;
  \item if $P=P_1\parallel P_2$, then $P\sigma=P_1\sigma \parallel P_2\sigma$.
\end{enumerate}
\end{definition}

\subsubsection{Operational Semantics}

The operational semantics is defined by LTSs (labelled transition systems), and it is detailed by the following definition.

\begin{definition}[Semantics]\label{semantics8}
The operational semantics of $\pi_{tc}$ with reversibility, probabilism and guards corresponding to the syntax in Definition \ref{syntax8} is defined by a series of transition rules, named $\textbf{PACT}$, $\textbf{PSUM}$, $\textbf{PBOX-SUM}$,
$\textbf{PIDE}$, $\textbf{PPAR}$, $\textbf{PRES}$ and named $\textbf{ACT}$, $\textbf{SUM}$,
$\textbf{IDE}$, $\textbf{PAR}$, $\textbf{COM}$, $\textbf{CLOSE}$, $\textbf{RES}$, $\textbf{OPEN}$ indicate that the rules are associated respectively with Prefix, Summation, Box-Summation,
Identity, Parallel Composition, Communication, and Restriction in Definition \ref{syntax8}. They are shown in Table \ref{PTRForPITC8} and \ref{TRForPITC8}.

\begin{center}
    \begin{table}
        \[\textbf{PTAU-ACT}\quad \frac{}{\langle \tau.P,s\rangle\rightsquigarrow \langle\breve{\tau}.P,s\rangle}\] 
        
        \[\textbf{POUTPUT-ACT}\quad \frac{}{\langle\overline{x}y.P,s\rangle\rightsquigarrow \langle\breve{\overline{x}y}.P,s\rangle}\]

        \[\textbf{PINPUT-ACT}\quad \frac{}{\langle x(z).P,s\rangle\rightsquigarrow \langle\breve{x(z)}.P,s\rangle}\]

        \[\textbf{PPAR}\quad \frac{\langle P,s\rangle\rightsquigarrow \langle P',s\rangle\quad \langle Q,s\rangle\rightsquigarrow \langle Q',s\rangle}{\langle P\parallel Q,s\rangle\rightsquigarrow \langle P'\parallel Q',s\rangle}\]

        \[\textbf{PSUM}\quad \frac{\langle P,s\rangle\rightsquigarrow \langle P',s\rangle\quad \langle Q,s\rangle\rightsquigarrow \langle Q',s\rangle}{\langle P+Q,s\rangle\rightsquigarrow \langle P'+Q',s\rangle}\]

        \[\textbf{PBOX-SUM}\quad \frac{\langle P,s\rangle\rightsquigarrow \langle P',s\rangle}{\langle P\boxplus_{\pi}Q,s\rangle\rightsquigarrow \langle P',s\rangle}\]

        \[\textbf{PIDE}\quad\frac{\langle P\{\widetilde{y}/\widetilde{x}\},s\rangle\rightsquigarrow \langle P',s\rangle}{\langle A(\widetilde{y}),s\rangle\rightsquigarrow \langle P',s\rangle}\quad (A(\widetilde{x})\overset{\text{def}}{=}P)\]

        \[\textbf{PRES}\quad \frac{\langle P,s\rangle\rightsquigarrow \langle P',s\rangle}{\langle (y)P,s\rangle\rightsquigarrow \langle (y)P',s\rangle}\quad (y\notin n(\alpha))\]

        \caption{Probabilistic transition rules}
        \label{PTRForPITC8}
    \end{table}
\end{center}

\begin{center}
    \begin{table}
        \[\textbf{TAU-ACT}\quad \frac{}{\langle\breve{\tau}.P,s\rangle\xrightarrow{\tau}\langle P,\tau(s)\rangle}\] 
        
        \[\textbf{OUTPUT-ACT}\quad \frac{}{\langle\breve{\overline{x}y}.P,s\rangle\xrightarrow{\overline{x}y}\langle P,s'\rangle}\]

        \[\textbf{INPUT-ACT}\quad \frac{}{\langle\breve{x(z)}.P,s\rangle\xrightarrow{x(w)}\langle P\{w/z\},s'\rangle}\quad (w\notin fn((z)P))\]

        \[\textbf{PAR}_1\quad \frac{\langle P,s\rangle\xrightarrow{\alpha}\langle P',s'\rangle\quad \langle Q,s\rangle\nrightarrow}{\langle P\parallel Q,s\rangle\xrightarrow{\alpha}\langle P'\parallel Q,s'\rangle}\quad (bn(\alpha)\cap fn(Q)=\emptyset)\]
        
        \[\textbf{PAR}_2\quad \frac{\langle Q,s\rangle\xrightarrow{\alpha}\langle Q',s'\rangle\quad \langle P,s\rangle\nrightarrow}{\langle P\parallel Q,s\rangle\xrightarrow{\alpha}\langle P\parallel Q',s'\rangle}\quad (bn(\alpha)\cap fn(P)=\emptyset)\]

        \[\textbf{PAR}_3\quad \frac{\langle P,s\rangle\xrightarrow{\alpha}\langle P',s'\rangle\quad \langle Q,s\rangle\xrightarrow{\beta}\langle Q',s''\rangle}{\langle P\parallel Q,s\rangle\xrightarrow{\{\alpha,\beta\}}\langle P'\parallel Q',s'\cup s''\rangle}\] $(\beta\neq\overline{\alpha}, bn(\alpha)\cap bn(\beta)=\emptyset, bn(\alpha)\cap fn(Q)=\emptyset,bn(\beta)\cap fn(P)=\emptyset)$

        \[\textbf{PAR}_4\quad \frac{\langle P,s\rangle\xrightarrow{x_1(z)}\langle P',s'\rangle\quad \langle Q,s\rangle\xrightarrow{x_2(z)}\langle Q',s''\rangle}{\langle P\parallel Q,s\rangle\xrightarrow{\{x_1(w),x_2(w)\}}\langle P'\{w/z\}\parallel Q'\{w/z\},s'\cup s''\rangle}\quad (w\notin fn((z)P)\cup fn((z)Q))\]

        \[\textbf{COM}\quad \frac{\langle P,s\rangle\xrightarrow{\overline{x}y}\langle P',s'\rangle\quad \langle Q,s\rangle\xrightarrow{x(z)}\langle Q',s''\rangle}{\langle P\parallel Q,s\rangle\xrightarrow{\tau}\langle P'\parallel Q'\{y/z\},s'\cup s''\rangle}\]

        \[\textbf{CLOSE}\quad \frac{\langle P,s\rangle\xrightarrow{\overline{x}(w)}\langle P',s'\rangle\quad \langle Q,s\rangle\xrightarrow{x(w)}\langle Q',s''\rangle}{\langle P\parallel Q,s\rangle\xrightarrow{\tau}\langle (w)(P'\parallel Q'),s'\cup s''\rangle}\]

        \caption{Forward action transition rules}
        \label{TRForPITC8}
    \end{table}
\end{center}

\begin{center}
    \begin{table}
        \[\textbf{SUM}_1\quad \frac{\langle P,s\rangle\xrightarrow{\alpha}\langle P',s'\rangle}{\langle P+Q,s\rangle\xrightarrow{\alpha}\langle P',s'\rangle}\]

        \[\textbf{SUM}_2\quad \frac{\langle P,s\rangle\xrightarrow{\{\alpha_1,\cdots,\alpha_n\}}\langle P',s'\rangle}{\langle P+Q,s\rangle\xrightarrow{\{\alpha_1,\cdots,\alpha_n\}}\langle P',s'\rangle}\]

        \[\textbf{IDE}_1\quad\frac{\langle P\{\widetilde{y}/\widetilde{x}\},s\rangle\xrightarrow{\alpha}\langle P',s'\rangle}{\langle A(\widetilde{y}),s\rangle\xrightarrow{\alpha}\langle P',s'\rangle}\quad (A(\widetilde{x})\overset{\text{def}}{=}P)\]

        \[\textbf{IDE}_2\quad\frac{\langle P\{\widetilde{y}/\widetilde{x}\},s\rangle\xrightarrow{\{\alpha_1,\cdots,\alpha_n\}}\langle P',s'\rangle} {\langle A(\widetilde{y}),s\rangle\xrightarrow{\{\alpha_1,\cdots,\alpha_n\}}\langle P',s'\rangle}\quad (A(\widetilde{x})\overset{\text{def}}{=}P)\]

        \[\textbf{RES}_1\quad \frac{\langle P,s\rangle\xrightarrow{\alpha}\langle P',s'\rangle}{\langle (y)P,s\rangle\xrightarrow{\alpha}\langle (y)P',s'\rangle}\quad (y\notin n(\alpha))\]

        \[\textbf{RES}_2\quad \frac{\langle P,s\rangle\xrightarrow{\{\alpha_1,\cdots,\alpha_n\}}\langle P',s'\rangle}{\langle (y)P,s\rangle\xrightarrow{\{\alpha_1,\cdots,\alpha_n\}}\langle (y)P',s'\rangle}\quad (y\notin n(\alpha_1)\cup\cdots\cup n(\alpha_n))\]

        \[\textbf{OPEN}_1\quad \frac{\langle P,s\rangle\xrightarrow{\overline{x}y}\langle P',s'\rangle}{\langle (y)P,s\rangle\xrightarrow{\overline{x}(w)}\langle P'\{w/y\},s'\rangle} \quad (y\neq x, w\notin fn((y)P'))\]

        \[\textbf{OPEN}_2\quad \frac{\langle P,s\rangle\xrightarrow{\{\overline{x}_1 y,\cdots,\overline{x}_n y\}}\langle P',s'\rangle}{\langle(y)P,s\rangle\xrightarrow{\{\overline{x}_1(w),\cdots,\overline{x}_n(w)\}}\langle P'\{w/y\},s'\rangle} \quad (y\neq x_1\neq\cdots\neq x_n, w\notin fn((y)P'))\]

        \caption{Forward action transition rules (continuing)}
        \label{TRForPITC82}
    \end{table}
\end{center}

\begin{center}
    \begin{table}
        \[\textbf{RTAU-ACT}\quad \frac{}{\langle\breve{\tau}.P,s\rangle\xtworightarrow{\tau}\langle P,\tau(s)\rangle}\]

        \[\textbf{ROUTPUT-ACT}\quad \frac{}{\langle\breve{\overline{x}y}[m].P,s\rangle\xtworightarrow{\overline{x}y[m]}\langle P,s'\rangle}\]

        \[\textbf{RINPUT-ACT}\quad \frac{}{\langle\breve{x(z)}[m].P,s\rangle\xtworightarrow{x(w)[m]}\langle P\{w/z\},s'\rangle}\quad (w\notin fn((z)P))\]

        \[\textbf{RPAR}_1\quad \frac{\langle P,s\rangle\xtworightarrow{\alpha[m]}\langle P',s'\rangle\quad \langle Q,s\rangle\nrightarrow}{\langle P\parallel Q,s\rangle\xtworightarrow{\alpha[m]}\langle P'\parallel Q,s'\rangle}\quad (bn(\alpha)\cap fn(Q)=\emptyset)\]
        
        \[\textbf{RPAR}_2\quad \frac{\langle Q,s\rangle\xtworightarrow{\alpha[m]}\langle Q',s'\rangle\quad \langle P,s\rangle\nrightarrow}{\langle P\parallel Q,s\rangle\xtworightarrow{\alpha[m]}\langle P\parallel Q',s'\rangle}\quad (bn(\alpha)\cap fn(P)=\emptyset)\]

        \[\textbf{RPAR}_3\quad \frac{\langle P,s\rangle\xtworightarrow{\alpha[m]}\langle P',s'\rangle\quad \langle Q,s\rangle\xtworightarrow{\beta[m]}\langle Q',s''\rangle}{\langle P\parallel Q,s\rangle\xtworightarrow{\{\alpha[m],\beta[m]\}}\langle P'\parallel Q',s'\cup s''\rangle}\] $(\beta\neq\overline{\alpha}, bn(\alpha)\cap bn(\beta)=\emptyset, bn(\alpha)\cap fn(Q)=\emptyset,bn(\beta)\cap fn(P)=\emptyset)$

        \[\textbf{RPAR}_4\quad \frac{\langle P,s\rangle\xtworightarrow{x_1(z)[m]}\langle P',s'\rangle\quad \langle Q,s\rangle\xtworightarrow{x_2(z)[m]}\langle Q',s''\rangle}{\langle P\parallel Q,s\rangle\xtworightarrow{\{x_1(w)[m],x_2(w)[m]\}}\langle P'\{w/z\}\parallel Q'\{w/z\},s'\cup s''\rangle}\quad (w\notin fn((z)P)\cup fn((z)Q))\]

        \[\textbf{RCOM}\quad \frac{\langle P,s\rangle\xtworightarrow{\overline{x}y[m]}\langle P',s'\rangle\quad \langle Q,s\rangle\xtworightarrow{x(z)[m]}\langle Q',s''\rangle}{\langle P\parallel Q,s\rangle\xtworightarrow{\tau}\langle P'\parallel Q'\{y/z\},s'\cup s''\rangle}\]

        \[\textbf{RCLOSE}\quad \frac{\langle P,s\rangle\xtworightarrow{\overline{x}(w)[m]}\langle P',s'\rangle\quad \langle Q,s\rangle\xtworightarrow{x(w)[m]}\langle Q',s''\rangle}{\langle P\parallel Q,s\rangle\xtworightarrow{\tau}\langle (w)(P'\parallel Q'),s'\cup s''\rangle}\]

        \caption{Reverse action transition rules}
        \label{TRForPITC83}
    \end{table}
\end{center}

\begin{center}
    \begin{table}
        \[\textbf{RSUM}_1\quad \frac{\langle P,s\rangle\xtworightarrow{\alpha[m]}\langle P',s'\rangle}{\langle P+Q,s\rangle\xtworightarrow{\alpha[m]}\langle P',s'\rangle}\]

        \[\textbf{RSUM}_2\quad \frac{\langle P,s\rangle\xtworightarrow{\{\alpha_1[m],\cdots,\alpha_n[m]\}}\langle P',s'\rangle}{\langle P+Q,s\rangle\xtworightarrow{\{\alpha_1[m],\cdots,\alpha_n[m]\}}\langle P',s'\rangle}\]

        \[\textbf{RIDE}_1\quad\frac{\langle P\{\widetilde{y}/\widetilde{x}\},s\rangle\xtworightarrow{\alpha[m]}\langle P',s'\rangle}{\langle A(\widetilde{y}),s\rangle\xtworightarrow{\alpha[m]}\langle P',s'\rangle}\quad (A(\widetilde{x})\overset{\text{def}}{=}P)\]

        \[\textbf{RIDE}_2\quad\frac{\langle P\{\widetilde{y}/\widetilde{x}\},s\rangle\xtworightarrow{\{\alpha_1[m],\cdots,\alpha_n[m]\}}\langle P',s'\rangle} {\langle A(\widetilde{y}),s\rangle\xtworightarrow{\{\alpha_1[m],\cdots,\alpha_n[m]\}}\langle P',s'\rangle}\quad (A(\widetilde{x})\overset{\text{def}}{=}P)\]

        \[\textbf{RRES}_1\quad \frac{\langle P,s\rangle\xtworightarrow{\alpha[m]}\langle P',s'\rangle}{\langle (y)P,s\rangle\xtworightarrow{\alpha[m]}\langle (y)P',s'\rangle}\quad (y\notin n(\alpha))\]

        \[\textbf{RRES}_2\quad \frac{\langle P,s\rangle\xtworightarrow{\{\alpha_1[m],\cdots,\alpha_n[m]\}}\langle P',s'\rangle}{\langle (y)P,s\rangle\xtworightarrow{\{\alpha_1[m],\cdots,\alpha_n[m]\}}\langle (y)P',s'\rangle}\quad (y\notin n(\alpha_1)\cup\cdots\cup n(\alpha_n))\]

        \[\textbf{ROPEN}_1\quad \frac{\langle P,s\rangle\xtworightarrow{\overline{x}y[m]}\langle P',s'\rangle}{\langle (y)P,s\rangle\xtworightarrow{\overline{x}(w)[m]}\langle P'\{w/y\},s'\rangle} \quad (y\neq x, w\notin fn((y)P'))\]

        \[\textbf{ROPEN}_2\quad \frac{\langle P,s\rangle\xtworightarrow{\{\overline{x}_1 y[m],\cdots,\overline{x}_n y[m]\}}\langle P',s'\rangle}{\langle(y)P,s\rangle\xtworightarrow{\{\overline{x}_1(w)[m],\cdots,\overline{x}_n(w)[m]\}}\langle P'\{w/y\},s'\rangle} \quad (y\neq x_1\neq\cdots\neq x_n, w\notin fn((y)P'))\]

        \caption{Reverse action transition rules (continuing)}
        \label{TRForPITC84}
    \end{table}
\end{center}
\end{definition}

\subsubsection{Properties of Transitions}

\begin{proposition}
\begin{enumerate}
  \item If $\langle P,s\rangle\xrightarrow{\alpha}\langle P',s'\rangle$ then
  \begin{enumerate}
    \item $fn(\alpha)\subseteq fn(P)$;
    \item $fn(P')\subseteq fn(P)\cup bn(\alpha)$;
  \end{enumerate}
  \item If $\langle P,s\rangle\xrightarrow{\{\alpha_1,\cdots,\alpha_n\}}\langle P',s\rangle$ then
  \begin{enumerate}
    \item $fn(\alpha_1)\cup\cdots\cup fn(\alpha_n)\subseteq fn(P)$;
    \item $fn(P')\subseteq fn(P)\cup bn(\alpha_1)\cup\cdots\cup bn(\alpha_n)$.
  \end{enumerate}
\end{enumerate}
\end{proposition}

\begin{proof}
By induction on the depth of inference.
\end{proof}

\begin{proposition}
Suppose that $\langle P,s\rangle\xrightarrow{\alpha(y)}\langle P',s'\rangle$, where $\alpha=x$ or $\alpha=\overline{x}$, and $x\notin n(P)$, then there exists some $P''\equiv_{\alpha}P'\{z/y\}$,
$\langle P,s\rangle\xrightarrow{\alpha(z)}\langle P'',s''\rangle$.
\end{proposition}

\begin{proof}
By induction on the depth of inference.
\end{proof}

\begin{proposition}
If $\langle P,s\rangle\xrightarrow{\alpha} \langle P',s'\rangle$, $bn(\alpha)\cap fn(P'\sigma)=\emptyset$, and $\sigma\lceil bn(\alpha)=id$, then there exists some $P''\equiv_{\alpha}P'\sigma$,
$\langle P,s\rangle\sigma\xrightarrow{\alpha\sigma}\langle P'',s''\rangle$.
\end{proposition}

\begin{proof}
By the definition of substitution (Definition \ref{subs8}) and induction on the depth of inference.
\end{proof}

\begin{proposition}
\begin{enumerate}
  \item If $\langle P\{w/z\},s\rangle\xrightarrow{\alpha}\langle P',s'\rangle$, where $w\notin fn(P)$ and $bn(\alpha)\cap fn(P,w)=\emptyset$, then there exist some $Q$ and $\beta$ with $Q\{w/z\}\equiv_{\alpha}P'$ and
  $\beta\sigma=\alpha$, $\langle P,s\rangle\xrightarrow{\beta}\langle Q,s'\rangle$;
  \item If $\langle P\{w/z\},s\rangle\xrightarrow{\{\alpha_1,\cdots,\alpha_n\}}\langle P',s'\rangle$, where $w\notin fn(P)$ and $bn(\alpha_1)\cap\cdots\cap bn(\alpha_n)\cap fn(P,w)=\emptyset$, then there exist some $Q$
  and $\beta_1,\cdots,\beta_n$ with $Q\{w/z\}\equiv_{\alpha}P'$ and $\beta_1\sigma=\alpha_1,\cdots,\beta_n\sigma=\alpha_n$, $\langle P,s\rangle\xrightarrow{\{\beta_1,\cdots,\beta_n\}}\langle Q,s'\rangle$.
\end{enumerate}

\end{proposition}

\begin{proof}
By the definition of substitution (Definition \ref{subs8}) and induction on the depth of inference.
\end{proof}

\begin{proposition}
\begin{enumerate}
  \item If $\langle P,s\rangle\xtworightarrow{\alpha[m]}\langle P',s'\rangle$ then
  \begin{enumerate}
    \item $fn(\alpha[m])\subseteq fn(P)$;
    \item $fn(P')\subseteq fn(P)\cup bn(\alpha[m])$;
  \end{enumerate}
  \item If $\langle P,s\rangle\xtworightarrow{\{\alpha_1[m],\cdots,\alpha_n[m]\}}\langle P',s\rangle$ then
  \begin{enumerate}
    \item $fn(\alpha_1[m])\cup\cdots\cup fn(\alpha_n[m])\subseteq fn(P)$;
    \item $fn(P')\subseteq fn(P)\cup bn(\alpha_1[m])\cup\cdots\cup bn(\alpha_n[m])$.
  \end{enumerate}
\end{enumerate}
\end{proposition}

\begin{proof}
By induction on the depth of inference.
\end{proof}

\begin{proposition}
Suppose that $\langle P,s\rangle\xtworightarrow{\alpha(y)[m]}\langle P',s'\rangle$, where $\alpha=x$ or $\alpha=\overline{x}$, and $x\notin n(P)$, then there exists some $P''\equiv_{\alpha}P'\{z/y\}$,
$\langle P,s\rangle\xtworightarrow{\alpha(z)[m]}\langle P'',s''\rangle$.
\end{proposition}

\begin{proof}
By induction on the depth of inference.
\end{proof}

\begin{proposition}
If $\langle P,s\rangle\xtworightarrow{\alpha[m]} \langle P',s'\rangle$, $bn(\alpha[m])\cap fn(P'\sigma)=\emptyset$, and $\sigma\lceil bn(\alpha[m])=id$, then there exists some $P''\equiv_{\alpha}P'\sigma$,
$\langle P,s\rangle\sigma\xtworightarrow{\alpha[m]\sigma}\langle P'',s''\rangle$.
\end{proposition}

\begin{proof}
By the definition of substitution (Definition \ref{subs8}) and induction on the depth of inference.
\end{proof}

\begin{proposition}
\begin{enumerate}
  \item If $\langle P\{w/z\},s\rangle\xtworightarrow{\alpha[m]}\langle P',s'\rangle$, where $w\notin fn(P)$ and $bn(\alpha)\cap fn(P,w)=\emptyset$, then there exist some $Q$ and $\beta$ with $Q\{w/z\}\equiv_{\alpha}P'$ and
  $\beta\sigma[m]=\alpha[m]$, $\langle P,s\rangle\xtworightarrow{\beta[m]}\langle Q,s'\rangle$;
  \item If $\langle P\{w/z\},s\rangle\xtworightarrow{\{\alpha_1[m],\cdots,\alpha_n[m]\}}\langle P',s'\rangle$, where $w\notin fn(P)$ and $bn(\alpha_1[m])\cap\cdots\cap bn(\alpha_n[m])\cap fn(P,w)=\emptyset$, then there exist some $Q$
  and $\beta_1[m],\cdots,\beta_n[m]$ with $Q\{w/z\}\equiv_{\alpha}P'$ and $\beta_1\sigma[m]=\alpha_1[m],\cdots,\beta_n\sigma[m]=\alpha_n[m]$, $\langle P,s\rangle\xtworightarrow{\{\beta_1[m],\cdots,\beta_n[m]\}}\langle Q,s'\rangle$.
\end{enumerate}

\end{proposition}

\begin{proof}
By the definition of substitution (Definition \ref{subs8}) and induction on the depth of inference.
\end{proof}

\subsection{Strong Bisimilarities}\label{s8}

\subsubsection{Laws and Congruence}

\begin{theorem}
$\equiv_{\alpha}$ are FR strongly probabilistic truly concurrent bisimulations. That is, if $P\equiv_{\alpha}Q$, then,
\begin{enumerate}
  \item $P\sim_{pp}^{fr} Q$;
  \item $P\sim_{ps}^{fr} Q$;
  \item $P\sim_{php}^{fr} Q$;
  \item $P\sim_{phhp}^{fr} Q$.
\end{enumerate}
\end{theorem}

\begin{proof}
By induction on the depth of inference, we can get the following facts:

\begin{enumerate}
  \item If $\alpha$ is a free action and $\langle P,s\rangle\rightsquigarrow\xrightarrow{\alpha}\langle P',s'\rangle$, then equally for some $Q'$ with $P'\equiv_{\alpha}Q'$, 
  $\langle Q,s\rangle\rightsquigarrow\xrightarrow{\alpha}\langle Q',s'\rangle$;
  \item If $\langle P,s\rangle\rightsquigarrow\xrightarrow{a(y)}\langle P',s'\rangle$ with $a=x$ or $a=\overline{x}$ and $z\notin n(Q)$, then equally for some $Q'$ with $P'\{z/y\}\equiv_{\alpha}Q'$,
  $\langle Q,s\rangle\rightsquigarrow\xrightarrow{a(z)}\langle Q',s'\rangle$;
  \item If $\alpha[m]$ is a free action and $\langle P,s\rangle\rightsquigarrow\xtworightarrow{\alpha[m]}\langle P',s'\rangle$, then equally for some $Q'$ with $P'\equiv_{\alpha}Q'$, 
  $\langle Q,s\rangle\rightsquigarrow\xtworightarrow{\alpha[m]}\langle Q',s'\rangle$;
  \item If $\langle P,s\rangle\rightsquigarrow\xtworightarrow{a(y)[m]}\langle P',s'\rangle$ with $a=x$ or $a=\overline{x}$ and $z\notin n(Q)$, then equally for some $Q'$ with $P'\{z/y\}\equiv_{\alpha}Q'$,
  $\langle Q,s\rangle\rightsquigarrow\xtworightarrow{a(z)[m]}\langle Q',s'\rangle$.
\end{enumerate}

Then, we can get:

\begin{enumerate}
  \item by the definition of FR strongly probabilistic pomset bisimilarity, $P\sim_{pp}^{fr} Q$;
  \item by the definition of FR strongly probabilistic step bisimilarity, $P\sim_{ps}^{fr} Q$;
  \item by the definition of FR strongly probabilistic hp-bisimilarity, $P\sim_{php}^{fr} Q$;
  \item by the definition of FR strongly probabilistic hhp-bisimilarity, $P\sim_{phhp}^{fr} Q$.
\end{enumerate}
\end{proof}

\begin{proposition}[Summation laws for FR strongly probabilistic pomset bisimulation] The Summation laws for FR strongly probabilistic pomset bisimulation are as follows.

\begin{enumerate}
  \item $P+Q\sim_{pp}^{fr} Q+P$;
  \item $P+(Q+R)\sim_{pp}^{fr} (P+Q)+R$;
  \item $P+P\sim_{pp}^{fr} P$;
  \item $P+\textbf{nil}\sim_{pp}^{fr} P$.
\end{enumerate}

\end{proposition}

\begin{proof}
\begin{enumerate}
  \item $P+Q\sim_{pp}^{fr} Q+P$. It is sufficient to prove the relation $R=\{(P+Q, Q+P)\}\cup \textbf{Id}$ is a FR strongly probabilistic pomset bisimulation, we omit it;
  \item $P+(Q+R)\sim_{pp}^{fr} (P+Q)+R$. It is sufficient to prove the relation $R=\{(P+(Q+R), (P+Q)+R)\}\cup \textbf{Id}$ is a FR strongly probabilistic pomset bisimulation, we omit it;
  \item $P+P\sim_{pp}^{fr} P$. It is sufficient to prove the relation $R=\{(P+P, P)\}\cup \textbf{Id}$ is a FR strongly probabilistic pomset bisimulation, we omit it;
  \item $P+\textbf{nil}\sim_{pp}^{fr} P$. It is sufficient to prove the relation $R=\{(P+\textbf{nil}, P)\}\cup \textbf{Id}$ is a FR strongly probabilistic pomset bisimulation, we omit it.
\end{enumerate}
\end{proof}

\begin{proposition}[Summation laws for FR strongly probabilistic step bisimulation] The Summation laws for FR strongly probabilistic step bisimulation are as follows.
\begin{enumerate}
  \item $P+Q\sim_{ps}^{fr} Q+P$;
  \item $P+(Q+R)\sim_{ps}^{fr} (P+Q)+R$;
  \item $P+P\sim_{ps}^{fr} P$;
  \item $P+\textbf{nil}\sim_{ps}^{fr} P$.
\end{enumerate}
\end{proposition}

\begin{proof}
\begin{enumerate}
  \item $P+Q\sim_{ps}^{fr} Q+P$. It is sufficient to prove the relation $R=\{(P+Q, Q+P)\}\cup \textbf{Id}$ is a FR strongly probabilistic step bisimulation, we omit it;
  \item $P+(Q+R)\sim_{ps}^{fr} (P+Q)+R$. It is sufficient to prove the relation $R=\{(P+(Q+R), (P+Q)+R)\}\cup \textbf{Id}$ is a FR strongly probabilistic step bisimulation, we omit it;
  \item $P+P\sim_{ps}^{fr} P$. It is sufficient to prove the relation $R=\{(P+P, P)\}\cup \textbf{Id}$ is a FR strongly probabilistic step bisimulation, we omit it;
  \item $P+\textbf{nil}\sim_{ps}^{fr} P$. It is sufficient to prove the relation $R=\{(P+\textbf{nil}, P)\}\cup \textbf{Id}$ is a FR strongly probabilistic step bisimulation, we omit it.
\end{enumerate}
\end{proof}

\begin{proposition}[Summation laws for FR strongly probabilistic hp-bisimulation] The Summation laws for FR strongly probabilistic hp-bisimulation are as follows.
\begin{enumerate}
  \item $P+Q\sim_{php}^{fr} Q+P$;
  \item $P+(Q+R)\sim_{php}^{fr} (P+Q)+R$;
  \item $P+P\sim_{php}^{fr} P$;
  \item $P+\textbf{nil}\sim_{php}^{fr} P$.
\end{enumerate}
\end{proposition}

\begin{proof}
\begin{enumerate}
  \item $P+Q\sim_{php}^{fr} Q+P$. It is sufficient to prove the relation $R=\{(P+Q, Q+P)\}\cup \textbf{Id}$ is a FR strongly probabilistic hp-bisimulation, we omit it;
  \item $P+(Q+R)\sim_{php}^{fr} (P+Q)+R$. It is sufficient to prove the relation $R=\{(P+(Q+R), (P+Q)+R)\}\cup \textbf{Id}$ is a FR strongly probabilistic hp-bisimulation, we omit it;
  \item $P+P\sim_{php}^{fr} P$. It is sufficient to prove the relation $R=\{(P+P, P)\}\cup \textbf{Id}$ is a FR strongly probabilistic hp-bisimulation, we omit it;
  \item $P+\textbf{nil}\sim_{php}^{fr} P$. It is sufficient to prove the relation $R=\{(P+\textbf{nil}, P)\}\cup \textbf{Id}$ is a FR strongly probabilistic hp-bisimulation, we omit it.
\end{enumerate}
\end{proof}

\begin{proposition}[Summation laws for FR strongly probabilistic hhp-bisimulation] The Summation laws for FR strongly probabilistic hhp-bisimulation are as follows.
\begin{enumerate}
  \item $P+Q\sim_{phhp}^{fr} Q+P$;
  \item $P+(Q+R)\sim_{phhp}^{fr} (P+Q)+R$;
  \item $P+P\sim_{phhp}^{fr} P$;
  \item $P+\textbf{nil}\sim_{phhp}^{fr} P$.
\end{enumerate}
\end{proposition}

\begin{proof}
\begin{enumerate}
  \item $P+Q\sim_{phhp}^{fr} Q+P$. It is sufficient to prove the relation $R=\{(P+Q, Q+P)\}\cup \textbf{Id}$ is a FR strongly probabilistic hhp-bisimulation, we omit it;
  \item $P+(Q+R)\sim_{phhp}^{fr} (P+Q)+R$. It is sufficient to prove the relation $R=\{(P+(Q+R), (P+Q)+R)\}\cup \textbf{Id}$ is a FR strongly probabilistic hhp-bisimulation, we omit it;
  \item $P+P\sim_{phhp}^{fr} P$. It is sufficient to prove the relation $R=\{(P+P, P)\}\cup \textbf{Id}$ is a FR strongly probabilistic hhp-bisimulation, we omit it;
  \item $P+\textbf{nil}\sim_{phhp}^{fr} P$. It is sufficient to prove the relation $R=\{(P+\textbf{nil}, P)\}\cup \textbf{Id}$ is a FR strongly probabilistic hhp-bisimulation, we omit it.
\end{enumerate}
\end{proof}

\begin{proposition}[Box-Summation laws for FR strongly probabilistic pomset bisimulation]
The Box-Summation laws for FR strongly probabilistic pomset bisimulation are as follows.

\begin{enumerate}
  \item $P\boxplus_{\pi} Q\sim_{pp}^{fr} Q\boxplus_{1-\pi} P$;
  \item $P\boxplus_{\pi}(Q\boxplus_{\rho} R)\sim_{pp}^{fr} (P\boxplus_{\frac{\pi}{\pi+\rho-\pi\rho}}Q)\boxplus_{\pi+\rho-\pi\rho} R$;
  \item $P\boxplus_{\pi}P\sim_{pp}^{fr} P$;
  \item $P\boxplus_{\pi}\textbf{nil}\sim_{pp}^{fr} P$.
\end{enumerate}
\end{proposition}

\begin{proof}
\begin{enumerate}
  \item $P\boxplus_{\pi} Q\sim_{pp}^{fr} Q\boxplus_{1-\pi} P$. It is sufficient to prove the relation $R=\{(P\boxplus_{\pi} Q, Q\boxplus_{1-\pi} P)\}\cup \textbf{Id}$ is a FR strongly probabilistic pomset bisimulation, we omit it;
  \item $P\boxplus_{\pi}(Q\boxplus_{\rho} R)\sim_{pp}^{fr} (P\boxplus_{\frac{\pi}{\pi+\rho-\pi\rho}}Q)\boxplus_{\pi+\rho-\pi\rho} R$. It is sufficient to prove the relation $R=\{(P\boxplus_{\pi}(Q\boxplus_{\rho} R), (P\boxplus_{\frac{\pi}{\pi+\rho-\pi\rho}}Q)\boxplus_{\pi+\rho-\pi\rho} R)\}\cup \textbf{Id}$ is a FR strongly probabilistic pomset bisimulation, we omit it;
  \item $P\boxplus_{\pi}P\sim_{pp}^{fr} P$. It is sufficient to prove the relation $R=\{(P\boxplus_{\pi}P, P)\}\cup \textbf{Id}$ is a FR strongly probabilistic pomset bisimulation, we omit it;
  \item $P\boxplus_{\pi}\textbf{nil}\sim_{pp}^{fr} P$. It is sufficient to prove the relation $R=\{(P\boxplus_{\pi}\textbf{nil}, P)\}\cup \textbf{Id}$ is a FR strongly probabilistic pomset bisimulation, we omit it.
\end{enumerate}
\end{proof}

\begin{proposition}[Box-Summation laws for FR strongly probabilistic step bisimulation]
The Box-Summation laws for FR strongly probabilistic step bisimulation are as follows.

\begin{enumerate}
  \item $P\boxplus_{\pi} Q\sim_{ps}^{fr} Q\boxplus_{1-\pi} P$;
  \item $P\boxplus_{\pi}(Q\boxplus_{\rho} R)\sim_{ps}^{fr} (P\boxplus_{\frac{\pi}{\pi+\rho-\pi\rho}}Q)\boxplus_{\pi+\rho-\pi\rho} R$;
  \item $P\boxplus_{\pi}P\sim_{ps}^{fr} P$;
  \item $P\boxplus_{\pi}\textbf{nil}\sim_{ps}^{fr} P$.
\end{enumerate}
\end{proposition}

\begin{proof}
\begin{enumerate}
  \item $P\boxplus_{\pi} Q\sim_{ps}^{fr} Q\boxplus_{1-\pi} P$. It is sufficient to prove the relation $R=\{(P\boxplus_{\pi} Q, Q\boxplus_{1-\pi} P)\}\cup \textbf{Id}$ is a FR strongly probabilistic step bisimulation, we omit it;
  \item $P\boxplus_{\pi}(Q\boxplus_{\rho} R)\sim_{ps}^{fr} (P\boxplus_{\frac{\pi}{\pi+\rho-\pi\rho}}Q)\boxplus_{\pi+\rho-\pi\rho} R$. It is sufficient to prove the relation $R=\{(P\boxplus_{\pi}(Q\boxplus_{\rho} R), (P\boxplus_{\frac{\pi}{\pi+\rho-\pi\rho}}Q)\boxplus_{\pi+\rho-\pi\rho} R)\}\cup \textbf{Id}$ is a FR strongly probabilistic step bisimulation, we omit it;
  \item $P\boxplus_{\pi}P\sim_{ps}^{fr} P$. It is sufficient to prove the relation $R=\{(P\boxplus_{\pi}P, P)\}\cup \textbf{Id}$ is a FR strongly probabilistic step bisimulation, we omit it;
  \item $P\boxplus_{\pi}\textbf{nil}\sim_{ps}^{fr} P$. It is sufficient to prove the relation $R=\{(P\boxplus_{\pi}\textbf{nil}, P)\}\cup \textbf{Id}$ is a FR strongly probabilistic step bisimulation, we omit it.
\end{enumerate}
\end{proof}

\begin{proposition}[Box-Summation laws for FR strongly probabilistic hp-bisimulation]
The Box-Summation laws for FR strongly probabilistic hp-bisimulation are as follows.

\begin{enumerate}
  \item $P\boxplus_{\pi} Q\sim_{php}^{fr} Q\boxplus_{1-\pi} P$;
  \item $P\boxplus_{\pi}(Q\boxplus_{\rho} R)\sim_{php}^{fr} (P\boxplus_{\frac{\pi}{\pi+\rho-\pi\rho}}Q)\boxplus_{\pi+\rho-\pi\rho} R$;
  \item $P\boxplus_{\pi}P\sim_{php}^{fr} P$;
  \item $P\boxplus_{\pi}\textbf{nil}\sim_{php}^{fr} P$.
\end{enumerate}
\end{proposition}

\begin{proof}
\begin{enumerate}
  \item $P\boxplus_{\pi} Q\sim_{php}^{fr} Q\boxplus_{1-\pi} P$. It is sufficient to prove the relation $R=\{(P\boxplus_{\pi} Q, Q\boxplus_{1-\pi} P)\}\cup \textbf{Id}$ is a FR strongly probabilistic hp-bisimulation, we omit it;
  \item $P\boxplus_{\pi}(Q\boxplus_{\rho} R)\sim_{php}^{fr} (P\boxplus_{\frac{\pi}{\pi+\rho-\pi\rho}}Q)\boxplus_{\pi+\rho-\pi\rho} R$. It is sufficient to prove the relation $R=\{(P\boxplus_{\pi}(Q\boxplus_{\rho} R), (P\boxplus_{\frac{\pi}{\pi+\rho-\pi\rho}}Q)\boxplus_{\pi+\rho-\pi\rho} R)\}\cup \textbf{Id}$ is a FR strongly probabilistic hp-bisimulation, we omit it;
  \item $P\boxplus_{\pi}P\sim_{php}^{fr} P$. It is sufficient to prove the relation $R=\{(P\boxplus_{\pi}P, P)\}\cup \textbf{Id}$ is a FR strongly probabilistic hp-bisimulation, we omit it;
  \item $P\boxplus_{\pi}\textbf{nil}\sim_{php}^{fr} P$. It is sufficient to prove the relation $R=\{(P\boxplus_{\pi}\textbf{nil}, P)\}\cup \textbf{Id}$ is a FR strongly probabilistic hp-bisimulation, we omit it.
\end{enumerate}
\end{proof}

\begin{proposition}[Box-Summation laws for FR strongly probabilistic hhp-bisimulation]
The Box-Summation laws for FR strongly probabilistic hhp-bisimulation are as follows.

\begin{enumerate}
  \item $P\boxplus_{\pi} Q\sim_{phhp}^{fr} Q\boxplus_{1-\pi} P$;
  \item $P\boxplus_{\pi}(Q\boxplus_{\rho} R)\sim_{phhp}^{fr} (P\boxplus_{\frac{\pi}{\pi+\rho-\pi\rho}}Q)\boxplus_{\pi+\rho-\pi\rho} R$;
  \item $P\boxplus_{\pi}P\sim_{phhp}^{fr} P$;
  \item $P\boxplus_{\pi}\textbf{nil}\sim_{phhp}^{fr} P$.
\end{enumerate}
\end{proposition}

\begin{proof}
\begin{enumerate}
  \item $P\boxplus_{\pi} Q\sim_{phhp}^{fr} Q\boxplus_{1-\pi} P$. It is sufficient to prove the relation $R=\{(P\boxplus_{\pi} Q, Q\boxplus_{1-\pi} P)\}\cup \textbf{Id}$ is a FR strongly probabilistic hhp-bisimulation, we omit it;
  \item $P\boxplus_{\pi}(Q\boxplus_{\rho} R)\sim_{phhp}^{fr} (P\boxplus_{\frac{\pi}{\pi+\rho-\pi\rho}}Q)\boxplus_{\pi+\rho-\pi\rho} R$. It is sufficient to prove the relation $R=\{(P\boxplus_{\pi}(Q\boxplus_{\rho} R), (P\boxplus_{\frac{\pi}{\pi+\rho-\pi\rho}}Q)\boxplus_{\pi+\rho-\pi\rho} R)\}\cup \textbf{Id}$ is a FR strongly probabilistic hhp-bisimulation, we omit it;
  \item $P\boxplus_{\pi}P\sim_{phhp}^{fr} P$. It is sufficient to prove the relation $R=\{(P\boxplus_{\pi}P, P)\}\cup \textbf{Id}$ is a FR strongly probabilistic hhp-bisimulation, we omit it;
  \item $P\boxplus_{\pi}\textbf{nil}\sim_{phhp}^{fr} P$. It is sufficient to prove the relation $R=\{(P\boxplus_{\pi}\textbf{nil}, P)\}\cup \textbf{Id}$ is a FR strongly probabilistic hhp-bisimulation, we omit it.
\end{enumerate}
\end{proof}

\begin{theorem}[Identity law for FR strongly probabilistic truly concurrent bisimilarities]
If $A(\widetilde{x})\overset{\text{def}}{=}P$, then

\begin{enumerate}
  \item $A(\widetilde{y})\sim_{pp}^{fr} P\{\widetilde{y}/\widetilde{x}\}$;
  \item $A(\widetilde{y})\sim_{ps}^{fr} P\{\widetilde{y}/\widetilde{x}\}$;
  \item $A(\widetilde{y})\sim_{php}^{fr} P\{\widetilde{y}/\widetilde{x}\}$;
  \item $A(\widetilde{y})\sim_{phhp}^{fr} P\{\widetilde{y}/\widetilde{x}\}$.
\end{enumerate}
\end{theorem}

\begin{proof}
\begin{enumerate}
  \item $A(\widetilde{y})\sim_{pp}^{fr} P\{\widetilde{y}/\widetilde{x}\}$. It is sufficient to prove the relation $R=\{(A(\widetilde{y}), P\{\widetilde{y}/\widetilde{x}\})\}\cup \textbf{Id}$ is a FR strongly probabilistic pomset bisimulation, we omit it;
  \item $A(\widetilde{y})\sim_{ps}^{fr} P\{\widetilde{y}/\widetilde{x}\}$. It is sufficient to prove the relation $R=\{(A(\widetilde{y}), P\{\widetilde{y}/\widetilde{x}\})\}\cup \textbf{Id}$ is a FR strongly probabilistic step bisimulation, we omit it;
  \item $A(\widetilde{y})\sim_{php}^{fr} P\{\widetilde{y}/\widetilde{x}\}$. It is sufficient to prove the relation $R=\{(A(\widetilde{y}), P\{\widetilde{y}/\widetilde{x}\})\}\cup \textbf{Id}$ is a FR strongly probabilistic hp-bisimulation, we omit it;
  \item $A(\widetilde{y})\sim_{phhp}^{fr} P\{\widetilde{y}/\widetilde{x}\}$. It is sufficient to prove the relation $R=\{(A(\widetilde{y}), P\{\widetilde{y}/\widetilde{x}\})\}\cup \textbf{Id}$ is a FR strongly probabilistic hhp-bisimulation, we omit it.
\end{enumerate}
\end{proof}

\begin{theorem}[Restriction Laws for FR strongly probabilistic pomset bisimilarity]
The restriction laws for FR strongly probabilistic pomset bisimilarity are as follows.

\begin{enumerate}
  \item $(y)P\sim_{pp}^{fr} P$, if $y\notin fn(P)$;
  \item $(y)(z)P\sim_{pp}^{fr} (z)(y)P$;
  \item $(y)(P+Q)\sim_{pp}^{fr} (y)P+(y)Q$;
  \item $(y)(P\boxplus_{\pi}Q)\sim_{pp}^{fr} (y)P\boxplus_{\pi}(y)Q$;
  \item $(y)\alpha.P\sim_{pp}^{fr} \alpha.(y)P$ if $y\notin n(\alpha)$;
  \item $(y)\alpha.P\sim_{pp}^{fr} \textbf{nil}$ if $y$ is the subject of $\alpha$.
\end{enumerate}
\end{theorem}

\begin{proof}
\begin{enumerate}
  \item $(y)P\sim_{pp}^{fr} P$, if $y\notin fn(P)$. It is sufficient to prove the relation $R=\{((y)P, P)\}\cup \textbf{Id}$, if $y\notin fn(P)$, is a FR strongly probabilistic pomset bisimulation, we omit it;
  \item $(y)(z)P\sim_{pp}^{fr} (z)(y)P$. It is sufficient to prove the relation $R=\{((y)(z)P, (z)(y)P)\}\cup \textbf{Id}$ is a FR strongly probabilistic pomset bisimulation, we omit it;
  \item $(y)(P+Q)\sim_{pp}^{fr} (y)P+(y)Q$. It is sufficient to prove the relation $R=\{((y)(P+Q), (y)P+(y)Q)\}\cup \textbf{Id}$ is a FR strongly probabilistic pomset bisimulation, we omit it;
  \item $(y)(P\boxplus_{\pi}Q)\sim_{pp}^{fr} (y)P\boxplus_{\pi}(y)Q$. It is sufficient to prove the relation $R=\{((y)(P\boxplus_{\pi}Q), (y)P\boxplus_{\pi}(y)Q)\}\cup \textbf{Id}$ is a FR strongly probabilistic pomset bisimulation, we omit it;
  \item $(y)\alpha.P\sim_{pp}^{fr} \alpha.(y)P$ if $y\notin n(\alpha)$. It is sufficient to prove the relation $R=\{((y)\alpha.P, \alpha.(y)P)\}\cup \textbf{Id}$, if $y\notin n(\alpha)$, is a FR strongly probabilistic pomset bisimulation, we omit it;
  \item $(y)\alpha.P\sim_{pp}^{fr} \textbf{nil}$ if $y$ is the subject of $\alpha$. It is sufficient to prove the relation $R=\{((y)\alpha.P, \textbf{nil})\}\cup \textbf{Id}$, if $y$ is the subject of $\alpha$, is a FR strongly probabilistic pomset bisimulation, we omit it.
\end{enumerate}
\end{proof}

\begin{theorem}[Restriction Laws for FR strongly probabilistic step bisimilarity]
The restriction laws for FR strongly probabilistic step bisimilarity are as follows.

\begin{enumerate}
  \item $(y)P\sim_{ps}^{fr} P$, if $y\notin fn(P)$;
  \item $(y)(z)P\sim_{ps}^{fr} (z)(y)P$;
  \item $(y)(P+Q)\sim_{ps}^{fr} (y)P+(y)Q$;
  \item $(y)(P\boxplus_{\pi}Q)\sim_{ps}^{fr} (y)P\boxplus_{\pi}(y)Q$;
  \item $(y)\alpha.P\sim_{ps}^{fr} \alpha.(y)P$ if $y\notin n(\alpha)$;
  \item $(y)\alpha.P\sim_{ps}^{fr} \textbf{nil}$ if $y$ is the subject of $\alpha$.
\end{enumerate}
\end{theorem}

\begin{proof}
\begin{enumerate}
  \item $(y)P\sim_{ps}^{fr} P$, if $y\notin fn(P)$. It is sufficient to prove the relation $R=\{((y)P, P)\}\cup \textbf{Id}$, if $y\notin fn(P)$, is a FR strongly probabilistic step bisimulation, we omit it;
  \item $(y)(z)P\sim_{ps}^{fr} (z)(y)P$. It is sufficient to prove the relation $R=\{((y)(z)P, (z)(y)P)\}\cup \textbf{Id}$ is a FR strongly probabilistic step bisimulation, we omit it;
  \item $(y)(P+Q)\sim_{ps}^{fr} (y)P+(y)Q$. It is sufficient to prove the relation $R=\{((y)(P+Q), (y)P+(y)Q)\}\cup \textbf{Id}$ is a FR strongly probabilistic step bisimulation, we omit it;
  \item $(y)(P\boxplus_{\pi}Q)\sim_{ps}^{fr} (y)P\boxplus_{\pi}(y)Q$. It is sufficient to prove the relation $R=\{((y)(P\boxplus_{\pi}Q), (y)P\boxplus_{\pi}(y)Q)\}\cup \textbf{Id}$ is a FR strongly probabilistic step bisimulation, we omit it;
  \item $(y)\alpha.P\sim_{ps}^{fr} \alpha.(y)P$ if $y\notin n(\alpha)$. It is sufficient to prove the relation $R=\{((y)\alpha.P, \alpha.(y)P)\}\cup \textbf{Id}$, if $y\notin n(\alpha)$, is a FR strongly probabilistic step bisimulation, we omit it;
  \item $(y)\alpha.P\sim_{ps}^{fr} \textbf{nil}$ if $y$ is the subject of $\alpha$. It is sufficient to prove the relation $R=\{((y)\alpha.P, \textbf{nil})\}\cup \textbf{Id}$, if $y$ is the subject of $\alpha$, is a FR strongly probabilistic step bisimulation, we omit it.
\end{enumerate}
\end{proof}

\begin{theorem}[Restriction Laws for FR strongly probabilistic hp-bisimilarity]
The restriction laws for FR strongly probabilistic hp-bisimilarity are as follows.

\begin{enumerate}
  \item $(y)P\sim_{php}^{fr} P$, if $y\notin fn(P)$;
  \item $(y)(z)P\sim_{php}^{fr} (z)(y)P$;
  \item $(y)(P+Q)\sim_{php}^{fr} (y)P+(y)Q$;
  \item $(y)(P\boxplus_{\pi}Q)\sim_{php}^{fr} (y)P\boxplus_{\pi}(y)Q$;
  \item $(y)\alpha.P\sim_{php}^{fr} \alpha.(y)P$ if $y\notin n(\alpha)$;
  \item $(y)\alpha.P\sim_{php}^{fr} \textbf{nil}$ if $y$ is the subject of $\alpha$.
\end{enumerate}
\end{theorem}

\begin{proof}
\begin{enumerate}
  \item $(y)P\sim_{php}^{fr} P$, if $y\notin fn(P)$. It is sufficient to prove the relation $R=\{((y)P, P)\}\cup \textbf{Id}$, if $y\notin fn(P)$, is a FR strongly probabilistic hp-bisimulation, we omit it;
  \item $(y)(z)P\sim_{php}^{fr} (z)(y)P$. It is sufficient to prove the relation $R=\{((y)(z)P, (z)(y)P)\}\cup \textbf{Id}$ is a FR strongly probabilistic hp-bisimulation, we omit it;
  \item $(y)(P+Q)\sim_{php}^{fr} (y)P+(y)Q$. It is sufficient to prove the relation $R=\{((y)(P+Q), (y)P+(y)Q)\}\cup \textbf{Id}$ is a FR strongly probabilistic hp-bisimulation, we omit it;
  \item $(y)(P\boxplus_{\pi}Q)\sim_{php}^{fr} (y)P\boxplus_{\pi}(y)Q$. It is sufficient to prove the relation $R=\{((y)(P\boxplus_{\pi}Q), (y)P\boxplus_{\pi}(y)Q)\}\cup \textbf{Id}$ is a FR strongly probabilistic hp-bisimulation, we omit it;
  \item $(y)\alpha.P\sim_{php}^{fr} \alpha.(y)P$ if $y\notin n(\alpha)$. It is sufficient to prove the relation $R=\{((y)\alpha.P, \alpha.(y)P)\}\cup \textbf{Id}$, if $y\notin n(\alpha)$, is a FR strongly probabilistic hp-bisimulation, we omit it;
  \item $(y)\alpha.P\sim_{php}^{fr} \textbf{nil}$ if $y$ is the subject of $\alpha$. It is sufficient to prove the relation $R=\{((y)\alpha.P, \textbf{nil})\}\cup \textbf{Id}$, if $y$ is the subject of $\alpha$, is a FR strongly probabilistic hp-bisimulation, we omit it.
\end{enumerate}
\end{proof}

\begin{theorem}[Restriction Laws for FR strongly probabilistic hhp-bisimilarity]
The restriction laws for FR strongly probabilistic hhp-bisimilarity are as follows.

\begin{enumerate}
  \item $(y)P\sim_{phhp}^{fr} P$, if $y\notin fn(P)$;
  \item $(y)(z)P\sim_{phhp}^{fr} (z)(y)P$;
  \item $(y)(P+Q)\sim_{phhp}^{fr} (y)P+(y)Q$;
  \item $(y)(P\boxplus_{\pi}Q)\sim_{phhp}^{fr} (y)P\boxplus_{\pi}(y)Q$;
  \item $(y)\alpha.P\sim_{phhp}^{fr} \alpha.(y)P$ if $y\notin n(\alpha)$;
  \item $(y)\alpha.P\sim_{phhp}^{fr} \textbf{nil}$ if $y$ is the subject of $\alpha$.
\end{enumerate}
\end{theorem}

\begin{proof}
\begin{enumerate}
  \item $(y)P\sim_{phhp}^{fr} P$, if $y\notin fn(P)$. It is sufficient to prove the relation $R=\{((y)P, P)\}\cup \textbf{Id}$, if $y\notin fn(P)$, is a FR strongly probabilistic hhp-bisimulation, we omit it;
  \item $(y)(z)P\sim_{phhp}^{fr} (z)(y)P$. It is sufficient to prove the relation $R=\{((y)(z)P, (z)(y)P)\}\cup \textbf{Id}$ is a FR strongly probabilistic hhp-bisimulation, we omit it;
  \item $(y)(P+Q)\sim_{phhp}^{fr} (y)P+(y)Q$. It is sufficient to prove the relation $R=\{((y)(P+Q), (y)P+(y)Q)\}\cup \textbf{Id}$ is a FR strongly probabilistic hhp-bisimulation, we omit it;
  \item $(y)(P\boxplus_{\pi}Q)\sim_{phhp}^{fr} (y)P\boxplus_{\pi}(y)Q$. It is sufficient to prove the relation $R=\{((y)(P\boxplus_{\pi}Q), (y)P\boxplus_{\pi}(y)Q)\}\cup \textbf{Id}$ is a FR strongly probabilistic hhp-bisimulation, we omit it;
  \item $(y)\alpha.P\sim_{phhp}^{fr} \alpha.(y)P$ if $y\notin n(\alpha)$. It is sufficient to prove the relation $R=\{((y)\alpha.P, \alpha.(y)P)\}\cup \textbf{Id}$, if $y\notin n(\alpha)$, is a FR strongly probabilistic hhp-bisimulation, we omit it;
  \item $(y)\alpha.P\sim_{phhp}^{fr} \textbf{nil}$ if $y$ is the subject of $\alpha$. It is sufficient to prove the relation $R=\{((y)\alpha.P, \textbf{nil})\}\cup \textbf{Id}$, if $y$ is the subject of $\alpha$, is a FR strongly probabilistic hhp-bisimulation, we omit it.
\end{enumerate}
\end{proof}

\begin{theorem}[Parallel laws for FR strongly probabilistic pomset bisimilarity]
The parallel laws for FR strongly probabilistic pomset bisimilarity are as follows.

\begin{enumerate}
  \item $P\parallel \textbf{nil}\sim_{pp}^{fr} P$;
  \item $P_1\parallel P_2\sim_{pp}^{fr} P_2\parallel P_1$;
  \item $(P_1\parallel P_2)\parallel P_3\sim_{pp}^{fr} P_1\parallel (P_2\parallel P_3)$;
  \item $(y)(P_1\parallel P_2)\sim_{pp}^{fr} (y)P_1\parallel (y)P_2$, if $y\notin fn(P_1)\cap fn(P_2)$.
\end{enumerate}
\end{theorem}

\begin{proof}
\begin{enumerate}
  \item $P\parallel \textbf{nil}\sim_{pp}^{fr} P$. It is sufficient to prove the relation $R=\{(P\parallel \textbf{nil}, P)\}\cup \textbf{Id}$ is a FR strongly probabilistic pomset bisimulation, we omit it;
  \item $P_1\parallel P_2\sim_{pp}^{fr} P_2\parallel P_1$. It is sufficient to prove the relation $R=\{(P_1\parallel P_2, P_2\parallel P_1)\}\cup \textbf{Id}$ is a FR strongly probabilistic pomset bisimulation, we omit it;
  \item $(P_1\parallel P_2)\parallel P_3\sim_{pp}^{fr} P_1\parallel (P_2\parallel P_3)$. It is sufficient to prove the relation $R=\{((P_1\parallel P_2)\parallel P_3, P_1\parallel (P_2\parallel P_3))\}\cup \textbf{Id}$ is a FR strongly probabilistic pomset bisimulation, we omit it;
  \item $(y)(P_1\parallel P_2)\sim_{pp}^{fr} (y)P_1\parallel (y)P_2$, if $y\notin fn(P_1)\cap fn(P_2)$. It is sufficient to prove the relation $R=\{((y)(P_1\parallel P_2), (y)P_1\parallel (y)P_2)\}\cup \textbf{Id}$, if $y\notin fn(P_1)\cap fn(P_2)$, is a FR strongly probabilistic pomset bisimulation, we omit it.
\end{enumerate}
\end{proof}

\begin{theorem}[Parallel laws for FR strongly probabilistic step bisimilarity]
The parallel laws for FR strongly probabilistic step bisimilarity are as follows.

\begin{enumerate}
  \item $P\parallel \textbf{nil}\sim_{ps}^{fr} P$;
  \item $P_1\parallel P_2\sim_{ps}^{fr} P_2\parallel P_1$;
  \item $(P_1\parallel P_2)\parallel P_3\sim_{ps}^{fr} P_1\parallel (P_2\parallel P_3)$;
  \item $(y)(P_1\parallel P_2)\sim_{ps}^{fr} (y)P_1\parallel (y)P_2$, if $y\notin fn(P_1)\cap fn(P_2)$.
\end{enumerate}
\end{theorem}

\begin{proof}
\begin{enumerate}
  \item $P\parallel \textbf{nil}\sim_{ps}^{fr} P$. It is sufficient to prove the relation $R=\{(P\parallel \textbf{nil}, P)\}\cup \textbf{Id}$ is a FR strongly probabilistic step bisimulation, we omit it;
  \item $P_1\parallel P_2\sim_{ps}^{fr} P_2\parallel P_1$. It is sufficient to prove the relation $R=\{(P_1\parallel P_2, P_2\parallel P_1)\}\cup \textbf{Id}$ is a FR strongly probabilistic step bisimulation, we omit it;
  \item $(P_1\parallel P_2)\parallel P_3\sim_{ps}^{fr} P_1\parallel (P_2\parallel P_3)$. It is sufficient to prove the relation $R=\{((P_1\parallel P_2)\parallel P_3, P_1\parallel (P_2\parallel P_3))\}\cup \textbf{Id}$ is a FR strongly probabilistic step bisimulation, we omit it;
  \item $(y)(P_1\parallel P_2)\sim_{ps}^{fr} (y)P_1\parallel (y)P_2$, if $y\notin fn(P_1)\cap fn(P_2)$. It is sufficient to prove the relation $R=\{((y)(P_1\parallel P_2), (y)P_1\parallel (y)P_2)\}\cup \textbf{Id}$, if $y\notin fn(P_1)\cap fn(P_2)$, is a FR strongly probabilistic step bisimulation, we omit it.
\end{enumerate}
\end{proof}

\begin{theorem}[Parallel laws for FR strongly probabilistic hp-bisimilarity]
The parallel laws for FR strongly probabilistic hp-bisimilarity are as follows.

\begin{enumerate}
  \item $P\parallel \textbf{nil}\sim_{php}^{fr} P$;
  \item $P_1\parallel P_2\sim_{php}^{fr} P_2\parallel P_1$;
  \item $(P_1\parallel P_2)\parallel P_3\sim_{php}^{fr} P_1\parallel (P_2\parallel P_3)$;
  \item $(y)(P_1\parallel P_2)\sim_{php}^{fr} (y)P_1\parallel (y)P_2$, if $y\notin fn(P_1)\cap fn(P_2)$.
\end{enumerate}
\end{theorem}

\begin{proof}
\begin{enumerate}
  \item $P\parallel \textbf{nil}\sim_{php}^{fr} P$. It is sufficient to prove the relation $R=\{(P\parallel \textbf{nil}, P)\}\cup \textbf{Id}$ is a FR strongly probabilistic hp-bisimulation, we omit it;
  \item $P_1\parallel P_2\sim_{php}^{fr} P_2\parallel P_1$. It is sufficient to prove the relation $R=\{(P_1\parallel P_2, P_2\parallel P_1)\}\cup \textbf{Id}$ is a FR strongly probabilistic hp-bisimulation, we omit it;
  \item $(P_1\parallel P_2)\parallel P_3\sim_{php}^{fr} P_1\parallel (P_2\parallel P_3)$. It is sufficient to prove the relation $R=\{((P_1\parallel P_2)\parallel P_3, P_1\parallel (P_2\parallel P_3))\}\cup \textbf{Id}$ is a FR strongly probabilistic hp-bisimulation, we omit it;
  \item $(y)(P_1\parallel P_2)\sim_{php}^{fr} (y)P_1\parallel (y)P_2$, if $y\notin fn(P_1)\cap fn(P_2)$. It is sufficient to prove the relation $R=\{((y)(P_1\parallel P_2), (y)P_1\parallel (y)P_2)\}\cup \textbf{Id}$, if $y\notin fn(P_1)\cap fn(P_2)$, is a FR strongly probabilistic hp-bisimulation, we omit it.
\end{enumerate}
\end{proof}

\begin{theorem}[Parallel laws for FR strongly probabilistic hhp-bisimilarity]
The parallel laws for FR strongly probabilistic hhp-bisimilarity are as follows.

\begin{enumerate}
  \item $P\parallel \textbf{nil}\sim_{phhp}^{fr} P$;
  \item $P_1\parallel P_2\sim_{phhp}^{fr} P_2\parallel P_1$;
  \item $(P_1\parallel P_2)\parallel P_3\sim_{phhp}^{fr} P_1\parallel (P_2\parallel P_3)$;
  \item $(y)(P_1\parallel P_2)\sim_{phhp}^{fr} (y)P_1\parallel (y)P_2$, if $y\notin fn(P_1)\cap fn(P_2)$.
\end{enumerate}
\end{theorem}

\begin{proof}
\begin{enumerate}
  \item $P\parallel \textbf{nil}\sim_{phhp}^{fr} P$. It is sufficient to prove the relation $R=\{(P\parallel \textbf{nil}, P)\}\cup \textbf{Id}$ is a FR strongly probabilistic hhp-bisimulation, we omit it;
  \item $P_1\parallel P_2\sim_{phhp}^{fr} P_2\parallel P_1$. It is sufficient to prove the relation $R=\{(P_1\parallel P_2, P_2\parallel P_1)\}\cup \textbf{Id}$ is a FR strongly probabilistic hhp-bisimulation, we omit it;
  \item $(P_1\parallel P_2)\parallel P_3\sim_{phhp}^{fr} P_1\parallel (P_2\parallel P_3)$. It is sufficient to prove the relation $R=\{((P_1\parallel P_2)\parallel P_3, P_1\parallel (P_2\parallel P_3))\}\cup \textbf{Id}$ is a FR strongly probabilistic hhp-bisimulation, we omit it;
  \item $(y)(P_1\parallel P_2)\sim_{phhp}^{fr} (y)P_1\parallel (y)P_2$, if $y\notin fn(P_1)\cap fn(P_2)$. It is sufficient to prove the relation $R=\{((y)(P_1\parallel P_2), (y)P_1\parallel (y)P_2)\}\cup \textbf{Id}$, if $y\notin fn(P_1)\cap fn(P_2)$, is a FR strongly probabilistic hhp-bisimulation, we omit it.
\end{enumerate}
\end{proof}

\begin{theorem}[Expansion law for truly concurrent bisimilarities]
Let $P\equiv\boxplus_i\sum_j \alpha_{ij}.P_{ij}$ and $Q\equiv\boxplus_{k}\sum_l\beta_{kl}.Q_{kl}$, where $bn(\alpha_{ij})\cap fn(Q)=\emptyset$ for all $i,j$, and
  $bn(\beta_{kl})\cap fn(P)=\emptyset$ for all $k,l$. Then,

\begin{enumerate}
  \item $P\parallel Q\sim_{pp}^{fr} \boxplus_{i}\boxplus_{k}\sum_j\sum_l (\alpha_{ij}\parallel \beta_{kl}).(P_{ij}\parallel Q_{kl})+\boxplus_i\boxplus_k\sum_{\alpha_{ij} \textrm{ comp }\beta_{kl}}\tau.R_{ijkl}$;
  \item $P\parallel Q\sim_{ps}^{fr} \boxplus_{i}\boxplus_{k}\sum_j\sum_l (\alpha_{ij}\parallel \beta_{kl}).(P_{ij}\parallel Q_{kl})+\boxplus_i\boxplus_k\sum_{\alpha_{ij} \textrm{ comp }\beta_{kl}}\tau.R_{ijkl}$;
  \item $P\parallel Q\sim_{php}^{fr} \boxplus_{i}\boxplus_{k}\sum_j\sum_l (\alpha_{ij}\parallel \beta_{kl}).(P_{ij}\parallel Q_{kl})+\boxplus_i\boxplus_k\sum_{\alpha_{ij} \textrm{ comp }\beta_{kl}}\tau.R_{ijkl}$;
  \item $P\parallel Q\nsim_{phhp} \boxplus_{i}\boxplus_{k}\sum_j\sum_l (\alpha_{ij}\parallel \beta_{kl}).(P_{ij}\parallel Q_{kl})+\boxplus_i\boxplus_k\sum_{\alpha_{ij} \textrm{ comp }\beta_{kl}}\tau.R_{ijkl}$.
\end{enumerate}

Where $\alpha_{ij}$ comp $\beta_{kl}$ and $R_{ijkl}$ are defined as follows:
\begin{enumerate}
  \item $\alpha_{ij}$ is $\overline{x}u$ and $\beta_{kl}$ is $x(v)$, then $R_{ijkl}=P_{ij}\parallel Q_{kl}\{u/v\}$;
  \item $\alpha_{ij}$ is $\overline{x}(u)$ and $\beta_{kl}$ is $x(v)$, then $R_{ijkl}=(w)(P_{ij}\{w/u\}\parallel Q_{kl}\{w/v\})$, if $w\notin fn((u)P_{ij})\cup fn((v)Q_{kl})$;
  \item $\alpha_{ij}$ is $x(v)$ and $\beta_{kl}$ is $\overline{x}u$, then $R_{ijkl}=P_{ij}\{u/v\}\parallel Q_{kl}$;
  \item $\alpha_{ij}$ is $x(v)$ and $\beta_{kl}$ is $\overline{x}(u)$, then $R_{ijkl}=(w)(P_{ij}\{w/v\}\parallel Q_{kl}\{w/u\})$, if $w\notin fn((v)P_{ij})\cup fn((u)Q_{kl})$.
\end{enumerate}

Let $P\equiv\boxplus_i\sum_j P_{ij}.\alpha_{ij}[m]$ and $Q\equiv\boxplus_{k}\sum_l Q_{kl}.\beta_{kl}[m]$, where $bn(\alpha_{ij}[m])\cap fn(Q)=\emptyset$ for all $i,j$, and
  $bn(\beta_{kl}[m])\cap fn(P)=\emptyset$ for all $k,l$. Then,

\begin{enumerate}
  \item $P\parallel Q\sim_{pp}^{fr} \boxplus_{i}\boxplus_{k}\sum_j\sum_l(P_{ij}\parallel Q_{kl}).(\alpha_{ij}[m]\parallel \beta_{kl}[m])+\boxplus_i\boxplus_k\sum_{\alpha_{ij} \textrm{ comp }\beta_{kl}} R_{ijkl}.\tau$;
  \item $P\parallel Q\sim_{ps}^{fr} \boxplus_{i}\boxplus_{k}\sum_j\sum_l(P_{ij}\parallel Q_{kl}).(\alpha_{ij}[m]\parallel \beta_{kl}[m])+\boxplus_i\boxplus_k\sum_{\alpha_{ij} \textrm{ comp }\beta_{kl}} R_{ijkl}.\tau$;
  \item $P\parallel Q\sim_{php}^{fr} \boxplus_{i}\boxplus_{k}\sum_j\sum_l(P_{ij}\parallel Q_{kl}).(\alpha_{ij}[m]\parallel \beta_{kl}[m])+\boxplus_i\boxplus_k\sum_{\alpha_{ij} \textrm{ comp }\beta_{kl}} R_{ijkl}.\tau$;
  \item $P\parallel Q\nsim_{phhp} \boxplus_{i}\boxplus_{k}\sum_j\sum_l(P_{ij}\parallel Q_{kl}).(\alpha_{ij}[m]\parallel \beta_{kl}[m])+\boxplus_i\boxplus_k\sum_{\alpha_{ij} \textrm{ comp }\beta_{kl}} R_{ijkl}.\tau$.
\end{enumerate}

Where $\alpha_{ij}[m]$ comp $\beta_{kl}[m]$ and $R_{ijkl}$ are defined as follows:
\begin{enumerate}
  \item $\alpha_{ij}[m]$ is $\overline{x}u$ and $\beta_{kl}[m]$ is $x(v)$, then $R_{ijkl}=P_{ij}\parallel Q_{kl}\{u/v\}$;
  \item $\alpha_{ij}[m]$ is $\overline{x}(u)$ and $\beta_{kl}[m]$ is $x(v)$, then $R_{ijkl}=(w)(P_{ij}\{w/u\}\parallel Q_{kl}\{w/v\})$, if $w\notin fn((u)P_{ij})\cup fn((v)Q_{kl})$;
  \item $\alpha_{ij}[m]$ is $x(v)$ and $\beta_{kl}[m]$ is $\overline{x}u$, then $R_{ijkl}=P_{ij}\{u/v\}\parallel Q_{kl}$;
  \item $\alpha_{ij}[m]$ is $x(v)$ and $\beta_{kl}[m]$ is $\overline{x}(u)$, then $R_{ijkl}=(w)(P_{ij}\{w/v\}\parallel Q_{kl}\{w/u\})$, if $w\notin fn((v)P_{ij})\cup fn((u)Q_{kl})$.
\end{enumerate}
\end{theorem}

\begin{proof}
According to the definition of FR strongly probabilistic truly concurrent bisimulations, we can easily prove the above equations, and we omit the proof.
\end{proof}

\begin{theorem}[Equivalence and congruence for FR strongly probabilistic pomset bisimilarity]
We can enjoy the full congruence modulo FR strongly probabilistic pomset bisimilarity.

\begin{enumerate}
  \item $\sim_{pp}^{fr}$ is an equivalence relation;
  \item If $P\sim_{pp}^{fr} Q$ then
  \begin{enumerate}
    \item $\alpha.P\sim_{pp}^{f} \alpha.Q$, $\alpha$ is a free action;
    \item $P.\alpha[m]\sim_{pp}^{r}Q.\alpha[m]$, $\alpha[m]$ is a free action;
    \item $\phi.P\sim_{pp}^{f} \phi.Q$;
    \item $P.\phi\sim_{pp}^{r}Q.\phi$;
    \item $P+R\sim_{pp}^{fr} Q+R$;
    \item $P\boxplus_{\pi} R\sim_{pp}^{fr}Q\boxplus_{\pi}R$;
    \item $P\parallel R\sim_{pp}^{fr} Q\parallel R$;
    \item $(w)P\sim_{pp}^{fr} (w)Q$;
    \item $x(y).P\sim_{pp}^{f} x(y).Q$;
    \item $P.x(y)[m]\sim_{pp}^{r}Q.x(y)[m]$.
  \end{enumerate}
\end{enumerate}
\end{theorem}

\begin{proof}
\begin{enumerate}
  \item $\sim_{pp}^{fr}$ is an equivalence relation, it is obvious;
  \item If $P\sim_{pp}^{fr} Q$ then
  \begin{enumerate}
    \item $\alpha.P\sim_{pp}^{f} \alpha.Q$, $\alpha$ is a free action. It is sufficient to prove the relation $R=\{(\alpha.P, \alpha.Q)\}\cup \textbf{Id}$ is a F strongly probabilistic pomset bisimulation, we omit it;
    \item $P.\alpha[m]\sim_{pp}^{r}Q.\alpha[m]$, $\alpha[m]$ is a free action. It is sufficient to prove the relation $R=\{(P.\alpha[m], Q.\alpha[m])\}\cup \textbf{Id}$ is a R strongly probabilistic pomset bisimulation, we omit it;
    \item $\phi.P\sim_{pp}^{f} \phi.Q$. It is sufficient to prove the relation $R=\{(\phi.P, \phi.Q)\}\cup \textbf{Id}$ is a F strongly probabilistic pomset bisimulation, we omit it;
    \item $P.\phi\sim_{pp}^{r}Q.\phi$. It is sufficient to prove the relation $R=\{(P.\phi, Q.\phi)\}\cup \textbf{Id}$ is a R strongly probabilistic pomset bisimulation, we omit it;
    \item $P+R\sim_{pp}^{fr} Q+R$. It is sufficient to prove the relation $R=\{(P+R, Q+R)\}\cup \textbf{Id}$ is a FR strongly probabilistic pomset bisimulation, we omit it;
    \item $P\boxplus_{\pi} R\sim_{pp}^{fr}Q\boxplus_{\pi}R$. It is sufficient to prove the relation $R=\{(P\boxplus_{\pi} R, Q\boxplus_{\pi} R)\}\cup \textbf{Id}$ is a FR strongly probabilistic pomset bisimulation, we omit it;
    \item $P\parallel R\sim_{pp}^{fr} Q\parallel R$. It is sufficient to prove the relation $R=\{(P\parallel R, Q\parallel R)\}\cup \textbf{Id}$ is a FR strongly probabilistic pomset bisimulation, we omit it;
    \item $(w)P\sim_{pp}^{fr} (w)Q$. It is sufficient to prove the relation $R=\{((w)P, (w)Q)\}\cup \textbf{Id}$ is a FR strongly probabilistic pomset bisimulation, we omit it;
    \item $x(y).P\sim_{pp}^{f} x(y).Q$. It is sufficient to prove the relation $R=\{(x(y).P, x(y).Q)\}\cup \textbf{Id}$ is a F strongly probabilistic pomset bisimulation, we omit it;
    \item $P.x(y)[m]\sim_{pp}^{r}Q.x(y)[m]$. It is sufficient to prove the relation $R=\{(P.x(y)[m], Q.x(y)[m])\}\cup \textbf{Id}$ is a R strongly probabilistic pomset bisimulation, we omit it.
  \end{enumerate}
\end{enumerate}
\end{proof}

\begin{theorem}[Equivalence and congruence for FR strongly probabilistic step bisimilarity]
We can enjoy the full congruence modulo FR strongly probabilistic step bisimilarity.

\begin{enumerate}
  \item $\sim_{ps}^{fr}$ is an equivalence relation;
  \item If $P\sim_{ps}^{fr} Q$ then
  \begin{enumerate}
    \item $\alpha.P\sim_{ps}^{f} \alpha.Q$, $\alpha$ is a free action;
    \item $P.\alpha[m]\sim_{ps}^{r}Q.\alpha[m]$, $\alpha[m]$ is a free action;
    \item $\phi.P\sim_{ps}^{f} \phi.Q$;
    \item $P.\phi\sim_{ps}^{r}Q.\phi$;
    \item $P+R\sim_{ps}^{fr} Q+R$;
    \item $P\boxplus_{\pi} R\sim_{ps}^{fr}Q\boxplus_{\pi}R$;
    \item $P\parallel R\sim_{ps}^{fr} Q\parallel R$;
    \item $(w)P\sim_{ps}^{fr} (w)Q$;
    \item $x(y).P\sim_{ps}^{f} x(y).Q$;
    \item $P.x(y)[m]\sim_{ps}^{r}Q.x(y)[m]$.
  \end{enumerate}
\end{enumerate}
\end{theorem}

\begin{proof}
\begin{enumerate}
  \item $\sim_{ps}^{fr}$ is an equivalence relation, it is obvious;
  \item If $P\sim_{ps}^{fr} Q$ then
  \begin{enumerate}
    \item $\alpha.P\sim_{ps}^{f} \alpha.Q$, $\alpha$ is a free action. It is sufficient to prove the relation $R=\{(\alpha.P, \alpha.Q)\}\cup \textbf{Id}$ is a F strongly probabilistic step bisimulation, we omit it;
    \item $P.\alpha[m]\sim_{ps}^{r}Q.\alpha[m]$, $\alpha[m]$ is a free action. It is sufficient to prove the relation $R=\{(P.\alpha[m], Q.\alpha[m])\}\cup \textbf{Id}$ is a R strongly probabilistic step bisimulation, we omit it;
    \item $\phi.P\sim_{ps}^{f} \phi.Q$. It is sufficient to prove the relation $R=\{(\phi.P, \phi.Q)\}\cup \textbf{Id}$ is a F strongly probabilistic step bisimulation, we omit it;
    \item $P.\phi\sim_{ps}^{r}Q.\phi$. It is sufficient to prove the relation $R=\{(P.\phi, Q.\phi)\}\cup \textbf{Id}$ is a R strongly probabilistic step bisimulation, we omit it;
    \item $P+R\sim_{ps}^{fr} Q+R$. It is sufficient to prove the relation $R=\{(P+R, Q+R)\}\cup \textbf{Id}$ is a FR strongly probabilistic step bisimulation, we omit it;
    \item $P\boxplus_{\pi} R\sim_{ps}^{fr}Q\boxplus_{\pi}R$. It is sufficient to prove the relation $R=\{(P\boxplus_{\pi} R, Q\boxplus_{\pi} R)\}\cup \textbf{Id}$ is a FR strongly probabilistic step bisimulation, we omit it;
    \item $P\parallel R\sim_{ps}^{fr} Q\parallel R$. It is sufficient to prove the relation $R=\{(P\parallel R, Q\parallel R)\}\cup \textbf{Id}$ is a FR strongly probabilistic step bisimulation, we omit it;
    \item $(w)P\sim_{ps}^{fr} (w)Q$. It is sufficient to prove the relation $R=\{((w)P, (w)Q)\}\cup \textbf{Id}$ is a FR strongly probabilistic step bisimulation, we omit it;
    \item $x(y).P\sim_{ps}^{f} x(y).Q$. It is sufficient to prove the relation $R=\{(x(y).P, x(y).Q)\}\cup \textbf{Id}$ is a F strongly probabilistic step bisimulation, we omit it;
    \item $P.x(y)[m]\sim_{ps}^{r}Q.x(y)[m]$. It is sufficient to prove the relation $R=\{(P.x(y)[m], Q.x(y)[m])\}\cup \textbf{Id}$ is a R strongly probabilistic step bisimulation, we omit it.
  \end{enumerate}
\end{enumerate}
\end{proof}

\begin{theorem}[Equivalence and congruence for FR strongly probabilistic hp-bisimilarity]
We can enjoy the full congruence modulo FR strongly probabilistic hp-bisimilarity.

\begin{enumerate}
  \item $\sim_{php}^{fr}$ is an equivalence relation;
  \item If $P\sim_{php}^{fr} Q$ then
  \begin{enumerate}
    \item $\alpha.P\sim_{php}^{f} \alpha.Q$, $\alpha$ is a free action;
    \item $P.\alpha[m]\sim_{php}^{r}Q.\alpha[m]$, $\alpha[m]$ is a free action;
    \item $\phi.P\sim_{php}^{f} \phi.Q$;
    \item $P.\phi\sim_{php}^{r}Q.\phi$;
    \item $P+R\sim_{php}^{fr} Q+R$;
    \item $P\boxplus_{\pi} R\sim_{php}^{fr}Q\boxplus_{\pi}R$;
    \item $P\parallel R\sim_{php}^{fr} Q\parallel R$;
    \item $(w)P\sim_{php}^{fr} (w)Q$;
    \item $x(y).P\sim_{php}^{f} x(y).Q$;
    \item $P.x(y)[m]\sim_{php}^{r}Q.x(y)[m]$.
  \end{enumerate}
\end{enumerate}
\end{theorem}

\begin{proof}
\begin{enumerate}
  \item $\sim_{php}^{fr}$ is an equivalence relation, it is obvious;
  \item If $P\sim_{php}^{fr} Q$ then
  \begin{enumerate}
    \item $\alpha.P\sim_{php}^{f} \alpha.Q$, $\alpha$ is a free action. It is sufficient to prove the relation $R=\{(\alpha.P, \alpha.Q)\}\cup \textbf{Id}$ is a F strongly probabilistic hp-bisimulation, we omit it;
    \item $P.\alpha[m]\sim_{php}^{r}Q.\alpha[m]$, $\alpha[m]$ is a free action. It is sufficient to prove the relation $R=\{(P.\alpha[m], Q.\alpha[m])\}\cup \textbf{Id}$ is a R strongly probabilistic hp-bisimulation, we omit it;
    \item $\phi.P\sim_{php}^{f} \phi.Q$. It is sufficient to prove the relation $R=\{(\phi.P, \phi.Q)\}\cup \textbf{Id}$ is a F strongly probabilistic hp-bisimulation, we omit it;
    \item $P.\phi\sim_{php}^{r}Q.\phi$. It is sufficient to prove the relation $R=\{(P.\phi, Q.\phi)\}\cup \textbf{Id}$ is a R strongly probabilistic hp-bisimulation, we omit it;
    \item $P+R\sim_{php}^{fr} Q+R$. It is sufficient to prove the relation $R=\{(P+R, Q+R)\}\cup \textbf{Id}$ is a FR strongly probabilistic hp-bisimulation, we omit it;
    \item $P\boxplus_{\pi} R\sim_{php}^{fr}Q\boxplus_{\pi}R$. It is sufficient to prove the relation $R=\{(P\boxplus_{\pi} R, Q\boxplus_{\pi} R)\}\cup \textbf{Id}$ is a FR strongly probabilistic hp-bisimulation, we omit it;
    \item $P\parallel R\sim_{php}^{fr} Q\parallel R$. It is sufficient to prove the relation $R=\{(P\parallel R, Q\parallel R)\}\cup \textbf{Id}$ is a FR strongly probabilistic hp-bisimulation, we omit it;
    \item $(w)P\sim_{php}^{fr} (w)Q$. It is sufficient to prove the relation $R=\{((w)P, (w)Q)\}\cup \textbf{Id}$ is a FR strongly probabilistic hp-bisimulation, we omit it;
    \item $x(y).P\sim_{php}^{f} x(y).Q$. It is sufficient to prove the relation $R=\{(x(y).P, x(y).Q)\}\cup \textbf{Id}$ is a F strongly probabilistic hp-bisimulation, we omit it;
    \item $P.x(y)[m]\sim_{php}^{r}Q.x(y)[m]$. It is sufficient to prove the relation $R=\{(P.x(y)[m], Q.x(y)[m])\}\cup \textbf{Id}$ is a R strongly probabilistic hp-bisimulation, we omit it.
  \end{enumerate}
\end{enumerate}
\end{proof}

\begin{theorem}[Equivalence and congruence for FR strongly probabilistic hhp-bisimilarity]
We can enjoy the full congruence modulo FR strongly probabilistic hhp-bisimilarity.

\begin{enumerate}
  \item $\sim_{phhp}^{fr}$ is an equivalence relation;
  \item If $P\sim_{phhp}^{fr} Q$ then
  \begin{enumerate}
    \item $\alpha.P\sim_{phhp}^{f} \alpha.Q$, $\alpha$ is a free action;
    \item $P.\alpha[m]\sim_{phhp}^{r}Q.\alpha[m]$, $\alpha[m]$ is a free action;
    \item $\phi.P\sim_{phhp}^{f} \phi.Q$;
    \item $P.\phi\sim_{phhp}^{r}Q.\phi$;
    \item $P+R\sim_{phhp}^{fr} Q+R$;
    \item $P\boxplus_{\pi} R\sim_{phhp}^{fr}Q\boxplus_{\pi}R$;
    \item $P\parallel R\sim_{phhp}^{fr} Q\parallel R$;
    \item $(w)P\sim_{phhp}^{fr} (w)Q$;
    \item $x(y).P\sim_{phhp}^{f} x(y).Q$;
    \item $P.x(y)[m]\sim_{phhp}^{r}Q.x(y)[m]$.
  \end{enumerate}
\end{enumerate}
\end{theorem}

\begin{proof}
\begin{enumerate}
  \item $\sim_{phhp}^{fr}$ is an equivalence relation, it is obvious;
  \item If $P\sim_{phhp}^{fr} Q$ then
  \begin{enumerate}
    \item $\alpha.P\sim_{phhp}^{f} \alpha.Q$, $\alpha$ is a free action. It is sufficient to prove the relation $R=\{(\alpha.P, \alpha.Q)\}\cup \textbf{Id}$ is a F strongly probabilistic hhp-bisimulation, we omit it;
    \item $P.\alpha[m]\sim_{phhp}^{r}Q.\alpha[m]$, $\alpha[m]$ is a free action. It is sufficient to prove the relation $R=\{(P.\alpha[m], Q.\alpha[m])\}\cup \textbf{Id}$ is a R strongly probabilistic hhp-bisimulation, we omit it;
    \item $\phi.P\sim_{phhp}^{f} \phi.Q$. It is sufficient to prove the relation $R=\{(\phi.P, \phi.Q)\}\cup \textbf{Id}$ is a F strongly probabilistic hhp-bisimulation, we omit it;
    \item $P.\phi\sim_{phhp}^{r}Q.\phi$. It is sufficient to prove the relation $R=\{(P.\phi, Q.\phi)\}\cup \textbf{Id}$ is a R strongly probabilistic hhp-bisimulation, we omit it;
    \item $P+R\sim_{phhp}^{fr} Q+R$. It is sufficient to prove the relation $R=\{(P+R, Q+R)\}\cup \textbf{Id}$ is a FR strongly probabilistic hhp-bisimulation, we omit it;
    \item $P\boxplus_{\pi} R\sim_{phhp}^{fr}Q\boxplus_{\pi}R$. It is sufficient to prove the relation $R=\{(P\boxplus_{\pi} R, Q\boxplus_{\pi} R)\}\cup \textbf{Id}$ is a FR strongly probabilistic hhp-bisimulation, we omit it;
    \item $P\parallel R\sim_{phhp}^{fr} Q\parallel R$. It is sufficient to prove the relation $R=\{(P\parallel R, Q\parallel R)\}\cup \textbf{Id}$ is a FR strongly probabilistic hhp-bisimulation, we omit it;
    \item $(w)P\sim_{phhp}^{fr} (w)Q$. It is sufficient to prove the relation $R=\{((w)P, (w)Q)\}\cup \textbf{Id}$ is a FR strongly probabilistic hhp-bisimulation, we omit it;
    \item $x(y).P\sim_{phhp}^{f} x(y).Q$. It is sufficient to prove the relation $R=\{(x(y).P, x(y).Q)\}\cup \textbf{Id}$ is a F strongly probabilistic hhp-bisimulation, we omit it;
    \item $P.x(y)[m]\sim_{phhp}^{r}Q.x(y)[m]$. It is sufficient to prove the relation $R=\{(P.x(y)[m], Q.x(y)[m])\}\cup \textbf{Id}$ is a R strongly probabilistic hhp-bisimulation, we omit it.
  \end{enumerate}
\end{enumerate}
\end{proof}

\subsubsection{Recursion}

\begin{definition}
Let $X$ have arity $n$, and let $\widetilde{x}=x_1,\cdots,x_n$ be distinct names, and $fn(P)\subseteq\{x_1,\cdots,x_n\}$. The replacement of $X(\widetilde{x})$ by $P$ in $E$, written
$E\{X(\widetilde{x}):=P\}$, means the result of replacing each subterm $X(\widetilde{y})$ in $E$ by $P\{\widetilde{y}/\widetilde{x}\}$.
\end{definition}

\begin{definition}
Let $E$ and $F$ be two process expressions containing only $X_1,\cdots,X_m$ with associated name sequences $\widetilde{x}_1,\cdots,\widetilde{x}_m$. Then,
\begin{enumerate}
  \item $E\sim_{pp}^{fr} F$ means $E(\widetilde{P})\sim_{pp}^{fr} F(\widetilde{P})$;
  \item $E\sim_{ps}^{fr} F$ means $E(\widetilde{P})\sim_{ps}^{fr} F(\widetilde{P})$;
  \item $E\sim_{php}^{fr} F$ means $E(\widetilde{P})\sim_{php}^{fr} F(\widetilde{P})$;
  \item $E\sim_{phhp}^{fr} F$ means $E(\widetilde{P})\sim_{phhp}^{fr} F(\widetilde{P})$;
\end{enumerate}

for all $\widetilde{P}$ such that $fn(P_i)\subseteq \widetilde{x}_i$ for each $i$.
\end{definition}

\begin{definition}
A term or identifier is weakly guarded in $P$ if it lies within some subterm $\alpha.Q$ or $Q.\alpha[m]$ or $(\alpha_1\parallel\cdots\parallel \alpha_n).Q$ or
$Q.(\alpha_1[m]\parallel\cdots\parallel \alpha_n[m])$ of $P$.
\end{definition}

\begin{theorem}
Assume that $\widetilde{E}$ and $\widetilde{F}$ are expressions containing only $X_i$ with $\widetilde{x}_i$, and $\widetilde{A}$ and $\widetilde{B}$ are identifiers with $A_i$, $B_i$. Then, for all $i$,
\begin{enumerate}
  \item $E_i\sim_{ps}^{fr} F_i$, $A_i(\widetilde{x}_i)\overset{\text{def}}{=}E_i(\widetilde{A})$, $B_i(\widetilde{x}_i)\overset{\text{def}}{=}F_i(\widetilde{B})$, then
  $A_i(\widetilde{x}_i)\sim_{ps}^{fr} B_i(\widetilde{x}_i)$;
  \item $E_i\sim_{pp}^{fr} F_i$, $A_i(\widetilde{x}_i)\overset{\text{def}}{=}E_i(\widetilde{A})$, $B_i(\widetilde{x}_i)\overset{\text{def}}{=}F_i(\widetilde{B})$, then
  $A_i(\widetilde{x}_i)\sim_{pp}^{fr} B_i(\widetilde{x}_i)$;
  \item $E_i\sim_{php}^{fr} F_i$, $A_i(\widetilde{x}_i)\overset{\text{def}}{=}E_i(\widetilde{A})$, $B_i(\widetilde{x}_i)\overset{\text{def}}{=}F_i(\widetilde{B})$, then
  $A_i(\widetilde{x}_i)\sim_{php}^{fr} B_i(\widetilde{x}_i)$;
  \item $E_i\sim_{phhp}^{fr} F_i$, $A_i(\widetilde{x}_i)\overset{\text{def}}{=}E_i(\widetilde{A})$, $B_i(\widetilde{x}_i)\overset{\text{def}}{=}F_i(\widetilde{B})$, then
  $A_i(\widetilde{x}_i)\sim_{phhp}^{fr} B_i(\widetilde{x}_i)$.
\end{enumerate}
\end{theorem}

\begin{proof}
\begin{enumerate}
  \item $E_i\sim_{ps}^{fr} F_i$, $A_i(\widetilde{x}_i)\overset{\text{def}}{=}E_i(\widetilde{A})$, $B_i(\widetilde{x}_i)\overset{\text{def}}{=}F_i(\widetilde{B})$, then
  $A_i(\widetilde{x}_i)\sim_{ps}^{fr} B_i(\widetilde{x}_i)$.

      We will consider the case $I=\{1\}$ with loss of generality, and show the following relation $R$ is a FR strongly probabilistic step bisimulation.

      $$R=\{(G(A),G(B)):G\textrm{ has only identifier }X\}.$$

      By choosing $G\equiv X(\widetilde{y})$, it follows that $A(\widetilde{y})\sim_{ps}^{fr} B(\widetilde{y})$. It is sufficient to prove the following:
      \begin{enumerate}
        \item If $\langle G(A),s\rangle\rightsquigarrow\xrightarrow{\{\alpha_1,\cdots,\alpha_n\}}\langle P',s'\rangle$, where $\alpha_i(1\leq i\leq n)$ is a free action or bound output action with
        $bn(\alpha_1)\cap\cdots\cap bn(\alpha_n)\cap n(G(A),G(B))=\emptyset$, then $\langle G(B),s\rangle\rightsquigarrow\xrightarrow{\{\alpha_1,\cdots,\alpha_n\}}\langle Q'',s''\rangle$ such that $P'\sim_{ps}^{fr} Q''$;
        \item If $\langle G(A),s\rangle\rightsquigarrow\xrightarrow{x(y)}\langle P',s'\rangle$ with $x\notin n(G(A),G(B))$, then $\langle G(B),s\rangle\rightsquigarrow\xrightarrow{x(y)}\langle Q'',s''\rangle$, such that for all $u$,
        $\langle P',s'\rangle\{u/y\}\sim_{ps}^{fr} \langle Q''\{u/y\},s''\rangle$;
        \item If $\langle G(A),s\rangle\rightsquigarrow\xtworightarrow{\{\alpha_1[m],\cdots,\alpha_n[m]\}}\langle P',s'\rangle$, where $\alpha_i[m](1\leq i\leq n)$ is a free action or bound output action with
        $bn(\alpha_1[m])\cap\cdots\cap bn(\alpha_n[m])\cap n(G(A),G(B))=\emptyset$, then $\langle G(B),s\rangle\rightsquigarrow\xtworightarrow{\{\alpha_1[m],\cdots,\alpha_n[m]\}}\langle Q'',s''\rangle$ such that $P'\sim_{ps}^{fr} Q''$;
        \item If $\langle G(A),s\rangle\rightsquigarrow\xtworightarrow{x(y)[m]}\langle P',s'\rangle$ with $x\notin n(G(A),G(B))$, then $\langle G(B),s\rangle\rightsquigarrow\xtworightarrow{x(y)[m]}\langle Q'',s''\rangle$, such that for all $u$,
        $P'\{u/y\}\sim_{ps}^{fr} Q''\{u/y\}$.
      \end{enumerate}

      To prove the above properties, it is sufficient to induct on the depth of inference and quite routine, we omit it.
  \item $E_i\sim_{pp}^{fr} F_i$, $A_i(\widetilde{x}_i)\overset{\text{def}}{=}E_i(\widetilde{A})$, $B_i(\widetilde{x}_i)\overset{\text{def}}{=}F_i(\widetilde{B})$, then
  $A_i(\widetilde{x}_i)\sim_{pp}^{fr} B_i(\widetilde{x}_i)$. It can be proven similarly to the above case.
  \item $E_i\sim_{php}^{fr} F_i$, $A_i(\widetilde{x}_i)\overset{\text{def}}{=}E_i(\widetilde{A})$, $B_i(\widetilde{x}_i)\overset{\text{def}}{=}F_i(\widetilde{B})$, then
  $A_i(\widetilde{x}_i)\sim_{php}^{fr} B_i(\widetilde{x}_i)$. It can be proven similarly to the above case.
  \item $E_i\sim_{phhp}^{fr} F_i$, $A_i(\widetilde{x}_i)\overset{\text{def}}{=}E_i(\widetilde{A})$, $B_i(\widetilde{x}_i)\overset{\text{def}}{=}F_i(\widetilde{B})$, then
  $A_i(\widetilde{x}_i)\sim_{phhp}^{fr} B_i(\widetilde{x}_i)$. It can be proven similarly to the above case.
\end{enumerate}
\end{proof}

\begin{theorem}[Unique solution of equations]
Assume $\widetilde{E}$ are expressions containing only $X_i$ with $\widetilde{x}_i$, and each $X_i$ is weakly guarded in each $E_j$. Assume that $\widetilde{P}$ and $\widetilde{Q}$ are
processes such that $fn(P_i)\subseteq \widetilde{x}_i$ and $fn(Q_i)\subseteq \widetilde{x}_i$. Then, for all $i$,
\begin{enumerate}
  \item if $P_i\sim_{pp}^{fr} E_i(\widetilde{P})$, $Q_i\sim_{pp}^{fr} E_i(\widetilde{Q})$, then $P_i\sim_{pp}^{fr} Q_i$;
  \item if $P_i\sim_{ps}^{fr} E_i(\widetilde{P})$, $Q_i\sim_{ps}^{fr} E_i(\widetilde{Q})$, then $P_i\sim_{ps}^{fr} Q_i$;
  \item if $P_i\sim_{php}^{fr} E_i(\widetilde{P})$, $Q_i\sim_{php}^{fr} E_i(\widetilde{Q})$, then $P_i\sim_{php}^{fr} Q_i$;
  \item if $P_i\sim_{phhp}^{fr} E_i(\widetilde{P})$, $Q_i\sim_{phhp}^{fr} E_i(\widetilde{Q})$, then $P_i\sim_{phhp}^{fr} Q_i$.
\end{enumerate}
\end{theorem}

\begin{proof}
\begin{enumerate}
  \item It is similar to the proof of unique solution of equations for FR strongly probabilistic pomset bisimulation in CTC, please refer to \cite{CTC2} for details, we omit it;
  \item It is similar to the proof of unique solution of equations for FR strongly probabilistic step bisimulation in CTC, please refer to \cite{CTC2} for details, we omit it;
  \item It is similar to the proof of unique solution of equations for FR strongly probabilistic hp-bisimulation in CTC, please refer to \cite{CTC2} for details, we omit it;
  \item It is similar to the proof of unique solution of equations for FR strongly probabilistic hhp-bisimulation in CTC, please refer to \cite{CTC2} for details, we omit it.
\end{enumerate}
\end{proof}

\subsection{Algebraic Theory}\label{a8}

\begin{definition}[STC]
The theory \textbf{STC} is consisted of the following axioms and inference rules:

\begin{enumerate}
  \item Alpha-conversion $\textbf{A}$.
  \[\textrm{if } P\equiv Q, \textrm{ then } P=Q\]
  \item Congruence $\textbf{C}$. If $P=Q$, then,
  \[\tau.P=\tau.Q\quad \overline{x}y.P=\overline{x}y.Q\quad P.\overline{x}y[m]=Q.\overline{x}y[m]\]
  \[P+R=Q+R\quad P\parallel R=Q\parallel R\]
  \[(x)P=(x)Q\quad x(y).P=x(y).Q\quad P.x(y)[m]=Q.x(y)[m]\]
  \item Summation $\textbf{S}$.
  \[\textbf{S0}\quad P+\textbf{nil}=P\]
  \[\textbf{S1}\quad P+P=P\]
  \[\textbf{S2}\quad P+Q=Q+P\]
  \[\textbf{S3}\quad P+(Q+R)=(P+Q)+R\]
  \item Box-Summation $\textbf(BS)$.
  \[\textbf{BS0}\quad P\boxplus_{\pi}\textbf{nil}= P\]
  \[\textbf{BS1}\quad P\boxplus_{\pi}P= P\]
  \[\textbf{BS2}\quad P\boxplus_{\pi} Q= Q\boxplus_{1-\pi} P\]
  \[\textbf{BS3}\quad P\boxplus_{\pi}(Q\boxplus_{\rho} R)= (P\boxplus_{\frac{\pi}{\pi+\rho-\pi\rho}}Q)\boxplus_{\pi+\rho-\pi\rho} R\]
  \item Restriction $\textbf{R}$.
  \[\textbf{R0}\quad (x)P=P\quad \textrm{ if }x\notin fn(P)\]
  \[\textbf{R1}\quad (x)(y)P=(y)(x)P\]
  \[\textbf{R2}\quad (x)(P+Q)=(x)P+(x)Q\]
  \[\textbf{R3}\quad (x)\alpha.P=\alpha.(x)P\quad \textrm{ if }x\notin n(\alpha)\]
  \[\textbf{R4}\quad (x)\alpha.P=\textbf{nil}\quad \textrm{ if }x\textrm{is the subject of }\alpha\]
  \item Expansion $\textbf{E}$.
  Let $P\equiv\boxplus_i\sum_j \alpha_{ij}.P_{ij}$ and $Q\equiv\boxplus_{k}\sum_l\beta_{kl}.Q_{kl}$, where $bn(\alpha_{ij})\cap fn(Q)=\emptyset$ for all $i,j$, and
  $bn(\beta_{kl})\cap fn(P)=\emptyset$ for all $k,l$. Then,

\begin{enumerate}
  \item $P\parallel Q\sim_{pp}^{fr} \boxplus_{i}\boxplus_{k}\sum_j\sum_l (\alpha_{ij}\parallel \beta_{kl}).(P_{ij}\parallel Q_{kl})+\boxplus_i\boxplus_k\sum_{\alpha_{ij} \textrm{ comp }\beta_{kl}}\tau.R_{ijkl}$;
  \item $P\parallel Q\sim_{ps}^{fr} \boxplus_{i}\boxplus_{k}\sum_j\sum_l (\alpha_{ij}\parallel \beta_{kl}).(P_{ij}\parallel Q_{kl})+\boxplus_i\boxplus_k\sum_{\alpha_{ij} \textrm{ comp }\beta_{kl}}\tau.R_{ijkl}$;
  \item $P\parallel Q\sim_{php}^{fr} \boxplus_{i}\boxplus_{k}\sum_j\sum_l (\alpha_{ij}\parallel \beta_{kl}).(P_{ij}\parallel Q_{kl})+\boxplus_i\boxplus_k\sum_{\alpha_{ij} \textrm{ comp }\beta_{kl}}\tau.R_{ijkl}$;
  \item $P\parallel Q\nsim_{phhp} \boxplus_{i}\boxplus_{k}\sum_j\sum_l (\alpha_{ij}\parallel \beta_{kl}).(P_{ij}\parallel Q_{kl})+\boxplus_i\boxplus_k\sum_{\alpha_{ij} \textrm{ comp }\beta_{kl}}\tau.R_{ijkl}$.
\end{enumerate}

Where $\alpha_{ij}$ comp $\beta_{kl}$ and $R_{ijkl}$ are defined as follows:
\begin{enumerate}
  \item $\alpha_{ij}$ is $\overline{x}u$ and $\beta_{kl}$ is $x(v)$, then $R_{ijkl}=P_{ij}\parallel Q_{kl}\{u/v\}$;
  \item $\alpha_{ij}$ is $\overline{x}(u)$ and $\beta_{kl}$ is $x(v)$, then $R_{ijkl}=(w)(P_{ij}\{w/u\}\parallel Q_{kl}\{w/v\})$, if $w\notin fn((u)P_{ij})\cup fn((v)Q_{kl})$;
  \item $\alpha_{ij}$ is $x(v)$ and $\beta_{kl}$ is $\overline{x}u$, then $R_{ijkl}=P_{ij}\{u/v\}\parallel Q_{kl}$;
  \item $\alpha_{ij}$ is $x(v)$ and $\beta_{kl}$ is $\overline{x}(u)$, then $R_{ijkl}=(w)(P_{ij}\{w/v\}\parallel Q_{kl}\{w/u\})$, if $w\notin fn((v)P_{ij})\cup fn((u)Q_{kl})$.
\end{enumerate}

Let $P\equiv\boxplus_i\sum_j P_{ij}.\alpha_{ij}[m]$ and $Q\equiv\boxplus_{k}\sum_l Q_{kl}.\beta_{kl}[m]$, where $bn(\alpha_{ij}[m])\cap fn(Q)=\emptyset$ for all $i,j$, and
  $bn(\beta_{kl}[m])\cap fn(P)=\emptyset$ for all $k,l$. Then,

\begin{enumerate}
  \item $P\parallel Q\sim_{pp}^{fr} \boxplus_{i}\boxplus_{k}\sum_j\sum_l(P_{ij}\parallel Q_{kl}).(\alpha_{ij}[m]\parallel \beta_{kl}[m])+\boxplus_i\boxplus_k\sum_{\alpha_{ij} \textrm{ comp }\beta_{kl}} R_{ijkl}.\tau$;
  \item $P\parallel Q\sim_{ps}^{fr} \boxplus_{i}\boxplus_{k}\sum_j\sum_l(P_{ij}\parallel Q_{kl}).(\alpha_{ij}[m]\parallel \beta_{kl}[m])+\boxplus_i\boxplus_k\sum_{\alpha_{ij} \textrm{ comp }\beta_{kl}} R_{ijkl}.\tau$;
  \item $P\parallel Q\sim_{php}^{fr} \boxplus_{i}\boxplus_{k}\sum_j\sum_l(P_{ij}\parallel Q_{kl}).(\alpha_{ij}[m]\parallel \beta_{kl}[m])+\boxplus_i\boxplus_k\sum_{\alpha_{ij} \textrm{ comp }\beta_{kl}} R_{ijkl}.\tau$;
  \item $P\parallel Q\nsim_{phhp} \boxplus_{i}\boxplus_{k}\sum_j\sum_l(P_{ij}\parallel Q_{kl}).(\alpha_{ij}[m]\parallel \beta_{kl}[m])+\boxplus_i\boxplus_k\sum_{\alpha_{ij} \textrm{ comp }\beta_{kl}} R_{ijkl}.\tau$.
\end{enumerate}

Where $\alpha_{ij}[m]$ comp $\beta_{kl}[m]$ and $R_{ijkl}$ are defined as follows:
\begin{enumerate}
  \item $\alpha_{ij}[m]$ is $\overline{x}u$ and $\beta_{kl}[m]$ is $x(v)$, then $R_{ijkl}=P_{ij}\parallel Q_{kl}\{u/v\}$;
  \item $\alpha_{ij}[m]$ is $\overline{x}(u)$ and $\beta_{kl}[m]$ is $x(v)$, then $R_{ijkl}=(w)(P_{ij}\{w/u\}\parallel Q_{kl}\{w/v\})$, if $w\notin fn((u)P_{ij})\cup fn((v)Q_{kl})$;
  \item $\alpha_{ij}[m]$ is $x(v)$ and $\beta_{kl}[m]$ is $\overline{x}u$, then $R_{ijkl}=P_{ij}\{u/v\}\parallel Q_{kl}$;
  \item $\alpha_{ij}[m]$ is $x(v)$ and $\beta_{kl}[m]$ is $\overline{x}(u)$, then $R_{ijkl}=(w)(P_{ij}\{w/v\}\parallel Q_{kl}\{w/u\})$, if $w\notin fn((v)P_{ij})\cup fn((u)Q_{kl})$.
\end{enumerate}
  \item Identifier $\textbf{I}$.
  \[\textrm{If }A(\widetilde{x})\overset{\text{def}}{=}P,\textrm{ then }A(\widetilde{y})= P\{\widetilde{y}/\widetilde{x}\}.\]
\end{enumerate}
\end{definition}

\begin{theorem}[Soundness]
If $\textbf{STC}\vdash P=Q$ then
\begin{enumerate}
  \item $P\sim_{pp}^{fr} Q$;
  \item $P\sim_{pp}^{fr} Q$;
  \item $P\sim_{php}^{fr} Q$;
  \item $P\sim_{phhp}^{fr} Q$.
\end{enumerate}
\end{theorem}

\begin{proof}
The soundness of these laws modulo strongly truly concurrent bisimilarities is already proven in Section \ref{s8}.
\end{proof}

\begin{definition}
The agent identifier $A$ is weakly guardedly defined if every agent identifier is weakly guarded in the right-hand side of the definition of $A$.
\end{definition}

\begin{definition}[Head normal form]
A Process $P$ is in head normal form if it is a sum of the prefixes:

$$P\equiv \boxplus_{i}\sum_j(\alpha_{ij1}\parallel\cdots\parallel\alpha_{ijn}).P_{ij}\quad P\equiv \boxplus_{i}\sum_jP_{ij}.(\alpha_{ij1}[m]\parallel\cdots\parallel\alpha_{ijn}[m])$$
\end{definition}

\begin{proposition}
If every agent identifier is weakly guardedly defined, then for any process $P$, there is a head normal form $H$ such that

$$\textbf{STC}\vdash P=H.$$
\end{proposition}

\begin{proof}
It is sufficient to induct on the structure of $P$ and quite obvious.
\end{proof}

\begin{theorem}[Completeness]
For all processes $P$ and $Q$,
\begin{enumerate}
  \item if $P\sim_{pp}^{fr} Q$, then $\textbf{STC}\vdash P=Q$;
  \item if $P\sim_{pp}^{fr} Q$, then $\textbf{STC}\vdash P=Q$;
  \item if $P\sim_{php}^{fr} Q$, then $\textbf{STC}\vdash P=Q$.
\end{enumerate}
\end{theorem}

\begin{proof}
\begin{enumerate}
  \item if $P\sim_{pp}^{fr} Q$, then $\textbf{STC}\vdash P=Q$.
  
  For the forward transition case.
  
Since $P$ and $Q$ all have head normal forms, let $P\equiv\boxplus_{j=1}^l\sum_{i=1}^k\alpha_{ji}.P_{ji}$ and $Q\equiv\boxplus_{j=1}^l\sum_{i=1}^k\beta_{ji}.Q_{ji}$. Then the depth of
$P$, denoted as $d(P)=0$, if $k=0$; $d(P)=1+max\{d(P_{ji})\}$ for $1\leq j,i\leq k$. The depth $d(Q)$ can be defined similarly.

It is sufficient to induct on $d=d(P)+d(Q)$. When $d=0$, $P\equiv\textbf{nil}$ and $Q\equiv\textbf{nil}$, $P=Q$, as desired.

Suppose $d>0$.

\begin{itemize}
  \item If $(\alpha_1\parallel\cdots\parallel\alpha_n).M$ with $\alpha_{ji}(1\leq j,i\leq n)$ free actions is a summand of $P$, then 
  $\langle P,s\rangle\rightsquigarrow\xrightarrow{\{\alpha_1,\cdots,\alpha_n\}}\langle M,s'\rangle$.
  Since $Q$ is in head normal form and has a summand $(\alpha_1\parallel\cdots\parallel\alpha_n).N$ such that $M\sim_{pp}^{fr} N$, by the induction hypothesis $\textbf{STC}\vdash M=N$,
  $\textbf{STC}\vdash (\alpha_1\parallel\cdots\parallel\alpha_n).M= (\alpha_1\parallel\cdots\parallel\alpha_n).N$;
  \item If $x(y).M$ is a summand of $P$, then for $z\notin n(P, Q)$, $\langle P,s\rangle\rightsquigarrow\xrightarrow{x(z)}\langle M',s'\rangle\equiv \langle M\{z/y\},s'\rangle$. Since $Q$ is in head normal form and has a summand
  $x(w).N$ such that for all $v$, $M'\{v/z\}\sim_{pp}^{fr} N'\{v/z\}$ where $N'\equiv N\{z/w\}$, by the induction hypothesis $\textbf{STC}\vdash M'\{v/z\}=N'\{v/z\}$, by the axioms
  $\textbf{C}$ and $\textbf{A}$, $\textbf{STC}\vdash x(y).M=x(w).N$;
  \item If $\overline{x}(y).M$ is a summand of $P$, then for $z\notin n(P,Q)$, $\langle P,s\rangle\rightsquigarrow\xrightarrow{\overline{x}(z)}\langle M',s'\rangle\equiv \langle M\{z/y\},s'\rangle$. Since $Q$ is in head normal form and
  has a summand $\overline{x}(w).N$ such that $M'\sim_{pp}^{fr} N'$ where $N'\equiv N\{z/w\}$, by the induction hypothesis $\textbf{STC}\vdash M'=N'$, by the axioms
  $\textbf{A}$ and $\textbf{C}$, $\textbf{STC}\vdash \overline{x}(y).M= \overline{x}(w).N$.
\end{itemize}

For the reverse transition case, it can be proven similarly, and we omit it.

  \item if $P\sim_{pp}^{fr} Q$, then $\textbf{STC}\vdash P=Q$. It can be proven similarly to the above case.
  \item if $P\sim_{php}^{fr} Q$, then $\textbf{STC}\vdash P=Q$. It can be proven similarly to the above case.
\end{enumerate}
\end{proof}

\end{document}